\documentclass{amsart}

\usepackage{slashed}
\usepackage{mathrsfs}
\usepackage{amsmath,amsthm}
\usepackage{amssymb}
\usepackage{MnSymbol}
\usepackage[margin=1.4in,dvips]{geometry}

\usepackage{tikz}

\usepackage{hyperref}

\allowdisplaybreaks
\makeindex

\newcommand{\NN}{\mathbb{N}}
\newcommand{\RR}{\mathbb{R}}
\newcommand{\ZZ}{\mathbb{Z}}

\DeclareMathOperator{\Sym}{Sym}

\newtheorem{definition}{Definition}[section]

\newtheorem{remark}{Remark}[section]
\newtheorem{lemma}{Lemma}[subsection]
\newtheorem{theorem}{Theorem}[section]
\newtheorem*{theorem*}{Theorem}
\newtheorem*{corollary*}{Corollary}
\newtheorem{proposition}{Proposition}[subsection]

\newtheorem{conjecture}{Conjecture}[section]

\newtheorem{corollary}{Corollary}[section]
\newtheorem{problem}{Open problem}

\title[Time-Periodic Einstein--Klein--Gordon Bifurcations Of Kerr]{Time-Periodic Einstein--Klein--Gordon Bifurcations Of Kerr}
\author{Otis Chodosh}
\author{Yakov Shlapentokh-Rothman}
\address{Princeton University, Department of Mathematics, Fine~Hall,~Washington~Road,~Princeton,~NJ~08544}
\date{\today}

\begin{document}

\renewcommand\indexname{Index of Notation}

\maketitle

\begin{abstract}We construct one-parameter families of solutions to the Einstein--Klein--Gordon equations bifurcating off the Kerr solution such that the underlying family of spacetimes are each an asymptotically flat, stationary, axisymmetric, black hole spacetime, and such that the corresponding scalar fields are non-zero and time-periodic. An immediate corollary is that for these Klein--Gordon masses, the Kerr family is not asymptotically stable as a solution to the Einstein--Klein--Gordon equations.
\end{abstract}

\section{Introduction}
 The sub-extremal Kerr family of spacetimes $\left(\mathcal{M},g_{a,M}\right)$ are a two-parameter family of Lorentzian manifolds which are solutions to the Einstein vacuum equations
\begin{equation}\label{eve}
Ric\left(g\right) = 0.
\end{equation}
The parameters $(a,M)$ are ``admissible'' if $|a| < M$; $M$ denotes the mass, and $a$ denotes the specific angular momentum. The so called ``exterior'' region of each these solutions is \emph{asymptotically flat}, \emph{stationary}, and is bounded by a \emph{non-degenerate event horizon}; in fact, the Kerr exteriors are expected to be the unique such solutions (see Conjecture~\ref{blackUnique} below), and thus they play a central role in physics via their representation of the possible end states of ``gravitational collapse'' (see the textbook~\cite{wald}).

 Surprisingly, despite the fundamental importance of the Kerr solution, some of the most basic questions have remained unanswered to this day. In particular, the question of the dynamics of small perturbations remains open. The following two conjectures represent two of the most fundamental open problems in classical general relativity.

\begin{conjecture}\label{blackStab} (Asymptotic Stability of the Kerr Family) The maximal Cauchy development of a small perturbation of sub-extremal Kerr initial data possesses a black hole and exterior region, and in the exterior region, the development remains close to the perturbed spacetime and asymptotically settles down to (a possibly different) Kerr exterior spacetime.
\end{conjecture}

\begin{conjecture}\label{blackUnique} (Uniqueness of the Kerr Family) Any sufficiently regular, asymptotically flat, and stationary solution to~(\ref{eve}) which is bounded by a non-degenerate event horizon is isometric to a Kerr exterior spacetime.
\end{conjecture}

While there has been little direct progress on Conjecture~\ref{blackStab},\footnote{There has however been a lot of progress concerning scalar model problems (see~\cite{waveKerr} and Section~\ref{linearscalar} below), and we do know that there \emph{exist} spacetimes which dynamically settle down to Kerr~\cite{bholescatter}.} Conjecture~\ref{blackUnique} is known to be true under various additional assumptions: if the spacetime is static~\cite{israel}, if there exists a suitable axisymmetric Killing vector field~\cite{carter3,robinson}, if the spacetime is assumed to be a small perturbation of a Kerr spacetime~\cite{aik,aik2,wongyu}, or if the spacetime is assumed to be real analytic~\cite{chrusciellopes,hawking}.\footnote{For the sake of brevity we have suppressed many important technical assumptions necessary for these results; we direct the interested reader to~\cite{reviewunique,costa} and the references therein for a very thorough discussion.}

Finally, it is important to note that a priori, as it would with any evolutionary PDE, an understanding of the long time behavior of solutions to~(\ref{eve}) also requires a classification of \emph{time-periodic} or other ``soliton'' like solutions. Here, heuristics suggest that due to the emission of ``gravitational waves'' all such solutions must in fact be stationary; in the time-periodic setting this issue has been recently addressed in the works~\cite{periodic,unique}.

Of course, one may also study the Einstein equations in the presence of matter:
\begin{equation}\label{einsteinmatter}
Ric\left(g\right) - \frac{1}{2}gR\left(g\right) = \mathbb{T}(g),
\end{equation}
where $R$ denotes the \emph{scalar curvature} of $g$ and $\mathbb{T}$ denotes the \emph{energy-momentum tensor} (which must be specified by the matter theory), and the equations~(\ref{einsteinmatter}) may require coupling to additional matter equations. It is of considerable mathematical and physical interest to understand the effect of matter on the validity of Conjectures~\ref{blackStab} and~\ref{blackUnique} (note that our universe is certainly not a vacuum!).

Scalar fields are essentially the simplest possible form of matter we can consider; a scalar field of mass $\mu^2 \geq 0$ is determined by a single function $\Psi : \mathcal{M} \to \mathbb{C}$ with energy-momentum tensor
\begin{equation}\label{emot}
\mathbb{T}_{\alpha\beta} \doteq \text{Re}\left(\partial_{\alpha}\Psi\overline{\partial_{\beta}\Psi}\right) - \frac{1}{2}g_{\alpha\beta}\left[g^{\gamma\delta}\text{Re}\left(\partial_{\gamma}\Psi\overline{\partial_{\delta}\Psi}\right) + \mu^2\left|\Psi\right|^2\right].
\end{equation}
Recalling that the twice contracted second Bianchi identity implies that the left hand side of~(\ref{einsteinmatter}) is always divergence free, we observe that any solution to~(\ref{einsteinmatter}) with $\mathbb{T}$ given by~(\ref{emot}) must satisfy
%\[\nabla^{\alpha}\mathbb{T}_{\alpha\beta} = 0 \Rightarrow \]
\begin{equation}\label{kg}
\Box_g\Psi - \mu^2\Psi = 0,
\end{equation}
i.e.~$\Psi$ must satisfy the Klein--Gordon equation. We thus refer to~(\ref{einsteinmatter}) with $\mathbb{T}$ given by~(\ref{emot}) as the Einstein--Klein--Gordon (EKG) equations.

In 1972, on the heels of the discovery of \emph{superradiant}\footnote{Cf.~the final two paragraphs of Section~\ref{linearscalar} below.} scattering~\cite{zeldovich}, Misner suggested the possibility of ``floating orbits'' where a massive ``particle'' orbiting a black hole produces an exact balance between energy extracted via superradiance and energy radiated to infinity. Soon after, Press and Teukolskly suggested that it was also possible for the energy extracted to dominate the energy radiated away and to thus produce an instability which they coined the ``Black-hole Bomb''~\cite{pressteuk}.

The instability can be naturally studied in the context of the linear massive Klein--Gordon equation~\eqref{kg} on a \underline{fixed} Kerr background, and following the pioneering heuristic and numerical works~\cite{ddr,det,ze}, the study of this (linear) instability reached a relatively refined state (see the recent book~\cite{super}); in particular, the instabilities were rigorously constructed in the work~\cite{shlapgrow}. However, the original question of Misner remained unresolved; do there exist time-periodic black hole solutions to the full Einstein--Klein--Gordon equations?\footnote{We note that the work~\cite{shlapgrow} produced exactly time-periodic solutions to the linear Klein--Gordon equation~\eqref{kg} on fixed Kerr backgrounds; however, it is not clear a priori if the necessary delicate balancing of energy extracted from the black hole and energy radiated to infinity can be maintained in the highly non-linear setting of Einstein--Klein--Gordon.}

Very recently, in a breakthrough numerical work~\cite{hairy}, Herdeiro and Radu have constructed the desired time-periodic black hole solutions to Einstein--Klein--Gordon. These solutions ``originate'' as bifurcations of the Kerr family, but the curve of solutions may be continued past the point where they are a ``small perturbation'' of Kerr. In fact, $1$-parameter families may be constructed which form a continuous bridge between Kerr black holes and boson stars (cf.~Section~\ref{uniqueblack} below). Following the paper~\cite{hairy}, there has been a series of works studying more refined properties of the time-periodic solutions and finding analogous constructions in different settings (see, e.g.,~\cite{hr,bhr,bchr,hrr,chrr}).

In this paper we initiate the mathematical study of these time-periodic black hole solutions to Einstein--Klein--Gordon and give a proof of their existence in a small neighborhood of the Kerr family. Our main result is the following:

\begin{theorem}\label{timeperiodicsoln} There exists Klein--Gordon masses $\mu^2 > 0$ such that there exists a $1$-parameter family of smooth spacetimes $(\mathcal{M},g_{\delta})$ and scalar fields $\Psi_{\delta} : \mathcal{M} \to \mathbb{C}$ indexed by $\delta \in [0,\epsilon)$ such that
\begin{enumerate}
    \item For each $\delta \geq 0$ the pair $(\mathcal{M},g_{\delta})$ and $\Psi_{\delta}$ yields a solution to the Einstein--Klein--Gordon equations with mass $\mu^2$.
    \item The spacetimes $(\mathcal{M},g_{\delta})$ are all stationary, axisymmetric, asymptotically flat, and posses a non-degenerate bifurcate event horizon.
    \item For $\delta > 0$ the scalar field $\Psi_{\delta}$ is non-zero, \underline{time-periodic}, and decays exponentially (in any asymptotically flat chart) along any asymptotically flat Cauchy hypersurface.
    \item The $1$-parameter family bifurcates off the Kerr family in the sense that $(\mathcal{M},g_0)$ is isometric to a sub-extremal Kerr exterior spacetime with $0<|a|<M$, the family is differentiable with respect to $\delta$ at $\delta = 0$, and $\lim_{\delta\to 0}\delta^{-1}\Psi_{\delta} = \hat{\Psi}$, where $\hat{\Psi}$ is a non-zero time-periodic solution to the Klein--Gordon equation~(\ref{kg}) on $(\mathcal{M},g_0)$.
\end{enumerate}
\end{theorem}

\begin{remark}
In fact, we produce solutions for a set of $\mu^{2}>0$ with positive Lebesgue measure. 
\end{remark}

\begin{remark}\label{lintononlin}As we have already remarked, the existence of time-periodic solutions to the linear Klein--Gordon equation on fixed Kerr spacetimes was previously shown in the work~\cite{shlapgrow}. Theorem~\ref{timeperiodicsoln} can be interpreted as showing that this ``linear hair'' can be integrated to yield ``nonlinear hair'' (while maintaining the relevant symmetries of underlying spacetime).
\end{remark}
\begin{remark} The statement that the family is differentiable with respect to $\delta$ at $\delta=0$ should be understood in the sense of a $1$-parameter family of sections of $\Sym^{2}T^{*}\mathcal{M}$. Alternatively, our construction yields coordinates in which the metric $g_{\delta}$ has coefficients (depending on $\delta$) which are differentiable at $\delta=0$. One can show that the $\delta$-dependence of the $1$-parameter family is much more regular than claimed; we restrict ourselves to the stated regularity class in this paper for the sake of the exposition (increasing the regularity does not involve any significantly new ideas).
\end{remark}

\begin{remark}
We note that all of the solutions constructed in Theorem \ref{timeperiodicsoln} are rotating. It is not possible for the Schwarzschild solution to bifurcate in this manner, even at the linear level; see \cite[Theorem 1.3]{shlapgrow}.
\end{remark}

An immediate consequence is the following.
\begin{corollary}There exist Klein--Gordon masses, a sub-extremal Kerr spacetime, and a small Einstein--Klein--Gordon perturbation such that the scalar field does \underline{not} decay to a stationary solution. In particular, as a family of solutions to the Einstein--Klein--Gordon equations, asymptotic stability does \underline{not} hold for the Kerr family.
\end{corollary}

This shows that the addition of even a relatively simple matter model may completely change the expectations regarding black hole stability and uniqueness.
\subsection{Background}
In this section we review background material on scalar fields, boson stars, and hairy black holes which helps to provide context for Theorem~\ref{timeperiodicsoln}.
\subsubsection{Linear Scalar Fields on Kerr exterior Backgrounds}\label{linearscalar}

The behavior of linear scalar fields on Kerr exterior backgrounds plays a fundamental role in the proof of Theorem~\ref{timeperiodicsoln}. In this section we will quickly review the relevant theory.

The previous decade has witnessed an intense period of study of boundedness and decay properties for the \emph{wave equation},
\begin{equation}\label{wave}
\Box_g\Psi = 0,
\end{equation}
i.e.~(\ref{kg}) with $\mu = 0$, on black hole hole spacetimes. The primary motivation for this is the connection between understanding dispersive properties of waves in a sufficiently robust fashion and Conjecture~\ref{blackStab} (cf.~the role of the wave equation in the proof of the stability of Minkowski space~\cite{ck},~\cite{cklinear}).

In joint work with Dafermos and Rodnianski, the second author has proved the following theorem.
\begin{theorem}[\cite{waveKerr}]\label{waveKerrT} For the full sub-extremal range $|a| < M$, finite energy solutions to the wave equation~(\ref{wave}) have uniformly bounded energy and satisfy an integrated energy decay statement.
\end{theorem}
Theorem~\ref{waveKerrT} is easily seen to imply that any finite energy solution to the wave equation decays to $0$ on any compact set; in particular, the natural analogue of the type of solution we construct in Theorem~\ref{timeperiodicsoln} \emph{cannot} exist when $\mu = 0$.

The proof of Theorem~\ref{waveKerrT} required in an essential way the previous work~\cite{realmodes} of the second author as well as the works~\cite{icmp} and~\cite{stabi} by Dafermos and Rodnianski and was preceded by the works~\cite{dr7},~\cite{anblue}, and~\cite{tattoh} which established an analogue of Theorem~\ref{waveKerrT} under the additional assumption that $|a| \ll M$ (where certain fundamental difficulties like \emph{superradiance} and \emph{trapping} are considerably simpler to handle). For more discussion of the history behind and motivation for Theorem~\ref{waveKerrT} we direct the reader to the introduction of~\cite{waveKerr}.

Before attempting to prove Theorem~\ref{timeperiodicsoln}, it is natural to first determine if Theorem~\ref{waveKerrT} is true for the Klein--Gordon equation. Somewhat surprisingly (given Theorem~\ref{waveKerrT}), the second author showed in the work~\cite{shlapgrow} that the answer is emphatically ``no.''

\begin{theorem}[\cite{shlapgrow}]\label{KGgrow} For every choice of parameters $(a,M)$ satisfying $0 < |a| < M$ there exists an open family of masses $\mu^2 \in (0,\infty)$ such that there exist solutions to the corresponding Klein--Gordon equation~(\ref{kg}), arising from localized initial data, which \underline{grow exponentially} in time. Furthermore, for every choice of parameters $(a,M)$ satisfying $0 < |a| < M$ there exists a countable sequence of masses $\mu^2$ such that there exist exactly \underline{time-periodic} solutions to the corresponding Klein--Gordon equation~(\ref{kg}).
\end{theorem}
The underlying mechanism for the exponential growth is that of \emph{superradiance}. This is a phenomenon by which a solution to~(\ref{kg}) or~(\ref{wave}) may extract energy from the black hole. Mathematically, the underlying reason this is possible is due to the lack of any globally timelike Killing vector field and the resulting loss of a positive definite conserved energy. For a more detailed discussion of superradiance we direct the reader to the introduction of~\cite{shlapgrow}. We will however take the chance to emphasize that, as shown in the recent work~\cite{scatter}, the phenomenon of superradiance also occurs for the wave equation even though in this case it does not lead to unbounded growth.\footnote{More precisely, it was shown that the energy of a solution to~(\ref{wave}) may strictly increase even though Theorem~\ref{waveKerrT} bounds the total amount of increase (cf.~the work~\cite{FKSY}).}

For this paper, it is the time-periodic solutions from Theorem~\ref{KGgrow} which are the most relevant. These solutions are precisely at the threshold of superradiance and have a vanishing energy flux along the event horizon (the technical ramifications of this is discussed further in Section \ref{subsubsec:compat-cond}). A key idea behind the construction of these solutions is to treat the Klein--Gordon mass as an eigenvalue. The proof of Theorem~\ref{timeperiodicsoln} will require us to revisit the construction of these time-periodic solutions for certain perturbations of the Kerr spacetime (see Section~\ref{secpsi}).
\subsubsection{Boson Stars and Hairy Black Holes}\label{uniqueblack}
The results of Israel, Carter, Hawking, and Robinson~\cite{israel,carter3,hawking,robinson} as well their corresponding extensions to the Einstein--Maxwell system~\cite{israel2,bunting,mazur} inspired an early belief in the so called ``generalized no hair conjecture.'' This informal conjecture stated that whenever the Einstein equations are coupled with any ``reasonable'' matter model, then the set of stationary and asymptotically flat solutions should be finite dimensional and continuously parameterized by mass, angular momentum, and various \emph{asymptotically defined} ``charges'' associated to the matter (see the discussion in~\cite{reviewunique}). As is well known, this conjecture turned out to be false. The most direct counter-examples consist of spherically symmetric infinite families of static \emph{globally regular} solutions to the Einstein--Yang--Mills equations, see~\cite{EYMnum,EYMprove}. Spherically symmetric infinite families of static \emph{black hole} solutions to the Einstein--Yang--Mills equations have also been constructed~\cite{EYMblack}.

More relevant to this paper, however, are so called \emph{boson stars}. These are solutions to the Einstein--Klein--Gordon equations where the metric is static, \emph{globally regular}, does not contain a black hole, and is asymptotically flat, and the scalar field is non-zero and time-periodic. Note that due to the scalar field not being stationary, boson stars are not, strictly speaking, counter-examples to the generalized no hair conjecture; however, they certainly violate the spirit of the conjecture. The first heuristic and numerical studies of spherically symmetric boson stars was carried out in 1968 and 1969 in the works~\cite{kaup,rufbon}, twenty years later, more involved heuristics, numerics, and a stability analysis were carried out in the works~\cite{flp,flp2,lp}, and finally, spherically symmetric boson stars were rigorously constructed in the work~\cite{bw}. We direct the reader to the review article~\cite{bosonreview} for a thorough discussion of the origins and role of boson stars in general relativity.

\subsection{Directions for Further Study}
In this section we will present some natural open problems and directions for further study which are suggested by Theorem~\ref{timeperiodicsoln} and~\cite{hairy}.

Let's first introduce some notation from~\cite{hairy} (see also \cite{rotbos}). Associated to the phase invariance of the Klein--Gordon Lagrangian is a conserved current:
\[j^{\alpha} \doteq 2\text{Im}\left(\overline{\psi}\partial^{\alpha}\psi\right).\]
Integrating this current along any Cauchy hypersurface yields the ``particle number'' $Q$. 

Let's now agree to specify our discussion to stationary and axisymmetric spacetimes of the form considered in this paper (see \ref{subsec:metric-ansatz}), and where the scalar field $\Psi$ takes the form
\begin{equation}\label{formscalar}
\Psi(t,\phi,\rho,z) = e^{-it\omega}e^{im\phi}\psi(\rho,z)
\end{equation}
with $\psi$ real. We will refer to $m$ as the ``azimuthal number'' and $\psi$ as the reduced scalar field. We will only be interested in the case when $m$ is a non-zero integer. In this case, the particle number takes the form (see \eqref{eq:metric-ansatz} below for the definition of $X$ and $W$)
\[Q = 4\pi\int_0^{\infty}\int_{-\infty}^{\infty}\left(\omega X - m W\right)\psi^2e^{2\lambda}\, d\rho\, dz.\]
We note that the ADM mass $M$, the angular momentum $J$, and the particle number $Q$ are the three natural conserved quantities associated to solutions of the Einstein--Klein--Gordon equations. Finally, we introduce the ratio
\begin{equation}\label{q}
q \doteq \frac{mQ}{J}.
\end{equation}
A straightforward and well-known calculation shows that rotating boson stars always satisfy $q = 1$~\cite{rotbos}.

Our first problem concerns extending the $1$-parameter family we found in Theorem~\ref{timeperiodicsoln}
\begin{problem}\label{qto1} Keeping in mind that the solutions constructed in Theorem~\ref{timeperiodicsoln} all satisfy $|q| \ll 1$, construct a $1$-parameter family of solutions to the Einstein--Klein--Gordon equations such that $q$ takes every value in $(0,1)$. Analyze the limit as $q \to 1$.
\end{problem}
\begin{remark} According to the numerical work~\cite{hairy}, the $q\to 1$ limit of the solutions should converge to a globally regular rotating boson star (in a suitable topology which allows in particular for the black hole to disappear in the limit). This would provide the first rigorous construction of a non-spherically symmetric boson star. Finally, based on \cite[Figure 3]{hairy}, we note that along such a $1$-parameter family, one expects to find solutions which violate the angular-momentum mass inequality $|J| \leq M^2$. See~\cite{dain1,dain2,schoenzhou} for proofs of the angular-momentum mass inequality for axisymmetric vacuum initial data sets.
\end{remark}

Next we consider the construction of ``excited states.''
\begin{problem}\label{excited} Construct solutions to the Einstein--Klein--Gordon equations for an arbitrary azimuthal number $m \in \mathbb{Z}_{\neq 0}$ and such that the reduced scalar field $\psi$ has an arbitrarily large number of nodal domains\footnote{Recall that the number of nodal domains of $\psi$ is the number of connected components of $\psi^{-1}(0)$. The solutions constructed in Theorem \ref{timeperiodicsoln} have a single nodal domain, corresponding to the fact that $\psi$ is constructed to be the \emph{first} eigenfunction of a suitable elliptic equation.}.
\end{problem}
\begin{remark}
Our proof of Theorem~\ref{timeperiodicsoln} works for any sufficiently large azimuthal number $m$ and always yields a reduced scalar field $\psi$ which is positive. However, we strongly believe that suitable modifications of the techniques used here would allow for $m$ to be any non-zero integer. The variational structure behind the construction of $\psi$ (see Section~\ref{secpsi}) should also naturally lead to the construction of $\psi$ with an arbitrarily large number of nodal domains.
\end{remark}

Now we turn the problem of uniqueness.
\begin{problem}\label{unique}Determine to what extent are solutions to the Einstein--Klein--Gordon equation with a stationary, axisymmetric spacetime and scalar field of the form~(\ref{formscalar}) determined by the values of the ADM mass $M$, the angular momentum $J$, the particle number $Q$, the azimuthal number $m$, and the number of nodal domains $n$ of $\psi$.
\end{problem}
\begin{remark}
The results of~\cite{hairy} suggest that at least within the class of solutions with $n = 1$, the values of $M$, $J$, $Q$, and $m$ do indeed determine the solution uniquely.
\end{remark}

Finally, we expect analogues of Theorem~\ref{timeperiodicsoln} to hold in many other settings.
\begin{problem}\label{analogue}Do analogues of Theorem~\ref{timeperiodicsoln} hold for extremal Kerr, Kerr-anti-de Sitter, and Kerr-de Sitter spacetimes?
\end{problem}
\begin{remark}
Of course, any analogue of Theorem~\ref{timeperiodicsoln} should be consistent with already established linear scalar field stability and instability results for these spacetimes, see e.g.~\cite{aretakisKerr,aretakisHor,holz-smul,holz-smul2,dyatlov1,dyatlov2,dyatlov-last,vasy}.
\end{remark}

\subsection{Acknowledgements}
OC was supported by an EPSRC Programme Grant entitled ÔSingularities of Geometric Partial Differential Equations,Õ number EP/K00865X/1 during part of the time this work was completed and is grateful to Simon Brendle for his encouragement concerning this work. YS acknowledges support from the NSF Postdoctoral Research Fellowship under award no.~1502569, and thanks Igor Rodnianski and Mihalis Dafermos for stimulating conversations about the paper. Finally, we are grateful to the referees whose many suggestions greatly improved the exposition and organization of the paper.

\section{Overview of the Proof}\label{sec:overview}
We now turn to a high level overview of the proof of Theorem~\ref{timeperiodicsoln}.
\subsection{Stationary Carter--Robinson Theory}
A fundamental role will be played by the structure behind Carter--Robinson's proof~\cite{carter3,robinson} of the uniqueness of Kerr as a stationary and axisymmetric solution to the Einstein vacuum equations. In this section we will review the salient points in their argument. Our discussion will be closely modeled on the approach of Weinstein from~\cite{weinstein,weinstein2}.

We first assume that the metric takes the form
\begin{equation}\label{astataximetric}
g = -Vdt^2 + 2Wdtd\phi + Xd\phi^2  + e^{2\lambda}\left(d\rho^2+dz^2\right),
\end{equation}
where $T \doteq \partial_t$ and $\Phi \doteq \partial_{\phi}$ are Killing vector fields. One can express the Einstein vacuum equations $Ric(g) = 0$ directly in terms of the metric functions $V$, $W$, $X$, and $e^{2\lambda}$, but the resulting system of quasilinear elliptic equations is quite complicated and does not readily admit a classification of solutions.

An important insight is that it instead pays to introduce a further unknown, the twist $1$-form $\theta$, defined by
\begin{equation}\label{twistthe1form}
\theta \doteq 2i_{\Phi}\left(*\nabla\Phi_{\flat}\right).
\end{equation}
Geometrically, $\theta$ is the obstruction to $\Phi$ being hypersurface orthogonal. It turns out that $Ric(g) = 0$ implies that
\[d\theta = 0,\]
and thus we can introduce a function $Y$ which vanishes at infinity and satisfies
\[dY = \theta.\]

The crucial and somewhat surprising fact is that $X$ and $Y$ satisfy an equation which decouples from the other metric components:
\begin{align}\label{harmonicmap1}
\rho^{-1}\partial_{\rho}\left(\rho\partial_{\rho}X\right) + \rho^{-1}\partial_z\left(\rho\partial_zX\right) &= \frac{\left(\partial_{\rho}X\right)^2 + \left(\partial_zX\right)^2 - \left(\partial_{\rho}Y\right)^2 - \left(\partial_zY\right)^2}{X},
\\ \label{harmonicmap2} \rho^{-1}\partial_{\rho}\left(\rho\partial_{\rho}Y\right) + \rho^{-1}\partial_z\left(\rho\partial_zY\right) &= \frac{2\left(\partial_{\rho}Y\right)\left(\partial_{\rho}X\right) + 2\left(\partial_zY\right)\left(\partial_zX\right)}{X}.
\end{align}
In fact, if $(\rho,z)$ are considered as being the $\rho$ and $z$ of $\mathbb{R}^3$ expressed in cylindrical coordinates, and $X\left(\rho,z\right)$ and $Y\left(\rho,z\right)$ are thus interpreted as axisymmetric functions on $\mathbb{R}^3$, then the above equations simply state that $(X,Y)$ forms a harmonic map from $\mathbb{R}^3$ to hyperbolic space $\mathbb{H}^2$. (We note that the decoupling of the equation for $(X,Y)$ from the other components is related to the fact that harmonic maps are conformally invariant in dimension $2$.) The requirements of asymptotic flatness and regular extensions to the axis and event horizon lead to natural boundary conditions for $X$ and $Y$.

The $2$-parameter Kerr family yields a $2$-parameter family of solutions to~\eqref{harmonicmap1} and~\eqref{harmonicmap2}. Relatively standard techniques allow one to show that these must in fact be the  unique $2$-parameter family of solutions to~\eqref{harmonicmap1} and~\eqref{harmonicmap2} and that this family is parametrized naturally by the boundary conditions of $X$ and $Y$. 

It remains to show that given a particular solution $(X,Y)$ to~\eqref{harmonicmap1} and~\eqref{harmonicmap2}, the rest of the metric coefficients are uniquely determined. For this we start with the non-obvious observation that
\[\sigma \doteq \sqrt{XV + W^2},\]
must be harmonic when considered as a function of $\rho$ and $z$:
\[\Delta_{\mathbb{R}^2}\sigma = 0.\]
The natural boundary conditions for $\sigma$ then force that $\sigma = \rho$. 

Next, the definition of the twist $\theta$ naturally leads to the following equation for $W$:
\[\partial_{\rho}\left(X^{-1}W\right)d\rho + \partial_z\left(X^{-1}W\right)dz = \frac{\rho}{X^2}\left(\left(\partial_{\rho}Y\right)dz - \left(\partial_zY\right)d\rho\right).\]
The necessary compatibility conditions to solve this equation is exactly~\eqref{harmonicmap2}. It then turns out that this equation and the natural boundary conditions uniquely determine $W$ in terms of $X$ and $Y$. 

Finally, for $\lambda$, one can derive the following equations
\begin{align*}
\partial_{\rho}\lambda &= \frac{1}{4}X^{-2}\left[\left(\partial_{\rho}X\right)^2 - \left(\partial_zX\right)^2 + \left(\partial_{\rho}Y\right)^2 - \left(\partial_zY\right)^2\right] -\frac{1}{2}\partial_{\rho}\log X,
\\ \nonumber \partial_z\lambda &= \frac{1}{2}X^{-2}\left(\left(\partial_{\rho}X\right)\left(\partial_zX\right) + \left(\partial_{\rho}Y\right)\left(\partial_zY\right)\right) - \frac{1}{2}\partial_z\log X.
\end{align*}
The compatibility conditions to solve this turns out to follow from the harmonic maps equation for $(X,Y)$. Then this equation as well as the natural boundary conditions for $\lambda$ show that $\lambda$ is uniquely specified in terms of $(X,Y)$. Thus, we conclude that the Kerr family is the unique set of stationary and axisymmetric solutions to the Einstein vacuum equations. 
\subsection{Time-periodic Solutions to the Linear Klein--Gordon Equation}
The second important ingredient in our proof comes from the work~\cite{shlapgrow} where Theorem~\ref{KGgrow} was proven. In this section we now briefly review how the existence of the time-periodic solution was obtained.

In view of later arguments, we will in fact survey how the proof works for metrics of the form~\eqref{astataximetric} which are small perturbations of a Kerr spacetime. Thus, we consider a metric of the form~\eqref{astataximetric} to be a fixed small perturbation of a Kerr metric with $a \neq 0$. We look for a solutions to the Klein--Gordon equation
\[\Box_g\Psi - \mu^2\Psi = 0,\]
of the form
\[\Psi\left(t,\phi,\rho,z\right) \doteq e^{-it\omega}e^{im\phi}\psi\left(\rho,z\right).\]
We refer to $\psi$ as the ``reduced scalar field''. In coordinates, the equation for $\psi$ becomes
\begin{equation}\label{ascalarfieldequation}
\sigma^{-1}\partial_{\rho}\left(\sigma\partial_{\rho}\psi\right) + \sigma^{-1}\partial_z\left(\sigma\partial_z\psi\right) + e^{2\lambda}\sigma^{-2}X^{-1}\left(X\omega+Wm\right)^2\psi - e^{2\lambda}m^2X^{-1}\psi- e^{2\lambda}\mu^2 \psi = 0,
\end{equation}
where $\sigma = \sqrt{XV + W^2}$ is defined as in the previous section. For the spacetimes under consideration one can take $\sigma \sim \rho$. 

The most naive approach to finding solutions to~\eqref{ascalarfieldequation} would be to consider the equation as an eigenvalue problem for the time frequency $\omega$. However, this has two fundamental problems:
\begin{enumerate}
	\item If $\omega$ is considered as a spectral parameter, one does not have a self-adjoint eingenvalue problem.
	\item A general argument shows that any time-periodic solution must have an exactly vanishing energy flux along the event horizon. This turns out to imply that $\omega = cm$ for a fixed constant $c$ which depends only on the metric~\eqref{astataximetric}. In particular, given $m$ (which must be integer valued) there is only one choice of $\omega$ which could possibly work.
\end{enumerate}

The key trick to overcome these difficulties is to instead consider~\eqref{ascalarfieldequation} as an eigenvalue problem for the mass $\mu^2$. The problem of constructing such solutions can be reduced to a minimization problem for an appropriate Lagrangian for which the direct method of the calculus of variations turns out to work. 

\subsection{The Nonlinear Coupling: An Easy Model Problem}
In order to prove Theorem~\ref{timeperiodicsoln} we will, of course, have to deal with the difficulty that the equations for the metric is coupled nonlinearly to the equations for the scalar field.

The following model problem indicates the general structure we hope to exploit.
\begin{proposition}\label{modelprop}Consider the following system for two functions $h_1,h_2: B_1 \to \mathbb{R}$ on the ball of radius $1$ and a real parameter $\kappa$:
\begin{align}\label{modelproblem1}
\Delta h_1 &= h_2^2,\qquad h_1|_{\partial B_1} = 0,
\\ \label{modelproblem2} \Delta h_2 - \left(\kappa + h_1\right)h_2 &= 0,\qquad \partial_nh_2|_{\partial B_1} = 0.
\end{align}

We have the trivial $1$-parameter family of solutions given by $(h_1,h_2) = (0,0)$ and $\kappa \in \mathbb{R}$. Furthermore, linearization around these trivial solutions with $\kappa = 0$ leads to 
\begin{align}\label{linmodelproblem1}
\Delta l_1 &= 0,\qquad l_1|_{\partial B_1} = 0,
\\ \label{linmodelproblem2} \Delta l_2 &= 0,\qquad \partial_nh_2|_{\partial B_1} = 0.
\end{align}
Because of the Neumann boundary conditions, the linear equations have a nontrivial kernel spanned by $l_1 = 0$ and $l_2 = 1$. 

The key fact is that this ``linear hair'' $(l_1,l_2)$ can be upgraded to ``nonlinear hair'' in the sense that there exists a $1$-parameter family $\left(h_1^{(\delta)},h_2^{(\delta)},\kappa^{(\delta)}\right)_{\delta \in[0,\delta_{0})}$ of solutions to~\eqref{modelproblem1} and~\eqref{modelproblem2} which equal the trivial solution when $\delta = 0$ and such that $\left\vert\left\vert h_2^{(\delta)}\right\vert\right\vert_{L^2} = \delta$.
\end{proposition}
\begin{proof}This problem is sufficiently simple that there are many ways to proceed. We sketch a proof which is close in spirit to the approach we will take later in the proof of Theorem~\ref{timeperiodicsoln}. 

First of all, for $h_2$ sufficiently regular, we will always be able to invert the Laplacian and then set
\[h_1 \doteq \Delta^{-1}\left(h_2^2\right)\]
where $h_{1}$ has Dirichlet boundary conditions. 

We are thus lead to solve the nonlinear eigenvalue problem
\begin{equation}\label{nonlineeig}
\Delta h_2 - \left(\kappa + \Delta^{-1}\left(h_2^2\right)\right)h_2 = 0,\qquad \partial_nh_2|_{\partial B_1} = 0.
\end{equation}

Let\footnote{Many choices are possible here for the underlying Banach space.} $\mathcal{L} = W^{2,2}(B_{1})$ and let $\delta \geq 0$ be sufficiently small. Letting $\mathcal{L}_s$ denote the ball of radius $s$ in $\mathcal{L}$, we then take $s$ sufficiently small and define a map (we will check below that the image is really in $\mathcal{L}_{s}$)
\[T : \mathcal{L}_s \to \mathcal{L}_s\]
as follows:

For every $f \in \mathcal{L}$, we define $T\left(f\right)$ to be the smallest eigenfunction with Neumann boundary conditions for the operator $\Delta - \Delta^{-1}\left(f^2\right)$. We normalize $u=T\left(f\right)$ to be positive (considering the nodal domains of $u$, we see that it cannot change sign) and have $\Vert u \Vert_{L^{2}(B_{1})} = \delta$. Recall that the variational characterization of the smallest eigenfunction $\kappa$ and eigenfunction $u$ shows that
\[
- \kappa = \delta^{-2} \int_{B_{1}} |\nabla u|^{2} + (\Delta^{-1}(f^{2})) u^{2} = \inf_{\{\tilde u \in H^{1}(B_{1}) : \Vert \tilde \Vert_{L^{2}(B_{1})}=\delta\}}  \delta^{-2} \int_{B_{1}} |\nabla \tilde u|^{2} + (\Delta^{-1}(f^{2})) \tilde u^{2} .
\]

To show that $T$ is a contraction for $\delta$ sufficiently small, we consider $f_{1},f_{2} \in \mathcal{L}_{s}$, and $u_{i} = T(f_{i})$ with associated eigenvalue $\kappa_{i}$ for $i=1,2$. Let $w_{i}$ solve $\Delta w_{i} = f_{i}^{2}$ with Dirichlet boundary conditions. By taking $u_{2}$ in the variational characterization of $\kappa_{1}$, we find that 
\[
- \kappa_{1} \leq -\kappa_{2} + c \Vert w_{1}-w_{2}\Vert_{C^{0}(B_{1})} \leq -\kappa_{2} + c\Vert f_{1}^{2}-f_{2}^{2}\Vert_{\mathcal{L}}.
\]
The second inequality follows from a standard application of the maximum principle (cf.\ Theorem 3.7 in \cite{giltru}) and Sobolev embedding. Thus (by symmetry),
\[
|\kappa_{1}-\kappa_{2}| + \Vert w_{1} - w_{2}\Vert_{C^{0}(B_{1})} \leq cs \Vert f_{1}-f_{2}\Vert_{\mathcal{L}}.
\]
Note that a simpler argument gives $|\kappa_{i}| \leq cs^{2}$ and $\Vert w_{i}\Vert_{C^{0}(B_{1})}\leq cs^{2}$ (note that this step uses the quadratic nature of $h_{1}$'s equation). In particular, elliptic estimates imply that
\[
\Vert u_{i} \Vert_{\mathcal{L}} \leq C(\delta + c s^{2})
\]
Thus, for any $s$ chosen sufficiently small, we can choose $\delta>0$ even smaller so that $T$ indeed maps into $\mathcal{L}_{s}$. 

To finish the proof that $T$ is a contraction map, we now rely on continuity of the first eigenfunction of a Schrodinger operator as the potential varies. There are several possible ways to proceed here (including arguments based on the spectral gap, or of a more functional theoretic flavor). Here, we give a proof that is more faithful to the argument we use in the proof of Theorem \ref{timeperiodicsoln} (cf.\ Proposition \ref{agmax}).

Multiply the equation for $u_{2}$ by $\frac{v^{2}}{u_{2}}$ where $v=u_{1}-u_{2}$ and integrate by parts. We thus find that
\[
\int_{B_{1}} \nabla \left( \frac{v^{2}}{u_{2}}\right)\cdot \nabla u_{2} + (\kappa_{2}+ w_{2}) v^{2} = 0.
\]
Note that 
\[
\nabla \left( \frac{v^{2}}{u_{2}}\right)\cdot \nabla u_{2}  = |\nabla v|^{2} - (u_{2})^{2} \left|\nabla \left( \frac{u_{1}}{u_{2}} \right)\right|^{2}.
\]
Thus, using the equations satisfied by $u_{i}$,
\begin{align*}
\int_{B_{1}} (u_{2})^{2}  \left|\nabla \left( \frac{u_{1}}{u_{2}} \right)\right|^{2} & \leq \int_{B_{1}} |\nabla v|^{2} + (\kappa_{2} +w_{2}) v^{2}\\
%& = \int_{B_{1}} - v\Delta v + (\kappa_{2} + w_{2})v^{2}\\
%& = \int_{B_{1}} v( -(\kappa_{2}+w_{2})v + (\kappa_{2}+w_{2})u_{1}-(\kappa_{1}+w_{1})u_{1} )+ (\kappa_{2} + w_{2})v^{2}\\
& = \int_{B_{1}} (\kappa_{2}-\kappa_{1} + w_{2}-w_{1}) u_{1} v\\
& \leq c(q) \delta^{2} \Vert f_{1}-f_{2}\Vert_{\mathcal{L}}^{2} + q  \int v^{2} 
\end{align*}
for $q>0$ to be chosen below. 

On the other hand, De Giorgi--Nash iteration (using the Neumann boundary condition, the normalization $\Vert u_{i}\Vert_{L^{2}(B_{1})}=\delta$, and the $C^{0}$ bound for $w_{i}$ alluded to above) implies that there are $c,C$ uniform for $s,\delta$ small so that $u_{i} \in [c \delta, C\delta]$. Now, the Poincar\'e inequality implies that
\begin{align*}
C^{-1} & \int_{B_{1}} |u_{1}-A u_{2}|^{2} \leq \delta^{2}\int_{B_{1}} \left| \frac{u_{1}}{u_{2}}-A \right|^{2} \leq c\delta^{2} \int_{B_{1}}  \left|\nabla \left( \frac{u_{1}}{u_{2}} \right)\right|^{2}\\
&  \leq \int_{B_{1}} (u_{2})^{2}  \left|\nabla \left( \frac{u_{1}}{u_{2}} \right)\right|^{2}  \leq c(q) \delta^{2} \Vert f_{1}- f_{2}\Vert_{\mathcal{L}}^{2} + q  \int v^{2} 
\end{align*}
for 
\[
A = \frac{1}{|B_{1}|} \int_{B_{1}} \frac{u_{1}}{u_{2}}.
\]
Observe that the triangle inequality and the normalization of the $u_{i}$'s implies that 
\[
 \delta^{2} |1-A|^{2} = \left| \Vert u_{1}\Vert_{L^{2}(B_{1})} - A \Vert u_{2}\Vert_{L^{2}(B_{1})}\right|^{2} \leq  \int_{B_{1}} |u_{1}-A u_{2}|^{2} \leq c(q) \delta^{2} \Vert f_{1}- f_{2}\Vert_{\mathcal{L}}^{2} + c q  \int v^{2}
 \]
so putting this all together, we find that
\[
\int_{B_{1}} v^{2} \leq 2 \int_{B_{1}} (u_{1}-A u_{2})^{2} +  2\delta^{2}|1-A|^{2} \leq c(q) \delta^{2} \Vert f_{1}- f_{2}\Vert_{\mathcal{L}}^{2} + c q  \int v^{2}.
\]
Choosing $q$ appropriately, we can absorb the second term into the right hand side. Standard $W^{2,2}$-elliptic estimates allow us to upgrade the resulting $L^{2}$-bound to the $\mathcal{L}$-estimate,
\[
\Vert T(f_{1})-T(f_{2})\Vert_{\mathcal{L}} = \Vert u_{1}-u_{2}\Vert_{\mathcal{L}}\leq c\delta\Vert f_{1}-f_{2}\Vert_{\mathcal{L}}
\]  
Thus, for $\delta >0$ sufficiently small, $T$ is a contraction and thus has a unique fixed point. This fixed point yields the desired solution.
\end{proof}

We now make an analogy with Theorem~\ref{timeperiodicsoln} by letting $h_1$ correspond to the metric $g$, $h_2$ to the reduced scalar field $\psi$, and the parameter $\kappa$ to the Klein--Gordon mass $\mu^2$. In this analogy, the nonlinear coupling between $h_1$ and $h_2$ mimics the way that $\psi$ is coupled to $g$ through the Einstein equations and the way in which $g$ is coupled to $\psi$ in the Klein--Gordon equation, the ability to invert the Laplacian and solve for $h_1$ corresponds the availability of the Carter--Robinson theory, and finally the variational structure used in solving for $h_2$ corresponds to the variational structure we have for $\psi$. 
\subsection{The Fundamental Difficulties}
We close this overview with a discussion of the key difficulties we will have to overcome if we can hope to prove Theorem~\ref{timeperiodicsoln} in the fashion of the proof of Proposition~\ref{modelprop}.
\subsubsection{Approximate Carter--Robinson Reduction}
We will have to show that the decoupling behind the Carter--Robinson theory continues to hold in an approximate setting even when we do not have a vacuum spacetime. 

We give one illustrative example of a difficulty which arises. Once our spacetime is not exactly Ricci flat, then one will no longer have that $d\theta = 0$ (recall that $\theta$ is the ``twist'' $1$-form~\eqref{twistthe1form}). However, in order to exploit the harmonic maps structure behind the Carter--Robinson theory, we will need to pick some function $Y$ so that $dY = \theta$. Equivalently, we need to make a ``gauge'' choice of a function $B$ which will satisfy $dB = \theta$. There is no ``obvious'' choice for $B$, but if $B$ is not chosen correctly, it becomes difficult to close the necessary estimates on the nonlinear error terms.
\subsubsection{Loss of Ellipticity at the Event Horizons and Singularities on the Axis}
The biggest deviation from the mode problem solved by Proposition~\ref{modelprop} is that essentially none of the equations we shall study are uniformly elliptic. This occurs for two reasons. First of all, on the event horizon of a Kerr black hole, the space of $T$ and $\Phi$ is \emph{not} timelike; hence, we expect to lose ellipticity as we approach the event horizon. Secondly, near the axis of symmetry, i.e., the fixed points of $\Phi$, our equations will have singular coefficients. Among the most worrisome points are where the axis of symmetry meets the event horizon.

Ultimately, these difficulties will be dealt with in an ad-hoc fashion for each equation. However, the general philosophy is already apparent in the equations for $(X,Y)$ in the vacuum setting~\eqref{harmonicmap1} and~\eqref{harmonicmap2}. Considered in the $(\rho,z)$ plane, these equations look very degenerate as $\rho \to 0$. However, once they are interpreted as axisymmetric functions on $\mathbb{R}^3$ the equations become regular. This trick of interpreting singular equations in lower dimensions as regular equations for symmetric functions in higher dimensions will occur repeatedly in our proof.
\subsubsection{Compatibility Conditions}\label{subsubsec:compat-cond}
In order for a metric of the form~\eqref{astataximetric} to have a regular extension to a bifurcate event horizon and to the axis of symmetry there are various compatibility conditions in the boundary conditions for the metric components that must be satisfied. Furthermore, the time-frequency parameter $\omega$ must be chosen so that the scalar field $\Psi$ has a vanishing energy flux along the event horizon.\footnote{See Theorem 1.3 in \cite{shlapgrow}.} This is crucial to close the estimates for $\Psi$ and various other terms. Thus, in the setting of a fixed point argument like in Proposition~\ref{modelprop} the parameter $\omega$ must be continuously modified. These couplings between the various unknowns will cause technical annoyances.

\subsubsection{Fixing the Klein--Gordon Mass}

Finally, the key variational structure behind the existence theory of the scalar field $\psi$ requires that we treat $\mu^2$ as an eigenvalue. Thus, it would naively seem that $\mu$ will vary along our $1$-parameter family of bifurcating solutions. However, Theorem~\ref{timeperiodicsoln} concerns a fixed Klein--Gordon mass. In order to keep the Klein--Gordon mass fixed and thus prove Theorem~\ref{timeperiodicsoln}, we will need to exploit the freedom to also vary the underlying $2$-parameter Kerr family that we are perturbing off of.

\section{Preliminaries}

We begin by \emph{fixing} a reference Kerr solution \index{Metric Quantities!$g_{(a,M)}$}$g_{(a,M)}$ with $0 < |a| < M$. These parameters can safely be thought of as being fixed until the final section of the paper, Section \ref{arrange}. 

We will briefly discuss the standard Boyer--Lindquist coordinates for $g_{(a,M)}$, and then introduce the isothermal coordinates that will form the backbone for our metric ansatz used in the proof of Theorem \ref{timeperiodicsoln}. Define \index{Metric Quantities!$\tilde r_{\pm}$}$\tilde r_{\pm} = M \pm\sqrt{M^2-a^2}$. The domain of outer communication (minus the axis of symmetry) of the Kerr spacetime of mass $M$ and angular momentum $a$ can be covered by a \index{Coordinates!Boyer--Lindquist}Boyer--Lindquist coordinate chart\footnote{We have used $\tilde r$ in place of $r$, because we will frequently use the definition $r^{2}=1+\rho^{2}+z^{2}$, for $\rho,z$ isothermal coordinates to be defined below.} on
\[\index{Coordinates!$\tilde r$}\{(t,\tilde r,\theta,\phi) \in \mathbb{R}\times (\tilde r_{+},\infty)\times (0,\pi) \times (0,2\pi)\}\]
where the metric takes the form
\begin{equation}\label{kerrmetric}
g_{a,M} = -\left( 1-\frac{2M\tilde r}{\Sigma^{2}}\right) dt^{2} - \frac{4Ma\tilde r \sin^{2}\theta}{\Sigma^{2}} dtd\phi+ \sin^{2}\theta \frac{\Pi}{\Sigma^{2}}d\phi^{2}  +\frac{\Sigma^{2}}{\Delta} d\tilde r^{2}  + \Sigma^{2} d\theta^{2}.
\end{equation}
Here
\begin{align*}
 \index{Metric Quantities!$\Delta$}\Delta & \doteq  \tilde r^{2} - 2M \tilde r+a^{2}\\
 \index{Metric Quantities!$\Sigma^{2}$} \Sigma^{2} & \doteq\tilde r^{2}+a^{2}\cos^{2}\theta\\
 \index{Metric Quantities!$\Pi$} \Pi & \doteq (\tilde r^{2}+a^{2})^{2}-a^{2}\sin^{2}\theta \Delta.
\end{align*}
Clearly, \index{Metric Quantities!$T$}$T\doteq \frac{\partial}{\partial t}$ and \index{Metric Quantities!$\Phi$}$\Phi \doteq \frac{\partial}{\partial \phi}$ are Killing vectors.

In order to bring~(\ref{kerrmetric}) into the form considered below, we introduce isothermal coordinates in the $(\tilde r,\theta)$ plane by\index{Coordinates!$\rho$}\index{Coordinates!$z$}
\[\rho\left(\tilde r,\theta\right) \doteq \sqrt{\Delta}\sin\theta,\]
\[z\left(\tilde r,\theta\right) \doteq  \left(\tilde r-M\right)\cos\theta.\]
After some straightforward calculations, one finds that the metric in $(t,\phi,\rho,z)$ coordinates is in the form
\begin{equation}\label{kerrmetric2}
g_{a,M} =  -V_{K} dt^{2} + 2W_{K}dtd\phi+ X_{K}d\phi^{2}  +e^{2\lambda_{K}}\left(d\rho^2+dz^2\right).
\end{equation}
 with
\begin{align*}
\index{Metric Quantities!$V_{K}$}V_{K} & = 1 - \frac{2M\tilde r}{\Sigma^{2}}\\
\index{Metric Quantities!$W_{K}$}W_{K} & = - \frac{2Ma\tilde r\sin^{2}\theta}{\Sigma^{2}}\\
\index{Metric Quantities!$X_{K}$}X_{K} & = \sin^{2}\theta \frac{\Pi}{\Sigma^{2}}\\
\index{Metric Quantities!$\lambda_{K}$} e^{2\lambda_{K}} & = \Sigma^{2}\Delta^{-1}\left(\frac{(\tilde r-M)^{2}}{\Delta} \sin^{2}\theta + \cos^{2}\theta \right)^{-1}.
\end{align*}
Here, we are considering $\tilde r$ and $\theta$ as implicit functions of $\rho,z$. See \cite[Lemma 3.0.1]{HBH:geometric} for some related computations. 

Let us also take the opportunity to note the important fact that $X_KV_K + W_K^2 = \rho^2$ (indeed, based on this relation, it is not hard to derive the above definition for $\rho$). 

It is convenient to define a constant \index{Metric Quantities!$\gamma$}$\gamma^{2} \doteq M^{2}-a^{2}$. 

\subsection{Metric ansatz} \label{subsec:metric-ansatz} We now describe our metric ansatz (as well as recall certain results from our companion article \cite{HBH:geometric}) used in the proof of Theorem \ref{timeperiodicsoln}. Roughly speaking, we will search for solutions to the Einstein--Klein--Gordon equations whose metric is of the form \eqref{kerrmetric2}. As we will see, the fact that $\rho,z$ are isotropic in an appropriate sense will play an important role in our analysis. We note that our choice of ansatz is based on the work \cite{weinstein,weinstein2}, and in particular the construction of the manifold with corners $\overline{\mathscr{B}}$ (e.g., the choice of $s$ and $\chi$ coordinates) is directly motivated by \cite{weinstein2}. 

We define \index{Coordinates!$\mathcal{M}$}$\mathcal{M} \doteq \{(t,\phi,\rho,z)  \in \mathbb{R} \times (0,2\pi) \times \mathscr{B}\}$, where $\mathscr{B} \doteq \{(\rho,z) \in \mathbb{R}^{2} : \rho > 0\}$. We will assume that the spacetimes, minus the axis of symmetry and horizon, are given by $(\mathcal{M},g)$ where the Lorentzian metrics $g$ take the form
\begin{equation}\label{eq:metric-ansatz}
g \doteq -V dt^{2} + 2W dtd\phi + X d\phi^{2} + e^{2\lambda}(d\rho^{2}+dz^{2})
\end{equation}
for suitable functions $V,W,X,\lambda : \mathscr{B}\to\mathbb{R}$. 

In \cite{HBH:geometric} we examine the Einstein--Klein--Gordon equations under this metric ansatz---we will return to this discussion below. However, we emphasize that even if we know that the metric $g$ (along with a scalar field) is a solution to the Einstein--Klein--Gordon equations, there is no reason (in general) to expect that $g$ extends smoothly across the axis of symmetry or to the horizon. Furthermore, the equations corresponding to the Einstein--Klein--Gordon equations derived in \cite{HBH:geometric} behave in a highly singular manner at the axis and horizon (essentially, due to the fact that the $(t,\phi,\rho,z)$ coordinates are not regular at the axis and horizon---this can already be seen by examining the Kerr metric in these coordinates). 

To handle both of these issues we will construct an appropriate surface with corners \index{Coordinates!$\overline{\mathscr{B}}$}$\overline{\mathscr{B}}$ so that \index{Coordinates!$\mathscr{B}$}$\mathscr{B}$ is the interior of $\overline{\mathscr{B}}$. We will then solve for $V,W,X,\lambda : \overline{\mathscr{B}} \to \RR$ which lie in certain function spaces. This will then allow us to refer to \cite[Proposition 2.2.1]{HBH:geometric}, which describes sufficient conditions for $(\mathcal{M},g)$ to extend to a larger Lorentzian manifold with boundary $(\tilde{\mathcal{M}},\tilde g)$ which is asymptotically flat and has a boundary consisting of a bifurcate Killing event horizon. In particular, our function spaces will turn out to guarantee that the horizon is non-degenerate. 

Recall that\footnote{The parameter \index{Coordinates!$\beta$}$\beta$ in \cite{HBH:geometric} will always be equal to $\gamma$ here, as it is for Kerr.} $\gamma = \sqrt{M^{2}-a^{2}}$. We now describe the surface with corners $\overline{\mathscr{B}}$ by describing three coordinate charts covering it. The reader may find it useful to also refer to \cite[Section 2]{HBH:geometric} for further details of this construction. The \index{Coordinates!$\rho$}\index{Coordinates!$z$}$\rho,z$ coordinates extend to cover the majority of $\overline{\mathscr{B}}$ as
\[
\{(\rho,z) : \rho\geq 0, z \in \RR\setminus\{\pm \gamma\}\}
\]
The interval \index{Coordinates!$\mathscr{H}$}$\mathscr{H} : = \{(0,z) : z \in (-\gamma,\gamma)\}$ will correspond to the horizon, while the intervals \index{Coordinates!$\mathscr{A}$}$\mathscr{A} = \{(0,z): |z| > \gamma \}$ will correspond to the axis. We will label the top and bottom components of the axis by \index{Coordinates!$\mathscr{A}_{N}$}$\mathscr{A}_{N}$ and  \index{Coordinates!$\mathscr{A}_{S}$}$\mathscr{A}_{S}$, respectively. 

\begin{figure}[h!]
\begin{tikzpicture}
	\filldraw [opacity = .2] (0,-3) rectangle (2,2);
	\filldraw [white] (0,0) circle (.1);
	\filldraw [white] (0,-1) circle (.1);
	\draw [thick, ->] (0,.1) -- node [left] {$\mathscr{A}_N$}  (0,2.1) node [above] {$\scriptstyle \rho=0$};
	\draw [thick] (0,-.1) -- node [left] {$\mathscr{H}$}  (0,-.9);
    	\draw [thick, ->] (0,-1.1) -- node [left] {$\mathscr{A}_S$}  (0,-3.1);
%	\draw (0,-.1) arc (-90:90:.1);
%	\draw (0,-1.1) arc (-90:90:.1);
	
	\node  at (-.5,0) {$\scriptstyle z=+\gamma$};
	\node  at (-.5,-1) {$\scriptstyle z=-\gamma$};	
	
	\draw [->] (1,-2.5) -- (1,-2) node [left] {$z$};
	\draw [->] (1,-2.5) -- (1.5,-2.5) node [below] {$\rho$};
	
	\draw  [dashed] (.05,.15) -- (3.45,1.85);
	\draw [dashed] (.05,-.15) -- (4,-.65);
	\draw [dashed] (3.95,.65) circle (1.3);

	\draw [thick, ->] (3.95,.65) -- (3.95,.65+1.3) node [below left] {$\chi$};
	\draw [thick, ->] (3.95,.65) -- (3.95 + 1.3,.65) node [below left] {$s$};
	\filldraw [opacity = .2] (3.95,.65) -- (3.95+1.3,.65) arc (0:90:1.3) -- cycle;
	\node at (3.9,.47) {$p_{N}$};
\end{tikzpicture}
\caption{A diagram of the manifold with corners $\overline{\mathscr{B}}$.}
\label{fig:conf-manifold}
\end{figure}
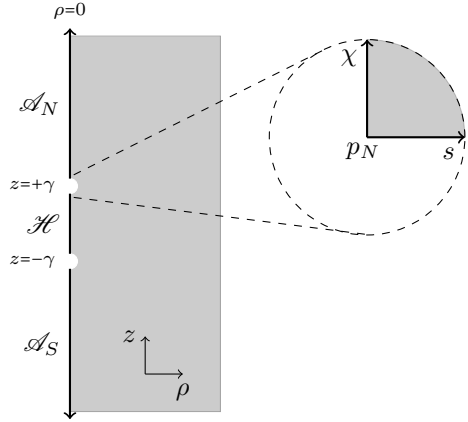

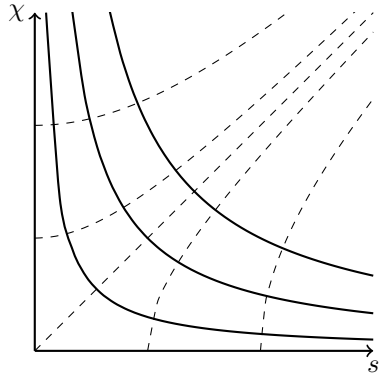
\begin{figure}[h!]
\begin{tikzpicture}[scale = 1.5]
	\begin{scope}
	\clip (0,0) rectangle (3,3);
	  \draw[thick,domain=.1:3,smooth,variable=\x] plot ({\x},{.3/\x});
	  \draw[thick, domain=.33:3,smooth,variable=\x] plot ({\x},{1/\x});
  	  \draw[thick, domain=.66:3,smooth,variable=\x] plot ({\x},{2/\x});
	   \draw[dashed, domain=0:3,smooth,variable=\x] plot ({\x},{(\x^2+1)^(1/2)});
   	   \draw[dashed, domain=0:3,smooth,variable=\x] plot ({\x},{(\x^2+4)^(1/2)});
	   \draw[dashed, domain=0:3,smooth,variable=\x] plot ({\x},{\x});
	   \draw[dashed, domain=1:3,smooth,variable=\x] plot ({\x},{(\x^2-1)^(1/2)});
   	   \draw[dashed, domain=2:3,smooth,variable=\x] plot ({\x},{(\x^2-4)^(1/2)});
	  \end{scope}
 	\draw [thick,->] (0,0) -- (3,0) node [below] {$s$};
	 \draw [thick,->] (0,0) -- (0,3) node [left] {$\chi$};
	%FINISH
\end{tikzpicture}
\caption{An illustration of the lines of constant $\rho$ (the solid lines) and $z$ (the dashed lines) in the $s,\chi$ plane.}
\label{fig:conf-manifold-rho-z}
\end{figure}

A careful examination of the form of the Kerr metric placed into the form \eqref{kerrmetric2} suggests that the axis meets the horizon at a right angle. Thus, we cannot simply include the points $(0,\pm\gamma)$, but instead must introduce new coordinate charts adapted to those points. To cover the region where $\mathscr{A}_{N}$ meets $\mathscr{H}$, we define coordinates on the surface with corners (see Figure \ref{fig:conf-manifold})
\[
\index{Coordinates!$s$}\index{Coordinates!$\chi$}\{(s,\chi) : s \geq 0, \chi\geq 0\}
\]
and observe that the coordinate change (see Figure \ref{fig:conf-manifold-rho-z})
\[
\rho  \doteq s \chi
\]
\[
z  \doteq \frac 12 (\chi^{2}-s^{2}) +\gamma
\]
is a smooth change of coordinates on $\{(s,\chi) : s > 0 ,\chi > 0\}$. Observe that we can invert the coordinate change to find
\[
s = \sqrt{-(z-\gamma) + \sqrt{(z-\gamma)^{2} + \rho^{2}}} 
\]
\[
\chi = \sqrt{(z-\gamma) + \sqrt{(z-\gamma)^{2} + \rho^{2}}} ,
\]
so the $s,\chi$ coordinate patch actually covers the entire $\rho,z$ coordinate patch.\footnote{The $\rho,z$ coordinates will be very important to our analysis on $\overline{\mathscr{B}}$ as they are much simpler than the $s,\chi$ coordinates away from the axis and horizon.}
Note also that the coordinate change is conformal
\[
d\rho^{2} + dz^{2} = (s^{2}+\chi^{2}) \left( ds^{2}+d\chi^{2}\right).
\]
It is useful to record the following expressions linking $\partial_{s}$ and $\partial_{\chi}$ with $\partial_{\rho}$ and $\partial_{z}$
\[
\partial_{s} = \chi \partial_{\rho} - s \partial_{z}, \qquad \partial_{\chi} = s\partial_{\rho}+ \chi\partial_{z}
\]
\[
\partial_{\rho} = \frac{\chi}{\chi^{2}+s^{2}} \partial_{s} + \frac{s}{\chi^{2}+s^{2}} \partial_{\chi}, \qquad \partial_{z} = \frac{-s}{s^{2}+\chi^{2}} \partial_{s}+ \frac{\chi}{s^{2}+\chi^{2}} \partial_{\chi}.
\]
We define the point where $s=\chi=0$ to be \index{Coordinates!$p_{N}$}$p_{N}$ and similarly set $s'=\chi'=0$ to be \index{Coordinates!$p_{S}$}$p_{S}$.

Finally, to cover the region where $\mathscr{A}_{S}$ meets $\mathscr{H}$, we define completely a analogous coordinate chart \index{Coordinates!$s'$}\index{Coordinates!$\chi'$}$\{(s',\chi') : s'\geq 0,\chi'\geq 0\}$ by 
\[
\rho \doteq s'\chi'
\]
\[
z\doteq \frac 12 \left((\chi')^{2}-(s')^{2}\right) - \gamma.
\]
It is easy to derive formulas analogous to those for the $s,\chi$ coordinates discussed above. 

We have thus defined a conformal structure on the manifold with corners $\overline{\mathscr{B}}$. It is useful to define the following open sets of $\overline{\mathscr{B}}$ 
\begin{align*}
\index{Coordinates!$\overline{\mathscr{B}_{A}}$}\overline{\mathscr{B}_{A}}& \doteq \left\{ (\rho,z) \in \overline{\mathscr{B}} : \rho^{2} + (z-\gamma)^{2} > \frac{\gamma}{200}, \rho^{2} + (z+\gamma)^{2} > \frac{\gamma}{200} , |z| + |\rho| > \frac{499}{500} \gamma \right\} \\
\index{Coordinates!$\overline{\mathscr{B}_{H}}$} \overline{\mathscr{B}_{H}}& \doteq \left\{ (\rho,z) \in \overline{\mathscr{B}} : \rho^{2} + (z-\gamma)^{2} > \frac{\gamma}{200}, \rho^{2} + (z+\gamma)^{2} > \frac{\gamma}{200} , |z| + |\rho| < \frac{501}{500} \gamma \right\}  \\
\index{Coordinates!$\overline{\mathscr{B}_{N}}$} \overline{\mathscr{B}_{N}}& \doteq \left\{ (\rho,z) \in \overline{\mathscr{B}} : z\not = \gamma, \rho^{2}+(z-\gamma)^{2} < \frac{\gamma}{100} \right\} \cup \left\{(s,\chi) \in \overline{\mathscr{B}} : 0 \leq s,\chi < \left( \frac{\gamma}{25}\right)^{\frac 1 4} \right\}\\
\index{Coordinates!$\overline{\mathscr{B}_{S}}$}  \overline{\mathscr{B}_{S}}& \doteq \left\{ (\rho,z) \in \overline{\mathscr{B}} : z\not = -  \gamma, \rho^{2}+(z+\gamma)^{2} < \frac{\gamma}{100} \right\} \cup \left\{(s',\chi') \in \overline{\mathscr{B}} : 0 \leq s',\chi'< \left( \frac{\gamma}{25}\right)^{\frac 1 4} \right\}
\end{align*}
The precise form of the sets will not be of much importance, so the reader can simply think of $\overline{\mathscr{B}_{A}}$ as being an open neighborhood of (most of) the axis, $\overline{\mathscr{B}_{H}}$ a neighborhood of (most of) the horizon, and $\overline{\mathscr{B}_{N}}$ (resp.\ $\overline{\mathscr{B}_{S}}$) an open neighborhood of the north pole $p_{N}$ (resp.\ the south pole $p_{S}$). Observe that $\overline{\mathscr{B}}$ is covered by the union of these four open sets. 

\subsection{Function Spaces on $\overline{\mathscr{B}}$}\label{funcspace}
In this section we will introduce some function spaces naturally associated to $\overline{\mathscr{B}}$ which will be convenient for our analysis.

By considering $(\rho,z,\phi)$ as cylindrical coordinates on $\mathbb{R}^3$, to any function $f(\rho,z) : \overline{\mathscr{B}} \to \mathbb{R}$, we may associate a function $f_{\mathbb{R}^3} : \mathbb{R}^3 \to \mathbb{R}$ by setting
\[f_{\mathbb{R}^3}\left(\rho,z,\phi\right) \doteq f\left(\rho,z\right).\]

Using this, we define the following space of functions $f: \overline{\mathscr{B}} \to \mathbb{R}$.
\begin{definition}\label{Lp}We define
\[\index{Function Spaces!$\dot{W}^{k,p}_{\rm axi}\left(\overline{\mathscr{B}}\right)$} \dot{W}^{k,p}_{\rm axi}\left(\overline{\mathscr{B}}\right) \doteq \left\{ f(\rho,z) : f_{\mathbb{R}^3} \in \dot{W}^{k,p}\left(\mathbb{R}^3\right)\right\},\]
\[\index{Function Spaces!$W^{k,p}_{\rm axi}\left(\overline{\mathscr{B}}\right)$} W^{k,p}_{\rm axi}\left(\overline{\mathscr{B}}\right) \doteq \left\{ f(\rho,z) : f_{\mathbb{R}^3} \in W^{k,p}\left(\mathbb{R}^3\right)\right\},\]
\[\index{Function Spaces!$C^{k,\alpha}_{\rm axi}\left(\overline{\mathscr{B}}\right)$} C^{k,\alpha}_{\rm axi}\left(\overline{\mathscr{B}}\right) \doteq \left\{ f(\rho,z) : f_{\mathbb{R}^3} \in C^{k,\alpha}\left(\mathbb{R}^3\right)\right\},\]
\[\index{Function Spaces!$C^{k,\alpha}_{0,\rm axi}\left(\overline{\mathscr{B}}\right) $} C^{k,\alpha}_{0,\rm axi}\left(\overline{\mathscr{B}}\right) \doteq \left\{ f(\rho,z) : f_{\mathbb{R}^3} \in C_0^{k,\alpha}\left(\mathbb{R}^3\right)\right\}.\]
We extend these norms to $1$-forms by applying the corresponding norm to each of the $x$, $y$, $z$ components of the form.
\end{definition}

It turns out that another family of function spaces is also useful. First of all, by considering  $(s,\phi_1,\chi,\phi_2)$ as polar coordinates\footnote{It is potentially surprising that ``four dimensional coordinates'' are useful near the poles. We can (partially) justify this by the fact that certain quantities naturally arising out of the Kerr metric in this form are regular in precisely the sense described by these ``four dimensional coordinates.'' See Lemma \ref{lem:xK-regularity} for an example of this.} on $\mathbb{R}^2 \times \mathbb{R}^2 = \mathbb{R}^4$, to any function $f(\rho,z) : \overline{\mathscr{B}_N} \to \mathbb{R}$ or $f(\rho,z) : \overline{\mathscr{B}_S} \to \mathbb{R}$, we may associate a function \index{Coordinates!$f_{\mathbb{R}^{4}}$}$f_{\mathbb{R}^4} : \mathbb{R}^4 \to \mathbb{R}$. For example, for $s,\chi \in \overline{\mathscr{B}_{N}}$ we set
\[
f_{\RR^{4}}(s,\phi_{1},\chi,\phi_{2}) = f(s,\chi). 
\]

Next we introduce some useful cut-offs.
\begin{definition}\label{cutoffs}
\index{Miscellaneous!$\chi_{N}$}\index{Miscellaneous!$\chi_{S}$} Let $\xi_N,\xi_S:\overline{\mathscr{B}} \to \mathbb{R}$ be smooth\footnote{Observe that it may not be clear what it means to be \emph{smooth} on $\overline{\mathscr{B}}$ in general. In this case, since $\xi_{N}\equiv 1$ near $p_{N}$ and $\xi_{N}\equiv 0$ near $p_{S}$, we simply require that $\xi_{N}$ is smooth in the $\rho,z$ coordinates (where they are defined). A similar discussion holds for $\xi_{S}$. } cut-off functions such that $\rho$ sufficiently small implies that $\partial_{\rho}\xi_N = \partial_{\rho}\xi_S = 0$, supp$\left(\xi_N\right) \subset \overline{\mathscr{B}_N}$, supp$\left(\xi_S\right) \subset \overline{\mathscr{B}_S}$, and supp$\left(1-\xi_N-\xi_S\right) \subset \overline{\mathscr{B}_A} \cup \overline{\mathscr{B}_H}$.
\end{definition}

Now we are ready for the second family of useful function spaces.
\begin{definition}\label{norm}
We define
\[\index{Function Spaces!$\hat{\dot{W}}_{\rm axi}^{k,p}\left(\overline{\mathscr{B}}\right)$} \hat{\dot{W}}_{\rm axi}^{k,p}\left(\overline{\mathscr{B}}\right) = \left\{ f(\rho,z) : \left(\xi_Nf\right)_{\mathbb{R}^4},\left(\xi_Sf\right)_{\mathbb{R}^4} \in \dot{W}^{k,p}\left(\mathbb{R}^4\right)\text{ and }\left(\left(1-\xi_N-\xi_S\right)f\right)_{\mathbb{R}^3} \in \dot{W}^{k,p}\left(\mathbb{R}^3\right)\right\},\]
\[\index{Function Spaces!$\hat{W}^{k,p}_{\rm axi}\left(\overline{\mathscr{B}}\right)$} \hat{W}^{k,p}_{\rm axi}\left(\overline{\mathscr{B}}\right) = \left\{ f(\rho,z) :  \left(\xi_Nf\right)_{\mathbb{R}^4},\left(\xi_Sf\right)_{\mathbb{R}^4} \in W^{k,p}\left(\mathbb{R}^4\right)\text{ and }\left(\left(1-\xi_N-\xi_S\right)f\right)_{\mathbb{R}^3} \in W^{k,p}\left(\mathbb{R}^3\right)\right\},\]
\[\index{Function Spaces!$\hat{C}^{k,\alpha}_{\rm axi}\left(\overline{\mathscr{B}}\right)$}  \hat{C}^{k,\alpha}_{\rm axi}\left(\overline{\mathscr{B}}\right) = \left\{ f(\rho,z) : \left(\xi_Nf\right)_{\mathbb{R}^4},\left(\xi_Sf\right)_{\mathbb{R}^4} \in C^{k,\alpha}\left(\mathbb{R}^4\right)\text{ and }\left(\left(1-\xi_N-\xi_S\right)f\right)_{\mathbb{R}^3} \in C^{k,\alpha}\left(\mathbb{R}^3\right)\right\}.\]
\[\index{Function Spaces!$\hat{C}^{k,\alpha}_{0,\rm axi}\left(\overline{\mathscr{B}}\right)$} \hat{C}^{k,\alpha}_{0,\rm axi}\left(\overline{\mathscr{B}}\right) = \left\{ f(\rho,z) : \left(\xi_Nf\right)_{\mathbb{R}^4},\left(\xi_Sf\right)_{\mathbb{R}^4} \in C_0^{k,\alpha}\left(\mathbb{R}^4\right)\text{ and }\left(\left(1-\xi_N-\xi_S\right)f\right)_{\mathbb{R}^3} \in C_0^{k,\alpha}\left(\mathbb{R}^3\right)\right\}.\]
We extend these norms to $1$-forms by cutting the forms off to the appropriate regions and then applying the corresponding norm to each of the components of the form in the relevant Cartesian orthonormal basis.
\end{definition}

\begin{remark}
When there is unlikely to be any confusion, we will often drop the subscript ``axi.''
\end{remark}

Similarly, we define \index{Function Spaces!$\hat{C}^{k}$} $\hat{C}^k$ and \index{Function Spaces!$\hat{C}^{k}_{0}$}$\hat{C}_0^k$ spaces, and then set \index{Function Spaces!$\hat{C}^{\infty}$}$\hat{C}^{\infty} \doteq \cap_{k=1}^{\infty}\hat{C}^k$ and \index{Function Spaces!$\hat{C}^{\infty}_{0}$}$\hat{C}_0^{\infty} \doteq \cap_{k=1}^{\infty}\hat{C}_0^k$.

For any function $f: \mathscr{B} \to \mathbb{R}$ we denote the gradient of $f$ by \index{Miscellaneous!$\partial f$}$\partial f$. We also introduce the convention that $\left|\partial f\right|$ always refers to the Euclidean norm,~i.e.
\begin{equation}\label{eq:def-partial}
\left|\partial f\right|^2 = \left(\partial_{\rho}f\right)^2 + \left(\partial_zf\right)^2.
\end{equation}

It is also sometimes useful to use the following notion of a renormalized gradient norm \index{Miscellaneous!$\hat{\partial}f$}$\left|\hat{\partial}f\right|$:
\begin{equation}\label{reder}
\left|\hat{\partial}f\right| \doteq \left|\nabla_{\mathbb{R}^4}\left(\xi_Nf\right)_{\mathbb{R}^4}\right| + \left|\nabla_{\mathbb{R}^4}\left(\xi_Sf\right)_{\mathbb{R}^4}\right| + \left|\nabla_{\mathbb{R}^3}\left(\left(1-\xi_N-\xi_S\right)f\right)_{\mathbb{R}^3}\right|.
\end{equation}
In an analogous fashion we may define $\left|\hat{\partial}^kf\right|$ for any $k \geq 1$.

We will use the notation \index{Miscellaneous!$\underline{\partial}f$}$\underline{\partial}$ to refer to vector of $(s,\chi)$ or $(s',\chi')$ derivatives,~i.e.
\[\underline{\partial}f|_{\overline{\mathscr{B}_N}} = \left(\partial_sf,\partial_{\chi}f\right).\]

The following lemma which concerns spherically symmetric functions $f : \mathbb{R}^2 \to \mathbb{R}$ will be used repeatedly in the paper, generally without explicit mention. 
\begin{lemma}\label{firstdecay}Let $f : \mathbb{R}^2 \to \mathbb{R}$ be a $C^2$ spherically symmetric function and $(\rho,\phi)$ denote polar coordinates on $\mathbb{R}^2$.\footnote{Here, by spherically symmetric, we mean that written in polar coordinates $(\rho,\phi)$, $f(\rho,\phi) = f(\rho)$.} Then, for any point $(x_0,y_0) \in \mathbb{R}^2$ we have
\[(x_0^2+y_0^2)^{-1/2}\partial_{\rho}f|_{\left\{(x,y) = (x_0,y_0)\right\}} = \partial_x^2f|_{\left\{(x,y) = (0,(x_0^2+y_0^2)^{1/2})\right\}}.\]
\end{lemma}
\begin{proof}This follows immediately from the straightforward computations
\[\partial_xf = \frac{x}{(x^2+y^2)^{1/2}}\partial_{\rho}f,\]
\[\partial_x^2f = \frac{y^2}{(x^2+y^2)^{3/2}}\partial_{\rho}f + \frac{x^2}{x^2+y^2}\partial_{\rho}^2f.\qedhere\]
\end{proof}

We also have the following relationship between the hatted and non-hatted H\"older spaces.
\begin{lemma}\label{lem:hat-non-hat-Holder}
For $\alpha \in [0,1]$ and $k \in \NN$, we have the continuous inclusion $C_{axi}^{k,\alpha}(\overline{\mathscr{B}})\subset \hat C_{axi}^{k,\alpha}(\overline{\mathscr{B}})$
\end{lemma}
\begin{proof}
Consider $f \in C^{k,\alpha}_{axi}(\overline{\mathscr{B}})$. Recall that this means that 
\[
F(u,v,z) \doteq f(\sqrt{u^{2}+v^{2}},z)
\]
is in $C^{k,\alpha}(\RR^{3})$ where $(u,v,z)$ are Cartesian coordinates on $\RR^{3}$. We claim that $f \in \hat C^{k,\alpha}_{axi}(\overline{\mathscr{B}})$. Clearly the the only issues are at the poles. To this end, define
\[
g(s,\chi) \doteq f\left(s\chi,\frac 12 (\chi^{2}-s^{2})+\gamma\right).
\]
We would like to show that
\[
G(x,y,\tilde x,\tilde y) \doteq g(\sqrt{x^{2}+y^{2}},\sqrt{\tilde x^{2}+\tilde y^{2}})
\]
is in\footnote{Really, we want an estimate for the $C^{k,\alpha}(\RR^{4})$ norm of $G$ by the $C^{k,\alpha}(\RR^{3})$ norm of $F$, but this will follow immediately from the argument below as well.} $C^{k,\alpha}(\RR^{4})$. Observe that 
\[
\sqrt{x^{2}+y^{2}}\sqrt{\tilde x^{2}+\tilde y^{2}} %= \sqrt{ x^{2}\tilde x^{2} + 2 xy\tilde x \tilde y + y^{2}\tilde y^{2} + x^{2}\tilde y^{2} - 2xy\tilde x\tilde y + y^{2}\tilde x^{2}} 
= \sqrt{(x\tilde x + y\tilde y)^{2} + (x\tilde y - y\tilde x)^{2}},
\]
so we may write
\[
G(x,y,\tilde x,\tilde y) = F\left(x\tilde x + y\tilde y, x\tilde y-y\tilde x, \frac 12 (\tilde x^{2}+\tilde y^{2}-x^{2}-y^{2})+\gamma \right)
\]
Thus, we have written $G$ as the composition of a $C^{k,\alpha}$ function with a smooth function, which implies the claim. 
\end{proof}

Finally, in order to discuss decay towards the asymptotically flat end, we introduce\index{Coordinates!$r$}
\begin{equation}\label{rdef}
r\left(\rho,z\right) \doteq \sqrt{1+\rho^2+z^2}.
\end{equation}

\subsection{Further analytic properties of the Kerr metric}\label{kerr}

Recall that $\gamma \doteq \sqrt{M^2-a^2}$. The following function will play an important role in our analysis\index{Metric Quantities!$h$}
\[h\left(\rho,z\right) \doteq \log\left(\sqrt{\rho^2+\left(z-\gamma\right)^2} - \left(z-\gamma\right)\right) + \log\left(\sqrt{\rho^2 + \left(z+\gamma\right)^2} + \left(z+\gamma\right)\right).\]
Observe that considered as a function on $\mathbb{R}^3$, $h$ is harmonic:
\[\rho^{-1}\partial_{\rho}\left(\rho\partial_{\rho}h\right) + \partial_z^2h = 0.\]
The importance of $h$ is that it precisely captures the singular behavior of $\log(X_{K})$ (recall that the coefficient of Kerr in the metric ansatz discussed above, see \eqref{kerrmetric2}) at the horizon. 

\begin{lemma}\label{lem:xK-regularity}Define a function\index{Metric Quantities!$x_{k}$}
\[x_K\left(\rho,z\right) \doteq \log\left(X_K\right) - h.\]
Then $x_K \in \hat{C}^{\infty}\left(\overline{\mathscr{B}}\right)$.
\end{lemma}
\begin{proof}
Recall that (see \eqref{kerrmetric2})
\[
X_{K} = \sin^{2}\theta \frac{\Pi}{\Sigma^{2}} = \Delta \sin^{2}\theta \frac{\Pi}{\Sigma^{2}\Delta}
\]
where 
\[
 \Delta  =  \tilde r^{2} - 2M \tilde r+a^{2}, \qquad \Sigma^{2} = \tilde r^{2}+a^{2}\cos^{2}\theta,\qquad  \Pi  \doteq (\tilde r^{2}+a^{2})^{2}-a^{2}\sin^{2}\theta \Delta
\]
and 
\[
\rho =  \sqrt{\Delta}\sin\theta, \qquad z=  \left(\tilde r-M\right)\cos\theta.
\]
define $\tilde r$ and $\theta$ implicitly in terms of $\rho$ and $z$. Elementary algebra implies that
\[
\tilde r = M + \frac{1}{\sqrt{2}} \sqrt{ \rho^{2} + z^{2} + \gamma^{2} + \sqrt{\rho^{4} + 2\rho^{2}(z^{2}+\gamma^{2}) + (z^{2}-\gamma^{2})^{2} } }.
\]
Using this expression, it is not hard to check that $x_{k}$ is smooth in $\{\rho>0\}$ with all derivatives bounded away from $\{\rho =0\}$. As such, we turn to smoothness of $x_{K}$ near the axis $\mathscr{A}$. Considering points near $\mathscr{A}_{N}$ i.e., points in $\overline{\mathscr{B}_{A}}\setminus\mathscr{A}$ with $z > \gamma$, we write
\begin{align*}
h  & = \log\left(\frac{\rho^{2}}{\sqrt{\rho^2+\left(z-\gamma\right)^2} + \left(z-\gamma\right)}\right) + \log\left(\sqrt{\rho^2 + \left(z+\gamma\right)^2} + \left(z+\gamma\right)\right)\\
&  = 2 \log \rho + \log\left(\frac{\sqrt{\rho^2 + \left(z+\gamma\right)^2} + \left(z+\gamma\right)}{\sqrt{\rho^2+\left(z-\gamma\right)^2} + \left(z-\gamma\right)}\right).
\end{align*}
The final term is in $\hat C^{\infty}$ near $\mathscr{A}_{N}$. A similar argument for the points near $\mathscr{A}_{S}$ gives
\[
h - 2\log \rho \in \hat C^{\infty}(\overline{\mathscr{B}_{A}}). 
\]
On the other hand, because $\Delta \sin^{2}\theta = \rho^{2}$, we see that 
\[
\log (X_{K}) -2 \log \rho = \log\left( \frac{\Pi}{\Sigma^{2}\Delta} \right)
\]
It is not hard to see that this quantity is in $\hat C^{\infty}(\overline{\mathscr{B}_{A}})$, using the above expression for $\tilde r$, as well as the fact that none of $\Delta,\Sigma^{2}$ or $\Pi$ vanishes near $\mathscr{A}$. A related (but more involved) computation can be found in the expansion of $X_{K}$ near the axis in the proof of \cite[Lemma 3.0.1]{HBH:geometric}. A similar (but easier) argument shows that both $\log(X_{K})$ and $h$ are regular near the horizon, so we easily check that $\log(X_{K}) - h \in \hat C^{\infty}(\overline{\mathscr{B}_{H}})$. 

Finally, we must consider near $p_{N}$ and $p_{S}$. We consider $p_{N}$, as the argument near $p_{S}$ is identical. Observe that
\[
h(s,\chi) =2 \log s + \log \left( \frac 12 \sqrt{  (\chi^{2}+s^{2})^{2} + 8 \gamma(\chi^{2}-s^{2}) + 4\gamma^{2}} + \frac 12 (\chi^{2}-s^{2} + 2 \gamma) \right).
\]
Thus, we see that $h - 2\log s \in \hat C^{\infty}(\overline{\mathscr{B}_{N}})$. On the other hand, from the above expression for $X_{K}$ we derive
\[
\log X_{K} = 2 \log s + \log  \frac{\chi^{2}}{\Delta} + \log \frac{\Pi}{\Sigma^{2}}.
\]
The final term is easily seen to be in $\hat C^{\infty}(\overline{\mathscr{B}_{N}})$. Note that both $\chi$ and $\Delta$ vanish at $p_{N}$, so we must take some care with the middle term. From the above expression for $\tilde r(\rho,z)$, it is not hard to show (see also \cite[Lemma 3.0.1]{HBH:geometric}) that for $s,\chi \in \overline{\mathscr{B}_{N}}$ sufficiently close to $p_{N}$,
\[
\tilde r = \tilde r_{+} + \chi^{2}\tilde r_{E} (s,\chi),
\]
for some smooth function $\tilde r_{E}(s,\chi) \in \hat C^{\infty}(\overline{\mathscr{B}_{N}})$. Because $\Delta = (\tilde r - \tilde r_{+})(\tilde r - \tilde r_{-})$, this shows that $\log  \frac{\chi^{2}}{\Delta} \in \hat C^{\infty}(\overline{\mathscr{B}_{N}})$. This completes the proof. 
%\newpage
%From the expression for $X_{K}$ in Section 3 of \cite{HBH:geometric}, it is clear that $x_{K}$ is smooth in $\{\rho > 0\}$. So it is sufficient to show smoothness in a neighborhood of $\mathscr{A}$, $\mathscr{H}$, $p_N$, and $p_S$.
%
%We start with $\mathscr{A}$. Let $U_{\mathscr{A}} \subset \overline{\mathscr{B}_{A}}$ be an open set containing the axis $\mathscr{A}$. It is easy to see from the formulas in Section 3 of~\cite{HBH:geometric} that $\log(X) - 2\log\rho \in \hat{C}^{\infty}\left(U_{\mathscr{A}}\right)$. Since we also have
%\[
%h_{\gamma} - 2\log\rho  = \log\left( \frac{\sqrt{\rho^{2} + (z-\gamma)^{2}}-(z-\gamma)}{\sqrt{\rho^{2}+z^{2}} - z} \right) + \log\left( \frac{\sqrt{\rho^{2} + (z+\gamma)^{2}}+(z+\gamma)}{\sqrt{\rho^{2}+z^{2}} + z}  \right) \in \hat{C}^{\infty}(U_{\mathscr{A}}),
%\]
%we conclude that $\log\left(X_K\right) - h_{\gamma} \in \hat{C}^{\infty}\left(U_{\mathscr{A}}\right)$.
%
%If $U_{\mathscr{H}} \subset \overline{\mathscr{B}_H}$ is an open set containing the horizon $\mathscr{H}$, it is easy to check that $h_{\gamma}\in \hat{C}^{\infty}(U_{\mathscr{H}})$. Furthermore, one finds that $X_K$ is smooth and does not vanish on $\mathscr{H}$. Thus $\log\left(X_K\right) - h_{\gamma} \in \hat{C}^{\infty}\left(U_{\mathscr{H}}\right)$.
%
%Near $p_N$, after switching to $(s,\chi)$ coordinates, the behavior of $X_K$ and $h_{\gamma}$ near $p_N$ is analogous to their behavior near the axis. We may easily check that $X_K - 2\log s,\ h_{\gamma}-2\log s \in \hat{C}^{\infty}(\overline{\mathscr{B}_{N}})$. An analogous argument near $p_S$ finishes the proof.
\end{proof}

It is useful to observe that the arguments in the previous proof (see also \cite[Lemma 3.0.1]{HBH:geometric}) prove
\begin{equation}\label{Xlikeh}
\frac{\left|\partial X_K\right|^2}{X_K^2} \sim \left|\partial h\right|^2,
\end{equation}
and that we have
\begin{equation}\label{hsing}
\partial h\big|_{\overline{\mathscr{B}_A}\cap \mathscr{A}} = \frac{2}{\rho}\partial_{\rho} + e_A,\qquad \underline{\partial} h\big|_{\overline{\mathscr{B}_N}} = \frac{2}{s}\underline{\partial}_s + e_N,\qquad \underline{\partial} h\big|_{\overline{\mathscr{B}_S}} = \frac{2}{s'}\underline{\partial}_{s'} + e_S,
\end{equation}
where the vector fields $e_A$, $e_N$, and $e_S$ all extend to smooth vector fields on $\overline{\mathscr{B}}$, and furthermore $\left|e_A\right| \leq Cr^{-2}$ as $r\to\infty$.

The \emph{twist} $\theta_{K}$ is a $1$-form defined on Kerr (see \cite[Section 2]{weinstein} and \cite[Section 4.2]{HBH:geometric}) that measures the failure of the axisymmetry to be hypersurface orthogonal. Let $\theta_K$ denote the twist form for the Kerr spacetime, which can be defined by
\[
\theta_{K} = \frac{X_{K}^{2}}{\rho} \left( \partial_{z}(X_{K}^{-1}W_{K}) d\rho - \partial_{\rho}(X_{K}^{-1}W_{K}) dz\right)
\]
Recall that (see \cite[p.\ 907]{weinstein}) $d\theta_K = 0$, so we may set $\theta_{K} =dY_{K}$ where $Y_{K}$ is uniquely determined by the requirement that $\lim_{r\to\infty}Y_K = 0$. This leads to the following lemma concerning various estimates for $Y_{K}$ (related computations can be found in Section 6.5.1 of \cite{chrusciellopes}).
\begin{lemma}\label{somestuffthatisuseful}
On $\overline{\mathscr{B}_{A}} \cup \overline{\mathscr{B}_{H}}$ we have the decay estimates
\[\frac{\rho}{X_K^2}\left|\partial Y_K\right| \leq Cr^{-4}, \qquad \left| \partial \left( X_{K}^{-1}\partial Y_{K}\right) \right| \leq C r^{-4}, \qquad \left| \partial^{2} \left( X_{K}^{-1}\partial Y_{K}\right) \right| \leq C r^{-5}.\]
In fact, we have the stronger bound for the $z$-component of the derivative of $Y_{K}$ on $\overline{\mathscr{B}_{A}}\cup\overline{\mathscr{B}_{H}}$,
\[\frac{1}{X_K^2}\left|\partial_{z} Y_K\right| \leq Cr^{-5}, \qquad |\partial (X_{K}^{-2}\partial_{z}Y_{K})| \leq C r^{-6}, \qquad |\partial^{2} (X_{K}^{-2}\partial_{z}Y_{K})| \leq C r^{-7} .\]
On $\overline{\mathscr{B}_{N}}$, we have the estimates
\[
|\underline{\partial} Y_{K}| \leq C s^{3}, \qquad |\underline\partial^{2}Y_{K}|\leq Cs^{2}, \qquad |\underline\partial^{3}Y_{K}| \leq C s
\]
Similar estimates hold on $\overline{\mathscr{B}_{S}}$. 
\end{lemma}
\begin{proof}
From the definition of $dY_{K} = \theta_{K}$ given above,
\begin{equation}\label{equaltosomething}
\frac{\rho}{|X_{K}|^{2}}|\partial Y_{K}| = \left|\partial\left( \frac{W_{K}}{X_{K}}\right)\right|.
\end{equation}
In Boyer--Lindquist coordinates (see \eqref{kerrmetric}),
\[
\frac{W_{K}}{X_{K}} = - \frac{2Ma\tilde r}{\Pi},
\]
so arguing as in Lemma \ref{lem:xK-regularity} (see also Section 3 of \cite{HBH:geometric}), it is easy to see that the right hand side of~\eqref{equaltosomething} is bounded by $Cr^{-4}$ for $r$ sufficiently large. Differentiating the definition of $\theta_{K}$ allows us to bound $ \left| \partial \left( X_{K}^{-1}\partial Y_{K}\right) \right| $ and $ \left| \partial^{2} \left( X_{K}^{-1}\partial Y_{K}\right) \right| $ by a similar argument.

Moreover, by the asymptotic falloff of $\tilde r$ proven in \cite[p.\ 13]{HBH:geometric}, 
\[
\left|\frac{\partial \tilde r}{\partial \rho} \right| \leq C\rho r^{-1}
\]
on $\overline{\mathscr{B}_{A}}$. Using this in the definition of $\theta_{K}$, we have that
\[\frac{\rho}{X_K^2}\left|\partial_{z} Y_K\right| = \left|\partial_{\rho}\left( \frac{W_{K}}{X_{K}}\right)\right| \leq C\rho r^{-5}, \]
as claimed. The differentiated estimates follow similarly. 

From the behavior of $X_{K}$ and $W_{K}+\Omega X_{K}$ near $p_{N}$ proven in \cite[p.\ 13]{HBH:geometric}, we may check that
\[
\left|\underline \partial\left( \frac{W_{K}}{X_{K}}\right)\right| \leq C \chi.
\]
Hence, a similar argument as above, in $\overline{\mathscr{B}_{N}}$ yields
\[
|\underline{\partial} Y_{K}| \leq C s^{3}
\]
and similarly for the higher derivatives. 
\end{proof}

\section{Two Fixed Point Lemmas}
In this section we will explain the functional analytic context for our existence argument. First of all, let's agree that for any Banach space $\mathcal{B}$, we will use \index{Miscellaneous!$B_{r}(\mathcal{B})$}$B_r\left(\mathcal{B}\right)$ to denote the ball of radius $r$ about the origin.

The following is a basic fixed point lemma.
\begin{lemma}\label{basic}Suppose we have Banach spaces $\mathcal{L}$, $\mathcal{Q}$, and $\mathcal{P}$, $\epsilon > 0$, and a map
\[\mathfrak{T} : B_{\epsilon}\left(\mathcal{L}\right) \times B_{\epsilon}\left(\mathcal{Q}\right) \times B_{\epsilon}\left(\mathcal{P}\right) \to B_{\epsilon}\left(\mathcal{L}\right).\]
Furthermore, suppose that
\begin{enumerate}
    \item There exists a constant $D > 0$ such that $\left(l,q,p\right) \in B_{\epsilon}\left(\mathcal{L}\right) \times B_{\epsilon}\left(\mathcal{Q}\right) \times B_{\epsilon}\left(\mathcal{P}\right)$ implies
        \[\left\Vert \mathfrak{T}\left(l,q,p\right)\right\Vert_{\mathcal{L}} \leq D\left[\left\Vert l\right\Vert_{\mathcal{L}}^2 + \left\Vert q\right\Vert_Q^2\right].\]
    \item There exists a constant $D > 0$ such that $\left(l_1,q_1,p_1\right),\left(l_2,q_2,p_2\right) \in B_{\epsilon}\left(\mathcal{L}\right) \times B_{\epsilon}\left(\mathcal{Q}\right) \times B_{\epsilon}\left(\mathcal{P}\right)$ implies
        \begin{align*}
        &\left\Vert \mathfrak{T}\left(l_1,q_1,p_1\right) - \mathfrak{T}\left(l_2,q_2,p_2\right)\right\Vert_{\mathcal{L}}\leq
        \\ \nonumber &\qquad  D\left[\left(\left\Vert l_1\right\Vert_{\mathcal{L}} + \left\Vert l_2\right\Vert_{\mathcal{L}}\right)\left\Vert l_1-l_2\right\Vert_{\mathcal{L}} +\left(\left\Vert q_1\right\Vert_{\mathcal{Q}} + \left\Vert q_2\right\Vert_{\mathcal{Q}}\right)\left\Vert q_1-q_2\right\Vert_{\mathcal{Q}} + \left\Vert p_1-p_2\right\Vert_{\mathcal{P}}\right].
        \end{align*}
\end{enumerate}
Then, after possibly shrinking $\epsilon$, there exists a ``solution map'' $\mathfrak{S} : B_{\epsilon}\left(\mathcal{Q}\right) \times B_{\epsilon}\left(\mathcal{P}\right) \to B_{\epsilon}\left(\mathcal{L}\right)$ such that
\begin{enumerate}
    \item $\left(q,p\right) \in B_{\epsilon}\left(\mathcal{Q}\right) \times B_{\epsilon}\left(\mathcal{P}\right)$ implies
        \[\mathfrak{T}\left(\mathfrak{S}\left(q,p\right),q,p\right) = \mathfrak{S}\left(q,p\right).\]
    \item There exists a constant $D > 0$ such that $\left(q,p\right) \in B_{\epsilon}\left(\mathcal{Q}\right) \times B_{\epsilon}\left(\mathcal{P}\right)$ implies
        \[\left\Vert \mathfrak{S}\left(q,p\right)\right\Vert_{\mathcal{L}}\leq D\left\Vert q\right\Vert_{\mathcal{Q}}^2.\]
    \item There exists a constant $D > 0$ such that $\left(q_1,p_1\right),\left(q_2,p_2\right) \in B_{\epsilon}\left(\mathcal{Q}\right) \times B_{\epsilon}\left(\mathcal{P}\right)$ implies
        \[\left\Vert \mathfrak{S}\left(q_1,p_1\right)-\mathfrak{S}\left(q_2,p_2\right)\right\Vert_{\mathcal{L}}\leq D\left[\left(\left\Vert q_1\right\Vert_{\mathcal{Q}} + \left\Vert q_2\right\Vert_{\mathcal{Q}}\right)\left\Vert q_1-q_2\right\Vert_{\mathcal{Q}} + \left\Vert p_1-p_2\right\Vert_{\mathcal{P}}\right].\]
\end{enumerate}
\end{lemma}
\begin{proof}This is a standard fixed point argument. For $\epsilon$ sufficiently small, the map $\mathfrak{T}(\cdot,p,q): B_{\epsilon}(\mathcal{L}) \to B_{\epsilon}(\mathcal{L})$ is easily seen to be a contraction map and thus has a fixed point $\mathfrak{S}(p,q)$. The asserted properties of $\mathfrak{S}$ are immediate consequences of the corresponding properties of $\mathfrak{T}$. 
\end{proof}
\begin{remark}\label{paramcontdepend}We refer to the Banach spaces $\mathcal{Q}$ and $\mathcal{P}$ as spaces of ``parameters.'' We refer specifically to $\mathcal{Q}$ as the space of ``genuinely nonlinear parameters.'' We refer to final two properties of the map $\mathfrak{S}$ as statements of ``continuous dependence on the parameters.''
\end{remark}

We will often apply this fixed point lemma in the context of linearizing a nonlinear equation we wish to solve. The basic set-up is as follows: We have Banach spaces $\mathcal{L}$, $\mathcal{Q}$, and $\tilde{\mathcal{L}}$ and, for some $\epsilon > 0$, a map
\[\mathfrak{E} : B_{\epsilon}\left(\mathcal{L}\right)\times B_{\epsilon}\left(\mathcal{Q}\right)\to \tilde{\mathcal{L}}.\]
As in Lemma~\ref{basic} the Banach spaces $\mathcal{Q}$ plays the role of a set of ``genuinely nonlinear parameters.'' For $\epsilon > 0$ sufficiently small, we are interested in finding a ``solution map''
\[\mathfrak{S} : B_{\epsilon}\left(\mathcal{Q}\right)  \to B_{\epsilon}\left(\mathcal{L}\right)\]
such that
\[\mathfrak{E}\left(\mathfrak{S}\left(q\right),q\right) = 0.\]

The following lemma provides a general framework for applying Lemma~\ref{basic} towards this problem.
\begin{lemma}\label{fixit}Suppose we have a linear operator $L : \mathcal{L} \to \tilde{\mathcal{L}}$ an operator $N : \mathcal{L}\times \mathcal{Q} \to \tilde{\mathcal{L}}$, and an operator $\mathcal{Q} : B_{\epsilon}\left(\mathcal{L}\right)\times B_{\epsilon}\left(\mathcal{Q}\right) \to \tilde{\mathcal{L}}$ for some $\epsilon>0$, such that
\begin{enumerate}
\item For all $\left(l,q\right) \in B_{\epsilon}\left(\mathcal{L}\right)\times B_{\epsilon}\left(\mathcal{Q}\right)$, we have
\[\mathfrak{E}\left(l,q\right) = L\left(l\right) - N\left(l,q\right).\]
\item We have a Banach space $\mathcal{N} \subset \tilde{\mathcal{L}}$ and a bounded map $L^{-1} : \mathcal{N} \to \mathcal{L}$ such that $H \in \mathcal{N}$ implies
\[L\left(L^{-1}\left(H\right)\right) = H.\]
\item We have $N\left(B_{\epsilon}\left(\mathcal{L}\right)\times B_{\epsilon}\left(\mathcal{Q}\right)\right) \subset \mathcal{N}$ and there exists a constant $D > 0$ such that $(l,q) \in B_{\epsilon}\left(\mathcal{L}\right)\times B_{\epsilon}\left(\mathcal{Q}\right)$ implies
    \[\left\Vert N\left(l,q\right)\right\Vert_{\mathcal{N}} \leq D\left[\left\Vert l\right\Vert_{\mathcal{L}}^2 + \left\Vert q\right\Vert_{\mathcal{Q}}^2\right].\]
\item There exists a constant $D > 0$ such that $(l_1,q_1),\left(l_2,q_2\right) \in B_{\epsilon}\left(\mathcal{L}\right)\times B_{\epsilon}\left(\mathcal{Q}\right)$ implies
    \[\left\Vert N\left(l_1,q_1\right)- N\left(l_2,q_2\right)\right\Vert_{\mathcal{N}} \leq D\left[\left(\left\Vert l_1\right\Vert + \left\Vert l_2\right\Vert\right)\left\Vert l_1-l_2\right\Vert_{\mathcal{L}} + \left(\left\Vert q_1\right\Vert + \left\Vert q_2\right\Vert\right)\left\Vert q_1-q_2\right\Vert_{\mathcal{Q}}\right].\]
\end{enumerate}
Then, after possibly shrinking $\epsilon$, there exists a ``solution map'' $\mathfrak{S} : B_{\epsilon}\left(\mathcal{Q}\right) \to B_{\epsilon}\left(\mathcal{L}\right)$ such that
\begin{enumerate}
    \item $q \in B_{\epsilon}\left(\mathcal{Q}\right)$ implies
        \[\mathfrak{E}\left(\mathfrak{S}\left(q\right),q\right) = 0.\]
    \item There exists a constant $D > 0$ such that $q \in B_{\epsilon}\left(\mathcal{Q}\right)$ implies
        \[\left\Vert \mathfrak{S}\left(q\right)\right\Vert_{\mathcal{L}}\leq D\left\Vert q\right\Vert_{\mathcal{Q}}^2.\]
    \item There exists a constant $D > 0$ such that $q_1,q_2 \in B_{\epsilon}\left(\mathcal{Q}\right)$ implies
        \[\left\Vert \mathfrak{S}\left(q_1\right)-\mathfrak{S}\left(q_2\right)\right\Vert_{\mathcal{L}}\leq D\left(\left\Vert q_1\right\Vert_{\mathcal{Q}} + \left\Vert q_2\right\Vert_{\mathcal{Q}}\right)\left\Vert q_1-q_2\right\Vert_{\mathcal{Q}}.\]
\end{enumerate}
\end{lemma}
\begin{proof}This is standard: one defines a map $\mathfrak{T} : B_{\epsilon}\left(\mathcal{L}\right) \times B_{\epsilon}\left(\mathcal{Q}\right) \to B_{\epsilon}\left(\mathcal{L}\right)$ by
\[\mathfrak{T}\left(l,q\right) \doteq L^{-1}\left(N\left(l,q\right)\right),\]
and checks that Lemma~\ref{basic} applies.
\end{proof}
\begin{remark}\label{contnonlindep} We will refer to the last two properties of the solutions as ``continuous nonlinear dependence on parameters.''
\end{remark}

\section{Metric quantities and Equations}
Recall that our metric ansatz \eqref{eq:metric-ansatz} takes the form 
\[
g = -V dt^{2} + 2W dtd\phi + X d\phi^{2} + e^{2\lambda}(d\rho^{2}+dz^{2}). 
\]
for \index{Metric Quantities!$V$}\index{Metric Quantities!$W$}\index{Metric Quantities!$X$}\index{Metric Quantities!$\lambda$}$V,W,X,\lambda$ functions $\mathscr{B}\to\RR$. As discussed in Section \ref{sec:overview}, for the study of such metrics solving the vacuum Einstein equations it is convenient to consider ``metric data'' which is not exactly $V,W,X$ and $\lambda$. In particular, motivated by the Carter--Robinson theory, we introduce the ``twist $1$-form''\index{Metric Quantities!$\theta$}
\[
\theta \doteq 2i_{\Phi}\left(*\nabla\Phi_{\flat}\right).
\]
Moreover, instead of $V$, we will consider the quantity \index{Metric Quantities!$\sigma$}$\sigma \doteq \sqrt{XV+W^{2}}$. We emphasize that unlike the vacuum case, $\sigma$ will not be equal to $\rho$. 

As such, we will refer to the data $(X,W,\theta,\sigma,\lambda)$ as ``metric data.'' It is clear that given metric data, we can solve for $V$ and thus reconstruct the metric $g$. Moreover, the reason for the choice of this exact form of the metric data is that the Einstein--Klein--Gordon equations with a scalar field of the form\index{Metric Quantities!$\Psi$}\index{Metric Quantities!$\omega$}\index{Metric Quantities!$m$}
\begin{equation}\label{formscalar-disp2}
\Psi(t,\phi,\rho,z) = e^{-it\omega} e^{im\phi} \psi(\rho,z)
\end{equation}
are equivalent to a system of equations for the metric data and $\psi$ (which we will call the ``reduced scalar field'') which will turn out to be possible to analyze. In \cite[Theorem 1.3]{HBH:geometric} we prove the following result\footnote{The analogous analysis for the vacuum equations was carried out in \cite{weinstein}; in particular, \cite{weinstein} contains the observation that $\lambda$ satisfies a first order equation in addition to the Liouville equation. This will be useful in our analysis of $\lambda$ below.}
\begin{theorem}\label{thm:HBH-geometric-main}
Suppose that $(\mathcal{M},g)$ solves the Einstein--Klein--Gordon equations with a scalar field $\Psi$ of the form \eqref{formscalar-disp2}. Then the metric data and reduced scalar field satisfy the following equations on $\mathscr{B} = \{(\rho,z) \in \mathbb{R}^{2} : \rho > 0 \}$:
\begin{enumerate}
\item $X$ satisfies
\[
 \sigma^{-1}\partial_{\rho}(\sigma\partial_{\rho}X) + \sigma^{-1}\partial_{z}(\sigma\partial_{z}X) = -e^{2\lambda}(2 m^{2} + \mu^{2} X) \psi^{2} + \frac{(\partial_{\rho}X)^{2}+(\partial_{z}X)^{2} -\theta_{\rho}^{2}-\theta_{z}^{2}}{X}.
\]
\item $W$ satisfies
\[
\partial_{\rho}(X^{-1}W) d\rho + \partial_{z} (X^{-1}W) dz = \frac{\sigma}{X^{2}}[\theta_{\rho}dz - \theta_{z} d\rho].
\]
\item $\theta$ satisfies
\[
d\theta =(\partial_{\rho}\theta_{z} - \partial_{z}\theta_{\rho}) d\rho\wedge dz = 2\sigma^{-1}e^{2\lambda} \left( X\omega m + W m^{2}\right) \psi^{2} d\rho\wedge dz,
\]
as well as
\[
\sigma^{-1}\partial_{\rho}(\sigma \theta_{\rho})+\sigma^{-1}\partial_{z}(\sigma \theta_{z}) = \frac{2\theta_{\rho}\partial_{\rho} X + 2 \theta_{z}\partial_{z}X}{X}.
\]
\item $\sigma$ satisfies
\[
X^{-1}e^{-2\lambda}\sigma\left(\partial_{\rho}^2\sigma + \partial_z^2\sigma\right) = \left(\left(\omega + X^{-1}Wm\right)^2 - \sigma^2\left(\frac{\mu^2}{X} + \frac{m^2}{X^2}\right)\right)\psi^2.
\]
\item $\lambda$ satisfies the following equations at the points where $|\partial\sigma|\not = 0$
\[
\partial_{\rho}\lambda = \alpha_{\rho} - \frac 12 \partial_{\rho}\log X, \qquad \partial_{z}\lambda = \alpha_{z} - \frac 12 \partial_{z} \log X
\]
where
\begin{align*}
& \left( (\partial_{\rho}\sigma)^{2}+(\partial_{z}\sigma)^{2}\right)  \alpha_{\rho}\\
& = \frac 12 (\partial_{\rho}\sigma) \sigma \left( (\partial_{\rho}\psi)^{2} - (\partial_{z}\psi)^{2} + \frac 12 X^{-2} \left[(\partial_{\rho}X)^{2} - (\partial_{z}X)^{2} + (\theta_{\rho})^{2} - (\theta_{z})^{2} \right] \right)\\
& + \partial_{\rho}\sigma (\partial^{2}_{\rho}\sigma - \partial^{2}_{z}\sigma) + \partial_{z}\sigma (\partial^{2}_{\rho,z}\sigma)\\
& + (\partial_{z}\sigma)\sigma \left[ (\partial_{\rho}\psi)(\partial_{z}\psi) + \frac 12 X^{-2}\left((\partial_{\rho}X)(\partial_{z}X) + (\theta_{\rho})(\theta_{z}) \right) \right],
\end{align*}
and
\begin{align*}
& \left( (\partial_{\rho}\sigma)^{2}+(\partial_{z}\sigma)^{2}\right)  \alpha_{z}\\
& = -\frac 12 (\partial_{z}\sigma)\sigma \left(  (\partial_{\rho}\psi)^{2} - (\partial_{z}\psi)^{2} + \frac 12 X^{-2} \left[(\partial_{\rho}X)^{2} - (\partial_{z}X)^{2} + (\theta_{\rho})^{2} - (\theta_{z})^{2} \right]  \right)\\
& - \partial_{z}\sigma (\partial^{2}_{\rho}\sigma - \partial^{2}_{z}\sigma) + \partial_{\rho}\sigma (\partial^{2}_{\rho,z}\sigma)\\
& + (\partial_{\rho}\sigma)\sigma \left[ (\partial_{\rho}\psi)(\partial_{z}\psi) + \frac 12 X^{-2} \left( (\partial_{\rho}X)(\partial_{z}X) + (\theta_{\rho})(\theta_{z}) \right)\right].
\end{align*}
Independent of the behavior of $\sigma$, $\lambda$ satisfies
\begin{align*}
2\partial^{2}_{\rho}\lambda + 2\partial^{2}_{z}\lambda & = - \partial^{2}_{\rho}\log X - \partial^{2}_{z}\log X +  \sigma^{-1}(\partial_{\rho}^{2}\sigma + \partial^{2}_{z}\sigma)\\
& - e^{2\lambda} \mu^{2} \psi^{2} - (\partial_{\rho}\psi)^{2} - (\partial_{z}\psi)^{2} - X^{-1}\left( 2m^{2} + \mu^{2} X\right)e^{2\lambda}\psi^2\\
& - \frac 12 X^{-2}\left[ (\partial_{\rho}X)^{2} + (\partial_{z}X)^{2} + (\theta_{\rho})^{2} + (\theta_{z})^{2} \right].
\end{align*}
\item $\psi$ satisfies
\[
\sigma^{-1}\partial_{\rho}\left(\sigma\partial_{\rho}\psi\right) + \sigma^{-1}\partial_z\left(\sigma\partial_z\psi\right) + e^{2\lambda}\sigma^{-2}X^{-1}\left(X\omega+Wm\right)^2\psi - e^{2\lambda}m^2X^{-1}\psi- e^{2\lambda}\mu^2 \psi = 0.
\]
\end{enumerate}
Conversely, if the metric data and reduced scalar field solves each of these equations, and $|\partial\sigma|\not = 0$ on $\mathscr{B}$, then we may recover the metric and scalar field $(\mathcal{M},g,\Psi)$, solving the Einstein--Klein--Gordon equations.
\end{theorem}

Of course, here we are primarily interested in the final conclusion. We emphasize that because the metrics we will construct will be a small perturbation of an appropriate Kerr metric (where we have that $\sigma = \rho$), the condition that $|\partial\sigma| \not = 0$ will be automatically satisfied. 

Finally, we emphasize that simply solving the above equations on $\mathscr{B}$ will be far from sufficient to prove Theorem \ref{timeperiodicsoln}, in particular because of the need to show that the metric extends to the axis and horizon, where the coordinates $(t,\phi,\rho,z)$ break down. Note, however, that once each of the metric data solves its relevant equations and we show that the metric extends to the axis and horizon, then the entire space-time will solve the Einstein--Klein--Gordon equations, since the above result guarantees that this holds on a dense set (namely, everything besides the axis and horizon). 

\section{Renormalized Quantities, Equations, and Norms}

In this section we will write the reduced Einstein--Klein--Gordon equations from the previous section in a form to which we can potentially apply Lemma~\ref{basic} or~\ref{fixit}. Recall that the basic unknowns are the set of ``metric data'' $(X,W,\theta,\sigma,\lambda)$ and the ``reduced scalar field'' $\psi$. We begin with a discussion of a gauge choice for the ``twist,'' $\theta$.

\subsection{The Equation for the Ernst Potential $Y$ and the $1$-form $B$}
In the vacuum case, the twist $\theta$ satisfies $d\theta =0$, allowing us to introduce the ``Ernst potential'' \index{Metric Quantities!$Y$}$Y$ by setting $\theta = dY$. It turns out that even when we are not studying the vacuum Einstein equations, it is useful to split the twist $1$-form $\theta$ into an Ernst potential piece $dY$ and another $1$-form $B$ which measures how far $\theta$ is from being closed. This choice of $B$ represents a gauge ambiguity in our problem. We will want $B$ to statisfy \eqref{dthatB}, but are free to choose appropriate boundary conditions. In general, there is a trade-off between decay at the axis and decay at infinity. It turns out that the following definition suffices for our purposes.

\begin{definition}\label{defB} \index{Metric Quantities!$B$} Let $(X,W,\theta,\sigma,\lambda)$ be a set of unknown metric data and $\psi$ be an unknown reduced scalar field. We define three $1$-forms $B^{(N)}$, $B^{(S)}$, and $B^{(A)}$ by the following formulas (interpreted in $(s,\chi)$ coordinates, $(s',\chi')$ coordinates, and $(\rho,z)$ coordinates respectively):
\begin{align*}
B^{(N)}_{\chi}(0,\chi) & \doteq 0,\\
B^{(N)}_{\chi}(s,0) & \doteq 0,\\
\partial_sB^{(N)}_{\chi}\left(s,\chi\right) &\doteq 2\left(\chi^2+s^2\right)\xi_N\sigma^{-1}e^{2\lambda}\left(X\omega m + Wm^2\right)\psi^2,
\\ \nonumber B^{(N)}_s\left(s,\chi\right) &\doteq 0,
\\  B^{(S)}_{\chi'}(0,\chi') & \doteq 0
\\  B^{(S)}_{\chi'}(s',0) & \doteq 0
\\ \nonumber \partial_{s'}B^{(S)}_{\chi'}\left(s',\chi'\right) &\doteq 2\left((\chi')^2+(s')^2\right)\xi_S\sigma^{-1}e^{2\lambda}\left(X\omega m + Wm^2\right)\psi^2,
\\ \nonumber B^{(S)}_{s'}\left(s',\chi'\right) &\doteq 0,
\\ B^{(A)}_{z}(0,z) & \doteq 0,
\\ \nonumber \partial_{\rho}B^{(A)}_z\left(\rho,z\right) &\doteq 2\left(1-\xi_N - \xi_S\right)\sigma^{-1}e^{2\lambda}\left(X\omega m + Wm^2\right)\psi^2
\\ \nonumber &\qquad \qquad -\left(\partial_{\rho}\xi_N\right)B^{(N)}_z + (\partial_z\xi_N)B_{\rho}^{(N)}- \left(\partial_{\rho}\xi_N\right)B^{(S)}_z + (\partial_z\xi_N)B_{\rho}^{(S)},
\\ \nonumber B^{(A)}_{\rho}(\rho,z) &\doteq 0.
\end{align*}
\begin{remark}\label{howdefB} Let us emphasize that the above definition should be thought of as defining a $1$-form in the coordinates $(s,\chi)$ and $(s',\chi')$. We then define functions functions $B^{(N)}_{\rho}$, $B^{(N)}_z$, $B^{(S)}_{\rho}$, and $B^{(S)}_z$ as the coefficients of the form after changing to $(\rho,z)$ coordinates
\[B^{(N)}_{\rho}d\rho + B^{(N)}_zdz = B^{(N)}_{\chi}d\chi,\qquad B^{(S)}_{\rho}d\rho + B^{(S)}_zdz = B^{(S)}_{\chi'}d\chi'.\]
This remark should be used in interpreting the expression for $\partial_{\rho}B^{(A)}_{z}$. Note that, for example, we have chosen to set $B^{(N)}_{s}=0$, which is why there is only a $d\chi$ component. 
\end{remark}

Then we define the $1$-form $B$ by
\begin{equation}\label{thisisB}
B \doteq B^{(A)} + \xi_NB^{(N)} + \xi_SB^{(S)}.
\end{equation}
\end{definition}
\begin{remark}A straightforward calculation yields
\begin{equation}\label{dthatB}
dB = 2\sigma^{-1}e^{2\lambda}\left(X\omega m + Wm^2\right)\psi^2d\rho\wedge dz% = d\theta.
\end{equation}
Thus, using the formula for $d\theta$ in Theorem \ref{thm:HBH-geometric-main}, $\theta-B$ is a closed $1$-form.
\end{remark}
\begin{remark}\label{Btriple} Note that Definition~\ref{defB} is really an equation for the the three functions $B^{(N)}_{\chi}$, $B^{(S)}_{\chi'}$, and $B^{(A)}_z$. In our fixed point argument, it is these three functions which we will work with as the actual unknowns. Given the triple $\left(B^{(N)}_{\chi},B^{(S)}_{\chi'},B^{(A)}_z\right)$, we will always consider $B$ to automatically be defined by~(\ref{thisisB}).
\end{remark}

\begin{definition}The Ernst potential $Y : \mathscr{B} \to \mathbb{R}$ is defined (up to a constant) by the formula\index{Metric Quantities!$Y$}
\[dY = \theta - B.\]
\end{definition}

\subsection{The Renormalized Unknowns and the Key Parameters}
An important point will be that we do not work directly at the level of the metric quantities $X$, $Y$, $W$, and $\sigma$, but instead apply certain ``renormalizations.'' Our definitions are chosen so as to subtract off the leading order singular behavior. Related renormalizations have been considered in \cite{weinstein,weinstein2,chrusciellopes,ionkla} and other places.

We begin by introducing the set of unknowns and parameters and explaining how they correspond to the metric quantities $(X,W,\theta,\sigma,\lambda)$ introduced in the previous section.

\begin{definition}The ``renormalized unknowns'' refer to the following set of functions and a $1$-form, which are all assumed to be continuous on $\overline{\mathscr{B}}$ and satisfy the following definitions on $\mathscr{B}$:
\begin{enumerate}
    \item The first renormalized unknown is \index{Metric Quantities!$\mathring{\sigma}$} $\mathring{\sigma}$, which is related to the metric quantity $\sigma$ by
        \[\mathring{\sigma} \doteq \frac{\sigma-\rho}{\rho}.\]
    \item The second renormalized unknown is the triple \index{Metric Quantities!$B^{(N)}$}\index{Metric Quantities!$B^{(S)}$}\index{Metric Quantities!$B^{(A)}$}   $\left(B^{(N)}_{\chi},B^{(S)}_{\chi'},B^{(A)}_z\right)$. (The reader may want to recall Remark~\ref{Btriple}).
    \item The third renormalized unknown is the pair of functions \index{Metric Quantities!$\mathring X$} $\mathring{X}$ and \index{Metric Quantities!$\mathring Y$} $\mathring{Y}$, which are related to the metric quantities $X$ and $Y$ by
        \[\mathring{X} \doteq X_K^{-1}\left(X - X_K\right),\qquad \mathring{Y} \doteq X_K^{-1}\left(Y-Y_K\right).\]
    \item The fourth renormalized unknown is the function \index{Metric Quantities!$\mathring \Theta$} $\mathring{\Theta}$, which is related to the metric quantity $W$ by
        \[\mathring{\Theta} \doteq X^{-1}W - X_K^{-1}W_K.\]
    \item The fifth renormalized unknown is the reduced scalar field \index{Metric Quantities!$\psi$} $\psi$ and the Klein--Gordon mass \index{Metric Quantities!$\mu^{2}$} $\mathring{\mu}^2$. The constant $\mathring{\mu}^2$ is related to $\mu^2$ by
        \[\mu^2 = \mathring{\mu}^2 + \mu_K^2,\]
        where \index{Metric Quantities!$\mu_{K}$} $\mu_K$ is a constant which will be fixed in Remark~\ref{defmuk}.
   \item The sixth renormalized unknown is the function $\mathring{\lambda}$, which is related to the metric quantity $\lambda$ by
        \[\mathring{\lambda} \doteq \lambda - \lambda_K.\]
\end{enumerate}
\end{definition}
\begin{remark}The order we have presented the renormalized unknowns reflects the order in which we will treat their equations. This order will be important for various reasons. In particular, it is important to solve $\mathring Y$ (and $B$)'s equation before treating $\Theta$, since $\mathring Y$'s equation provides a necessary compatibility condition to integrate $\Theta$'s equation (see Proposition \ref{fixth}). Moreover, it is important to solve $\Theta$ before considering $\psi$, since this will allow us to show $\mathring\omega$ is constant (see Corollary \ref{omconst}). 

The choice of the order of the other equations is also important, but somewhat more subtle. The key reason for the given order is that we will (for most of the unknowns) have to verify the quadratic estimates needed for Lemma \ref{fixit}; the given order allows for this. For example, when solving for $\mathring \sigma$, the nonlinear term will \underline{all} be multiplied by $\psi^{2}$ which will automatically guarantee such estimates (see the proof of Proposition \ref{fixsig}). This then implies (via Lemma \ref{fixit}) that $\mathring \sigma$ is \underline{quadratically} bounded by the remaining unknowns. This is useful, for example, when solving for $\Theta$ (see Proposition \ref{fixth}). 
\end{remark}
\begin{remark}\label{unrenormalize}Note that given renormalized unknowns \[\left(\mathring{\sigma},\left(B^{(N)}_{\chi},B^{(S)}_{\chi'},B^{(A)}_z\right),\left(\mathring{X},\mathring{Y}\right),\mathring{\Theta},\psi,\mathring{\mu}^2,\mathring{\lambda}\right),\] we can easily compute the corresponding quantities
\begin{equation}\label{unrenormalize2}
\left(\sigma,B,\left(X,Y\right),W,\psi,\mu^2,\lambda \right),
\end{equation}
from which we obtain the metric data $(X,W,\theta,\sigma,\lambda)$. In what follows, given a set of renormalized unknowns, we will always consider the quantities~(\ref{unrenormalize2}) to be automatically defined by reversing the renormalization process.
\end{remark}

Finally, we have four important parameters. First there is a small parameter \index{Miscellaneous!$\delta$} $0 \leq \delta \ll 1$ which will eventually be defined by the requirement that
\[\int_{\mathscr{B}}\psi^2e^{2\lambda}\sigma\, d\rho\, dz = \delta^2.\]
Note that $\delta = 0$ implies that the scalar field vanishes\footnote{Note that $\sigma$ cannot vanish away from the horizon/axis as at such points if $\sigma=0$, it is easy to see that the metric $g$ cannot be Lorentzian}. We will eventually consider $\delta$ as a ``bifurcation parameter.''

Second, we have a parameter \index{Metric Quantities!$m$} $m \in \mathbb{Z}_{\neq 0}$ which determines the azimuthal number of the scalar field. We will fix the value of $m$ in Remark~\ref{fixmalpha0}.

Lastly, we have the parameters \index{Metric Quantities!$a$}\index{Metric Quantities!$M$} $(a,M)$ which determine the Kerr metric quantities. Recall that we have picked any choice of  $(a,M)$ satisfying $0 < |a| < M$ and then consider $(a,M)$ fixed until Section~\ref{arrange}. As remarked in the introduction, the fact that $|a| > 0$ will be necessary for our construction. Until Section~\ref{arrange} the dependence of various constants on $(a,M)$ will be suppressed.

\subsection{The Renormalized Equations}
In this section we will explicitly define the equations associated to the renormalized unknowns. In certain cases we will also define the ``linear'' and ``nonlinear'' parts of the equations $L_i$ and $N_i$.

Before we dive into the equations, it is useful to introduce the following function \index{Metric Quantities!$\mathring \omega$} $\mathring{\omega}$ as follows. We first define $\mathring{\omega}$ on the horizon $\mathscr{H}$ by
\begin{equation}\label{mathringomega}
 \mathring{\omega}\left(0,z\right)|_{\mathscr{H}} = -m\left(\mathring{\Theta} + X_K^{-1}W_K\right)|_{\mathscr{H}}.
 \end{equation}
If $\mathring{\Theta}$ is in $\hat C^{2,\alpha}$ near $\mathscr{H}$ (which we will arrange), then $\mathring{\omega}$ will be $C^{2,\alpha}$ on $\mathscr{H}$. Now let $E : C^2[-\gamma,\gamma] \to C^2\left(\mathbb{R}\right)$ be a linear bounded ``extension map,'' which sends the constant function to the constant function,\footnote{One can easily build such an extension operator from the half-line extension operator $\tilde E : C^2(-\infty,0) \to C^2(-\infty,\infty)$ defined, for $x > 0$, by  $\tilde Ef(x) = 6f(-x) - 32f(-x/2) + 27f(-x/3)$.} and use this to extend $\mathring{\omega}$ to a function $\mathring{\omega}\left(0,z\right) : \mathbb{R}\to\mathbb{R}$. Lastly, we extend $\mathring{\omega}$ to all $\overline{\mathscr{B}}$ by letting it be independent of $\rho$.

%We will replace the constant $\omega$ which appears in the equations from \cite{HBH:geometric} with the function $\mathring{\omega}$. Later, we will show that $\mathring{\omega}$ is in fact a constant (See Corollary~\ref{omconst}).

We will see that if $\mathring \Theta$ solves an appropriate equation (see \eqref{thetaringeqnz}), then $\mathring \omega$ is indeed a constant. However, until we have arranged for $\mathring\Theta$ to solve its equation, it is useful to allow for a non-constant $\mathring \omega$. As such, in the preliminary portion of our fixed point argument, we will \underline{change} the equations in Theorem \ref{thm:HBH-geometric-main} by writing $\mathring\omega$ instead of $\omega$. Once we have solved this system, we will conclude \emph{a posteriori} that $\mathring\omega$ was \underline{constant} and thus we have indeed solved the equations considered in Theorem \ref{thm:HBH-geometric-main}. 

\subsubsection{The Equation for $\mathring{\sigma}$}
We start with the equation for $\mathring{\sigma}$

If we write the flat metric on $\mathbb{R}^4$ as $d\rho^2 + \rho^2d\mathbb{S}^2 + dz^2$, then, the equation for $\sigma$ from Theorem 1.3 in \cite{HBH:geometric} can be rewritten as
\begin{equation}\label{sigmaeqn2}
\Delta_{\mathbb{R}^4}\left(\frac{\sigma}{\rho}\right) = \rho^{-1}\sigma^{-1}Xe^{2\lambda}\left(\left(\mathring{\omega}+X^{-1}Wm\right)^2 - \sigma^2\left(\frac{\mu^2}{X} + \frac{m^2}{X^2}\right)\right)\psi^2.
\end{equation}
\begin{remark}
The reason this equation has a four dimensional Laplacian is simply because when going from the $\partial^{2}_{\rho}\sigma + \partial^{2}_{z}\sigma$ expression in Theorem 1.3 in \cite{HBH:geometric} to a PDE in terms of $\frac{\sigma}{\rho}$, we use the expression
\[
\partial^{2}_{\rho} \sigma + \partial^{2}_{z}\sigma = \rho \left( \partial^{2}_{\rho} \left( \frac{\sigma}{\rho}\right)+ \frac2 \rho \partial_{\rho} \left( \frac{\sigma}{\rho}\right)   + \partial^{2}_{z} \left( \frac{\sigma}{\rho}\right)  \right),
\]
which is precisely the Laplacian on $\mathbb{R}^{4}$ in the above described coordinates. We will repeatedly use this argument below without comment. We emphasize that the four dimensionality here is an artifact of the equations, rather than a deep physical or geometric manifestation. A similar idea is used in \cite{weinstein,weinstein2} among other places. 
\end{remark}

In particular, we obtain (keeping Remark~\ref{unrenormalize} in mind)
\begin{equation}\label{renormalsigmaeqn}
\Delta_{\mathbb{R}^4}\mathring{\sigma} = \rho^{-1}\sigma^{-1}Xe^{2\lambda}\left(\left(\mathring{\omega}+X^{-1}Wm\right)^2 - \sigma^2\left(\frac{\mu^2}{X} + \frac{m^2}{X^2}\right)\right)\psi^2.
\end{equation}
The left hand side of the equation is defined to be $L_{\sigma}$, the right hand side of the equation is defined to $N_{\sigma}$.
\subsubsection{The Equations for $\left(B^{(N)}_{\chi},B^{(S)}_{\chi'},B^{(A)}_z\right)$}
Definition~\ref{defB} yields the following equations for $\left(B^{(N)}_{\chi},B^{(S)}_{\chi'},B^{(A)}_z\right)$.
\begin{align}\label{Btripleeqn}
\nonumber \partial_sB^{(N)}_{\chi}\left(s,\chi\right) &= 2\left(\chi^2+s^2\right)\xi_N\sigma^{-1}e^{2\lambda}\left(X\mathring{\omega} m + Wm^2\right)\psi^2,
\\  \partial_{s'}B^{(S)}_{\chi'}\left(s',\chi'\right) &= 2\left((\chi')^2+(s')^2\right)\xi_S\sigma^{-1}e^{2\lambda}\left(X\mathring{\omega} m + Wm^2\right)\psi^2,
\\ \nonumber \partial_{\rho}B^{(A)}_z\left(\rho,z\right)&+\left(\partial_{\rho}\xi_N\right)B^{(N)}_z + (\partial_z\xi_N)B_{\rho}^{(N)}- \left(\partial_{\rho}\xi_N\right)B^{(S)}_z + (\partial_z\xi_N)B_{\rho}^{(S)}
\\ \nonumber &=2\left(1-\xi_N - \xi_S\right)\sigma^{-1}e^{2\lambda}\left(X\mathring{\omega} m + Wm^2\right)\psi^2.
\end{align}
Recall (see Remark~\ref{defB}) that we have
\[B^{(N)}_{\rho}d\rho + B^{(N)}_zdz \doteq B^{(N)}_{\chi}d\chi,\qquad B^{(S)}_{\rho}d\rho + B^{(S)}_zdz \doteq B^{(S)}_{\chi'}d\chi'.\]

The operator $L_B$ will simply denote the left hand side of~(\ref{Btripleeqn}), and $N_B$ will denote the right hand side.
\subsubsection{The Equations for $\mathring{X}$ and $\mathring{Y}$}\label{subsec:XY-renorm-eqn}
We now turn to the equations for $\mathring{X}$ and $\mathring{Y}$. To compute their equations, one simply replaces $X$ by $X_{K}(1+\mathring X)$ and $Y$ by $Y = Y_{K}+X_{K}\mathring Y$ and expands the equations for $X$ and $Y$ (the equations are stated for $\theta$, but we use $\theta = dY + B$) from Theorem \ref{thm:HBH-geometric-main}. A completely analogous computation has been preformed in \cite{ionkla} (see the derivation of equation (2.1) in \cite{ionkla}) for the matter-free case (we emphasize that we have used the same renormalization as \cite{ionkla}), so we omit the details. We obtain
\begin{align}\label{X1}
\Delta_{\mathbb{R}^3}\mathring{X} + \frac{2\partial Y_K\cdot\partial\mathring{Y}}{X_K} - \frac{2\left|\partial Y_K\right|^2}{X_K^2}\mathring{X} + 2\frac{\partial X_K\cdot\partial Y_K}{X_K^2}\mathring{Y} &=  N_X \doteq N^{(1)}_X + N^{(2)}_X,
\end{align}
\begin{align}\label{Y1}
\Delta_{\mathbb{R}^3}\mathring{Y} - \frac{2\partial Y_K\cdot\partial\mathring{X}}{X_K} - \frac{\left[\left|\partial X_K\right|^2 + \left|\partial Y_K\right|^2\right]}{X_K^2}\mathring{Y} &= N_Y \doteq N^{(1)}_Y + N^{(2)}_Y,
\end{align}
where for a function $f(\rho,z)$ we define
\[\Delta_{\mathbb{R}^3}f \doteq \rho^{-1}\partial_{\rho}\left(\rho\partial_{\rho}f\right) + \partial_z^2f,\]
as well as

\begin{align*}
N^{(1)}_X &\doteq \frac{X_K^2\left(\left|\partial\mathring{X}\right|^2 - \left|\partial\mathring{Y}\right|^2\right) + \left(\mathring{X}\partial Y_K - \mathring{Y}\partial X_K\right)\cdot\left(2X_K\partial\mathring{Y} - \mathring{X}\partial Y_K + \mathring{Y}\partial X_K\right)}{X_K^2\left(1+\mathring{X}\right)}
\\ \nonumber N^{(2)}_X &\doteq \left(\rho^{-1}-\sigma^{-1}\partial_{\rho}\sigma\right)\frac{\partial_{\rho}\left(X_K\left(1+\mathring{X}\right)\right)}{X_K} - \sigma^{-1}\partial_z\sigma\frac{\partial_z\left(X_K\left(1+\mathring{X}\right)\right)}{X_K}- \frac{2\partial_{\rho}\left(Y_K+X_K\mathring{Y}\right)B_{\rho}}{X_K^2\left(1+\mathring{X}\right)}
\\ \nonumber &\qquad - \frac{2\partial_z\left(Y_K+X_K\mathring{Y}\right)B_z}{X_K^2\left(1+\mathring{X}\right)} - \frac{B_{\rho}^2 + B_z^2}{X_K^2\left(1+\mathring{X}\right)} - X_K^{-1}e^{2\lambda}\left(2m^2+X_K\left(1+\mathring{X}\right)\mu^2\right)\psi^2,
\\ \nonumber N^{(1)}_Y &\doteq \frac{2X_K^2\partial\mathring{X}\cdot\partial\mathring{Y} + 2X_K\left(\mathring{Y}\partial X_K - \mathring{X}\partial Y_K\right)\cdot\partial \mathring{X}}{X_K^2\left(1+\mathring{X}\right)},
\\ \nonumber N^{(2)}_Y &\doteq \left(\rho^{-1}-\sigma^{-1}\partial_{\rho}\sigma\right)\partial_{\rho}\left(Y_K+X_K\mathring{Y}\right)X_K^{-1} - \sigma^{-1}\partial_{\rho}\sigma B_{\rho}X_K^{-1} - \sigma^{-1}\partial_z\sigma\left(\partial_z\left(Y_K + X_K\mathring{Y}\right) + B_z\right)X_K^{-1}
\\ \nonumber &\qquad +\frac{B_{\rho}\partial_{\rho}\left(X_K\left(1+\mathring{X}\right)\right) + B_z\partial_z\left(X_K\left(1+\mathring{X}\right)\right)}{X_K^2\left(1+\mathring{X}\right)}.
\end{align*}

Here $L_{X,Y}$ denotes the operator corresponding to the left hand side of equations~(\ref{X1}) and~(\ref{Y1}), and $N_{X,Y}$ denotes the operator corresponding to $N_X$ and $N_Y$ on the right hand side of~(\ref{X1}) and~(\ref{Y1}).
\subsubsection{The Equation for $\mathring{\Theta}$}
Now, starting from the equation for $W$ in Theorem \ref{thm:HBH-geometric-main}, a straightforward calculation yields the following pair of equations for $\mathring{\Theta}$:
\begin{align}\label{thetaringeqnrho}
\partial_{\rho}\mathring{\Theta} = -\frac{\sigma}{X^2}\left(\partial_zY + B_z\right) + \frac{\rho}{X_K^2}\partial_zY_K,
\end{align}
\begin{align}\label{thetaringeqnz}
\partial_z\mathring{\Theta} = \frac{\sigma}{X^2}\left(\partial_{\rho}Y + B_{\rho}\right) - \frac{\rho}{X_K^2}\partial_{\rho}Y_K,
\end{align}

The left hand side is $L_{\Theta}$, and the right hand side is $N_{\Theta}$.

\subsubsection{The Equation for $\psi$ and $\mu^2$}
Finally, we turn to the equation we will use for $\psi$ and $\mu^2$ in our fixed point argument.

The equation for $\psi$ and $\mu^2$ from Theorem \ref{thm:HBH-geometric-main} becomes
\begin{equation}\label{theeqn10}
\sigma^{-1}\partial_{\rho}\left(\sigma\partial_{\rho}\psi\right) + \sigma^{-1}\partial_z\left(\sigma\partial_z\psi\right) + e^{2\lambda}\sigma^{-2}X^{-1}\left(X\mathring{\omega}+Wm\right)^2\psi - e^{2\lambda}m^2X^{-1}\psi- e^{2\lambda}\mu^2 \psi = 0.
\end{equation}

In addition to~(\ref{theeqn10}), we will require that
\begin{equation}\label{howbigpsi}
\int_{\mathscr{B}}\psi^2e^{2\lambda}\sigma\, d\rho\, dz = \delta^2,
\end{equation}
and that
\begin{equation}\label{wantpos}
\psi |_{\mathscr{B}}> 0.
\end{equation}
(Note that we will not need to specifically pose boundary conditions at the horizon.)

\subsubsection{The Equation for $\mathring{\lambda}$}\label{subsubsection-mathring-lambda}
The equation for $\lambda$ from Theorem \ref{thm:HBH-geometric-main} yields
\begin{align}
\label{lamdef-rho} \partial_{\rho}\mathring\lambda & = \alpha_{\rho} - (\alpha_{K})_{\rho} - \frac 12 \partial_{\rho}\log (1+\mathring X), \\
 \label{lambdef-z} \partial_{z}\mathring\lambda & = \alpha_{z} - (\alpha_{K})_{z} - \frac 12 \partial_{z} \log(1+\mathring X).
\end{align}
We denote the left hand side of~(\ref{lamdef-rho}) and~(\ref{lambdef-z}) by $L_{\lambda}$, and we denote all terms of the right hand side by $N_{\lambda}$. Recall that, by Theorem \ref{thm:HBH-geometric-main},
\begin{align}\label{eq:alpha-rho-defn}
& \left((\partial_{\rho}\sigma)^2 + (\partial_z\sigma)^2\right) \alpha_{\rho} \\
& \nonumber = \frac{1}{2}(\partial_{\rho}\sigma)\sigma\left((\partial_{\rho}\psi)^2 - (\partial_z\psi)^2 + \frac{1}{2}X^{-2}\left[(\partial_{\rho}X)^2 - (\partial_zX)^2 + (\theta_{\rho})^2 - (\theta_z)^2\right]\right)
\\ \nonumber &+ \partial_{\rho}\sigma(\partial_{\rho}^2\sigma - \partial_z^2\sigma) + \partial_z\sigma(\partial_{\rho,z}^2\sigma)
\\ \nonumber &+(\partial_z\sigma)\sigma\left[(\partial_{\rho}\psi)(\partial_z\psi) + \frac{1}{2}X^{-2}\left((\partial_{\rho}X)(\partial_zX) + (\theta_{\rho})(\theta_z)\right)\right],
\end{align} and \begin{align}\label{eq:alpha-z-defn}
& \left((\partial_{\rho}\sigma)^2 + (\partial_z\sigma)^2\right) \alpha_{z}  \\
& \nonumber = \frac{1}{2}(\partial_z\sigma)\sigma\left((\partial_z\psi)^2 - (\partial_{\rho}\psi)^2 + \frac{1}{2}X^{-2}\left[(\partial_zX)^2 - (\partial_{\rho}X)^2 + (\theta_z)^2 - (\theta_{\rho})^2\right]\right)
\\ \nonumber &+ \partial_z\sigma(\partial_z^2\sigma - \partial_{\rho}^2\sigma) + \partial_{\rho}\sigma(\partial_{\rho,z}^2\sigma)
\\ \nonumber &+(\partial_{\rho}\sigma)\sigma\left[(\partial_{\rho}\psi)(\partial_z\psi) + \frac{1}{2}X^{-2}\left((\partial_{\rho}X)(\partial_zX) + (\theta_{\rho})(\theta_z)\right)\right].
\end{align}
Moreover, we record
\begin{align*}
(\alpha_{K})_{\rho}  &= \frac{1}{4}\rho X_{K}^{-2}\left((\partial_{\rho}X_K)^2 - (\partial_zX_K)^2 + (\partial_{\rho}Y_{K})^2 - (\partial_zY_{K})^2\right),\\
(\alpha_{K})_{z}  &= \frac 12 \rho X_{K}^{-2}\left((\partial_{\rho}X_{K})(\partial_zX_{K}) + (\partial_{\rho}Y_{K})(\partial_zY_{K})\right).
\end{align*}

\subsection{The Function Spaces}\label{functhespace}
In this section we will introduce the relevant function spaces for each renormalized quantity. Most of our function spaces will involve H\"{o}lder spaces of a certain order \index{Miscellaneous!$\alpha_{0}$}$\alpha_0 \in (0,1)$. The parameter $\alpha_0$ will remain unfixed until we fix it in Remark~\ref{fixmalpha0}.

Most (but not all) of the renormalized quantities will be solved via a fixed point argument (i.e., using Lemma \ref{fixit}). To apply Lemma \ref{fixit} we will first solve a inhomogeneous linear problem, showing that for inhomogoneneities in $\mathcal{N}$ we can solve the linear problem in $\mathcal{L}$. We then will show that the inhomogeneities in the relevant equations satisfy appropriate quadratic estimates in terms of the other renormalized quantities, as required for Lemma \ref{fixit}. This process is perhaps most cleanly illustrated when we solve for $\mathring \sigma$ (see Section \ref{section:solving-for-sigma}). 

The exact nature of the function spaces has been chosen to balance both of these processes: solving the linear problem and establishing the non-linear estimates. The majority of the function spaces are H\"older spaces of certain orders, with particular decay assumptions enforced. 
\subsubsection{Function Spaces for $\mathring{\sigma}$}
The relevant Banach spaces are 
\begin{definition}The Banach space \index{Function Spaces!$\mathcal{L}_{\sigma}$} $\left(\mathcal{L}_{\sigma},\left\Vert\cdot\right\Vert_{\mathcal{L}_{\sigma}}\right)$ is defined to be the completion of smooth functions $f \in \hat{C}_0^{\infty}\left(\overline{\mathscr{B}}\right)$ under the norm
    \begin{align*}
    \left\Vert f\right\Vert_{\mathcal{L}_\sigma} \doteq &\left\Vert f\right\Vert_{C^{3,\alpha_0}\left(\overline{\mathscr{B}}\right)} + \left\Vert r^2f\right\Vert_{L^{\infty}\left(\overline{\mathscr{B}}\right)} + \left\Vert r^3\partial f\right\Vert_{L^{\infty}\left(\overline{\mathscr{B}}\right)} +  \left\Vert r^4\log^{-1}\left(4r\right)\partial^2f\right\Vert_{L^{\infty}\left(\overline{\mathscr{B}}\right)} +
    \\ \nonumber &\left\Vert r^4\log^{-1}\left(4r\right)\partial^3f\right\Vert_{C^{0,\alpha_0}\left(\overline{\mathscr{B}}\right)}.
    \end{align*}
\end{definition}

\begin{definition}The Banach space \index{Function Spaces!$\mathcal{N}_{\sigma}$} $\left(\mathcal{N}_{\sigma},\left\Vert\cdot\right\Vert\right)$ is defined to be the completion of smooth functions $f \in \hat{C}_0^{\infty}\left(\overline{\mathscr{B}}\right)$ under the norm
\begin{equation*}
    \left\Vert f\right\Vert_{\mathcal{N}_\sigma} \doteq \left\Vert r^4f\right\Vert_{C^{1,\alpha_0}\left(\overline{\mathscr{B}}\right)}.
    \end{equation*}
\end{definition}

\subsubsection{Function Spaces for $\left(B^{(N)}_{\chi},B^{(S)}_{\chi'},B^{(A)}_z\right)$}

The relevant Banach space for $\left(B^{(N)}_{\chi},B^{(S)}_{\chi'},B^{(A)}_z\right)$ is given by the following,.
\begin{definition}The Banach space \index{Function Spaces!$\mathcal{L}_{B}$} $\left(\mathcal{L}_B,\left\Vert\cdot\right\Vert_{\mathcal{L}_B}\right)$ is defined to be the completion of triples $\left(F^{(A)}_z,F^{(N)}_{\chi},F^{(N)}_{\chi'}\right) \in \left(\hat{C}^{\infty}\left(\overline{\mathscr{B}}\right)\right)^3$ under the norm
    \begin{align*}
    \left\Vert \left(F^{(A)}_z,F^{(N)}_{\chi},F^{(N)}_{\chi'}\right)\right\Vert_{\mathcal{L}_B} &\doteq \left\Vert \left(\frac{(1+\rho^{10})(1+r^{10})}{\rho^{10}}\right)F^{(A)}_z\right\Vert_{\hat{C}^{1,\alpha_0}\left(\overline{\mathscr{B}_A} \cup \overline{\mathscr{B}_H}\right)}
    \\ \nonumber &\ \ + \left\Vert s^{-10}F^{(N)}_{\chi}\right\Vert_{\hat{C}^{1,\alpha_0}\left(\overline{\mathscr{B}_N}\right)}+ \left\Vert (s')^{-10}F^{(S)}_{\chi'}\right\Vert_{\hat{C}^{1,\alpha_0}\left(\overline{\mathscr{B}_S}\right)}.
    \end{align*}
\end{definition}

\begin{definition}The Banach space \index{Function Spaces!$\mathcal{N}_{B}$}. $\left(\mathcal{N}_B,\left\Vert\cdot\right\Vert_{\mathcal{N}_B}\right)$ is defined to be the completion of triples $\left(F^{(A)}_z,F^{(N)}_{\chi},F^{(S)}_{\chi'}\right) \in \left(\hat{C}^{\infty}\left(\overline{\mathscr{B}}\right)\right)^3$ under the norm
    \begin{align*}
    \left\Vert \left(F^{(A)}_z,F^{(N)}_{\chi},F^{(S)}_{\chi'}\right)\right\Vert_{\mathcal{N}_B} &\doteq \left\Vert \left(\frac{(1+\rho^{15})(1+r^{10})}{\rho^{15}}\right)F^{(A)}_z\right\Vert_{\hat{C}^{1,\alpha_0}\left(\overline{\mathscr{B}_A} \cup \overline{\mathscr{B}_H}\right)}
    \\ \nonumber &\ \ + \left\Vert s^{-15}F^{(N)}_{\chi}\right\Vert_{\hat{C}^{1,\alpha_0}\left(\overline{\mathscr{B}_N}\right)} + \left\Vert (s')^{-15}F^{(S)}_{\chi'}\right\Vert_{\hat{C}^{1,\alpha_0}\left(\overline{\mathscr{B}_S}\right)}.
    \end{align*}

We also set $\hat{\mathcal{N}}_B \doteq \mathcal{N}_B$.
\end{definition}

\subsubsection{Function Spaces for $\mathring{X}$ and $\mathring{Y}$}
The relevant Banach spaces are
\begin{definition}The Banach space \index{Function Spaces!$\mathcal{L}_{X}$} $\left(\mathcal{L}_X,\left\Vert\cdot\right\Vert_{\mathcal{L}_X}\right)$ is defined to be the completion of smooth functions $f \in \hat{C}_0^{\infty}\left(\overline{\mathscr{B}}\right)$ under the norm
    \begin{align*}
    \left\Vert f\right\Vert_{\mathcal{L}_X} \doteq &\left\Vert f\right\Vert_{\dot{H}^1_{\rm axi}\left(\mathbb{R}^3\right)}+ \left\Vert f\right\Vert_{\hat{C}_0^{2,\alpha_0}\left(\overline{\mathscr{B}}\right)}+ \left\Vert rf\right\Vert_{L^{\infty}\left(\overline{\mathscr{B}}\right)} + \left\Vert r^2\hat{\partial}f\right\Vert_{L^{\infty}\left(\overline{\mathscr{B}}\right)} + \left\Vert r^3\log^{-1}\left(4r\right)\hat{\partial}^2f\right\Vert_{C^{0,\alpha_0}\left(\overline{\mathscr{B}}\right)}.
    \end{align*}
\end{definition}
\begin{definition}The Banach space \index{Function Spaces!$\mathcal{L}_{Y}$} $\left(\mathcal{L}_Y,\left\Vert\cdot\right\Vert_{\mathcal{L}_Y}\right)$ is defined to be the completion of smooth functions $f \in \hat{C}_0^{\infty}\left(\overline{\mathscr{B}}\right)$ under the norm
    \begin{align*}
    \left\Vert f\right\Vert_{\mathcal{L}_Y} \doteq &\left\Vert f\right\Vert_{\dot{H}^1_{\rm axi}\left(\mathbb{R}^3\right)} + \left\Vert \left|\partial h\right|f\right\Vert_{L^2\left(\mathbb{R}^3\right)}  + \left\Vert f\right\Vert_{\hat{C}_0^{2,\alpha_0}\left(\overline{\mathscr{B}}\right)} + \left\Vert X_K^{-1}f\right\Vert_{\hat{C}_0^{2,\alpha_0}\left(\overline{\mathscr{B}}\right)}+
    \\ \nonumber &\left\Vert r^3X_K^{-1}f\right\Vert_{L^{\infty}\left(\overline{\mathscr{B}}\right)} + \left\Vert r^4\hat{\partial}\left(X_K^{-1}f\right)\right\Vert_{L^{\infty}\left(\overline{\mathscr{B}}\right)} + \left\Vert r^5\log^{-1}\left(4r\right)\hat{\partial}^2\left(X_K^{-1}f\right)\right\Vert_{C^{0,\alpha_0}\left(\overline{\mathscr{B}}\right)}.
    \end{align*}

    Recall that the function $h$ is defined in Section~\ref{kerr}.
\end{definition}

We then define\index{Function Spaces!$\mathcal{L}_{X,Y}$}
\[\mathcal{L}_{X,Y} \doteq \mathcal{L}_X \times \mathcal{L}_Y.\]
\begin{definition}The Banach space \index{Function Spaces!$\mathcal{N}_{X}$} $\left(\mathcal{N}_X,\left\Vert\cdot\right\Vert_{\mathcal{N}_X}\right)$ is defined to be the completion of smooth functions $f \in \hat{C}_0^{\infty}\left(\overline{\mathscr{B}}\right)$ under the norm
\begin{equation*}
\left\Vert f\right\Vert_{\mathcal{N}_X} \doteq \left\Vert r^3\left(1-\xi_N-\xi_S\right)f\right\Vert_{C^{0,\alpha_0}\left(\mathbb{R}^3\right)} + \left\Vert \left(\chi^2+s^2\right)\xi_Nf\right\Vert_{C^{0,\alpha_0}\left(\overline{\mathscr{B}_{N}}\right)} + \left\Vert \left((\chi')^2+(s')^2\right)\xi_Sf\right\Vert_{C^{0,\alpha_0}\left(\overline{\mathscr{B}_{S}}\right)}
\end{equation*}
\end{definition}

\begin{definition}The Banach space \index{Function Spaces!$\mathcal{N}_{Y}$} $\left(\mathcal{N}_Y,\left\Vert\cdot\right\Vert_{\mathcal{N}_Y}\right)$ is defined to be the completion of smooth functions $f \in \hat{C}_0^{\infty}\left(\overline{\mathscr{B}}\right)$ under the norm
\begin{align*}
 \Vert f\Vert_{\mathcal{N}_{Y}} \doteq &\left\Vert fr^5X_K^{-1}\right\Vert_{\hat{C}_0^{0,\alpha_0}\left(
\left(\overline{\mathscr{B}_A} \cup\overline{\mathscr{B}_H}\right)\cap \{\rho \leq 1\}\right)} + \left\Vert fr^4\right\Vert_{\hat{C}_0^{0,\alpha_0}\left(\overline{\mathscr{B}}\cap \{\rho \geq 1\}\right)}
 \\ \nonumber &\qquad +\left\Vert \left(\chi^2+s^2\right)X_K^{-1}f\right\Vert_{C^{0,\alpha_0}\left(\overline{\mathscr{B}_N}\right)}+ \left\Vert \left((\chi')^2+(s')^2\right)X_K^{-1}f\right\Vert_{C^{0,\alpha_0}\left(\overline{\mathscr{B}_S}\right)}.
\end{align*}
\end{definition}

We then define\index{Function Spaces!$\mathcal{N}_{X,Y}$}
\[\mathcal{N}_{X,Y} \doteq \mathcal{N}_X \oplus \mathcal{N}_Y,\]
\subsubsection{Function Spaces for $\mathring{\Theta}$}

\begin{definition}The Banach space \index{Function Spaces!$\mathcal{L}_{\Theta}$}$\left(\mathcal{L}_{\Theta},\left\Vert\cdot\right\Vert_{\mathcal{L}_{\Theta}}\right)$ is defined to be the completion of smooth functions $f \in \hat{C}_0^{\infty}\left(\overline{\mathscr{B}}\right)$ under the norm
    \begin{align*}
    \left\Vert f\right\Vert_{\mathcal{L}_{\Theta}} \doteq &\left\Vert r^2f\right\Vert_{\hat{C}^{2,\alpha_0}\left(\overline{\mathscr{B}}\right)}.
    \end{align*}
\end{definition}

\begin{definition}The Banach space \index{Function Spaces!$\mathcal{N}_{\Theta}$}$\left(\mathcal{N}_{\Theta},\left\Vert\cdot\right\Vert_{\mathcal{L}_{\Theta}}\right)$ is defined to be the completion of pairs of smooth compactly supported \underline{closed} $1$-forms $F$ under the norm
    \begin{align*}
    \left\Vert F\right\Vert_{\mathcal{N}_{\Theta}} \doteq &\left\Vert r^3\left(1+\rho^{-1}\right)F_{\rho}\right\Vert_{\hat{C}^{1,\alpha_0}\left(\overline{\mathscr{B}_A}\cup\overline{\mathscr{B}_H}\right)}+\left\Vert r^3F_z\right\Vert_{\hat{C}^{1,\alpha_0}\left(\overline{\mathscr{B}_A}\cup\overline{\mathscr{B}_H}\right)}
    \\ \nonumber &+\left\Vert s^{-1}F_s\right\Vert_{\hat{C}^{1,\alpha_0}\left(\overline{\mathscr{B}_N}\right)}+\left\Vert F_{\chi}\right\Vert_{\hat{C}^{1,\alpha_0}\left(\overline{\mathscr{B}_N}\right)}
    \\ \nonumber &+\left\Vert (s')^{-1}F_{s'}\right\Vert_{\hat{C}^{1,\alpha_0}\left(\overline{\mathscr{B}_S}\right)}+\left\Vert F_{\chi'}\right\Vert_{\hat{C}^{1,\alpha_0}\left(\overline{\mathscr{B}_S}\right)}.
    \end{align*}
\end{definition}

\subsubsection{Function Spaces for $\psi$ and $\mu^2$}

The relevant function spaces are given by the following.
\begin{definition}\label{functhepsi} The Banach space \index{Function Spaces!$\mathcal{L}_{\psi}$}$\left(\mathcal{L}_{\psi},\left\Vert\cdot\right\Vert_{\mathcal{L}_{\psi}}\right)$ is defined to be the completion of functions $f \in \hat{C}_0^{\infty}\left(\overline{\mathscr{B}}\right)\times \mathbb{R}$ under the norm
   \begin{align*}
    \left\Vert f\right\Vert_{\mathcal{L}_{\psi}} &\doteq \left\Vert \left(\frac{(1+\rho^{10})(1+r^{10})}{\rho^{10}}\right)\psi\right\Vert_{\hat{C}^{2,\alpha_0}\left(\overline{\mathscr{B}_A}\right)} + \left\Vert \psi\right\Vert_{\hat{C}^{2,\alpha_0}\left(\overline{\mathscr{B}_H}\right)}
    \\ \nonumber &\qquad + \left\Vert s^{-10}\psi\right\Vert_{\hat{C}^{2,\alpha_0}\left(\overline{\mathscr{B}_N}\right)} + \left\Vert (s')^{-10}\psi\right\Vert_{\hat{C}^{2,\alpha_0}\left(\overline{\mathscr{B}_S}\right)}.
    \end{align*}
\end{definition}

\begin{definition}\label{functhemu} We define \index{Function Spaces!$\mathcal{L}_{\mu^{2}}$}$\mathcal{L}_{\mu^2} \doteq \left(\mathbb{R},|\cdot|\right)$.
\end{definition}

\subsubsection{Function Space for $\mathring{\lambda}$}
\begin{definition}
The Banach space \index{Function Spaces!$\mathcal{L}_{\lambda}$}$(\mathcal{L}_{\lambda},\Vert \cdot \Vert_{\mathcal{L}_{\lambda}})$ is defined to be the completion of smooth functions $f\in \hat C^{\infty}(\overline{\mathscr{B}})$ under the norm
\[
\Vert f\Vert_{\mathcal{L}_{\lambda}} = \Vert f\Vert_{\hat C^{1,\alpha_{0}}(\overline{\mathscr{B}})}.
\]
\end{definition}
\section{Solving for $\mathring{\sigma}$}\label{section:solving-for-sigma}
In this section will solve for $\mathring{\sigma}$ in terms of the renormalized unknowns
\[\left(\left(B^{(N)}_{\chi},B^{(S)}_{\chi'},B^{(A)}_z\right),\left(\mathring{X},\mathring{Y}\right),\mathring{\Theta},\psi,\mathring{\mu}^2,\mathring{\lambda}\right).\]
\subsection{Linear Estimates}
The linear problem we need to study is\index{Metric Quantities!$H_{\sigma}$}
\begin{equation}\label{linsig}
\Delta_{\mathbb{R}^4}\mathring{\sigma} = H_{\sigma}.
\end{equation}

We have
\begin{proposition}\label{prop:lin-est-sigma}
Suppose that $H_{\sigma} \in \mathcal{N}_{\sigma}$. Then, if we let $\Delta_{\mathbb{R}^4}^{-1}$ denote convolution with the fundamental solution of $\Delta_{\mathbb{R}^4}$, and set $\mathring{\sigma} \doteq \Delta_{\mathbb{R}^4}^{-1}H_{\sigma}$. Then there exists a constant $D(\alpha_0)$, independent of $H_{\sigma}$ and depending on $\alpha_0$, such that
\[\left\Vert \mathring{\sigma}\right\Vert_{\mathcal{L}_{\sigma}} \leq D\left(\alpha_0\right)\left\Vert H_{\sigma}\right\Vert_{\mathcal{N}_{\sigma}}.\]
\end{proposition}
\begin{proof}This follows immediately from Lemmas~\ref{newt1} and~\ref{newt2} in the Appendix and local Schauder estimates.
\end{proof}
\subsection{The Fixed Point}
We are ready to run the fixed point argument.
\begin{proposition}\label{fixsig}For each $m \in \mathbb{Z}_{\neq 0}$ and $\alpha_0 \in (0,1)$, there exists $\epsilon > 0$ sufficiently small so that given \[\left(\left(B^{(N)}_{\chi},B^{(S)}_{\chi'},B^{(A)}_z\right),\left(\mathring{X},\mathring{Y}\right),\mathring{\Theta},\psi,\mathring{\mu}^2,\mathring{\lambda}\right) \in B_{\epsilon}\left(\mathcal{L}_B\right)\times \cdots  \times B_{\epsilon}\left(\mathcal{L}_{\mathring\lambda}\right),\]
we may find $\mathring{\sigma} \in B_{\epsilon}\left(\mathcal{L}_{\sigma}\right)$ which solves~(\ref{renormalsigmaeqn}). Furthermore, $\mathring{\sigma}$ has ``continuous nonlinear dependence on parameters''
\[\left(\left(B^{(N)}_{\chi},B^{(S)}_{\chi'},B^{(A)}_z\right),\left(\mathring{X},\mathring{Y}\right),\mathring{\Theta},\psi,\mathring{\mu}^2,\mathring{\lambda}\right),\]
in the sense of Lemma~\ref{fixit} and Remark~\ref{contnonlindep}.
\end{proposition}
\begin{proof}This follows easily from Lemma~\ref{fixit}. 

Recalling that $\mathring{\omega}$ is defined so that $\mathring{\omega} + X^{-1}Wm$ vanishes on the horizon, the required nonlinear estimates of $N_{\sigma}$ are easily verified. Because we will use an argument of this flavor repeatedly in the subsequent sections, we give some of the details of this proof, as a guide for the reader. We begin by defining the spaces and operators that we will use to apply Lemma~\ref{fixit}. First, we set $\mathcal{L} = \mathcal{L}_{\sigma}$ and $\tilde{\mathcal{L}} =  \mathcal{N}_{\sigma}$. We also set
\[
\mathcal{Q} = \mathcal{L}_{B} \times \mathcal{L}_{X}\times \mathcal{L}_{Y}\times \mathcal{L}_{\Theta}\times \mathcal{L}_{\psi} \times \mathcal{L}_{\mu^{2}}\times \mathcal{L}_{\lambda}
\]
This allows us to define $\mathfrak{E}:B_{\epsilon}(\mathcal{L}) \times B_{\epsilon}(\mathcal{Q})\to\tilde{\mathcal{L}}$ by
\[
\mathfrak{E}(\mathring \sigma,( B, \mathring X,\mathring Y, \mathring \Theta,\psi,\mathring\mu^{2}, \mathring \lambda)) = \Delta_{\RR^{4}} \mathring \sigma - N_{\sigma},
\]
where 
\[
N_{\sigma} = \rho^{-1}\sigma^{-1}Xe^{2\lambda}\left(\left(\mathring{\omega}+X^{-1}Wm\right)^2 - \sigma^2\left(\frac{\mu^2}{X} + \frac{m^2}{X^2}\right)\right)\psi^2
\]
(see \eqref{renormalsigmaeqn}). This is well defined (we will check below that the image is indeed contained in $\mathcal{N}_{\sigma}$), as long as we take $\epsilon>0$ sufficiently small. Note that $N_{\sigma}$ is, in particular, a map $N:\mathcal{L}\times \mathcal{Q}\to\tilde{\mathcal{L}}$ so that $\mathfrak{E} = L - N$, as desired. We have proven the linear estimates in Proposition \ref{prop:lin-est-sigma} above, which establishes part (2) of Lemma~\ref{fixit} for $\mathcal{N} = \mathcal{N}_{\sigma}$. 

Now, we come to part (3) of Lemma~\ref{fixit}, namely the non-linear estimates on $N$. Reviewing the definition of the norm for $\mathcal{N}_{\sigma}$, we would like that 
\[
\Vert r^{4} N_{\sigma} \Vert_{C^{1,\alpha_{0}}(\overline{\mathscr{B}})} \leq D \left[ \Vert \mathring \sigma \Vert_{\mathcal{L}_{\sigma}}^{2} + \Vert(B,\mathring X,\mathring Y,\mathring\Theta,\psi,\mathring\mu^{2},\mathring\sigma)\Vert_{\mathcal{Q}}^{2}\right].
\]
Before establishing these estimates, we make several remarks. To obtain such quadratic bounds, it is useful to note that everything is multiplied by $\psi^{2}$, and as such it basically suffices to bound the remaining terms linearly. Note also that $\psi$ vanishes to high order at the poles and the axis, by the definition of $\mathcal{L}_{\psi}$. For example, we have that $\psi = O(\rho^{10})$ in $\hat C^{2,\alpha_{0}}(\overline{\mathscr{B}_{A}})$. 

To prove the quadratic estimate, we consider various regimes. We begin with points near the axis, at a bounded distance away from the horizon. Thanks to the bounds on $\mathring X$ from the $\mathcal{L}_{X}$ norm, we have that $X \sim \rho^{2}$ in $C^{2,\alpha_{0}}$ near these points. Writing
\[
\rho^{-1}\sigma^{-1} = \rho^{-2}(1+\mathring\sigma)^{-1},
\]
these terms combine into a quantity that we can bound in $C^{1,\alpha_{0}}$. It is easy to bound
\[
\left(\left(\mathring{\omega}+X^{-1}Wm\right)^2 - \sigma^2\left(\frac{\mu^2}{X} + \frac{m^2}{X^2}\right)\right)
\]
in $C^{1,\alpha_{0}}$ (note that $X^{-1}W = \mathring \Theta + X^{-1}_{K}W_{K}$ and the second (purely Kerr) quantity is smooth and bounded at the points we are considering by Lemma 3.0.1 in \cite{HBH:geometric}). Finally, the $\psi^{2}$ term in $C^{1,\alpha}$ can be bounded by the $\mathcal{L}_{\psi}$ norm. This term multiplies the other quantities, leading to the quadratic form of the estimates. 

A more difficult regime to consider is near where the axis meets the horizon. Essentially, the main issue here corresponds to the fact that we control $\psi$ in $\hat C^{2,\alpha_{0}}$, but we would like estimates for $N_{\sigma}$ in $C^{1,\alpha_{0}}$. However, this is easily handled thanks to the rapid decay of $\psi$ as $\rho\to0$. For example, we note that
\[
|\rho \partial_{z} \psi| \lesssim s | \partial_{s}\psi |+ \chi|\partial_{\chi}\psi|
\] 
Both of these quantities are controlled by the $\hat C^{1}$ norm of $\psi$ (we can argue similarly for the H\"older norm). Using the remaining $\rho$ decay of $\psi$, we can control the other terms similarly.

The other regimes follow along analogous (but simpler arguments). This proves (3) in Lemma~\ref{fixit}. Finally, item (4) in Lemma~\ref{fixit} follows from a nearly identical argument. Thus, the proof is completed by applying Lemma~\ref{fixit}. 
\end{proof}
\section{Solving for $\left(B^{(N)}_{\chi},B^{(S)}_{\chi'},B^{(A)}_z\right)$}\label{solvingforb}
In this section will solve for $\left(B^{(N)}_{\chi},B^{(S)}_{\chi'},B^{(A)}_z\right)$ in terms of the renormalized unknowns
\[\left(\left(\mathring{X},\mathring{Y}\right),\mathring{\Theta},\psi,\mathring{\mu}^2,\mathring{\lambda}\right).\]

\subsection{Linear Estimates for $B$}
In this section will provide the necessary estimates for the $1$-form $B$'s equation. The linear problem we need to study is
\begin{align}\label{linB}\index{Metric Quantities!$H^{(1)}_{B}$}\index{Metric Quantities!$H^{(2)}_{B}$}\index{Metric Quantities!$H^{(3)}_{B}$}
\nonumber \partial_sB^{(N)}_{\chi}\left(s,\chi\right) &= H^{(1)}_B,
\\  \partial_{s'}B^{(S)}_{\chi'}\left(s',\chi'\right)&= H^{(2)}_B,
\\ \nonumber \partial_{\rho}B^{(A)}_z\left(\rho,z\right) &= \left(1-\xi_N - \xi_S\right)H^{(3)}_B
\\ \nonumber &\qquad \qquad -\left(\partial_{\rho}\xi_N\right)B^{(N)}_z + (\partial_z\xi_N)B_{\rho}^{(N)}- \left(\partial_{\rho}\xi_N\right)B^{(S)}_z + (\partial_z\xi_N)B_{\rho}^{(S)}.
\end{align}
These equations are essentially those from Definition~\ref{defB}, with the dependence on the other parameters suppressed for now. 

We have
\begin{proposition}Suppose that $(H^{(1)}_B,H^{(2)}_B,H^{(3)}_B) \in \mathcal{N}_B$ and we define the solution \[\left(B^{(N)}_{\chi},B^{(S)}_{\chi'},B^{(A)}_z\right)\]
to~(\ref{linB}) simply by integrating the equations in~(\ref{linB}) from $0$ to $s$, $0$ to $s'$, and $0$ to $\rho$ respectively, then there exists a constant $D(\alpha_0)$, independent of $(H^{(1)}_B,H^{(2)}_B,H^{(3)}_B)$ and depending on $\alpha_0$, such that
\[\left\Vert \left(B^{(N)}_{\chi},B^{(S)}_{\chi'},B^{(A)}_z\right)\right\Vert_{\mathcal{L}_B} \leq D(\alpha_0)\left\Vert\left(H^{(1)}_B,H^{(2)}_B,H^{(3)}_B\right)\right\Vert_{\mathcal{N}_B}.\]
\end{proposition}
\begin{proof}This estimate is a trivial consequence of the fundamental theorem of calculus, and the formulas
\[\partial_x = \frac{x}{\sqrt{x^2+y^2}}\partial_{\rho},\qquad \partial_y = \frac{y}{\sqrt{x^2+y^2}},\]
which arise when switching from cylindrical coordinates $(\rho,\phi,z)$ to Cartesian coordinates $(x,y,z)$ (as well as the analogous formulas in $(s,\chi)$ and $(s',\chi')$ coordinates).
\end{proof}
\subsection{The Fixed Point}
We are ready to run the fixed point argument.
\begin{proposition}\label{fixb}For each $m \in \mathbb{Z}_{\neq 0}$ and $\alpha_0 \in (0,1)$, there exists $\epsilon > 0$ sufficiently small so that given \[\left(\left(\mathring{X},\mathring{Y}\right),\mathring{\Theta},\psi,\mathring{\mu}^2,\mathring{\lambda}\right) \in B_{\epsilon}\left(\mathcal{L}_{X,Y}\right)\times \cdots\times B_{\epsilon}\left(\mathcal{L}_{\psi}\right) \times B_{\epsilon}\left(\mathcal{L}_{\mu^2}\right)\times B_{\epsilon}(\mathcal{L}_{\lambda}),\]
we may find $\left(\mathring{\sigma},\left(B^{(N)}_{\chi},B^{(S)}_{\chi'},B^{(A)}_z\right)\right) \in B_{\epsilon}\left(\mathcal{L}_{\sigma}\right)\times B_{\epsilon}\left(\mathcal{L}_B\right)$ which solve~(\ref{renormalsigmaeqn}) and~(\ref{Btripleeqn}) respectively. Furthermore, $\mathring{\sigma}$ and $\left(B^{(N)}_{\chi},B^{(S)}_{\chi'},B^{(A)}_z\right)$ have ``continuous nonlinear dependence on parameters''
\[\left(\left(\mathring{X},\mathring{Y}\right),\mathring{\Theta},\psi,\mathring{\mu}^2,\mathring{\lambda}\right),\]
in the sense of Lemma~\ref{fixit} and Remark~\ref{contnonlindep}.
\end{proposition}
\begin{proof}First we apply Proposition~\ref{fixsig} to solve for $\sigma$ in terms of the other unknowns. Then the proof follows easily from Lemma~\ref{fixit}. In particular, we take $\mathcal{L} = \mathcal{L}_{B}$, $\tilde{\mathcal{L}}= \mathcal{N}_{B}$ and 
\[
\mathcal{Q} = \mathcal{L}_{X}\times\mathcal{L}_{Y} \times \mathcal{L}_{\Theta} \times \mathcal{L}_{\psi}\times\mathcal{L}_{\mu^{2}}\times\mathcal{L}_{\lambda}.
\]
Then, we define $\mathfrak{E}:B_{\epsilon}(\mathcal{L}) \times B_{\epsilon}(\mathcal{Q})\to\tilde{\mathcal{L}}$ by integrating the equations \eqref{linB} with the source terms given by
\begin{align*}
H_{B}^{(1)} &\doteq 2\left(\chi^2+s^2\right)\xi_N\sigma^{-1}e^{2\lambda}\left(X\omega m + Wm^2\right)\psi^2
\\ H_{B}^{(2)} &\doteq 2\left((\chi')^2+(s')^2\right)\xi_S\sigma^{-1}e^{2\lambda}\left(X\omega m + Wm^2\right)\psi^2,
\\ H_{B}^{(3)} &\doteq 2\sigma^{-1}e^{2\lambda}\left(X\omega m + Wm^2\right)\psi^2
\end{align*}
Note that $\mathring\sigma$ (and thus $\sigma$) is a function of the other parameters (and satisfies ``continuous nonlinear dependence on parameters'' in the sense of Lemma \ref{fixit}). 

We now explain how to check the nonlinear estimates required for Lemma \ref{fixit}. Near $p_{N}$, we claim that
\[
\Vert s^{-15} H_{B}^{(1)} \Vert_{\hat C^{1,\alpha_{0}}(\overline{\mathscr{B}_{N}})} \leq D \Vert ((B^{N}_{\chi},B^{S}_{\chi'},B^{A}_{z}),(\mathring X,\mathring Y),\mathring \Theta,\psi,\mathring\mu^{2},\mathring\lambda)\Vert_{\mathcal{L}_{B}\times\dots\times\mathcal{L}_{\lambda}}^{2}.
\]
This follows easily from the fact that the $\hat C^{2,\alpha}$ norm of $\psi$ decays like $s^{10}$ near $p_{N}$ and in the definition of $H^{(1)}_{B}$, $\psi$ appears squared (and multiplying the rest of the terms); note that we can bound $\Vert \mathring\sigma\Vert_{\mathcal{L}_{\sigma}}$ by the right hand side of the above expression, by Proposition \ref{fixsig}. The remaining estimates for $H_{B}$ in $\mathcal{N}_{B}$ as well as for norm of the difference follow similar arguments. Having this verified the nonlinear estimates (3) and (4) in the hypothesis of Lemma \ref{fixit}, this completes the proof. 
\end{proof}
\section{Solving for $\left(\mathring{X},\mathring{Y}\right)$}
In this section will solve for $\left(\mathring{X},\mathring{Y}\right)$ in terms of the renormalized unknowns
\[\left(\mathring{\Theta},\psi,\mathring{\mu}^2,\mathring{\lambda}\right).\]

Throughout this section we will identify functions of $\rho$ and $z$ as functions on $\mathbb{R}^3$ via cylindrical coordinates, often without saying so explicitly.
\subsection{Linear Estimates for $X$ and $Y$}
In this section we will provide the necessary linear estimates for $\mathring{X}$ and $\mathring{Y}$.

The linear equations we need to study are \index{Metric Quantities!$H_{X}$}\index{Metric Quantities!$H_{Y}$}

\begin{equation}\label{linX}
\Delta_{\mathbb{R}^3}\mathring{X} + \frac{2\partial Y_K\cdot\partial\mathring{Y}}{X_K} - \frac{2\left|\partial Y_K\right|^2}{X_K^2}\mathring{X} + 2\frac{\partial X_K\cdot\partial Y_K}{X_K^2}\mathring{Y} = H_X,
\end{equation}
\begin{equation}\label{linY}
\Delta_{\mathbb{R}^3}\mathring{Y} - \frac{2\partial Y_K\cdot\partial\mathring{X}}{X_K} - \frac{\left[\left|\partial X_K\right|^2 + \left|\partial Y_K\right|^2\right]}{X_K^2}\mathring{Y} = H_Y.
\end{equation}
These equations were derived in Section \ref{subsec:XY-renorm-eqn}.

\subsubsection{Existence of an Inverse: The Variational Structure}\label{variational}
In this section we will show that the equations~(\ref{linX}) and~(\ref{linY}) have a useful variational structure. Using this, we will show that~(\ref{linX}) and~(\ref{linY}) are (in an appropriate sense) uniquely solvable in terms of $H_X$ and $H_Y$.

The following proposition is taken from~\cite{ionkla}.\index{Miscellaneous!$\mathcal{L}(\mathring X,\mathring Y)$}  
\begin{proposition}The system~(\ref{linX}) and~(\ref{linY}) are formally the Euler--Lagrange equations of
\begin{align*}
\mathcal{L}\left(\mathring{X},\mathring{Y}\right) = \int_{\mathbb{R}^3}\Big[&\left|\partial \mathring{X} + X_K^{-1}\left(\partial Y_k\right) \mathring{Y}\right|^2 + \left|\partial \mathring{Y} - X_K^{-1}\left(\partial Y_K\right)\mathring{X}\right|^2
\\ \nonumber &+ X_K^{-2}\left|\mathring{X}\partial Y_K - \mathring{Y}\partial X_K\right|^2 + 2\mathring{X}H_X + 2\mathring{Y}H_Y\Big].
\end{align*}
\end{proposition}
\begin{proof}This is a straightforward (if tedious) computation. See (1.15) and Section 2.1 in  \cite{ionkla}.
\end{proof}
\begin{remark}
We will specify below (in Lemma \ref{solve}) precisely what we mean to solve \eqref{linX} and \eqref{linY}. In this setting, it will be clear that indeed the system is the Euler--Lagrange equations for $\mathcal{L}$. Until then, we may simply consider this as a formal computation. 
\end{remark}

\begin{remark}It is possible to derive this Lagrangian using the second variation of energy formula for a harmonic map~\cite{eelslemaire}. In this context, the non-negativity of the lower order term is a result of the negative sectional curvature of hyperbolic space $\mathbb{H}^2$. One would then need to derive a Lagrangian for the renormalized quantities (which must still be non-negative). We refer instead to \cite{ionkla} for a direct computation proving this result.
\end{remark}
Our goal is to apply the direct method of the calculus of variations. We start with a sequence of preliminary lemmas.

This lemma establish a Poincar\'{e} inequality in the spirit of Lemma 1 from~\cite{weinstein2}.
\begin{lemma}\label{poin} Let $f(\rho,z) \in C^{\infty}_0\left(\mathbb{R}^3\right)$  Then
\[\int_{\mathbb{R}^3}\left|\partial X_K\right|^2f^2 \leq C\int_{\mathbb{R}^3}X_K^2\left|\partial f\right|^2.\]
\end{lemma}
\begin{proof}First of all, by Lemma \ref{lem:xK-regularity}, we have
\[\left|\partial X_K\right| \sim e^h\left|\partial h\right|.\]
Thus, it suffices to establish
\[\int_{\mathbb{R}^3}e^{2h}\left|\partial h\right|^2f^2 \leq C\int_{\mathbb{R}^3}e^{2h}\left|\partial f\right|^2.\]

Next, using that $h$ is harmonic, we note that
\[{\rm div}_{\mathbb{R}^3}\left(e^{2h}f^2\partial h\right) = 2e^{2h}f^2\left|\partial h\right|^2 + 2e^{2h}f\partial f\cdot\partial h.\]

Applying the divergence theorem, we obtain
\begin{align*}
2 \int_{\mathbb{R}^3}\left[e^{2h}f^2\left|\partial h\right|^2 + e^{2h}f\partial f\cdot\partial h\right] &= -\liminf_{\epsilon\to 0}\int_{-\infty}^{\infty}\left(e^{2h}f^2\partial_{\rho}h\right)|_{(\rho,z) = (\epsilon,\tilde z)}\, d\tilde z
\\ \nonumber &= 0.
\end{align*}
Thus,
\[\int_{\mathbb{R}^3}e^{2h}f^2\left|\partial h\right|^2 \leq \int_{\mathbb{R}^3}e^{2h}\left|f\right|\left|\partial f\right|\left|\partial h\right|.\]
The proof then concludes via an absorbing inequality (as in Lemma 1 in~\cite{weinstein2}).
\end{proof}

Next, we quote a Hardy inequality from Theorem 1.3 in \cite{bartnik}.
\begin{lemma}\label{bar}Let $f \in C^{\infty}_0\left(\mathbb{R}^3\right)$. Set \index{Miscellaneous!$\langle r \rangle$} $\langle r\rangle \doteq (1+x^2+y^2+z^2)^{1/2}$. Then
\[\int_{\mathbb{R}^3}\langle r\rangle^{-2}f^2 \leq C\int_{\mathbb{R}^3}\left|\partial f\right|^2.\]
\end{lemma}

Now we come to the key lemma of the section.
\begin{lemma}\label{coercive}There exists a constant $C$ such that if
\begin{equation}\label{finitestuff}
\mathring{X} \in C^{\infty}_0\left(\mathbb{R}^3\right),\qquad \mathring{Y} \in C^{\infty}_0\left(\mathbb{R}^3\right),\qquad \int_{\mathbb{R}^3}\frac{\left|\partial X_K\right|^2}{X_K^2}\mathring{Y}^2 < \infty,
\end{equation}
then
\begin{equation}\label{tocontrol}
\int_{\mathbb{R}^3}\left[\left|\partial\mathring{X}\right|^2 + \left|\partial \mathring{Y}\right|^2 + \left|\partial h\right|^2\mathring{Y}^2\right] \leq C\left[\mathcal{L}\left(\mathring{X},\mathring{Y}\right) + 2\int_{\mathbb{R}^3}\left(\left|\mathring{X}H_X\right| + \left|\mathring{Y}H_Y\right|\right)\right].
\end{equation}
\end{lemma}
\begin{proof}Let's introduce the notation\index{Miscellaneous!$\mathcal{L}_{0}(\mathring X,\mathring Y)$}
\[\mathcal{L}_0\left(\mathring{X},\mathring{Y}\right) \doteq \mathcal{L}\left(\mathring{X},\mathring{Y}\right) - 2\int_{\mathbb{R}^3}\left(\mathring{X}H_X + \mathring{Y}H_Y\right).\]

We start by observing that
\begin{align*}
\left|\partial \mathring{Y} - X_K^{-1}\partial X_K \mathring{Y}\right|^2 &= \left|\partial \mathring{Y} - X_K^{-1}\partial Y_K\mathring{X} + X_K^{-1}\partial Y_K \mathring{X} - X_K^{-1}\partial X_K \mathring{Y}\right|^2
\\ \nonumber &\leq C\left|\partial \mathring{Y} - X_K^{-1}\partial Y_K\mathring{X}\right|^2 + C\left|X_K^{-1}\partial Y_K \mathring{X} - X_K^{-1}\partial X_K \mathring{Y}\right|^2.
\end{align*}
Thus,
\[\int_{\mathbb{R}^3}\left|\partial \mathring{Y} - X_K^{-1}\partial X_K \mathring{Y}\right|^2 \leq C\mathcal{L}_0\left(\mathring{X},\mathring{Y}\right).\]
Now, define $\tilde Y$ by
\[\mathring{Y} \doteq X_K\tilde{Y},\]
and observe that
\[\partial \mathring{Y} - X_K^{-1}\partial X_K \mathring{Y} = X_K\partial\tilde{Y} +\partial X_K \tilde{Y} - \partial X_K \tilde{Y} = X_K\partial\tilde{Y}.\]
Thus,
\[\int_{\mathbb{R}^3}X_K^2\left|\partial\tilde{Y}\right|^2 \leq C\mathcal{L}_0\left(\mathring{X},\mathring{Y}\right).\]

Now, Lemma~\ref{poin} (and the assumptions~(\ref{finitestuff})) imply that
\[\int_{\mathbb{R}^3}\left|\partial X_K\right|^2\tilde{Y}^2 \leq C\int_{\mathbb{R}^3}X_K^2\left|\partial\tilde{Y}\right|^2 \leq C\mathcal{L}_0\left(\mathring{X},\mathring{Y}\right).\]
This already yields the necessary estimate for the third term on the left hand side of~(\ref{tocontrol}).

Next, combining this with the third term in $\mathcal{L}_0$, we find that
\[\int_{\mathbb{R}^3}X_K^{-2}\mathring{X}^2\left|\partial Y_K\right|^2 \leq C\mathcal{L}_0\left(\mathring{X},\mathring{Y}\right).\]

Combining this with the second term in $\mathcal{L}_0$ yields the desired $\dot{H}^1_{\rm axi}$ bound\footnote{Recall that $\dot H_{\rm axi}^{1}$ (as defined in Definition \ref{Lp}) simply means that $\mathring Y(\rho,z)$ thought of as an equation on $\mathbb{R}^{3}$ in cylindrical coordinates, has first derivative in $L^{2}$. } for $\mathring{Y}$:
\[\int_{\mathbb{R}^3}\left|\partial \mathring{Y}\right|^2 \leq C\mathcal{L}_0\left(\mathring{X},\mathring{Y}\right).\]

Next, we may apply Lemma~\ref{bar} to conclude that
\[\int_{\mathbb{R}^3}\frac{\left|\partial Y_K\right|^2}{X_K^2}\mathring{Y}^2 \leq C \int_{\mathbb{R}^3}\langle r\rangle^{-2}\mathring{Y}^2 \leq C \int_{\mathbb{R}^3}\left|\partial\mathring{Y}\right|^2 \leq C\mathcal{L}_0\left(\mathring{X},\mathring{Y}\right).\]
Here, we have used Lemma \ref{somestuffthatisuseful} in the first inequality to obtain decay of $\frac{|\partial Y_{K}|}{X_{K}}$ for  $r$ large. 

Combining this with the first term in $\mathcal{L}_0$ immediately yields that
\[\int_{\mathbb{R}^3}\left|\partial \mathring{X}\right|^2 \leq C\mathcal{L}_0\left(\mathring{X},\mathring{Y}\right),\]
and finishes the proof.
\end{proof}

The $\left|\partial h\right|$ weight which appears in the previous lemma motivates the following definitions.
\begin{definition}Let $u$ be a function and $v$ be a vector field on $\mathbb{R}^{n}$. We introduce the weighted norms \index{Function Spaces!$L_{u}^{p}$}\index{Function Spaces!$L_{v}^{p}$}$L^p_u\left(\mathbb{R}^n\right)$ and $L^p_v\left(\mathbb{R}^n\right)$ by
\[\left\Vert f\right\Vert_{L^p_u\left(\mathbb{R}^n\right)} \doteq \left\Vert  \left|u\right| f\right\Vert_{L^p\left(\mathbb{R}^n\right)}, \qquad \left\Vert f\right\Vert_{L^p_v\left(\mathbb{R}^n\right)} \doteq \left\Vert  \left|v\right| f\right\Vert_{L^p\left(\mathbb{R}^n\right)}.\]

\end{definition}

Lemma~\ref{coercive} allows us to solve the system~(\ref{linX}) and~(\ref{linY}) variationally, when $H_X$ and $H_Y$ are smooth and compactly supported.

\begin{lemma}\label{solve}Suppose that $H_X, H_Y \in C^{\infty}_0\left(\mathbb{R}^3\right)$. Then~(\ref{linX}) and~(\ref{linY}) has a unique weak solution $\left(\mathring{X},\mathring{Y}\right) \in \dot{H}^1_{\rm axi}\left(\mathbb{R}^3\right) \times \left(\dot{H}^1_{\rm axi}\left(\mathbb{R}^3\right)\cap L^2_{\nabla h}\left(\mathbb{R}^3\right)\right)$. By weak solution we mean that for any $(\varphi_1,\varphi_2) \in \dot{H}^1_{\rm axi}\left(\mathbb{R}^3\right) \times \left( \dot{H}^1_{\rm axi}\left(\mathbb{R}^3\right)\cap L^2_{\nabla h}\left(\mathbb{R}^3\right)\right)$ we have
\[\int_{\mathbb{R}^3}\Bigg(\partial\mathring{X}\cdot\partial\varphi_1 - \frac{2\partial Y_K\cdot\partial\mathring{Y}}{X_K}\varphi_1 + \frac{2\left|\partial Y_K\right|^2}{X_K^2}\mathring{X}\varphi_1 - 2\frac{\partial X_K\cdot\partial Y_K}{X_K^2}\mathring{Y}\varphi_1 + H_X\varphi_1\]
\[+ \partial\mathring{Y}\cdot\partial\varphi_2  + \frac{2\partial Y_K\cdot\partial\mathring{X}}{X_K}\varphi_2 + \frac{\left[\left|\partial X_K\right|^2 + \left|\partial Y_K\right|^2\right]}{X_K^2}\mathring{Y}\varphi_2 + H_Y\varphi_2\Bigg)= 0.\]
Finally, the solution is uniquely determined in the class $\dot{H}^1_{\rm axi}\left(\mathbb{R}^3\right)\times \left(\dot{H}^1_{\rm axi}\left(\mathbb{R}^3\right)\cap L^2_{\nabla h}\left(\mathbb{R}^3\right)\right)$.
\end{lemma}
\begin{remark}Note that it follows immediately from elliptic theory that $\mathring{X}$ and $\mathring{Y}$ are smooth classical solutions in $\mathbb{R}^3\setminus \{ \{\rho = 0\} \cap z\not\in (-\gamma,\gamma)\}$.
\end{remark}
\begin{proof}It easily follows from Lemma~\ref{bar}, Lemma~\ref{coercive}, and the compact support of $H_X$ and $H_Y$ that
\[\hat{\alpha} \doteq \inf_{\left(f,g\right) \in \dot{H}^1_{\rm axi} \times \left(\dot{H}^1_{\rm axi}\left(\mathbb{R}^3\right)\cap L^2_{\nabla h}\left(\mathbb{R}^3\right)\right)}\mathcal{L}\left(f,g\right),\]
exists and is finite. Let $\{(f_k,g_k)\}_{k=1}^{\infty}$ be a minimizing sequence. Passing to a subsequence we can assume that $(f_k,g_k)$ converge pointwise almost everywhere to some $(\mathring{X},\mathring{Y})$. Furthermore, we can assume the convergence happens weakly in $\dot{H}^1_{\rm axi}$ and strongly in $L^2$ on compact subsets of $\mathbb{R}^3$. A standard argument implies that
\[\mathcal{L}\left(\mathring{X},\mathring{Y}\right) = \hat{\alpha}.\]

Let $(\varphi_1,\varphi_2) \in \dot{H}^1_{\rm axi}\left(\mathbb{R}^3\right) \times \left(\dot{H}^1_{\rm axi}\left(\mathbb{R}^3\right)\cap L^2_{\nabla h}\left(\mathbb{R}^3\right)\right)$. By definition of $\hat\alpha$, we must have
\[\frac{d}{ds}\mathcal{L}\left(\mathring{X} + s\varphi_1,\mathring{Y} + s\varphi_2\right) = 0.\]
Since~(\ref{linX}) and~(\ref{linY}) are the Euler--Lagrange equations of $\mathcal{L}$, the existence statement of the lemma immediately follows.

Now we turn to uniqueness. Since the equations are linear, it suffices to show that $H_X = H_Y = 0$ implies that $(\mathring{X},\mathring{Y}) = 0$. In this case, plugging in $\left(\varphi_1,\varphi_2\right) = (\mathring{X},\mathring{Y})$ into the weak formulation of the equations, it follows immediately that $\mathcal{L}\left(\mathring{X},\mathring{Y}\right) = 0$. Then Lemma~\ref{coercive} implies that $(\mathring{X},\mathring{Y})$ vanishes.
\end{proof}
\subsubsection{A Second Renormalization and a Conformal Change of Coordinates near $p_N$ and $p_S$}
In the equation~(\ref{linY}) the term multiplying $\mathring{Y}$ is very singular along $\{\rho = 0\}$ (the reader may wish to recall~(\ref{Xlikeh}) and~(\ref{hsing})). Before we can proceed with further estimates with $\mathring{Y}$, we will require an additional renormalization and function spaces.

We define\index{Metric Quantities!$\tilde Y$}
\[\tilde Y \doteq X_K^{-1}\mathring{Y}.\]
(This renormalization has in fact already appeared in the proof of Lemma~\ref{coercive}).

A straightforward calculation shows that $\tilde Y$ weakly satisfies the following equation on $\mathbb{R}^3\setminus\{\rho =0\}$:
\begin{equation}\label{tildeeqnY}
\Delta_{\mathbb{R}^3}\tilde Y + \frac{2\partial X_K\cdot\partial \tilde Y}{X_K} - \frac{2\left|\partial Y_K\right|^2}{X_K^2}\tilde Y - \frac{2\partial Y_K\cdot\partial\mathring{X}}{X_K^2} = H_YX_K^{-1}.
\end{equation}
Next, see~(\ref{hsing}), we note that
\[\frac{2\partial X_K\cdot\partial \tilde Y}{X_K} = \frac{4}{\rho}\partial_{\rho}\tilde Y + \tilde{e}_A\cdot \partial \tilde Y,\]
for some vector field \index{Metric Quantities!$\tilde{e}_{A}$}$\tilde{e}_A$ which is smooth in $\overline{\mathscr{B}_A}$.

In particular, if we write the flat metric on $\mathbb{R}^7$ as
\[d\rho^2 + \rho^2d\mathbb{S}^5 + dz^2,\]
and identify $\tilde Y\left(\rho,z\right)$ as a function on $\mathbb{R}^7$ in the obvious way, then we may write~(\ref{tildeeqnY}) as
\begin{equation}\label{tildeeqnYA}
\Delta_{\mathbb{R}^7}\tilde Y + \tilde{e}_A\cdot\partial \tilde Y - \frac{2\left|\partial Y_K\right|^2}{X_K^2}\tilde Y - \frac{2\partial Y_K\cdot\partial\mathring{X}}{X_K^2} = H_YX_K^{-1}.
\end{equation}
Furthermore, note that
\begin{equation}\label{xterm}
\frac{\partial Y_K\cdot\partial \mathring{X}}{X_K^2} = \tilde{d}_A\cdot \partial\mathring{X} + \frac{\check{d}_A}{\rho}\partial_{\rho}\mathring{X},
\end{equation}
where in $\overline{\mathscr{B}_A}$, the vector field \index{Metric Quantities!$\tilde{d}_{A}$}$\tilde{d}_A$ and function \index{Metric Quantities!$\check{d}_{A}$}$\check{d}_A$ are smooth and all of their derivatives are bounded by a constant times $r^{-4}$. Thus, within the set $\overline{\mathscr{B}_A}$, the equation~(\ref{tildeeqnYA}) has smooth coefficients multiplying $\tilde Y$ (note also that $\rho^{-1}\partial_{\rho}\mathring{X}$ is bounded by the $\mathbb{R}^3$ Hessian of $\mathring{X}$ via Lemma~\ref{firstdecay}).

It is this equation that we will use to ultimately establish the $C^{2,\alpha_0}\left(\overline{\mathscr{B}_A}\right)$ estimates for $\tilde Y$. To emphasize that this identification with $\mathbb{R}^7$ only makes sense within the set $\overline{\mathscr{B}_A}\cup\overline{\mathscr{B}_H}$, we will write \index{Miscellaneous!$\mathbb{R}^7_A$}$\mathbb{R}^7_A$. Then, for any function $f(\rho,z) : \overline{\mathscr{B}_A}\cup\overline{\mathscr{B}_H} \to \mathbb{R}$, we define Sobolev norms
\[\left\Vert f\right\Vert_{\dot{W}^{k,p}\left(\mathbb{R}^7_A\right)} \doteq \left(\int\int_{\overline{\mathscr{B}_A}}\left|\partial^kf\right|^p\, \rho^5\, dz\, d\rho\right)^{1/p}.\]
We similarly define $W^{k,p}\left(\mathbb{R}^7_A\right)$ and $C^{k,\alpha_0}\left(\mathbb{R}^7_A\right)$ norms.

It also turns out\footnote{See \cite{chrusciellopes} for a thorough discussion of this point.} that both $\mathring{X}$'s and $\mathring{Y}$'s (and $\tilde Y$'s) equation have singularities near the points $p_N$ and $p_S$. Geometrically, this singularities are due to the fact that in the actual black hole spacetime, the axis of symmetry should meet the event horizon orthogonally, but in $(\rho,z)$ coordinates, they appear to be parallel. This difficulty is resolved by the use of $(s,\chi)$ coordinates. Indeed, using the conformal invariance of $\mathring{Y}$'s equation, and writing the flat metric on $\mathbb{R}^4$ in double polar coordinates:
\[ds^2 + s^2d\phi_1^2 + d\chi^2 + \chi^2d\phi_2^2,\]
we find that
\begin{equation}\label{tildeeqnYp}
\underline{\Delta}_{\mathbb{R}^4_N}\tilde Y + \frac{2\underline{\partial} X_K\cdot\underline{\partial} \tilde Y}{X_K} - \frac{2\left|\underline{\partial} Y_K\right|^2}{X_K^2}\tilde Y - \frac{2\underline{\partial} Y_K\cdot\underline{\partial}\mathring{X}}{X_K^2} = \left(\chi^2+s^2\right)H_YX_K^{-1}.
\end{equation}
We underline the Laplacian to emphasize that it arises from $(s,\chi)$ coordinates. Also, similarly to the above we shall write \index{Miscellaneous!$\mathbb{R}^{4}_{N}$}$\mathbb{R}^4_N$ to emphasize that this only makes sense on the set $\overline{\mathscr{B}_N}$.

This time we have
\[\frac{2\underline{\partial} X_K\cdot\underline{\partial} \tilde Y}{X_K} = \frac{4}{s}\underline{\partial}_s\tilde Y + \tilde{e}_N\cdot \underline{\partial} \tilde Y,\]
for a vector field \index{Metric Quantities!$\tilde{e}_N$}$\tilde{e}_N$ which is smooth in $\overline{\mathscr{B}_N}$. Thus, analogously to before, we obtain
\begin{equation}\label{tildeeqnYN}
\underline{\Delta}_{\mathbb{R}^8}\tilde Y + \tilde{e}_N\cdot\underline{\partial} \tilde Y - \frac{2\left|\underline{\partial} Y_K\right|^2}{X_K^2}\tilde Y - \frac{2\underline{\partial} Y_K\cdot\underline{\partial}\mathring{X}}{X_K^2} = \left(\chi^2+s^2\right)H_YX_K^{-1}.
\end{equation}
Finally, we also may introduce the $\dot{W}^{k,p}\left(\mathbb{R}^8_N\right)$ and $C^{k,\alpha_0}\left(\mathbb{R}^8_N\right)$ norms.

Of course, the exact same procedure may be carried out in the set $\overline{\mathscr{B}_S}$. The relevant equation is

\begin{equation}\label{tildeeqnYS}
\underline{\Delta}_{\mathbb{R}^8}\tilde Y + \tilde{e}_S\cdot\underline{\partial} \tilde Y - \frac{2\left|\underline{\partial} Y_K\right|^2}{X_K^2}\tilde Y - \frac{2\underline{\partial} Y_K\cdot\underline{\partial}\mathring{X}}{X_K^2} = \left((\chi')^2+(s')^2\right)H_YX_K^{-1},
\end{equation}
and we obtain the $\dot{W}^{k,p}\left(\mathbb{R}^8_S\right)$ and $C^{k,\alpha_0}\left(\mathbb{R}^8_S\right)$ norms.

Finally, we may also carry out the same procedure for $\mathring{X}$'s equation. We obtain
\begin{equation}\label{linXcorrect}
\underline{\Delta}_{\mathbb{R}^4}\mathring{X} + \frac{2\underline{\partial} Y_K\cdot\partial\mathring{Y}}{X_K} - \frac{2\left|\underline{\partial} Y_K\right|^2}{X_K^2}\mathring{X} + 2\frac{\underline{\partial} X_K\cdot\underline{\partial} Y_K}{X_K^2}\mathring{Y} = \left(\chi^2+s^2\right)H_X,
\end{equation}
and we obtain the $\dot{W}^{k,p}\left(\mathbb{R}^4_N\right)$, $\dot{W}^{k,p}\left(\mathbb{R}^4_S\right)$, etc., norms.
\subsubsection{The Estimates}
In this section we will establish estimates for solutions to~(\ref{linX}) and~(\ref{linY}). We start with the most basic $L^2$ estimates and work our way up to $C^{2,\alpha_0}$ estimates. Throughout this section we will suppose that $H_X,H_Y \in C^{\infty}_0\left(\mathbb{R}^3\right)$ (at the very end we will easily be able to weaken these assumptions on $H_X$ and $H_Y$ with a standard density argument), and let $(\mathring{X},\mathring{Y}) \in \dot{H}^1_{\rm axi}\left(\mathbb{R}^3\right) \times \left(\dot{H}^1_{\rm axi}\left(\mathbb{R}^3\right) \cap L^2_{\nabla h}\left(\mathbb{R}^3\right)\right)$ be the unique solution to~(\ref{linX}) and~(\ref{linY}).

We start with the ``free'' energy estimate which comes from the variational structure of Section~\ref{variational}.
\begin{lemma}\label{energy}There exists a constant $C > 0$ such that we have
\[\left\Vert \mathring{X}\right\Vert_{\dot{H}^1_{\rm axi}\left(\mathbb{R}^3\right)} + \left\Vert \mathring{Y}\right\Vert_{\dot{H}^1_{\rm axi}\left(\mathbb{R}^3\right) \cap L^2_{\nabla h}\left(\mathbb{R}^3\right)} \leq C\left[\left\Vert H_X\mathring{X}\right\Vert_{L^1\left(\mathbb{R}^3\right)}^{1/2} + \left\Vert H_Y \mathring{Y}\right\Vert_{L^1\left(\mathbb{R}^3\right)}^{1/2}\right].\]
\end{lemma}
\begin{proof}This is an immediate consequence of Lemma~\ref{coercive} and the observation that we must have $\mathcal{L}\left(\mathring{X},\mathring{Y}\right) \leq 0$.
\end{proof}
\begin{remark}Note that due to the compact support of $H_X$ and $H_Y$ and Lemma~\ref{bar}, the right hand side of this estimate is finite.
\end{remark}

Next, working away from the points $p_N$ and $p_S$ where $\mathring{X}$'s equation is singular, we apply a standard elliptic Hessian estimate to~(\ref{linX}), and apply Lemma~\ref{bar} to handle the lower order terms.

\begin{lemma}There exists a constant $C > 0$ such that
\[\left\Vert \left(1-\xi_N-\xi_S\right)\mathring{X}\right\Vert_{\dot{H}^2\left(\mathbb{R}^3\right)} \leq C\left[\left\Vert \left(1-\xi_N-\xi_S\right)H_X\right\Vert_{L^2\left(\mathbb{R}^3\right)}+ \left\Vert H_X\mathring{X}\right\Vert_{L^1\left(\mathbb{R}^3\right)}^{1/2} + \left\Vert H_Y \mathring{Y}\right\Vert_{L^1\left(\mathbb{R}^3\right)}^{1/2}\right].\]
\end{lemma}
Recall that the cut-offs $\xi_N$, and $\xi_S$ where defined in Definition~\ref{cutoffs}.

Near the points $p_N$ and $p_S$ we use the equation~(\ref{linXcorrect}) to carry out a Hessian estimate. We obtain
\begin{lemma}\label{nearpnps} There exists a constant $C > 0$ such that
\[\left\Vert \xi_N\mathring{X}\right\Vert_{\dot{H}^2\left(\mathbb{R}^4_N\right)} + \left\Vert \xi_S\mathring{X}\right\Vert_{\dot{H}^2\left(\mathbb{R}^4_S\right)} \leq \]
\[C\left[\left\Vert \xi_N\left(\chi^2+s^2\right)H_X\right\Vert_{L^2\left(\mathbb{R}^4_N\right)}+ \left\Vert \xi_S\left(\chi^2+s^2\right)H_X\right\Vert_{L^2\left(\mathbb{R}^4_S\right)} + \left\Vert H_X\mathring{X}\right\Vert_{L^1\left(\mathbb{R}^3\right)}^{1/2} + \left\Vert H_Y \mathring{Y}\right\Vert_{L^1\left(\mathbb{R}^3\right)}^{1/2}\right].\]
\end{lemma}

Naively, one would expect the next estimate to be a Hessian estimate for $\mathring{Y}$ using~(\ref{linY}). However, we cannot proceed with this directly since we do not yet know that $\frac{\left|\partial X_K\right|^2}{X_K^2}\mathring{Y} \in L^2\left(\mathbb{R}^3\right)$. Instead, using the equations~(\ref{tildeeqnYA}),~(\ref{tildeeqnYN}), and~(\ref{tildeeqnYS}), we will do separate Hessian estimates in each of the regions $\overline{\mathscr{B}_A}\cup\overline{\mathscr{B}_H}$, $\overline{\mathscr{B}_N}$, and $\overline{\mathscr{B}_S}$.

We start with an estimate away from the axis where the dangerous lower order term is not actually singular. Let $\xi_{A}$ be a smooth cut-off function which is identically $1$ in a sufficiently small neighborhood of the axis $\mathscr{A}$ and vanishes in a slightly larger neighborhood. 
\begin{lemma}\label{hessYE}There exists a constant $C > 0$ such that
\[\left\Vert \left(1-\xi_N-\xi_S\right)\left(1-\xi_{A}\right)\mathring{Y}\right\Vert_{\dot{H}^2\left(\mathbb{R}^3\right)} \leq \]
\[C\left[\left\Vert \left(1-\xi_N-\xi_S\right)\left(1-\xi_{A}\right)H_Y\right\Vert_{L^2\left(\mathbb{R}^3\right)} + \left\Vert H_X\mathring{X}\right\Vert_{L^1\left(\mathbb{R}^3\right)}^{1/2} + \left\Vert H_Y \mathring{Y}\right\Vert_{L^1\left(\mathbb{R}^3\right)}^{1/2}\right].\]
\end{lemma}
\begin{proof}On the support of $\left(1-\xi_N-\xi_S\right)\left(1-\xi_{A}\right)$,~(\ref{linY}) is a uniformly elliptic equation with smooth coefficients, so the lemma follows from a standard elliptic estimate.
\end{proof}

Next, we carry out a Hessian estimate for $\tilde Y$ in $\mathbb{R}^7_A$.
\begin{lemma}\label{hessYA}There exists a constant $C > 0$ such that

\[\left\Vert \left(1-\xi_N-\xi_S\right)\xi_{A}\tilde Y\right\Vert_{\dot{H}^2\left(\mathbb{R}^7_A\right)} \leq \]
\[C\left[\left\Vert \left(1-\xi_N-\xi_S\right)\xi_{A}H_YX_K^{-1}\right\Vert_{L^2\left(\mathbb{R}^7_A\right)}+ \left\Vert H_X\mathring{X}\right\Vert_{L^1\left(\mathbb{R}^3\right)}^{1/2} + \left\Vert H_Y \mathring{Y}\right\Vert_{L^1\left(\mathbb{R}^3\right)}^{1/2}\right].\]
\end{lemma}
\begin{proof}A straightforward calculation using that the volume form on $\mathbb{R}^7_A$ is $\rho^5\, d\rho\, dz$ yields
\begin{align*}
\left\Vert \tilde Y\right\Vert_{\dot{H}^1_{\rm axi}\left(\mathbb{R}^7_A\right) \cap L^2_{\rho^{-1}}\left(\mathbb{R}^7_A\right)} &\leq C\left\Vert \mathring{Y}\right\Vert_{\dot{H}^1_{\rm axi}\left(\overline{\mathscr{B}_A}\right) \cap L^2_{\nabla h}\left(\overline{\mathscr{B}_A}\right)}
\\ \nonumber &\leq C\left[\left\Vert H_X\mathring{X}\right\Vert_{L^1\left(\mathbb{R}^3\right)}^{1/2} + \left\Vert H_Y \mathring{Y}\right\Vert_{L^1\left(\mathbb{R}^3\right)}^{1/2}\right].
\end{align*}
Next, keeping in mind that $\rho$ is bounded in the set $\mathbb{R}^7_A$, the expression~(\ref{xterm}) implies that
\begin{align*}
\left\Vert \frac{\partial Y_K\cdot \partial\mathring{X}}{X_K^2}\right\Vert_{L^2\left(\mathbb{R}^7_A\right)} &\leq C\left\Vert \mathring{X}\right\Vert_{\dot{H}^1_{\rm axi}\left(\mathbb{R}^3\right)}
\\ \nonumber &\leq  C\left[\left\Vert H_X\mathring{X}\right\Vert_{L^1\left(\mathbb{R}^3\right)}^{1/2} + \left\Vert H_Y \mathring{Y}\right\Vert_{L^1\left(\mathbb{R}^3\right)}^{1/2}\right].
\end{align*}
Conversely, we could also use the fact that $\left|\rho^{-1}\partial_{\rho}\mathring{X}\right| \leq C\left|\hat{\partial}^2\mathring{X}\right|$ (this approach is useful when we do higher $L^p$ estimates).

Thus, the lemma follows from a straightforward elliptic estimate (note that a standard argument yields that $\tilde Y$'s equation is weakly satisfied on $\mathbb{R}^7_A$ rather than just $\mathbb{R}^7_A\setminus \{\rho = 0\}$).
\end{proof}

Repeating the same procedure in the sets $\overline{\mathscr{B}_N}$ and $\overline{\mathscr{B}_S}$ immediately yields the following.
\begin{lemma}\label{hessYA}There exists a constant $C > 0$ such that
\[\left\Vert \xi_N\tilde{Y}\right\Vert_{\dot{H}^2\left(\mathbb{R}^8_N\right)} \leq C\left[\left\Vert H_X\mathring{X}\right\Vert_{L^1\left(\mathbb{R}^3\right)}^{1/2} + \left\Vert H_Y \mathring{Y}\right\Vert_{L^1\left(\mathbb{R}^3\right)}^{1/2} + \left\Vert \left(\chi^2+s^2\right)H_Y\right\Vert_{L^2\left(\overline{\mathscr{B}_N}\right)}\right],\]

\[\left\Vert \xi_S\tilde{Y}\right\Vert_{\dot{H}^2\left(\mathbb{R}^8_S\right)} \leq C\left[\left\Vert H_X\mathring{X}\right\Vert_{L^1\left(\mathbb{R}^3\right)}^{1/2} + \left\Vert H_Y \mathring{Y}\right\Vert_{L^1\left(\mathbb{R}^3\right)}^{1/2} + \left\Vert \left((\chi')^2+(s')^2\right)H_Y\right\Vert_{L^2\left(\overline{\mathscr{B}_S}\right)}\right].\]
\end{lemma}

Next we rephrase the estimates for $\tilde Y$ in terms of $\mathring{Y}$.
\begin{lemma}There exists a constant $C > 0$ such that
\[\left\Vert \left(1-\xi_N-\xi_S\right)\xi_A\mathring{Y} \right\Vert_{\dot{H}^2\left(\mathbb{R}^3\right)} \leq C\left[\left\Vert H_X\mathring{X}\right\Vert_{L^1\left(\mathbb{R}^3\right)}^{1/2} + \left\Vert H_Y \mathring{Y}\right\Vert_{L^1\left(\mathbb{R}^3\right)}^{1/2} + \left\Vert \left(1-\xi_N-\xi_S\right)\xi_{A}H_Y\right\Vert_{L^2\left(\mathbb{R}^3\right)}\right],\]

\[\left\Vert \xi_N\mathring{Y}\right\Vert_{\dot{H}^2\left(\mathbb{R}^4_N\right)} \leq C\left[\left\Vert H_X\mathring{X}\right\Vert_{L^1\left(\mathbb{R}^3\right)}^{1/2} + \left\Vert H_Y \mathring{Y}\right\Vert_{L^1\left(\mathbb{R}^3\right)}^{1/2} + \left\Vert \left(\chi^2+s^2\right)H_Y\right\Vert_{L^2\left(\overline{\mathscr{B}_N}\right)}\right],\]

\[\left\Vert \xi_S\mathring{Y}\right\Vert_{\dot{H}^2\left(\mathbb{R}^4_S\right)} \leq C\left[\left\Vert H_X\mathring{X}\right\Vert_{L^1\left(\mathbb{R}^3\right)}^{1/2} + \left\Vert H_Y \mathring{Y}\right\Vert_{L^1\left(\mathbb{R}^3\right)}^{1/2} + \left\Vert \left((\chi')^2+(s')^2\right)H_Y\right\Vert_{L^2\left(\overline{\mathscr{B}_S}\right)}\right].\]
\end{lemma}
\begin{proof}First of all, by repeated use of the $1$-dimensional inequalities
\[\int_0^{\infty}f^2\rho^3\, d\rho\, dz \leq C\int_0^{\infty}(\partial_{\rho}f)^2\rho^5\, d\rho\, dz,\qquad \int_0^{\infty}f^2\rho\, d\rho\, dz \leq C\int_0^{\infty}(\partial_{\rho}f)^2\rho^3\, d\rho\, dz,\]
which hold for any smooth compactly supported function $f$ which vanishes for large $\rho$, one easily obtains
\[\left\Vert \left(1-\xi_N-\xi_S\right)\xi_A\partial\tilde{Y}\right\Vert_{L^2_{\rho^{-1}}\left(\mathbb{R}^7_A\right)} +
\left\Vert \left(1-\xi_N-\xi_S\right)\xi_A\tilde{Y}\right\Vert_{L^2_{\rho^{-2}}\left(\mathbb{R}^7_A\right)} \leq\]
\[C\left[\left\Vert H_X\mathring{X}\right\Vert_{L^1\left(\mathbb{R}^3\right)}^{1/2} + \left\Vert H_Y \mathring{Y}\right\Vert_{L^1\left(\mathbb{R}^3\right)}^{1/2} + \left\Vert \left(1-\xi_N-\xi_S\right)\xi_{A}H_Y\right\Vert_{L^2\left(\mathbb{R}^3\right)}\right].\]

Next, we observe that a straightforward calculation yields
\[\left|\partial^2\tilde Y\right|^2 \geq \left|\partial^2\mathring{Y}\right|^2\rho^{-4} - C\left(\left|\partial\tilde Y\right|^2\rho^{-2} + \tilde Y^2\rho^{-4}\right). \]
Thus,
\[\left\Vert \left(1-\xi_N-\xi_S\right)\xi_A\mathring{Y} \right\Vert_{\dot{H}^2\left(\overline{\mathscr{B}_A}\right)} \leq C\left[\left\Vert H_X\mathring{X}\right\Vert_{L^1\left(\mathbb{R}^3\right)}^{1/2} + \left\Vert H_Y \mathring{Y}\right\Vert_{L^1\left(\mathbb{R}^3\right)}^{1/2} + \left\Vert \left(1-\xi_N-\xi_S\right)\xi_{A}H_Y\right\Vert_{L^2\left(\mathbb{R}^3\right)}\right].\]

The other two estimates follow similarly.
\end{proof}

Next, using Sobolev estimates, we obtain $L^p$ bounds for $\mathring{X}$, $\mathring{Y}$, and $\tilde Y$. Carrying out elliptic estimates as above \emph{mutatis mutandis} then yields $\dot{W}^{2,p}$ estimates. Iterating yields $\dot{W}^{2,p}$ estimates for any $p \geq 2$ and then, via Sobolev inequalities, we obtain $\hat{C}^1$ estimates. We record this in the following lemma.
\begin{lemma}\label{w1infty}There exists a constant $C > 0$ such that
\begin{align*}
&\left\Vert \mathring{X}\right\Vert_{\hat{C}^1\left(\overline{\mathscr{B}}\right)} + \left\Vert \mathring{Y}\right\Vert_{\hat{C}^1\left(\overline{\mathscr{B}}\right)} + \left\Vert \tilde{Y}\right\Vert_{\hat{C}^1\left(\overline{\mathscr{B}}\right)}
\\ \nonumber &\qquad \qquad \leq C\Big[\left\Vert H_X\mathring{X}\right\Vert_{L^1\left(\mathbb{R}^3\right)}^{1/2} + \left\Vert H_Y \mathring{Y}\right\Vert_{L^1\left(\mathbb{R}^3\right)}^{1/2}+ \left\Vert  \left(1-\xi_N-\xi_S\right)H_X\right\Vert_{L^2\left(\mathbb{R}^3\right) \cap L^{\infty}\left(\mathbb{R}^3\right)} +
\\ \nonumber &\qquad \qquad \qquad  \left\Vert \left(1-\xi_N-\xi_S\right)\left(1-\xi_A\right)H_Y\right\Vert_{L^2\left(\mathbb{R}^3\right) \cap L^{\infty}\left(\mathbb{R}^3\right)}  + \left\Vert\left(1-\xi_N-\xi_S\right)\xi_A X_K^{-1}H_Y\right\Vert_{L^2\left(\mathbb{R}^7_A\right) \cap L^{\infty}\left(\mathbb{R}^7_A\right)} +
\\ \nonumber &\qquad \qquad \qquad \left\Vert \left(\chi^2+s^2\right)\xi_NX_K^{-1}H_Y\right\Vert_{L^{\infty}\left(\mathbb{R}^8_N\right)} + \left\Vert \left((\chi')^2+(s')^2\right)X_K^{-1}\xi_SH_Y\right\Vert_{L^{\infty}\left(\mathbb{R}^8_S\right)}
\\ \nonumber &\qquad \qquad \qquad \left\Vert \left(\chi^2+s^2\right)\xi_NH_X\right\Vert_{L^{\infty}\left(\mathbb{R}^4_N\right)} + \left\Vert \left((\chi')^2+(s')^2\right)X_K^{-1}\xi_SH_X\right\Vert_{L^{\infty}\left(\mathbb{R}^4_S\right)}\Big].
\end{align*}
\end{lemma}

It is now easy to add in Schauder estimates, pointwise decay estimates, and use a standard density argument in order to finally conclude with the following proposition.
\begin{proposition}\label{c2alpha}Suppose that $\left(H_X,H_Y\right) \in \mathcal{N}_{X,Y}$. Then, if we let $\left(\mathring{X},\mathring{Y}\right)$ denote the unique solutions to~(\ref{linX}) and~(\ref{linY}), then there exists a constant $D(\alpha_0)$, independent of $\left(H_X,H_Y\right)$ and depending on $\alpha_0$, such that
\[\left\Vert \left(\mathring{X},\mathring{Y}\right)\right\Vert_{\mathcal{L}_{X,Y}} \leq D(\alpha_0)\left\Vert \left(H_X,H_Y\right)\right\Vert_{\mathcal{N}_{X,Y}}.\]

\end{proposition}
\begin{proof}
First of all, we observe that other than the terms $\left\Vert H_X\mathring{X}\right\Vert_{L^1\left(\mathbb{R}^3\right)}^{1/2} + \left\Vert H_Y \mathring{Y}\right\Vert_{L^1\left(\mathbb{R}^3\right)}^{1/2}$, some straightforward calculations show that the right hand side of the estimate in Lemma~\ref{w1infty} is easily seen to be controlled by $\left\Vert \left(H_X,H_Y\right)\right\Vert_{N_{X,Y}}$. For the terms $\left\Vert H_X\mathring{X}\right\Vert_{L^1\left(\mathbb{R}^3\right)}^{1/2} + \left\Vert H_Y \mathring{Y}\right\Vert_{L^1\left(\mathbb{R}^3\right)}^{1/2}$, we use the estimate
\begin{align*}
\left\Vert H_X\mathring{X}\right\Vert_{L^1\left(\mathbb{R}^3\right)}^{1/2} \leq C\Big[&q^{-1}\left\Vert \left(1-\xi_N-\xi_S\right)H_X\right\Vert_{L^{6/5}\left(\mathbb{R}^3\right)} + q\left\Vert \left(1-\xi_N-\xi_S\right)\mathring{X}\right\Vert_{\dot{H}^1\left(\mathbb{R}^3\right)} +
\\ \nonumber &q^{-1}\left\Vert \xi_N H_X\right\Vert_{L^1\left(\mathbb{R}^3\right)} + q\left\Vert \xi_N\mathring{X}\right\Vert_{L^{\infty}\left(\mathbb{R}^3\right)}
\\ \nonumber &q^{-1}\left\Vert \xi_S H_X\right\Vert_{L^1\left(\mathbb{R}^3\right)} + q\left\Vert \xi_S\mathring{X}\right\Vert_{L^{\infty}\left(\mathbb{R}^3\right)}\Big],
\end{align*}
which holds for any $q > 0$. Here we have used the Sobolev embedding $L^6\left(\mathbb{R}^3\right) \hookrightarrow \dot{H}^1\left(\mathbb{R}^3\right)$. The terms with the small parameter $q$ may be absorbed into terms we have estimated previously and the other terms may be controlled by $\left\Vert \left(H_X,H_Y\right)\right\Vert_{N_{X,Y}}$.\footnote{It may be useful for the reader to keep in mind that $\rho = s\chi$ in $(s,\chi)$ coordinates and thus one may easily check that $\left\Vert \xi_N f\right\Vert_{L^1\left(\mathbb{R}^3\right)} \leq C \left\Vert \xi_N \left(\chi^2+s^2\right)f\right\Vert_{L^{\infty}\left(\mathbb{R}^3\right)}$.}A similar argument works for the term $\left\Vert H_Y \mathring{Y}\right\Vert_{L^1\left(\mathbb{R}^3\right)}^{1/2}$.

Next, we note that the $C^{2,\alpha_0}$ estimates follows immediately from  Lemma~\ref{w1infty} and local Schauder estimates on balls of radius $1$. More precisely, we apply interior Schauder estimates (cf.\ Theorem 6.6 in~\cite{giltru}) in balls of a fixed radius. The coordinates in which we apply these estimates vary with the location of the balls, but we use the coordinates in which the equation is uniformly elliptic with regular lower order terms. For example, near (and on) the axis and horizon but away from the points $p_{N}$ and $p_{S}$, we use equations \eqref{linX} and \eqref{tildeeqnYA}. Similarly, near the points $p_{N}$ and $p_{S}$, we use \eqref{linXcorrect} and \eqref{tildeeqnYN}. 

The pointwise decay arguments are standard so we will be brief in our presentation:

Pointwise decay for $\mathring{X}$ follows immediately from Lemma~\ref{newt1} and Lemma~\ref{newt2} from the Appendix, the already established $W^{1,\infty}$ estimates for $\mathring{X}$ and $\mathring{Y}$, and the estimates for $(X_K,Y_K)$ from Section \ref{kerr}.

For $\tilde{Y}$, we first note that $\left\Vert \tilde Y\right\Vert_{\dot{H}^1_{\rm axi}\left(\mathbb{R}^7_A\right)} \leq C\left\Vert \mathring{Y}\right\Vert_{\dot{H}^1_{\rm axi}\left(\mathbb{R}^3\right) \cap L^2_{\nabla h}\left(\mathbb{R}^3\right)}$. A Sobolev inequality then yields control of $\tilde Y$ in $L^{\frac{14}{5}}\left(\mathbb{R}^7_A\right)$. Next, we note that
\[\left|\frac{2\partial Y_K\cdot\partial \mathring{X}}{X_K^2}\right|\leq  Cr^{-6}\left[\left\Vert r^2\hat{\partial}\mathring{X}\right\Vert_{L^{\infty}\left(\overline{\mathscr{B}}\right)} +\left\Vert r^2\hat{\partial}^2\mathring{X}\right\Vert_{L^{\infty}\left(\overline{\mathscr{B}}\right)}\right].\]
Since, we also have $|\tilde e_A| \leq Cr^{-2}$ and $\frac{|\partial Y_K|^2}{X_K^2}\leq Cr^{-6}$, we can apply the mean value inequality (Theorem 8.17 of~\cite{giltru}) with $n = 7$ and $p = \frac{14}{5}$, to obtain a decay rate of $r^{-5/2}$ for $\tilde Y$. Local Schauder estimates then yield a decay rate of $r^{-5/2}$ for both $\hat{\partial}\tilde Y$ and $\hat{\partial}^2\tilde Y$. In particular, we have obtained a decay rate of $r^{-9/2}$ for $\tilde e_A \cdot\partial\tilde Y$ and a decay rate $r^{-17/2}$ for $\frac{|\partial Y_K|^2}{X_K^2}\tilde Y$. Thus, treating these two terms as errors, we can use Lemma~\ref{newt1} from the Appendix to obtain a decay rate of $-7/2$ for $\partial\tilde Y$. Using this improved estimate for $\partial\tilde Y$ and applying the Newton potential estimate again yields a decay rate of $-7/2$ for $\tilde Y$ and a decay rate of $-9/2$ for $\partial\tilde Y$. One final iteration using also Lemma~\ref{newt2} yields the desired decay rate of $r^{-3}$ for $\tilde Y$, $r^{-4}$ for $\partial \tilde Y$, and $r^{-5}\log\left(4r\right)$ for $\hat{\partial}^2\tilde Y$.
\end{proof}
\subsection{The Fixed Point}
We are ready to run the fixed point argument.
\begin{proposition}\label{fixxy}For each $m \in \mathbb{Z}_{\neq 0}$ and $\alpha_0 \in (0,1)$, there exists $\epsilon > 0$ sufficiently small so that given \[\left(\mathring{\Theta},\psi,\mathring{\mu}^2,\mathring{\lambda}\right) \in B_{\epsilon}\left(\mathcal{L}_{\Theta}\right)\times \cdots\times B_{\epsilon}\left(\mathcal{L}_{\lambda}\right),\]
we may find
\[\left(\mathring{\sigma},\left(B^{(N)}_{\chi},B^{(S)}_{\chi'},B^{(A)}_z\right),\left(\mathring{X},\mathring{Y}\right) \right)\in B_{\epsilon}\left(\mathcal{L}_{\sigma}\right)\times \cdots \times B_{\epsilon}\left(\mathcal{L}_{X,Y}\right),\]
which solve~(\ref{renormalsigmaeqn}), (\ref{Btripleeqn}), and~(\ref{X1}) and~(\ref{Y1}) respectively.

Furthermore, $\mathring{\sigma}$, $\left(B^{(N)}_{\chi},B^{(S)}_{\chi'},B^{(A)}_z\right)$, and $(\mathring{X},\mathring{Y})$ have ``continuous nonlinear dependence on parameters''
\[\left(\mathring{\Theta},\psi,\mathring{\mu}^2,\mathring{\lambda}\right),\]
in the sense of Lemma~\ref{fixit} and Remark~\ref{contnonlindep}.
\end{proposition}
\begin{proof}We first apply Propositions~\ref{fixsig} and~\ref{fixb} to solve for $\mathring{\sigma}$ and $\left(B^{(N)}_{\chi},B^{(S)}_{\chi'},B^{(A)}_z\right)$ in terms of the other unknowns.

The remaining part of the proof will follow from Lemma~\ref{fixit} once we establish the necessary nonlinear estimates for $N_X$ and $N_Y$; these estimates follow from a straightforward (but somewhat lengthy) analysis of each term, which we now describe. 

Recall that $N_{X}=N_{X}^{(1)}+N_{X}^{(2)}$ where 
\[
N_{X}^{(1)} = \frac{|\partial\mathring X|^{2}-|\partial \mathring Y|^{2}}{1+\mathring X} + \frac{(\mathring X \partial Y_{K} - \mathring Y \partial X_{K})(2X_{K}\partial \mathring Y - \mathring X \partial Y_{K} + \mathring Y \partial X_{K})}{X_{K}^{2}(1+\mathring X)} 
\]
and 
\begin{align*}
N_{X}^{(2)} & = (\rho^{-1}-\sigma^{-1}\partial_{\rho}\sigma) \frac{\partial_{\rho}(X_{K}(1+\mathring X))}{X_{K}} - \sigma^{-1}\partial_{z}\sigma \frac{\partial_{z}(X_{K}(1+\mathring X))}{X_{K}} - 2 \frac{\partial_{\rho}(Y_{K}+X_{K}\mathring Y)B_{\rho}}{X_{K}^{2}(1+\mathring X)}\\
& \qquad - 2 \frac{\partial_{z}(Y_{K}+X_{K}\mathring Y)B_{z}}{X_{K}^{2}(1+\mathring X)} - \frac{B_{\rho}^{2}+ B_{z}^{2}}{X_{K}^{2}(1+\mathring X)} - X_{K}^{-1}e^{2\lambda}(2m^{2} + X_{K}(1+\mathring X)\mu^{2})\psi^{2}
\end{align*}
To establish the non-linear estimates for these terms (corresponding to (3) in Proposition~\ref{fixit}), we must show that 
\begin{equation}\label{eq:NX-nonlinear-estimate-quadratic}
\Vert N_{X}\Vert_{\mathcal{N}_{X}} \leq D \Vert (\mathring X, \mathring Y)\Vert_{\mathcal{L}_{X}\times \mathcal{L}_{Y}}^{2} + D \Vert \mathring \Theta, \psi,\mathring \mu^{2},\mathring \lambda\Vert_{\mathcal{L}_{\Theta}\times \dots\times \mathcal{L}_{\lambda}}^{2}.
\end{equation}
Note that here we are thinking of $\mathring \sigma$ and $\left(B^{(N)}_{\chi},B^{(S)}_{\chi'},B^{(A)}_z\right)$ as functions on the other unknowns, thanks to Propositions~\ref{fixsig} and~\ref{fixb} (these propositions include, among other things, the fact that $\mathring \sigma$ and $B$ have ``nonlinear dependence on parameters,'' in the sense of Lemma~\ref{fixit}). 

Recall that 
\begin{align*}
& \Vert N_{X}\Vert_{\mathcal{N}_{X}} \\
& = \Vert r^{3}(1-\xi_{N}-\xi_{S}) N_{X}\Vert_{C^{0,\alpha_{0}}(\RR^{3})} + \Vert (\chi^{2}+s^{2})\xi_{N}N_{X}\Vert_{C^{0,\alpha_{0}}(\RR^{3})} + \Vert ((\chi')^{2}+(s')^{2})\xi_{S}N_{X}\Vert_{C^{0,\alpha_{0}}(\RR^{3})}.
\end{align*}
The first term $\frac{|\partial\mathring X|^{2}-|\partial \mathring Y|^{2}}{1+\mathring X}$ in $N_{X}^{(1)}$ is easily bounded away from $p_{N},p_{S}$ using the fact that $|\partial\mathring X| \leq \Vert\mathring X\Vert_{\mathcal{L}_{X}} r^{-2}$ (with a similar estimate for the H\"older norm) and $|\partial\mathring Y| \leq D \Vert \mathring Y\Vert_{\mathcal{L}_{Y}} r^{-2}$ (and similarly for the H\"older norm).\footnote{Note that to derive the estimate for $\mathring Y$, we should write $\partial \mathring Y = \partial(X_{K}X_{K}^{-1}\mathring Y) = (\partial X_{K})(X_{K}^{-1}\mathring Y) + X_{K} \partial (X_{K}^{-1}\mathring Y)$. Now, recall that away from $p_{N},p_{S}$, $X_{K}/\rho^{2}$ is a smooth function (on $\RR^{3}$) with all derivatives bounded. Moreover, the $\mathcal{L}_{Y}$-norm bounds $r^{3} X_{K}^{-1}\mathring Y$ and $r^{4} \partial(X_{K}^{-1}\mathring Y)$ (away from $p_{N},p_{S}$, there is no difference between $\partial$ and $\hat\partial$, see \eqref{reder}). Thus, we obtain $|\partial \mathring Y| \leq D \Vert \mathring Y\Vert_{\mathcal{L}_{Y}} (\rho r^{-3} + \rho^{2}r^{-4})$. This yields the desired estimate. A similar argument handles the H\"older norm.} Now, we consider the behavior of this term near $p_{N}$, i.e., in the $(s,\chi)$ coordinates. Recall that 
\[
|\partial  f|^{2} = (\chi^{2}+s^{2})^{-1} |\underline \partial f|^{2} 
\]
(see \eqref{eq:def-partial}). Thus, we find that, for example, $(s^{2}+\chi^{2})|\partial \mathring X|^{2}$ is bounded by the $\hat C^{1}(\overline{\mathscr{B}})$ norm of $\mathring X$, which is bounded by the $\mathcal{L}_{X}$-norm of $\mathring X$. We may similarly bound the H\"older norm, and a similar argument controls the $\mathring Y$ term. This yields the desired bounds near $p_{N}$, a similar argument handles $p_{S}$. 

We now turn to the second term in $N_{X}^{(1)}$, which is similar although slightly more involved. Here, it is useful to recall that $|\partial Y_{K}| \leq \rho^{3}r^{-4}$ away from $p_{N},p_{S}$, by Lemma~\ref{somestuffthatisuseful}. Then,
\[
|\mathring X \partial Y_{K} - \mathring Y \partial X_{K}| + | 2X_{K} \partial\mathring Y - \mathring X \partial Y_{K} + \mathring Y \partial X_{K}| \leq D ( \Vert \mathring X\Vert_{\mathcal{L}_{X}} \rho^{3} r^{-5} + \Vert \mathring Y\Vert_{\mathcal{L}_{Y}}\rho^{3}r^{-3})
\]
Recalling that $X_{K}\sim \rho^{2}$, the second term in $N^{(1)}_{X}$ is thus bounded (away from $p_{N},p_{S}$) by $D \Vert (\mathring X,\mathring Y)\Vert_{\mathcal{L}_{X}\times \mathcal{L}_{Y}}^{2} \rho^{2} r^{-6}$. A similar argument bounds the H\"older norm. An argument similar to that used for the first term provides the requisite bounds near $p_{N},p_{S}$; we omit the details. 

We now turn to the estimates for $N_{X}^{(2)}$. It is useful to note that
\[
\rho^{-1} - \sigma^{-1} \partial_{\rho}\sigma = - \frac{\partial_{\rho}\mathring \sigma}{1+\mathring \sigma}
\]
(we will use this again later). We begin with the estimate for the first term in $N_{X}^{(2)}$ away from $p_{N},p_{S}$. The only difficult case is when the derivative hits $X_{K}$, yielding the expression
\[
-\frac{1+\mathring X}{1+\mathring\sigma} (\partial_{\rho} \log X_{K})\partial_{\rho}\mathring\sigma.
\]
Because $\partial_{\rho} \log X_{K} \sim \frac 1 \rho$, this term looks potentially worrisome. However, at this point, we make use of the rotational symmetry in the form of Lemma \ref{firstdecay}, which implies that $\rho^{-1}\partial_{\rho}\mathring \sigma|_{(x,y)} = \partial^{2}_{\rho}\mathring \sigma|_{(0,(x^{2}+y^{2})^{1/2})}$, where $x,y$ are the polar coordinates associated to $\rho$. As such, this expression is controlled in $C^{0,\alpha_{0}}$ away from $p_{N},p_{S}$ by $\Vert \mathring\sigma\Vert_{\mathcal{L}_{\sigma}}$. This, in turn, is bounded by the terms on the right hand side of \eqref{eq:NX-nonlinear-estimate-quadratic}, by the ``nonlinear dependence on parameters'' of $\mathring\sigma$. Note that the second term in $N_{X}^{(2)}$ is bounded similarly. 

We now explain how to control the first and second term near $p_{N}$ (i.e., in $s,\chi$ coordinates). As we will see, it is necessary to consider both terms simultaneously. As above, the difficult case will be when the derivatives hit the $X_{K}$ terms. We will ignore the other terms (which are easily bounded) and explain how to bound 
\[
- (\chi^{2}+s^{2}) \frac{1+\mathring X}{1+\mathring \sigma} ((\partial_{\rho}\log X_{K}) \partial_{\rho} \mathring \sigma+ (\partial_{z}\log X_{K} ) \partial_{z}\mathring \sigma).  
\]
At this point, we may use the following transformation rule: for any two $1$-forms $C = C_{\rho}d\rho + C_zdz$ and $E = E_{\rho}d\rho + E_zdz$, we have
\begin{equation}\label{rhotos}
C_{\rho}E_{\rho} + C_zE_z = \left(\chi^2+s^2\right)^{-1}\left(C_sE_s + C_{\chi}E_{\chi}\right) = \left((\chi')^2+(s')^2\right)^{-1}\left(C_{s'}E_{s'} + C_{\chi'}E_{\chi'}\right) .
\end{equation}
In particular, the above expression becomes 
\[
-  \frac{1+\mathring X}{1+\mathring \sigma} ((\partial_{s}\log X_{K}) \partial_{s} \mathring \sigma+ (\partial_{\chi}\log X_{K} ) \partial_{\chi}\mathring \sigma).  
\]
The potentially troublesome term $\partial_{s}\log X_{K}\sim s^{-1}$ is bounded (as above) using Lemma \ref{firstdecay} applied to $\partial_{s}\mathring\sigma$. The remaining terms are easily bounded. Estimates at $p_{S}$ follow similarly.

We now turn to the third term in $N_{X}^{(2)}$. Away from $p_{N},p_{S}$, the desired estimates follow immediately from the rapid decay of $B$ as $\rho \to 0$ and $r\to\infty$ (the term involving $\partial Y_{K}$ may be estimated using Lemma \ref{somestuffthatisuseful}). Similarly, near $p_{N}$, in $(s,\chi)$ coordinates this term (along with an extra factor of $s^{2}+\chi^{2}$, which appears in $\mathcal{N}_{X}$) becomes 
\[
- \frac{2}{1+\mathring X} (\chi^{2}+s^{2})^{-1}X_{K}^{-2} (\chi \partial_{s}(Y_{K} + X_{K}\mathring Y) + s\partial_{\chi}(Y_{K}+X_{K}\mathring Y))(sB_{\chi} + \chi B_{s})
\]
Because the $C^{0,\alpha_{0}}$ norm of $B$ in the $(s,\chi)$ coordinates vanishes like $s^{10}$, the $C^{0,\alpha_{0}}$ norm of this quantity near $p_{N}$ is readily bounded by $D \Vert B, (\mathring X,\mathring Y)\Vert_{\mathcal{L}_{B}\times\dots\times\mathcal{L}_{Y}}$, as desired. The fourth and fifth terms in $N_{X}^{(2)}$ are bounded in a similar manner. Finally, we turn to the last term in $N_{X}^{(2)}$. Away from $p_{N},p_{S}$, the rapid decay of $\psi$ as $\rho\to0,r\to\infty$ yields the desired bounds. A similar argument works near $p_{N},p_{S}$.

We now consider similar bounds for $N_{Y}$. We begin with $N_{Y}^{(1)}$. We would first like $\hat C^{0,\alpha_{0}}$ bounds for $r^{5}X_{K}^{-1}N_{Y}^{(1)}$ in $(\overline{\mathscr{B}_{A}}\cup\overline{\mathscr{B}_{H}} )\cap \{\rho\leq 1\}$. Note that
\[
X_{K}^{-1}N_{Y}^{(1)} =2 \frac{\partial \mathring X \cdot \partial (X_{K}^{-1}\mathring Y)}{1+\mathring X} + 4\frac{(X_{K}^{-1}\mathring Y) \partial X_{K}\cdot \partial \mathring X}{X_{K}(1+\mathring X)} - 2\frac{\mathring X \partial Y_{K} \cdot \partial\mathring X}{X_{K}^{2}(1+\mathring X)}
\]
The first term is clearly bounded as asserted (thanks to the decay enforced by $\mathcal{L}_{X},\mathcal{L}_{Y}$). The second term is slightly more complicated, in particular the $\rho$-component, since $\partial_{\rho}\log X_{K}\sim \rho^{-1}$. Here, we bound $\rho^{-1}\partial_{\rho}\mathring X$ in $\hat C^{0,\alpha_{0}}$ using Lemma \ref{firstdecay}, since $\partial^{2}\mathring X$ has H\"older norm decaying at least like $r^{-2}$. For the last term, we use Lemma \ref{somestuffthatisuseful} to bound $\partial Y_{K}$. We explain how to obtain point-wise decay, as the H\"older estimate follows similarly. As before, we must consider the $\rho$ and $z$ terms separately. We rewrite the third term as
\[
2(1+\mathring X)^{-1}\frac{\rho \partial_{\rho} Y_{K}}{X_{K}^{2}} \rho^{-1}\partial_{\rho}\mathring X + 2(1+\mathring X)^{-1}\frac{ \partial_{z} Y_{K}}{X_{K}^{2}} \partial_{z}\mathring X.
\]
Using Lemma \ref{firstdecay}, we may estimate $\rho^{-1}\partial_{\rho}\mathring X$ by $\partial^{2}_{\rho}\mathring X$, as above. The fractions involving $Y_{K}$ are bounded by Lemma \ref{somestuffthatisuseful}. The $C^{0,\alpha_{0}}$ estimate for $N_{Y}^{(1)}$ in $\{\rho\geq 1\}$ is similar, but more straightforward (since we do not need to consider the behavior as $\rho\to0$). We now consider the estimate near $p_{N}$. We will use the above expression for $X_{K}^{-1}N_{Y}^{(1)}$ (and then multiply by a single factor of $\chi^{2}+s^{2}$). The first term is easily bounded, since using \eqref{rhotos},
\[
(\chi^{2}+s^{2})\partial \mathring X \cdot \partial (X_{K}^{-1}\mathring Y) = (\partial_{s} \mathring X)(\partial _{s}(X_{K}^{-1}\mathring Y)) + (\partial_{\chi} \mathring X)(\partial _{\chi}(X_{K}^{-1}\mathring Y)) 
\]
and both of these terms are bounded in $C^{0,\alpha_{0}}(\overline{\mathscr{B}_{N}})$ by the $\mathcal{L}_{X}$ and $\mathcal{L}_{Y}$ norms. The second term is similarly bounded, except we must use Lemma \ref{firstdecay} to handle the term of the form $\partial_{s}\log X_{K}\partial_{s}\mathring X$. The third term is bounded using Lemma \ref{somestuffthatisuseful}. A similar argument works near $p_{S}$. 

Thus, we turn to $N_{Y}^{(2)}$. These terms are bounded in a similar manner to those above so we will be brief. Firstly, the terms involving $B$ will be bounded in a straightforward manner, thanks to the rapid decay of $B$ in the relevant regimes. Secondly, the terms involving $\partial Y_{K}$ are bounded via Lemma \ref{somestuffthatisuseful} and Lemma \ref{firstdecay}. For example, the first term in $X_{K}^{-1}N^{(2)}_{Y}$ contains a term of the form
\[
-\frac{1}{1+\mathring \sigma} \frac{\partial_{\rho}Y_{K}}{X_{K}^{2}} \partial_{\rho}\mathring \sigma = -\frac{1}{1+\mathring \sigma} \frac{\rho \partial_{\rho}Y_{K}}{X_{K}^{2}} \rho^{-1}\partial_{\rho}\mathring \sigma.
\] 
This is bounded as above. The remaining terms are similar. 

Putting this together, we have established the non-linear estimates corresponding to (3) in Proposition \ref{fixit}. As explained above, property (4) in Proposition \ref{fixit} follows from a similar argument. Thus, the assertion follows from an application of Proposition \ref{fixit}. 
\end{proof}

\section{Solving for $\mathring{\Theta}$}
In this section will solve for $\mathring{\Theta}$ in terms of the renormalized unknowns
\[\left(\psi,\mathring{\mu}^2,\mathring{\lambda}\right).\]
\subsection{Linear Estimates}
The linear equation we need to study is\index{Metric Quantities!$H_{\Theta}$}
\begin{equation}\label{linTh}
d\mathring{\Theta} = H_{\Theta},
\end{equation}
where $H_{\Theta}$ is a closed $1$-form.

\begin{proposition}Suppose that $H_{\Theta} \in \mathcal{N}_{\Theta}$. Set
\[\mathring{\Theta}\left(\rho,z\right) \doteq -\int_{\rho}^{\infty}\left(H_{\Theta}\right)_{\rho}\left(\tau,z\right)\, d\tau.\]

Then $\mathring{\Theta}$ solves~(\ref{linTh}) and there exists a constant $D$, depending on $\alpha_0$ and independent of $H_{\Theta}$, such that
\[\left\Vert \mathring{\Theta}\right\Vert_{\mathcal{L}_{\Theta}} \leq D(\alpha_0)\left\Vert H_{\Theta}\right\Vert_{\mathcal{N}_{\Theta}}.\]
\end{proposition}
\begin{proof}The assumed pointwise decay of $H_{\Theta}$ is easily seen to justify the following computation
\begin{align*}
\partial_z\mathring{\Theta} &= -\int_{\rho}^{\infty}\partial_z\left(H_{\Theta}\right)_{\rho}\, d\tau
\\ \nonumber &= -\int_{\rho}^{\infty}\partial_{\rho}\left(H_{\Theta}\right)_z\, d\tau
\\ \nonumber &= \left(H_{\Theta}\right)_z.
\end{align*}
Thus, $\mathring{\Theta}$ solves~(\ref{linTh}).

Next, we discuss $\hat{C}^{2,\alpha_0}$ estimates. First of all, away from $\{\rho = 0 \}$, the desired estimates follow immediately from the given expression for $\Theta$. Near the axis, the desired estimates for $\partial_z\mathring{\Theta}$ are similarly easy to obtain. For $\partial_x\mathring{\Theta}$ and $\partial_y\mathring{\Theta}$ we simply use the formula for $\partial_{\rho}\mathring{\Theta}$ and the formulas
\[\partial_x = \frac{x}{\sqrt{x^2+y^2}}\partial_{\rho},\qquad \partial_y = \frac{y}{x\sqrt{x^2+y^2}}\partial_{\rho}.\]
Near where the axis and horizon meet, one simply writes the equation~(\ref{linTh}) in either $(s,\chi)$ or $(s',\chi')$ coordinates and argues similarly.
\end{proof}
\subsection{The Fixed Point}
We are ready to run the fixed point argument.
\begin{proposition}\label{fixth}For each $m \in \mathbb{Z}_{\neq 0}$ and $\alpha_0 \in (0,1)$, there exists $\epsilon > 0$ sufficiently small so that given \[\left(\psi,\mathring{\mu}^2,\mathring{\lambda}\right) \in B_{\epsilon}\left(\mathcal{L}_{\lambda}\right)\times \cdots\times B_{\epsilon}\left(\mathcal{L}_{\psi}\right)\times B_{\epsilon}\left(\mathcal{L}_{\mu^2}\right),\]
we may find
\[\left(\mathring{\sigma},\left(B^{(N)}_{\chi},B^{(S)}_{\chi'},B^{(A)}_z\right),\left(\mathring{X},\mathring{Y}\right),\mathring{\Theta} \right)\in B_{\epsilon}\left(\mathcal{L}_{\sigma}\right)\times \cdots \times B_{\epsilon}\left(\mathcal{L}_{\Theta}\right),\]
which solve~(\ref{renormalsigmaeqn}), (\ref{Btripleeqn}), ~(\ref{X1}) and~(\ref{Y1}), and~(\ref{thetaringeqnrho}) and~(\ref{thetaringeqnz}) respectively.

Furthermore, $\mathring{\sigma}$, $\left(B^{(N)}_{\chi},B^{(S)}_{\chi'},B^{(A)}_z\right)$, $(\mathring{X},\mathring{Y})$, and $\mathring{\Theta}$  have ``continuous nonlinear dependence on parameters''
\[\left(\psi,\mathring{\mu}^2,\mathring{\lambda}\right),\]
in the sense of Lemma~\ref{fixit} and Remark~\ref{contnonlindep}.
\end{proposition}
\begin{proof}We first apply Propositions~\ref{fixsig},~\ref{fixb}, and~\ref{fixxy} to solve for $\mathring{\sigma}$ and $\left(B^{(N)}_{\chi},B^{(S)}_{\chi'},B^{(A)}_z\right)$, and $(\mathring{X},\mathring{Y})$ in terms of the other unknowns.

The remaining part of the proof will follow from Lemma~\ref{fixit} once we establish the necessary properties of $N_{\Theta}$. First of all, we note that $\mathring{Y}$'s equation immediately implies that $N_{\Theta}\left(\mathring{\Theta},\psi,\mathring{\mu}^2,e^{2\mathring{\lambda}}\right)$ is always a closed $1$-form (cf.\ \cite[p.\ 917]{weinstein}). Indeed, it is easiest to check this starting from the equation for $W$ in Theorem \ref{thm:HBH-geometric-main} (which will then immediately imply that $N_{\Theta}$ is closed, since it is the same quantity just expressed in terms of the renormalized variables). Now, that $\frac{\sigma}{X^{2}} [\theta_{\rho}dz - \theta_{z}d\rho]$ is closed follows immediately from the second equation for $\theta$ in Theorem \ref{thm:HBH-geometric-main} (which is satisfied, since $\mathring Y$ and $B$ solve their respective equations at this point). 

We now must verify the nonlinear estimates for $N_{\Theta}$. We begin with the estimates for 
\[
(N_{\Theta})_{\rho} = \left(  \frac{2\mathring X + \mathring X^{2}  - \mathring\sigma}{(1+\mathring X)^{2}} \right)  \frac{\rho}{X_{K}^{2}}  \partial_{z}Y_{K} - \frac{\sigma }{(1+\mathring X)^{2}} \partial_{z}(X_{K}^{-1}\mathring Y)- 2 \frac{\sigma}{(1+\mathring X)^{2}} X_{K}^{-1} \mathring Y \partial_{z}\log X_{K}  - \frac{\sigma}{X^{2}}B_{z}
\]
We would like to prove than $r^{3}(1+\rho^{-1})$ times this quantity is bounded in $\hat C^{1,\alpha_{0}}$ away from $p_{N},p_{S}$. Note that since $\mathring X,\mathring Y,\mathring \sigma$ obey ``nonlinear dependence on parameters'' in the sense of Lemma \ref{fixit}, we only need to bound this by $D\Vert \mathring X, \mathring Y,\mathring\sigma \Vert_{\mathcal{L}_{X}\times\dots\times \mathcal{L}_{\sigma}}$ to establish (3) in Lemma \ref{fixit} (establishing (4) follows similarly). Using Lemma \ref{somestuffthatisuseful} (and, in particular, the improved estimates for $\partial_{z}Y_{K}$), the first term in $(N_{\Theta})_{\rho}$ is bounded as asserted. The remaining terms are bounded in a straightforward manner. The analogous estimate for $(N_{\theta})_{z}$ follows similarly (note we are not asserting the added decay as $\rho\to 0$, which allows us to use the weaker bounds present in Lemma \ref{somestuffthatisuseful}). 

We now turn to estimates near $p_{N}$. A straightforward computation (for example, using the fact that the Hodge star operator is conformally invariant for one forms in two variables) implies that 
\[
(N_{\Theta})_{s} = \left(  \frac{2\mathring X + \mathring X^{2}  - \mathring\sigma}{(1+\mathring X)^{2}} \right)  \frac{\chi s}{X_{K}^{2}}  \partial_{\chi}Y_{K} - \frac{\sigma}{(1+\mathring X)^{2}} \partial_{\chi}(X_{K}^{-1}\mathring Y)- 2 \frac{\sigma}{(1+\mathring X)^{2}} X_{K}^{-1} \mathring Y \partial_{\chi}\log X_{K}  - \frac{\sigma}{X^{2}}B_{\chi}
\]
We would like to estimate $s^{-1}$ times this quantity. The only non-trivial term is the first one. As in Lemma \ref{somestuffthatisuseful} and Lemma \ref{firstdecay}, if $(x,y,\tilde x,\tilde y)$ denote Cartesian coordinates on $\mathbb{R}^{4}$ associated to $(s,\phi_{1},\chi,\phi_{2})$, we compute at $(x,y,\tilde x,\tilde y)$
\[
\frac{\chi s}{X_{K}^{2}}  \partial_{\chi}Y_{K} = \partial_{s}\left( \frac {W_{K}}{X_{K}}\right) = s \partial_{x}^{2}\left( \frac {W_{K}}{X_{K}}\right)\Big|_{(0,\sqrt{x^{2}+y^{2}},\tilde x,\tilde y)}
\]
After dividing by $s$, this is seen to be the composition of a smooth function with a $C^{1,1}$ function, which yields the desired $C^{1,\alpha_{0}}$ estimates. The remaining terms are easily bounded. The other component, $(N_{\Theta})_{\chi}$, is bounded similarly. The estimates at $p_{S}$ are identical. This completes the proof of (3) in Lemma \ref{fixit}, and as noted above, (4) in Lemma \ref{fixit} follows similarly. Thus, we may apply Lemma \ref{fixit} to complete the proof.
\end{proof}

\begin{corollary}\label{omconst}Let $m \in \mathbb{Z}_{\neq 0}$ and $\alpha_0 \in (0,1)$ and choose $\epsilon > 0$ sufficiently small. Then, given any \[\left(\psi,\mathring{\mu}^2,\mathring{\lambda}\right) \in B_{\epsilon}\left(\mathcal{L}_{\lambda}\right)\times \cdots\times B_{\epsilon}\left(\mathcal{L}_{\psi}\right)\times B_{\epsilon}\left(\mathcal{L}_{\mu^2}\right),\]
we may solve for
\[\left(\mathring{\sigma},\left(B^{(N)}_{\chi},B^{(S)}_{\chi'},B^{(A)}_z\right),\left(\mathring{X},\mathring{Y}\right),\mathring{\Theta} \right)\in B_{\epsilon}\left(\mathcal{L}_{\sigma}\right)\times \cdots \times B_{\epsilon}\left(\mathcal{L}_{\Theta}\right),\]
using Proposition~\ref{fixth}.

Then, the function \index{Metric Quantities!$\mathring \omega$}$\mathring \omega$ is constant along $\mathscr{H}$.
\end{corollary}
\begin{proof}We have
\begin{equation}\label{dzth}
\partial_z\left(\mathring{\Theta} + \frac{W_K}{X_K}\right) = \frac{\sigma}{X^2}\left(\partial_{\rho}Y + B_{\rho}\right).
\end{equation}
The proof concludes by noting that the right hand side of~(\ref{dzth}) vanishes on the horizon.
\end{proof}

From now on we will replace $\mathring{\omega}$ with the constant $\omega$. It is important to observe that for $\epsilon$ sufficiently small, $\omega$ cannot vanish, as long as we started with a Kerr solution with non-zero angular momentum. This follows by direct computation of the value of $\omega$ on Kerr, cf.~Lemma 3.0.1 in \cite{HBH:geometric}. 

\section{Solving for $\psi$}\label{secpsi}
The main goal of the section is to prove the following proposition.
\begin{proposition}\label{solveforpsi}Let $\delta \geq 0$. Then there exists $\alpha_0 \in (0,1)$ and $m \in \mathbb{Z}_{\neq 0}$ such that if \[\left(\mathring{\sigma},\left(B^{(N)}_{\chi},B^{(S)}_{\chi'},B^{(A)}_z\right),\left(\mathring{X},\mathring{Y}\right),\mathring{\Theta},\mathring{\lambda}\right) \in B_{\epsilon}\left(\mathcal{L}_{\sigma}\right)\times\cdots\times B_{\epsilon}\left(\mathcal{L}_{\mathring{\lambda}}\right),\]
for sufficiently small $\epsilon > 0$, then there exists $\psi : \overline{\mathscr{B}} \to \mathbb{R}$ and $\mu^2 \in \mathbb{R}$ solving~(\ref{theeqn10}) with
\[\int_0^{\infty}\int_{-\infty}^{\infty}\psi^2(\rho,z)e^{2\lambda}\sigma\, \ dz\ d\rho = \delta^2,\]
\[\left\Vert \psi\right\Vert_{\mathcal{L}_{\psi}}\leq C\delta\left[1 + \left\Vert \left(\mathring{\sigma},\left(B^{(N)}_{\chi},B^{(S)}_{\chi'},B^{(A)}_z\right),\left(\mathring{X},\mathring{Y}\right),\mathring{\Theta},{\mathring{\lambda}}\right) \right\Vert_{\mathcal{L}_{\sigma}\times\cdots\times\mathcal{L}_{\lambda}}\right],\]
\[\mu^2 \leq C\left[1 + \left\Vert \left(\mathring{\sigma},\left(B^{(N)}_{\chi},B^{(S)}_{\chi'},B^{(A)}_z\right),\left(\mathring{X},\mathring{Y}\right),\mathring{\Theta},{\mathring{\lambda}}\right) \right\Vert_{\mathcal{L}_{\sigma}\times\cdots\times\mathcal{L}_{\lambda}}\right].\]

Furthermore, we can arrange for $\psi$ and $\mu^2$ to depend continuously on the renormalized quantities in the sense that if $(\psi_1,\mu_1^2)$ and $(\psi_2,\mu_2^2)$ are two pairs of scalar fields and Klein--Gordon masses associated to two sets of renormalized quantities \index{Metric Quantities!$\mathfrak{a}$}\[\mathfrak{a}_1 \doteq \left(\mathring{\sigma}^{(1)},\left(B^{(N),(1)}_{\chi},B^{(S),(1)}_{\chi'},B^{(A),(1)}_z\right),\left(\mathring{X}^{(1)},\mathring{Y}^{(1)}\right),\mathring{\Theta}^{(1)},\mathring{\lambda}^{(1)}\right)\] and \[\mathfrak{a}_2 \doteq \left(\mathring{\sigma}^{(2)},\left(B^{(N),(2)}_{\chi},B^{(S),(2)}_{\chi'},B^{(A),(2)}_z\right),\left(\mathring{X}^{(2)},\mathring{Y}^{(2)}\right),\mathring{\Theta}^{(2)},\mathring{\lambda}^{(2)}\right),\] then we have
\[\left\Vert \left(\psi_1-\psi_2\right)\right\Vert_{\mathcal{L}_{\psi}} \leq D\delta\left\Vert \mathfrak{a}_1-\mathfrak{a}_2\right\Vert_{\mathcal{L}_{\sigma}\times\cdots\times\mathcal{L}_{\lambda}},\]
\[\left|\mu_1^2-\mu_2^2\right| \leq D\left\Vert \mathfrak{a}_1-\mathfrak{a}_2\right\Vert_{\mathcal{L}_{\sigma}\times\cdots\times\mathcal{L}_{\lambda}}.\]

\end{proposition}

In Section~\ref{fixthepsinow} we will use Proposition~\ref{solveforpsi} to solve for $\psi$ (and the other renormalized quantities) all in terms of ${\mathring{\lambda}}$.

\subsection{Existence}

We start by defining a useful functional.
\begin{definition}For every $\mu^2 > 0$ we define the following functional \index{Miscellaneous!$\mathscr{L}_{\mu^{2}}$}$\mathscr{L}_{\mu^2}$ on smooth functions $f(\rho,z) : \mathscr{B} \to \mathbb{R}$:
\begin{align}\label{func}
\mathscr{L}_{\mu^2}\left(f\right) = \int_0^{\infty}\int_{-\infty}^{\infty}\left[(\partial_{\rho}f)^2+(\partial_zf)^2 + e^{2\lambda}\left(\mu^2 + X^{-1}m^2 - \sigma^{-2}X^{-1}\left(X\omega + Wm\right)^2\right)f^2\right]\sigma\ dz\ d\rho.
\end{align}
\end{definition}
\begin{remark}Note that we have $X\omega + Wm = O\left(\rho^2\right)$ as $\rho \to 0$.
\end{remark}
\begin{remark}
Note that critical points of $\mathscr{L}_{\mu^2}$ correspond to solutions of~(\ref{theeqn10}).
\end{remark}

The following lemma shows if $\epsilon$ is sufficiently small and $\mu^2$ is taken sufficiently close to but larger than $\omega^2$, then there exists a function $f$ which is compactly supported in the region $(\rho,z) \in (0,\infty)\times (0,\infty)$ such that $\mathscr{L}_{\mu^2}\left(f\right) < 0$.
\begin{lemma}\label{neg}There exists $\hat{\mu}^2$ satisfying $\hat{\mu}^2 > \omega^2$ and a function $f(\rho,z)$ which is compactly supported in the region $(\rho,z) \in (0,\infty)\times (0,\infty)$ such that for all $\epsilon \geq 0$ sufficiently small
\[\mathscr{L}_{\hat{\mu}^2}\left(f\right) < 0.\]
\end{lemma}
\begin{proof}First of all, it is clear that by continuity, it suffices to prove the lemma for $\epsilon = 0$, i.e.~the case when we are exactly on the Kerr spacetime. In this case, straightforward calculations yield
\[X_K = \rho^2\left(1 + \frac{2M}{r} + O\left(r^{-2}\right)\right)\text{ as }r\to\infty,\]
\[X_K^{-1} = \rho^{-2}\left(1-\frac{2M}{r} + O\left(r^{-2}\right)\right)\text{ as }r\to\infty,\]
\[W_K = \rho^2\left(-\frac{2Ma}{r^3} + O\left(r^{-4}\right)\right)\text{ as }r\to\infty.\]
Keeping in mind that on the Kerr spacetime we have $\sigma = \rho$, we then obtain
\begin{align*}\label{keystruct1}
\mu^2 &+ X_K^{-1}m^2 - \rho^{-2}X_K^{-1}\left(X_K\omega + W_Km\right)^2
\\ \nonumber &=
\mu^2 +\frac{m^2}{\rho^2\left(1+\frac{2M}{r} + O\left(r^{-2}\right)\right)}- \omega^2 - \frac{2M\omega^2}{r} + O\left(r^{-2}\right)\text{ as }r\to\infty.
\end{align*}
The negative sign in front of the $\frac{2M\omega^2}{r}$ term is the key structure behind the proof of this lemma (note that this term vanishes for vanishing angular momentum, i.e. $a=0$).

Since we also have
\begin{equation*}\label{keystruct2b}
e^{2\lambda_K} = 1 + O\left(r^{-1}\right)\text{ as }r\to\infty,
\end{equation*}
we conclude
\begin{align}\label{keystruct3}
e^{2\lambda_K}&\left(\mu^2 + X_K^{-1}m^2 - \rho^{-2}X_K^{-1}\left(X_K\omega + W_Km\right)^2\right)
\\ \nonumber &= \left(1 + O\left(r^{-1}\right)\right)\left(\mu^2 +\frac{m^2}{\rho^2\left(1+\frac{2M}{r} + O\left(r^{-2}\right)\right)}- \omega^2 - \frac{2M\omega^2}{r} + O\left(r^{-2}\right)\right)\text{ as }r\to\infty.
\end{align}

Now let $\chi(\rho,z)$ be a bump function which is identically $1$ in the ball of radius $1$ around $(\rho,z) = (10,0)$ and identically $0$ outside a ball of radius $2$ around $(\rho,z) = (10,0)$. Then, for every $\nu > 0$ we set $f_{\nu} \doteq \chi\left(\frac{\rho-\nu}{\nu},\frac{z}{\nu}\right)$. Note that for sufficiently large $\nu$, on the support of the ball of radius $\nu$ around $(\rho,z) = (10+\nu,0)$,~(\ref{keystruct3}) implies
\[e^{2\lambda_K}\left(\mu^2 + X_K^{-1}m^2 - \rho^{-2}X_K^{-1}\left(X_K\omega + W_Km\right)^2\right) \leq \frac{-b_1}{r} + \left|\mu^2 - \omega^2\right| + O\left(r^{-2}\right) \leq \frac{-b_2}{\nu} + \left|\mu^2 - \omega^2\right|,\]
for some small constant $b_1,b_2 > 0$. Furthermore, one easily finds that
\[\left(\partial_{\rho}f_{\nu}\right)^2 + \left(\partial_zf_{\nu}\right)^2 \leq C\nu^{-2}.\]
Using these facts, one easily finds that $\mathscr{L}_{\mu^2}\left(f_{\nu}\right) < 0$ if $\mu^2$ is taken sufficiently close to $\omega^2$ and if $\nu$ is taken sufficiently large.
\end{proof}
\begin{definition}Using Lemma~\ref{neg} we pick and fix a choice of \index{Metric Quantities!$\hat{\mu}^{2}$}$\hat{\mu}^2$ such that for all $\epsilon \geq 0$ sufficiently small, there exists a smooth compactly supported $f : \mathscr{B} \to \mathbb{R}$ with $\mathscr{L}_{\hat{\mu}^2}\left(f\right) < 0$. Note that we can (and shall) pick $\hat{\mu}^2$ to depend in a Lipschitz manner on the Kerr parameters $a$ and $M$.
\end{definition}

We are now ready for the key result of the section.
\begin{proposition}\label{itexists} Let $\delta \geq 0$. Then there exists $\mu^2 > 0$ and a function $\psi(\rho,z) \in \dot{H}^1_{\rm axi}\left(\mathbb{R}^3\right)$ satisfying
\begin{equation}\label{normalization}
\int_0^{\infty}\int_{-\infty}^{\infty}\psi^2(\rho,z)e^{2\lambda}\sigma\, dz\ d\rho = \delta^2,
\end{equation}
such that $\psi$ is a weak solution to~(\ref{theeqn10}) in the sense that for every smooth compactly supported function $\varphi : \overline{\mathscr{B}} \to \mathbb{R}$, we have
\[\int_0^{\infty}\int_{-\infty}^{\infty}\left[\partial_{\rho}\psi\partial_{\rho}\varphi +\partial_z\psi\partial_z\varphi - e^{2\lambda}\left(\mu^2 + X^{-1}m^2 - \sigma^{-2}X^{-1}\left(X\omega + Wm\right)^2\right)\psi\varphi\right]\sigma\, dz\, d\rho = 0.\]

Furthermore, we can (and do) take $\psi|_{\mathscr{B}} > 0$.
\end{proposition}
\begin{proof}
Since the Klein--Gordon equation is linear, it suffices here to take $\delta = 1$ without loss of generality (in general, we can replace $\psi$ by $\delta \psi$).
Set\index{Miscellaneous!$\nu$}
\[\nu \doteq \inf\left\{\mathcal{L}_{\hat{\mu}^2}\left(f\right) : f\in \dot{H}^1_{\rm axi}\left(\mathbb{R}^3\right)\text{ and } \int_0^{\infty}\int_{-\infty}^{\infty}f^2e^{2\lambda}\sigma\ dz\ d\rho = 1\right\}.\]
Of course, by the definition of $\hat{\mu}^2$ and the previous lemma, we must have $\nu < 0$.

We will also be interested in a lower bound on $\nu$. The key point is that the only negative term in the integrand of $\mathscr{L}_{\hat{\mu}^2}\left(f\right)$ is the term $- \sigma^{-2}X^{-1}\left(X\omega + Wm\right)^2f^2$, and that we have
\begin{equation}\label{boundneg}
\left|\sigma^{-2}X^{-1}\left(X\omega + Wm\right)^2\right| \leq C < \infty,
\end{equation}
for some constant $C$ which depends on the choice of renormalized quantities.

In particular,  we obtain
\[\int_0^{\infty}\int_{-\infty}^{\infty}f^2e^{2\lambda}\sigma\ dz\ d\rho = 1 \Rightarrow \mathscr{L}_{\hat{\mu}^2}\left(f\right) \geq -C.\]
Consequently,
\begin{equation}\label{nuControl}
-C < \nu < 0.
\end{equation}

Our next goal is to show that the infimum $\nu$ is actually achieved on a function $\psi$. We proceed via the direct method. Let $\{\psi_j\}_{j=1}^{\infty}$ be a sequence of smooth functions vanishing for sufficiently large $r$, which satisfy
\[\int_0^{\infty}\int_{-\infty}^{\infty}\psi_j^2e^{2\lambda}\sigma\, dz\, d\rho = 1,\]
and such that
\[\lim_{j\to\infty}\mathscr{L}_{\hat{\mu}^2}\left(\psi_j\right) = \nu.\]
Next, we note that the bound~(\ref{boundneg}) on the negative terms in the integrand of $\mathscr{L}_{\hat{\mu}^2}$, and the fact that $\sigma$ is comparable to $\rho$, are easily seen to imply that for each element of the minimizing sequence we have
\[\int_0^{\infty}\int_{-\infty}^{\infty}\left[(\partial_{\rho}\psi_j)^2 + (\partial_z\psi_j)^2\right]\rho\ dz\ d\rho \leq \mathscr{L}_{\hat{\mu}^2}\left(\psi_i\right) + C\int_0^{\infty}\int_{-\infty}^{\infty}f^2e^{2\lambda}\sigma\ dz\ d\rho \leq C.\]
In particular, defining a sequence of rectangles $\mathscr{C}_i = \{(\rho,z) \in [1/i,i] \times [-i,i]\}$, applying the Rellich embedding theorem, a standard diagonalization argument, and relabeling the corresponding subsequence, we may find a $\psi \in \dot{H}^1_{\rm axi}\left(\mathbb{R}^3\right)$ such that the $\psi_j$ converge weakly to $\psi$ in $H^1_{\rm axi}\left(\mathbb{R}^3\right)$ and, for each fixed $i$, converge strongly to $\psi$ in $L^2\left(\mathscr{C}_i\right)$.

Now, it is easy to see that the support of the negative part of
\[e^{2\lambda_K}\left(\mu^2 + X_K^{-1}m^2 - \rho^{-2}X_K^{-1}\left(X_K\omega + W_Km\right)^2\right)\]
is contained in some $\mathscr{C}_i$ for $i$ sufficiently large. Thus, using the lower-semicontinuity of any Hilbert space norm under weak convergence and the strong convergence of $\psi_j$ to $\psi$ in $L^2\left(\mathscr{C}_i\right)$ for each $i$, we easily obtain
\begin{equation}\label{Llessnu}
\mathscr{L}_{\hat{\mu}^2}\left(\psi\right) \leq \nu.
\end{equation}

Now we want to argue that $\psi$ does not ``lose any mass in the limit'', i.e.
\begin{equation}\label{nomasslost}
\int_0^{\infty}\int_{-\infty}^{\infty}\psi^2e^{2\lambda}\sigma\, dz\, d\rho = 1.
\end{equation}

For the sake of contradiction, suppose that~(\ref{nomasslost}) is false. Then there exists $\hat{\epsilon} > 0$ such that we can find a subsequence $\{\psi_{j_k}\}_{k=1}^{\infty}$ satisfying
\begin{equation}\label{tocontradict}
\int_{\mathbb{R}\setminus[1/k,k]}\int_{\mathbb{R}\setminus [-k,k]}\psi_{j_k}^2e^{\lambda}\sigma\ dz\ d\rho \geq \hat{\epsilon} \qquad \forall\ k \in \mathbb{Z}_{> 0}.
\end{equation}

First we will show that no mass can escape through the boundary $\{\rho = 0\}$. Let $\chi\left(\rho\right)$ be a non-negative cut-off function which is identically $1$ for $\rho \leq 1$ and vanishes for $\rho \geq 2$. Now, for all $\tilde\rho \in (0,1/2)$ and arbitrarily small $\tilde\epsilon$ we have
\begin{align}\label{escapetorho0}
\int_{\mathbb{R}}\psi_{j_k}^2\left(\tilde\rho,z\right)\, dz &= \int_{\mathbb{R}}\chi\left(\tilde\rho\right)\psi_{j_k}^2\left(\tilde\rho,z\right)\, dz
\\ \nonumber &\leq \int_{\tilde\rho}^{\infty}\int_{\mathbb{R}}\left|\partial_{\rho}\left(\chi\left(\rho\right)\psi_{j_k}^2\right)\right|\, dz\, d\rho
\\ \nonumber &\leq C\int_{\tilde\rho}^2\int_{\mathbb{R}}\left|\psi_{j_k}\partial_{\rho}\psi_{j_k}\right|\, dz\, d\rho + C
\\ \nonumber &\leq C\tilde\epsilon \left|\log^{-1}\left(\tilde\rho\right)\right|\int_{\tilde\rho}^2\int_{\mathbb{R}}\psi_{j_k}^2\rho^{-1}\, dz\, d\rho + C\tilde\epsilon^{-1}\left|\log\left(\tilde\rho\right)\right|\int_{\tilde\rho}^2\int_{\mathbb{R}}\left(\partial_{\rho}\psi_{j_k}\right)^2\rho\, dz\, d\rho + C
\\ \nonumber &\leq C\tilde\epsilon\sup_{\rho \in (\tilde\rho,1/2)}\int_{\mathbb{R}}\psi_{j_k}^2\, dz + C\tilde\epsilon^{-1}\left|\log\left(\tilde\rho\right)\right| + C.
\end{align}

An easy argument then implies that for all $\tilde \rho \in (0,1/2)$ we have
\begin{equation}
\int_{\mathbb{R}}\psi_{j_k}^2\left(\tilde\rho,z\right)\, dz  \leq C\left|\log\left(\tilde\rho\right)\right|.
\end{equation}
In particular,
\begin{equation}\label{escapetorho0}
\limsup_{k\to\infty}\int_0^{1/k}\int_{\mathbb{R}}\psi_{j_k}^2\left(\rho,z\right)\, dz\, d\rho = 0.
\end{equation}

Combining with~(\ref{tocontradict}), we conclude that
\begin{equation}\label{tocontradict2}
\int_{\mathbb{R}\setminus[0,k]}\int_{\mathbb{R}\setminus [-k,k]}\psi_{j_k}^2e^{\lambda}\sigma\ dz\ d\rho \geq \hat{\epsilon} \qquad \forall\ k \in \mathbb{Z}_{> 0},
\end{equation}
i.e.~we can only lose mass in the large $r$ limit.

Next, let $\tilde{\chi}(x)$ be a cut-off function which is identically $1$ for $\left|x\right| < 1$ and identically $0$ for $\left|x\right| > 2$. Then define
\[\tilde{\psi}_{j_k}(\rho,z) \doteq \tilde{\chi}\left(\frac{10\rho}{k}\right)\tilde{\chi}\left(\frac{10z}{k}\right)\psi_{j_k}.\]
It follows from the compact support of the negative part of the integrand of $\mathscr{L}_{\hat{\mu}^2}$ and~(\ref{tocontradict2}) that for sufficiently large $j_k$, we have
\begin{equation}\label{normtildef}
0 < b \leq \left(\int_0^{\infty}\int_{-\infty}^{\infty}\tilde{\psi}_{j_k}^2e^{2\lambda}\sigma\ dz\ d\rho\right)^{1/2} \leq 1-\tilde{\epsilon},
\end{equation}
for some positive constants $b$ and $\tilde{\epsilon}$ (independent of $j_k$).

Keeping in mind still where the negative part of the integrand of $\mathscr{L}_{\hat{\mu}^2}$ is supported, we find that for sufficiently large $j_k$
\begin{align}\label{tildefcontradiction}
\mathscr{L}_{\hat{\mu}^2}\left(\frac{\tilde{\psi}_{j_k}}{ \left(\int_0^{\infty}\int_{-\infty}^{\infty}\tilde{\psi}_{j_k}^2e^{2\lambda}\sigma\ dz\ d\rho\right)^{1/2}}\right) &= \left(\int_0^{\infty}\int_{-\infty}^{\infty}\tilde{\psi}_{j_k}^2e^{2\lambda}\sigma\ dz\ d\rho\right)^{-1}\mathscr{L}_{\hat{\mu}^2}\left(\tilde{\psi}_{j_k}\right)
\\ \nonumber &\leq \left(1-\tilde{\epsilon}\right)^{-1}\nu + \frac{C}{k^2}.
\end{align}

If we take $k$ sufficiently large, then this contradicts the definition of $\nu$. We thus conclude that~(\ref{nomasslost}) holds.

By definition of $\nu$ we immediately obtain that
\[\nu \leq \int_0^{\infty}\int_{-\infty}^{\infty}\left[(\partial_{\rho}\psi)^2+(\partial_z\psi)^2 + e^{2\lambda}\left(\hat{\mu}^2 + X^{-1}m^2 - \sigma^{-2}X^{-1}\left(X\omega + Wm\right)^2\right)\psi^2\right]\sigma\ dz\ d\rho.\]
Combined with~(\ref{Llessnu}) we conclude that
\begin{equation}\label{psigivesnu}
\nu = \int_0^{\infty}\int_{-\infty}^{\infty}\left[(\partial_{\rho}\psi)^2+(\partial_z\psi)^2 + e^{2\lambda}\left(\hat{\mu}^2 + X^{-1}m^2 - \sigma^{-2}X^{-1}\left(X\omega + Wm\right)^2\right)\psi^2\right]\sigma\ dz\ d\rho.
\end{equation}

Setting $\mu^2 \doteq \hat{\mu}^2 - \nu$ (keep in mind that $\nu$ is negative), the verification of $\psi$'s equation concludes with a standard Euler--Lagrange equation calculation for $\psi$.

Finally, we show that $\psi$ may be taken to be positive. Recall the standard facts that $\psi \in H^1_{\rm axi}\left(\mathbb{R}^3\right)$ implies that $\left|\psi\right| \in H^1_{\rm axi}\left(\mathbb{R}^3\right)$ and that $\left|\nabla\left|\psi\right|\right| = \left|\nabla \psi\right|$. These imply that we have
\[\mathscr{L}_{\hat{\mu}^2}\left(\left|\psi\right|\right) \leq \nu.\]
In particular, $\left|\psi\right|$ solves the same minimization problem as $\psi$, and we thus conclude that $\left|\psi\right|$ is also a solution to~(\ref{theeqn10}) with the same $\mu^2$ as $\psi$.

Now we observe that by continuity (note that Schauder theory immediately implies that $\psi$ is continuous away from the axis), the set of zeros of $\left|\psi\right|$ is closed. However, since $\left|\psi\right|$ is non-negative, $\left|\psi\right|$ satisfies a Harnack inequality (Theorem 8.20 in~\cite{giltru}) which immediately implies that the set of zeros is open. Thus, since $\left|\psi\right|$ is not identically $0$ and $\mathscr{B}$ is connected, we conclude that $\left|\psi\right|$ cannot vanish. Since $\psi$ is continuous, we conclude that $\psi$ may be taken to be positive.
\end{proof}
\begin{remark}\label{defmuk}Running Proposition~\ref{itexists} on the exact Kerr spacetime, i.e.~with vanishing renormalized quantities, produces a Klein--Gordon mass which we define to be the number $\mu^2_K$ used in the definition of $\mathring{\mu}^2$.
\end{remark}
\subsection{Regularity, Decay at the Axis, and Decay as $r\to\infty$}\label{regularitydecay}
In this section we will establish some H\"{o}lder regularity for $\psi$, and establish strong decay of $\psi$ towards the axis and towards the asymptotically flat end.

We start with a weak estimate.
\begin{lemma}\label{psireg}Consider the $(\psi,\mu^2)$ produced by Proposition~\ref{itexists}. There exists $\hat{\alpha} \in (0,1)$ such that we have the following estimate:
\begin{equation}\label{c0apsi}
\left\Vert \psi\right\Vert_{\hat{C}^{0,\hat{\alpha}}\left(\overline{\mathscr{B}}\right)} \leq  C\delta\left[1 + \left\Vert \left(\mathring{\sigma},\left(B^{(N)}_{\chi},B^{(S)}_{\chi'},B^{(A)}_z\right),\left(\mathring{X},\mathring{Y}\right),\mathring{\Theta},e^{2\mathring{\lambda}}\right) \right\Vert_{\mathcal{L}_{\sigma}\times\cdots\times\mathcal{L}_{\lambda}}\right].
\end{equation}

The H\"{o}lder constant $\hat{\alpha}$ can be taken to depend on an upper bound for 
\[\left\Vert \left(\mathring{\sigma},\left(B^{(N)}_{\chi},B^{(S)}_{\chi'},B^{(A)}_z\right),\left(\mathring{X},\mathring{Y}\right),\mathring{\Theta},e^{2\mathring{\lambda}}\right) \right\Vert_{\mathcal{L}_{\sigma}\times\cdots\times\mathcal{L}_{\lambda}}.\]
\end{lemma}
\begin{proof}Without loss of generality we take $\delta = 1$.

Away from the axis, the estimate~(\ref{c0apsi}) immediately follows from Schauder estimates. We thus focus on the regions near the axis $\mathscr{A}$. As we have seen before, separate arguments will be required at the noth and south axis $\mathscr{A}_N$ and $\mathscr{A}_S$ and where the axis and horizon meet.

We start with $\mathscr{A}_N \cap \overline{\mathscr{B}_A}$. We define a new function $\psi_m : \overline{\mathscr{B}_A} \times (0,2\pi) \to \mathbb{C}$ in by
\[\psi_m\left(\rho,z,\phi\right) \doteq e^{im\phi}\psi\left(\rho,z\right).\]
It is easy to see that we then have
\begin{equation}\label{psimeqn}
\sigma^{-1}\partial_{\rho}\left(\sigma\partial_{\rho}\psi_m\right) + \sigma^{-1}\partial_z\left(\sigma\partial_z\psi_m\right) + \frac{e^{2\lambda}}{X}\partial_{\phi}^2\psi_m + e^{2\lambda}\left(\sigma^{-2}X^{-1}\left(X\omega + Wm\right)^2 - \mu^2\right)\psi_m = 0.
\end{equation}
Our plan is to now to replace the variables $(\rho,\phi)$ with Cartesian coordinates $(x,y)$ defined, as usual, by
\[x \doteq \rho\cos\phi,\qquad y \doteq \rho\sin\phi.\]
We are interested in the regularity of the coefficients of~(\ref{psimeqn}) when expressed in Cartesian coordinates. Instead of doing an explicit calculation we will reason as follows: First of all we write
\[\sigma^{-1}\partial_{\rho}\left(\sigma\partial_{\rho}\psi_m\right) + \sigma^{-1}\partial_z\left(\sigma\partial_z\psi_m\right) = \] \[X^{-1/2}\partial_{\rho}\left(X^{1/2}\partial_{\rho}\psi_m\right) + X^{-1/2}\partial_z\left(X^{1/2}\partial_z\psi_m\right) + \left(\frac{X^{1/2}}{\sigma}\right)\partial_{\rho}\left(\frac{\sigma}{X^{1/2}}\right)\partial_{\rho}\psi_m + \left(\frac{X^{1/2}}{\sigma}\right)\partial_z\left(\frac{\sigma}{X^{1/2}}\right)\partial_z\psi_m.\]
The following formulas are easily established
\begin{equation}\label{formrhophi}
\partial_{\rho} = \frac{x}{\sqrt{x^2+y^2}}\partial_x + \frac{y}{\sqrt{x^2+y^2}}\partial_y,\qquad \partial_{\phi} = -y\partial_x + x\partial_y.
\end{equation}
Next, we note that it follows from our estimates for $\mathring{X}$ and $\mathring{\sigma}$, ~(\ref{formrhophi}), and Lemma~\ref{firstdecay} that, as a differential operator in Cartesian coordinates $(x,y) \in \mathbb{R}^2$, both $\left(\frac{X^{1/2}}{\sigma}\right)\partial_{\rho}\left(\frac{\sigma}{X^{1/2}}\right)\partial_{\rho}$ and $\left(\frac{X^{1/2}}{\sigma}\right)\partial_z\left(\frac{\sigma}{X^{1/2}}\right)\partial_z$ are first order differential operators with $\hat{C}^{0,\alpha}$ coefficients. Next we turn to the terms
\[X^{-1/2}\partial_{\rho}\left(X^{1/2}\partial_{\rho}\psi_m\right) + X^{-1/2}\partial_z\left(X^{1/2}\partial_z\psi_m\right) + e^{2\lambda}X^{-1}\partial_{\phi}^2\psi_m.\]
For these terms, we begin by noting that
\[e^{-2\lambda}X^{-1/2}\partial_{\rho}\left(X^{1/2}\partial_{\rho}\psi_m\right) + e^{-2\lambda}X^{-1/2}\partial_z\left(X^{1/2}\partial_z\psi_m\right) + X^{-1}\partial_{\phi}^2\psi_m,\]
is the Laplacian associated to the metric
\[Xd\phi^2 + e^{2\lambda}\left(d\rho^2+dz^2\right).\]
In Cartesian coordinates, this becomes
\[\left(\frac{\left(\frac{X}{\rho^2}\right)y^2+e^{2\lambda}x^2}{x^2+y^2}\right)dx^2 + \left(\frac{\left(\frac{X}{\rho^2}\right)x^2+e^{2\lambda}y^2}{x^2+y^2}\right)dy^2 + \frac{2xy}{x^2+y^2}\left(e^{2\lambda}-\frac{X}{\rho^2}\right)dxdy.\]
The key observation is that even though this metric may not even be continuous at $(x,y) = (0,0)$, the coefficients are bounded away from $0$ and $\infty$, and the determinant of the metric is bounded away from $0$ and $\infty$. Thus the corresponding Laplacian is uniformly elliptic as an operator on $\mathbb{R}^3$. In particular, when written in Cartesian coordinates, the equation~(\ref{psimeqn}) is in a form where we can apply the classical De Giorgi--Nash regularity theory. Furthermore, since the set $\{\rho = 0\}$ is co-dimension $2$ in $\mathbb{R}^3$, a standard argument\footnote{To show this, we can multiply equation \eqref{psimeqn} by the product of test function compactly supported in $\mathbb{R}^{3}$ and a function cutting off (between $\rho=0$ and $\rho=2\epsilon$) at $\rho=0$. Integrating by parts and using that $\psi_{m}$ is in $\dot H^{1}_{\rm axi}(\mathbb{R}^{3})$, we find that the resulting error term tends to zero as $\epsilon\to 0$.} shows that $\psi_m$ extends to a weak $H^1_{\rm loc}$ solution along $\{\rho = 0\}$. In particular, taking $\hat{\alpha}$ sufficiently small, we can apply Theorem 8.22 of~\cite{giltru} to establish a $C^{0,\hat{\alpha}}$ estimate for $e^{im\phi}\psi$. Note that the constant $\hat{\alpha}$ can be taken to depend on an upper bound for \[\left\Vert \left(\mathring{\sigma},\left(B^{(N)}_{\chi},B^{(S)}_{\chi'},B^{(A)}_z\right),\left(\mathring{X},\mathring{Y}\right),\mathring{\Theta},e^{2\mathring{\lambda}}\right) \right\Vert_{\mathcal{L}_{\sigma}\times\cdots\times\mathcal{L}_{\lambda}}.\] 

One easily checks that also implies a $C^{0,\hat{\alpha}}$ estimate for $\psi$. Of course, a similar argument works for $\mathscr{A}_S \cap \overline{\mathscr{B}_A}$.

Now we consider the region $\mathscr{H} \cap \overline{\mathscr{B}_H}$. Here, due to the non-vanishing of $X$, we may easily interpret~(\ref{theeqn10}) as an elliptic equation in $\mathbb{R}^3$ written in cylindrical coordinates. Arguing as in the previous paragraph, we also obtain $C^{0,\hat{\alpha}}$ bounds for $\psi$ (in fact we obtain $C^{2,\hat{\alpha}}$ bounds).

Now consider the regions $\overline{\mathscr{B}_N}$ and $\overline{\mathscr{B}_S}$ where the horizon and axis meet. In $\overline{\mathscr{B}_N}$, we may introduce $(s,\chi)$ coordinates where~(\ref{theeqn10}) becomes
\begin{equation}\label{schipsi}
\sigma^{-1}\partial_s\left(\sigma\partial_s\psi\right) + \sigma^{-1}\partial_{\chi}\left(\sigma\partial_{\chi}\psi\right) + \left(\chi^2+s^2\right)e^{2\lambda}\left(\sigma^{-2}X^{-1}\left(X\omega+Wm\right)^2 -m^2X^{-1}\psi- \mu^2\right) \psi = 0.
\end{equation}
Next, we observe that
\[\sigma^{-1}\partial_s\left(\sigma\partial_s\psi\right) + \sigma^{-1}\partial_{\chi}\left(\sigma\partial_{\chi}\psi\right) = (s\chi)^{-1}\partial_s\left((s\chi)\partial_s\psi\right) + (s\chi)^{-1}\partial_{\chi}\left((s\chi)\partial_{\chi}\psi\right) + D\psi,\]
for some $1$st order differential operator $D$ with coefficients in $\hat{C}^{0,\alpha_0}\left(\overline{\mathscr{B}}\right)$. Observe that \[(s\chi)^{-1}\partial_s\left((s\chi)\partial_s\psi\right) + (s\chi)^{-1}\partial_{\chi}\left((s\chi)\partial_{\chi}\psi\right)\] is the Laplacian for the metric
\[ds^2 + s^2d\phi_1^2 + d\chi^2 + \chi^2d\phi_2^2,\]
when applied to a function which only depends on $s$ and $\chi$. One may now check that~(\ref{schipsi}) may be treated with the same trick we used in the region $\mathscr{A}_N \cap \overline{\mathscr{B}_A}$. A similar argument works for $\overline{\mathscr{B}_S}$.
\end{proof}

\begin{remark}\label{fullregpsi}Note that if we knew that $e^{2\lambda}$ and $X$ satisfied the appropriate compatibility conditions along the axis then, after passing to Cartesian coordinates, we would obtain an elliptic equation with coefficients as regular as $e^{2\lambda}$ and $X/\rho^2$ (see the formulas for the metric in Appendix A in \cite{HBH:geometric}). Then we could just apply standard Schauder theory. However, at this stage of the argument, we have \underline{not} established these compatibility conditions.

It is possible that a more careful fixed point scheme could encode these compatibility conditions in some way. This would allow for the argument to be done with only standard Schauder theory. Here, we rely on De Giorgi--Nash regularity theory as a way to avoid such an issue. 
\end{remark}

In this next lemma, we will show that if $m$ is taken sufficiently large, then $\psi$ must decay quickly as one approaches the axis $\mathscr{A}$.
\begin{lemma}\label{decayax}Let $m$ be sufficiently large and $(\psi,\mu^2)$ be produced from Proposition~\ref{itexists}. Suppose that
\[\left\Vert \left(\mathring{\sigma},\left(B^{(N)}_{\chi},B^{(S)}_{\chi'},B^{(A)}_z\right),\left(\mathring{X},\mathring{Y}\right),\mathring{\Theta},e^{2\mathring{\lambda}}\right) \right\Vert_{\mathcal{L}_{\sigma}\times\cdots\times\mathcal{L}_{\lambda}}\]
is sufficiently small.

Then there exists a choice of $\alpha_0$ such that we have
\[\left\Vert \rho^{-10}\psi\right\Vert_{\hat{C}^{0,\alpha_0}\left(\overline{\mathscr{B}_A}\right)} + \left\Vert s^{-10}\psi\right\Vert_{\hat{C}^{0,\alpha_0}\left(\overline{\mathscr{B}_N}\right)}+ \left\Vert (s')^{-10}\psi\right\Vert_{\hat{C}^{0,\alpha_0}\left(\overline{\mathscr{B}_S}\right)} \leq \]
\[C\delta\left[1 + \left\Vert \left(\mathring{\sigma},\left(B^{(N)}_{\chi},B^{(S)}_{\chi'},B^{(A)}_z\right),\left(\mathring{X},\mathring{Y}\right),\mathring{\Theta},\mathring{\lambda}\right) \right\Vert_{\mathcal{L}_{\sigma}\times\cdots\times\mathcal{L}_{\lambda}}\right].\]
\end{lemma}
\begin{proof}Without loss of generality we take $\delta = 1$.

Near the set $\overline{\mathscr{B}_A} \cap \mathscr{A}_N$ we can write~(\ref{theeqn10}) as
\begin{equation}\label{tointbyparts}
\rho^{-1}\partial_{\rho}\left(\rho\partial_{\rho}\psi\right) + \rho^{-1}\partial_z\left(\rho\partial_z\psi\right)- e^{2\lambda_K}m^2X_K^{-1}\psi = O\left(\rho\right)\partial\psi + O\left(1\right)\psi.
\end{equation}
Furthermore, in the region $\overline{\mathscr{B}_A}\cap\mathscr{A}_N$, we have that $X_K$ and $\rho^2$ are comparable.

Now let $\epsilon_0 > 0$ be a small positive number, $\eta$ be a large positive number,  and $\xi_0\left(\rho,z\right)$ be a cut-off which is identically $1$ near $\mathscr{A}$, $0$ outside an open set of $\mathscr{A}$, and such that the support of $\partial_z\xi_0$ is contained in $\overline{\mathscr{B}_N}$. Next multiply~(\ref{tointbyparts}) by $\xi_0\rho\left(\rho+\epsilon_0\right)^{-\eta+2}\psi$, and integrate by parts. One obtains
\begin{align}\label{firstpass}
\int_{\mathscr{B}}&\xi_0\left[\rho\left(\rho+\epsilon_0\right)^{-\eta+2}\left(\partial_{\rho}\psi\right)^2 + \xi_0\rho\left(\rho+\epsilon_0\right)^{-\eta+2}\left(\partial_z\psi\right)^2 + \rho m^2\left(\rho+\epsilon_0\right)^{-\eta}\psi^2\right] \leq
\\ \nonumber &C\left(\eta\right)\left[1 + \left\Vert \left(\mathring{\sigma},\left(B^{(N)}_{\chi},B^{(S)}_{\chi'},B^{(A)}_z\right),\left(\mathring{X},\mathring{Y}\right),\mathring{\Theta},e^{2\mathring{\lambda}}\right) \right\Vert_{\mathcal{L}_{\sigma}\times\cdots\times\mathcal{L}_{\lambda}}\right] + \\ \nonumber &C\int_{\mathscr{B}}\left[\left|\partial_z\xi_0\right|\rho\left(\rho+\epsilon_0\right)^{-\eta+2}\left|\partial_z\psi\psi\right| + \left|\eta\right|\xi_0\rho\left(\rho+\epsilon_0\right)^{-\eta+1}\left|\partial_{\rho}\psi\psi\right|\right].
\end{align}
If $m$ is sufficiently large relative to $\eta$, then it is easy to see that~(\ref{firstpass}) implies
\begin{align}\label{secondpass}
\int_{\mathscr{B}}&\xi_0\left[\rho\left(\rho+\epsilon_0\right)^{-\eta+2}\left(\partial_{\rho}\psi\right)^2 + \xi_0\rho\left(\rho+\epsilon_0\right)^{-\eta+2}\left(\partial_z\psi\right)^2 + \rho m^2\left(\rho+\epsilon_0\right)^{-\eta}\psi^2\right] \leq
\\ \nonumber &C\left(\eta\right)\left[1 + \left\Vert \left(\mathring{\sigma},\left(B^{(N)}_{\chi},B^{(S)}_{\chi'},B^{(A)}_z\right),\left(\mathring{X},\mathring{Y}\right),\mathring{\Theta},e^{2\mathring{\lambda}}\right) \right\Vert_{\mathcal{L}_{\sigma}\times\cdots\times\mathcal{L}_{\lambda}}\right] +
\\ \nonumber &C\int_{\mathscr{B}}\left[\left|\partial_z\xi_0\right|\rho\left(\rho+\epsilon_0\right)^{-\eta+2}\left|\partial_z\psi\psi\right|\right].
\end{align}

Next, we move to the set $\overline{\mathscr{B}_N}$ where we write the equation in $(s,\chi)$ coordinates as
\begin{equation}\label{tointbyparts2}
s^{-1}\partial_s\left(s\partial_s\psi\right) + \chi^{-1}\partial_z\left(\chi\partial_z\psi\right)- \left(\chi^2+s^2\right)e^{2\lambda_K}m^2X_K^{-1}\psi = O\left(s\right)\partial\psi + O\left(1\right)\psi.
\end{equation}
In $\overline{\mathscr{B}_N}$ we have that $X_K$ is comparable to $s^2$. Let $\xi_1$ be a cut-off function which is identically $1$ in $\overline{\mathscr{B}_N}$, vanishes outside an open set containing $\overline{\mathscr{B}_N}$, and such that $\partial_{\chi}\xi_1$ is supported in $\overline{\mathscr{B}_A} \cap \mathscr{A}$. Analogously to the above we obtain
\begin{align}\label{thirdpass}
\int_{\mathscr{B}}&\xi_1\left[s\chi\left(s+\epsilon_0\right)^{-\eta+2}\left(\partial_s\psi\right)^2 + \xi_1s\chi\left(s+\epsilon_0\right)^{-\eta+2}\left(\partial_{\chi}\psi\right)^2 + s\chi m^2\left(s+\epsilon_0\right)^{-\eta}\psi^2\right] \leq
\\ \nonumber &C\left(\eta\right)\left[1 + \left\Vert \left(\mathring{\sigma},\left(B^{(N)}_{\chi},B^{(S)}_{\chi'},B^{(A)}_z\right),\left(\mathring{X},\mathring{Y}\right),\mathring{\Theta},e^{2\mathring{\lambda}}\right) \right\Vert_{\mathcal{L}_{\sigma}\times\cdots\times\mathcal{L}_{\lambda}}\right] +
\\ \nonumber &C\int_{\mathscr{B}}\left[\left|\partial_\chi\xi_1\right|s\chi\left(s+\epsilon_0\right)^{-\eta+2}\left|\partial_\chi\psi\psi\right|\right].
\end{align}
Now we observe that along the support of $\nabla\xi_0$ and $\nabla\xi_1$, we have that $s$ is comparable to $\rho$. Thus, it is easy to see that if we add~(\ref{secondpass}) and~(\ref{thirdpass}) together, and then take $\epsilon_0 \to 0$, we may obtain
\begin{align}\label{fourthpass}
&\int_{\mathscr{B}}\left[\xi_0\rho^{-\eta}\psi^2\rho + \xi_1s^{-\eta}\psi^2s\chi\right] \leq
 \\ \nonumber &C\left(\eta\right)\left[1 + \left\Vert \left(\mathring{\sigma},\left(B^{(N)}_{\chi},B^{(S)}_{\chi'},B^{(A)}_z\right),\left(\mathring{X},\mathring{Y}\right),\mathring{\Theta},e^{2\mathring{\lambda}}\right) \right\Vert_{\mathcal{L}_{\sigma}\times\cdots\times\mathcal{L}_{\lambda}}\right].
\end{align}
Analogous estimates may be proved along $\mathscr{A}_S$.

Let $(x_0,y_0,z_0)$ be a point in Cartesian coordinates near the axis and set $\rho_0 \doteq \sqrt{x_0^2+y^2}$. Moser's mean-value inequality (see Theorem 8.17 from~\cite{giltru}, also remember that in Lemma~\ref{psireg} we observed that $\psi$ satisfies an equation in an open set around $\mathscr{A}$ for which the De Giorgi-Nash regularity theory holds) yields
\[\sup_{B_{\rho_0/20}\left((x_0,y_0,z_0)\right)}\left|\psi\right| \leq C\rho_0^{-3/2}\left(\int_{B_{\rho_0/10}((x_{0},y_{0},z_{0}))}\psi^2\right)^{1/2}.\]
A similar estimate holds in Cartesian coordinates near $p_N$ and $p_S$.

Combining these mean-value inequalities with the weighted $L^2$-estimates (which give decaying $L^{2}$-estimates on $B_{\rho_{0}/10}((x_{0},y_{0},z_{0}))$) leads to the following $L^{\infty}$ estimates
\[\left\Vert \rho^{-\eta/2+C_0}\psi\right\Vert_{L^{\infty}\left(\overline{\mathscr{B}_A}\right)} + \left\Vert s^{-\eta/2+C_0}\psi\right\Vert_{L^{\infty}\left(\overline{\mathscr{B}_N}\right)}+ \left\Vert (s')^{-\eta/2+C_0}\psi\right\Vert_{L^{\infty}\left(\overline{\mathscr{B}_S}\right)} \leq \]
\[C\left(\eta\right)\left[1 + \left\Vert \left(\mathring{\sigma},\left(B^{(N)}_{\chi},B^{(S)}_{\chi'},B^{(A)}_z\right),\left(\mathring{X},\mathring{Y}\right),\mathring{\Theta},e^{2\mathring{\lambda}}\right) \right\Vert_{\mathcal{L}_{\sigma}\times\cdots\times\mathcal{L}_{\lambda}}\right],\]
where we emphasize that the constant $C_0$ is independent of $\eta$.

Interpolating with the already proved $\hat{C}^{0,\hat{\alpha}}$ estimate and taking $\eta$ large enough finishes the proof with $\alpha_0 = \hat{\alpha}/2$.
\end{proof}
\begin{remark}\label{fixmalpha0} We now consider  \index{Metric Quantities!$m$}$m$ and \index{Miscellaneous!$\alpha_{0}$}$\alpha_0$ to be fixed.
\end{remark}

Finally, we obtain the full Schauder estimate.
\begin{lemma}\label{psiregfull}Consider the $(\psi,\mu^2)$ produced by Proposition~\ref{itexists}. Suppose that
\[\left\Vert \left(\mathring{\sigma},\left(B^{(N)}_{\chi},B^{(S)}_{\chi'},B^{(A)}_z\right),\left(\mathring{X},\mathring{Y}\right),\mathring{\Theta},e^{2\mathring{\lambda}}\right) \right\Vert_{\mathcal{L}_{\sigma}\times\cdots\times\mathcal{L}_{\lambda}}\]
is sufficiently small.

Then
\begin{equation}\label{c2apsi}
\left\Vert \psi\right\Vert_{\hat{C}^{2,\alpha_0}\left(\overline{\mathscr{B}}\right)} \leq  C\delta\left[1 + \left\Vert \left(\mathring{\sigma},\left(B^{(N)}_{\chi},B^{(S)}_{\chi'},B^{(A)}_z\right),\left(\mathring{X},\mathring{Y}\right),\mathring{\Theta},\mathring{\lambda}\right) \right\Vert_{\mathcal{L}_{\sigma}\times\cdots\times\mathcal{L}_{\lambda}}\right].
\end{equation}
\end{lemma}
\begin{proof}Without loss of generality we take $\delta = 1$.

The key point is that Lemma~\ref{decayax} implies that $\psi/X^2 \in \hat{C}^{0,\alpha_0}\left(\overline{\mathscr{B}}\right)$. Thus, we may repeat the analysis in the proof of Lemma~\ref{psiregfull} except that we just treat the $\psi/X^2$ term as an error and then use standard Schauder estimates.
\end{proof}

We close the section establishing decay in the large $r$ region.
\begin{lemma}\label{lem:poly-decay-scalar}Consider the $(\psi,\mu^2)$ produced by Proposition~\ref{itexists}. Suppose that
\[\left\Vert \left(\mathring{\sigma},\left(B^{(N)}_{\chi},B^{(S)}_{\chi'},B^{(A)}_z\right),\left(\mathring{X},\mathring{Y}\right),\mathring{\Theta},e^{2\mathring{\lambda}}\right) \right\Vert_{\mathcal{L}_{\sigma}\times\cdots\times\mathcal{L}_{\lambda}}\]
is sufficiently small.

We have
\[\left\Vert r^{10}\psi\right\Vert_{\hat{C}^{2,\alpha_0}\left(\overline{\mathscr{B}}\right)} \leq C\delta\left[1 + \left\Vert \left(\mathring{\sigma},\left(B^{(N)}_{\chi},B^{(S)}_{\chi'},B^{(A)}_z\right),\left(\mathring{X},\mathring{Y}\right),\mathring{\Theta},\mathring{\lambda}\right) \right\Vert_{\mathcal{L}_{\sigma}\times\cdots\times\mathcal{L}_{\lambda}}\right].\]
\end{lemma}
\begin{proof}Without loss of generality we take $\delta = 1$.

This is straightforward. The key point is the lower order term in~(\ref{theeqn10}) is comparable to $1$ when $r$ is sufficiently large, i.e.~$r$ sufficiently large implies
\[e^{2\lambda}\sigma^{-2}X^{-1}\left(X\omega+Wm\right)^2 - e^{2\lambda}m^2X^{-1}- e^{2\lambda}\mu^2 \geq b\left(1+\rho^{-2}\right).\]
In particular, for all $R$ sufficiently large, we may define a cut-off $\xi_R$ which is identically $0$ for $r \leq R/2$ and identically $1$ for $r \geq R$ and such that $\left|\nabla\xi_R\right| \leq CR^{-1}$. Multiplying~(\ref{theeqn10}) through by $\xi_R$ and using the above positivity property, we easily may establish
\[-\sigma^{-1}\partial_{\rho}\left(\sigma\partial_{\rho}\psi\right) - \sigma^{-1}\partial_z\left(\sigma\partial_z\psi\right) + b\left(1+\rho^{-2}\right)\psi \leq C\left[R^{-1}\left|\partial\psi\right| + R^{-2}\left|\psi\right|\right].\]
A straightforward elliptic estimate then establishes
\begin{equation}\label{decayestpoly}
\int_{\mathscr{B}\cap \{r \geq R\}}\left[\left(\partial_{\rho}\psi\right)^2 + \left(\partial_z\psi\right)^2 + \left(1+\rho^{-2}\right)\psi^2\right]\rho\, d\rho\, dz \leq CR^{-1}\left[1+\left\Vert(\mathring{X},\mathring{Y},\mathring{W},\mathring{\sigma},\mathring{\lambda})\right\Vert_{\mathcal{L}}\right].
\end{equation}

It is clear that this process can be iterated so that the $R^{-1}$ on the right hand side of~(\ref{decayestpoly}) may be replaced by $R^{-k}$. Carrying the out the same Schauder estimate scheme we have done before is easily seen to finish the proof.
\end{proof}
\subsection{Continuous Dependence on the Renormalized Quantities}
In this section, we will show that $\psi$ and $\mu^2$, as produced by Proposition~\ref{itexists}, depend continuously on the renormalized quantities.

Let's introduce some convenient notation. Throughout this section, we let $(\psi_1,\mu_1^2)$ and $(\psi_2,\mu_2^2)$ be two pairs of scalar fields and Klein--Gordon masses associated to two sets of renormalized quantities\index{Metric Quantities!$\mathfrak{a}$}
\[\mathfrak{a}_1 \doteq \left(\mathring{\sigma}^{(1)},\left(B^{(N),(1)}_{\chi},B^{(S),(1)}_{\chi'},B^{(A),(1)}_z\right),\left(\mathring{X}^{(1)},\mathring{Y}^{(1)}\right),\mathring{\Theta}^{(1)},\mathring{\lambda}^{(1)}\right)\] and \[\mathfrak{a}_2 \doteq \left(\mathring{\sigma}^{(2)},\left(B^{(N),(2)}_{\chi},B^{(S),(2)}_{\chi'},B^{(A),(2)}_z\right),\left(\mathring{X}^{(2)},\mathring{Y}^{(2)}\right),\mathring{\Theta}^{(2)},\mathring{\lambda}^{(2)}\right).\] Then we set
\index{Metric Quantities!$\mathfrak{A}$}\[\mathfrak{A} \doteq \left\Vert \mathfrak{a}_1-\mathfrak{a}_2\right\Vert_{\mathcal{L}_{\sigma}\times\cdots\times\mathcal{L}_{\lambda}}.\]

Finally we denote the metric quantities associated to the renormalized quantities by
\[(X_i,Y_i,W_i,\sigma_i,\lambda_i).\]

We start by proving that the Klein--Gordon mass depends continuously on the renormalized quantities.
\begin{lemma}\label{itisindeedquitecontinuous}We have
\[\left|\mu^2_1-\mu^2_2\right| \leq C\mathfrak{A}.\]
\end{lemma}
\begin{proof}Without loss of generality we take $\delta = 1$.

We follow arguments from~\cite{shlapgrow}. Let $\mathscr{L}^{(1)}_{\hat{\mu}^2}$ and $\mathscr{L}^{(2)}_{\hat{\mu}^2}$ denote the functionals from Proposition~\ref{itexists} corresponding to $\psi_1$ and $\psi_2$ respectively. Observe the following easy estimate:
\[\left|\mathscr{L}^{(1)}_{\hat{\mu}^2}\left(\psi_1\right) - \mathscr{L}^{(2)}_{\hat{\mu}^2}\left(\psi_1\right)\right| + \left|\mathscr{L}^{(1)}_{\hat{\mu}^2}\left(\psi_2\right) - \mathscr{L}^{(2)}_{\hat{\mu}^2}\left(\psi_2\right)\right| \leq C\mathfrak{A}.\]
Using $\psi_1$ and $\psi_2$ as test functions in the variational formulations of $\mu_1$ and $\mu_2$, we immediately obtain
\[\mu^2_1 \leq \mu^2_2 + C\mathfrak{A},\]
\[\mu^2_2 \leq \mu^2_1 + C\mathfrak{A}.\]
The lemma then easily follows.
\end{proof}

The following result is the technical heart of the paper and is very closely related to estimates from~\cite{agmon}. At this point, it may be useful to review the model problem in the introduction, Proposition \ref{modelprop}. 
\begin{proposition}\label{agmax}Continuing to assume that $\psi_{1},\psi_{2}$ solve their respective equations, as above, we have that for every $q > 0$, then
\begin{equation}\label{agest}
\int_{\mathscr{B}}\psi_2^2\left|\partial\left(\frac{\psi_1}{\psi_2}\right)\right|^2\, \rho\, d\rho\, dz \leq \delta^2C(q)\mathfrak{A}^2 + q\delta^2\int_{\mathscr{B}}\left(\psi_1-\psi_2\right)^2\, \rho\, d\rho\, dz.
\end{equation}
\end{proposition}
\begin{proof}Without loss of generality we take $\delta = 1$.  It will be convenient to set $v \doteq \psi_1-\psi_2$.

We now multiply $\psi_2$'s equation~(\ref{theeqn10}) by $\sigma_2 v^2/\psi_2$ and integrating by parts. We obtain
\begin{equation}\label{howitallstarted}
\int_{\mathscr{B}}\left[\partial_{\rho}\psi_2\partial_{\rho}\left(\frac{v^2}{\psi_2}\right) + \partial_z\psi_2\partial_z\left(\frac{v^2}{\psi_2}\right) -e^{2\lambda_2}\left(\sigma_2^{-2}X_2^{-1}\left(X_2\omega_2+W_2m\right)^2 - m^2X_2^{-1}-\mu_2^2\right)v^2 \right]\sigma_2\, d\rho\, dz = 0.
\end{equation}

Observe the following identity which holds for any two functions $f(x)$ and $g(x)$ with $f > 0$
\begin{align*}
f'\left(\frac{g^2}{f}\right)' &= \frac{2f'gg'}{f} - \frac{(f')^2g^2}{f^2}
\\ \nonumber &= (g')^2 - (g')^2 + \frac{2f'gg'}{f} - \frac{(f')^2g^2}{f^2}
\\ \nonumber &= (g')^2 -f^2\left(\frac{(g')^2}{f^2} - \frac{2f'gg'}{f^3} + \frac{(f')^2g^2}{f^4}\right)
\\ \nonumber &= (g')^2 -f^2\left(\left(\frac{g}{f}\right)'\right)^2.
\end{align*}

In particular,
\begin{align*}
\partial_{\rho}&\psi_2\partial_{\rho}\left(\frac{v^2}{\psi_2}\right) + \partial_z\psi_2\partial_z\left(\frac{v^2}{\psi_2}\right) = (\partial_{\rho}v)^2 + (\partial_zv)^2 - \psi_2^2\left[\left(\partial_{\rho}\left(\frac{v}{\psi_2}\right)\right)^2 + \left(\partial_z\left(\frac{v}{\psi_2}\right)\right)^2\right].
\end{align*}

Thus,~(\ref{howitallstarted}) implies
\begin{align}\label{whathappenednext}
\int_{\mathscr{B}}&\psi_2^2\left[\left(\partial_{\rho}\left(\frac{v}{\psi_2}\right)\right)^2 + \left(\partial_z\left(\frac{v}{\psi_2}\right)\right)^2\right]\sigma_2\, d\rho\, dz =
\\ \nonumber &\int_{\mathscr{B}}\left[(\partial_{\rho}v)^2 + (\partial_zv)^2 -e^{2\lambda_2}\left(\sigma_2^{-2}X_2^{-1}\left(X_2\omega_2+W_2m\right)^2 - m^2X_2^{-1}-\mu_2^2\right)v^2 \right]\sigma_2\, d\rho\, dz.
\end{align}

Next, we note that using our previous estimates, one may easily establish
\[\int_{\mathscr{B}}\left[\left|\sigma_2^{-1}\partial_{\rho}\left(\sigma_2\partial_{\rho}v\right) + \sigma_2^{-1}\partial_z\left(\sigma_2\partial_zv\right) -e^{2\lambda_2}\left(\sigma_2^{-2}X_2^{-1}\left(X_2\omega_2+W_2m\right)^2 - m^2X_2^{-1}-\mu_2^2\right)v \right|^2\right]\rho\, d\rho\, dz \leq C\mathfrak{A}^2.\]

In particular, a straightforward elliptic estimate yields
\[\int_{\mathscr{B}}\left[(\partial_{\rho}v)^2 + (\partial_zv)^2 -e^{2\lambda_2}\left(\sigma_2^{-2}X_2^{-1}\left(X_2\omega_2+W_2m\right)^2 - m^2X_2^{-1}-\mu_2^2\right)v^2 \right]\sigma_2\, d\rho\, dz \leq C(q)\mathfrak{A}^2 + q\int_{\mathscr{B}}v^2\, \rho\, d\rho\, dz.\]

Using that $\rho$ and $\sigma_2$ are comparable, we immediately conclude
\begin{equation}\label{thefallout}
\int_{\mathscr{B}}\psi_2^2\left[\left(\partial_{\rho}\left(\frac{v}{\psi_2}\right)\right)^2 + \left(\partial_z\left(\frac{v}{\psi_2}\right)\right)^2\right]\rho\, d\rho\, dz \leq C(q)\mathfrak{A}^2+ q\int_{\mathscr{B}}v^2\, \rho\, d\rho\, dz.
\end{equation}
Equivalently,
\[\int_{\mathscr{B}}\psi_2^2\left|\partial\left(\frac{\psi_1-\psi_2}{\psi_2}\right)\right|^2\, \rho\, d\rho\, dz \leq C(q)\mathfrak{A}^2+ q\int_{\mathscr{B}}(\psi_1-\psi_2)^2\, \rho\, d\rho\, dz.\qedhere\]
\end{proof}

Now we are ready for the main result of the section.
\begin{proposition}We have
\[\left\Vert \psi_1-\psi_2\right\Vert_{\mathcal{L}_{\psi}} \leq D\delta\mathfrak{A}.\]
\end{proposition}
\begin{proof}Without loss of generality we take $\delta = 1$.

Let \index{Coordinates!$K_{\minus}$}$K_{\minus}$ be an open set which contains the the support of the negative part of both \[e^{2\lambda_2}\left(m^2X_2^{-1} + \mu^2_2-\sigma_2^{-2}X_2^{-1}\left(X_2\omega_2+W_2m\right)^2\right) \quad  \text{and}  \quad e^{2\lambda_1}\left(m^2X_1^{-1} + \mu^2_1-\sigma_1^{-2}X_1^{-1}\left(X_1\omega_1+W_1m\right)^2\right).\]
We may easily arrange for the size of this set to be uniform in $\epsilon$ and for $K_{\minus}$ to be bounded away from $\{\rho = 0\}$.

It follows immediately from the variational proof of Lemma~\ref{itexists} and the fact that $\nu < 0$, that there exists a constant \index{Miscellaneous!$b$}$b > 0$, uniform in $\epsilon > 0$, such that
\[\int_{K_{\minus}}\psi_2^2\, \rho\, d\rho\, dz \geq b,\]
\[\int_{K_{\minus}}\psi_1^2\, \rho\, d\rho\, dz \geq b.\]

It then follows immediately from the reverse mean-value inequality (see Theorem 8.18 of~\cite{giltru}) that
\begin{equation}\label{lower}
\inf_{K_{\minus}}\psi_2^2 \geq b,
\end{equation}
\[\inf_{K_{\minus}}\psi_1^2 \geq b,\]
for a possibly different constant $b$.

Combining~(\ref{lower}) with Proposition~\ref{agmax}, we obtain
\[\int_{K_{\minus}}\left|\partial\left(\frac{\psi_1}{\psi_2}\right)\right|^2\, \rho\, d\rho\, dz \leq C(q)\mathfrak{A}^2 + q\int_{\mathscr{B}}(\psi_1-\psi_2)^2\, \rho\, d\rho\, dz.\]
Set\index{Miscellaneous!$A$}
\[A \doteq \frac{1}{\left|K_{\minus}\right|}\int_{K_{\minus}}\frac{\psi_1}{\psi_2}\, \rho\, d\rho\, dz.\]

A Poincar\'{e} inequality yields
\[\int_{K_{\minus}}\left|\frac{\psi_1}{\psi_2} - A\right|^2\, \rho\, d\rho\, dz \leq C(q)\mathfrak{A}^2+q\int_{\mathscr{B}}(\psi_1-\psi_2)^2\, \rho\, d\rho\, dz.\]
Applying~(\ref{lower}) again (and previous estimates on the scalar field) yields
\begin{equation}\label{negativepart}
\int_{K_{\minus}}\left|\psi_1- A\psi_2\right|^2\, \rho\, d\rho\, dz \leq C(q)\mathfrak{A}^2+q\int_{\mathscr{B}}(\psi_1-\psi_2)^2\, \rho\, d\rho\, dz.
\end{equation}

Next, we observe
\begin{align*}
\int_{\mathscr{B}}&\Big[\Big|\sigma_2^{-1}\partial_{\rho}\left(\sigma_2\partial_{\rho}\left(\psi_1- A\psi_2\right)\right) + \sigma_2^{-1}\partial_z\left(\sigma_2\partial_z\left(\psi_1- A\psi_2\right)\right)
\\ \nonumber &-e^{2\lambda_2}\left(\sigma_2^{-2}X_2^{-1}\left(X_2\omega_2+W_2m\right)^2 - m^2X_2^{-1}-\mu_2^2\right)\left(\psi_1- A\psi_2\right) \Big|\left|\psi_1- A\psi_2\right|\Big] \leq
\\ \nonumber &C(q)\mathfrak{A}^2 + q\int_{\mathscr{B}}(\psi_1-\psi_2)^2\, \rho\, d\rho\, dz.
\end{align*}

Thus, a straightforward elliptic estimate and~(\ref{negativepart}) easily imply
\[\int_{\mathscr{B}}\left[\left(\partial_{\rho}\left(\psi_1- A\psi_2\right)\right)^2 + \left(\partial_z\left(\psi_1- A\psi_2\right)\right)^2 + \left(\frac{m^2}{X_K^2}+1\right)\left(\psi_1- A\psi_2\right)^2\right]\, \rho\, d\rho\, dz \leq\]
 \[C(q)\mathfrak{A}^2+ q\int_{\mathscr{B}}(\psi_1-\psi_2)^2\, \rho\, d\rho\, dz.\]

Using the normalizations~(\ref{normalization}) of $\psi_1$ and $\psi_2$ in \eqref{negativepart} combined with the triangle inequality, we find that $\left|A-1\right| \leq C(q)\mathfrak{A}^2+ q\int_{\mathscr{B}}(\psi_1-\psi_2)^2\, \rho\, d\rho\, dz$. Finally, the remaining estimates may be established in a straightforward manner by repeating the arguments of Section~\ref{regularitydecay} (where we eventually absorb the term $q\int_{\mathscr{B}}(\psi_1-\psi_2)^2\, \rho\, d\rho\, dz$ into the left hand side and then fix $q$).
\end{proof}

Note that putting all of the results of the section we have proved together finishes the proof of Proposition~\ref{solveforpsi}.
\subsection{The Fixed Point}\label{fixthepsinow}
In this section we will use Proposition~\ref{solveforpsi} to solve for $\psi$ and all of the other unknowns in terms of $\mathring{\lambda}$.
\begin{proposition}\label{fixpsi}There exists $\epsilon > 0$ sufficiently small so that given
\[\mathring{\lambda} \in B_{\epsilon}\left(\mathcal{L}_{\lambda}\right),\]
and $\delta \geq 0$ sufficiently small,
we may find
\[\left(\mathring{\sigma},\left(B^{(N)}_{\chi},B^{(S)}_{\chi'},B^{(A)}_z\right),\left(\mathring{X},\mathring{Y}\right),\mathring{\Theta},\psi,\mathring{\mu}^2 \right)\in B_{\epsilon}\left(\mathcal{L}_{\sigma}\right)\times \cdots \times B_{\epsilon}\left(\mathcal{L}_{\psi}\right)\times B_{\epsilon}\left(\mathcal{L}_{\mu^2}\right),\]
which solve~(\ref{renormalsigmaeqn}), (\ref{Btripleeqn}), ~(\ref{X1}) and~(\ref{Y1}),~(\ref{thetaringeqnrho}) and~(\ref{thetaringeqnz}), and~(\ref{theeqn10}) respectively.

Furthermore, $\mathring{\sigma}$, $\left(B^{(N)}_{\chi},B^{(S)}_{\chi'},B^{(A)}_z\right)$, $(\mathring{X},\mathring{Y})$, and $\mathring{\Theta}$  have ``continuous nonlinear dependence on the parameters'' $\mathring\lambda$ and $\delta$ in the following sense: for\index{Metric Quantities!$\mathfrak{b}$}
\[\mathfrak{b}_1 \doteq \left(\mathring{\sigma}^{(1)},\left(B^{(N),(1)}_{\chi},B^{(S),(1)}_{\chi'},B^{(A),(1)}_z\right),\left(\mathring{X}^{(1)},\mathring{Y}^{(1)}\right),\mathring{\Theta}^{(1)}\right)\] and \[\mathfrak{b}_2 \doteq \left(\mathring{\sigma}^{(2)},\left(B^{(N),(2)}_{\chi},B^{(S),(2)}_{\chi'},B^{(A),(2)}_z\right),\left(\mathring{X}^{(2)},\mathring{Y}^{(2)}\right),\mathring{\Theta}^{(2)}\right)\]
corresponding to $(\mathring\lambda_{1},\delta_{1})$ and $(\mathring\lambda_{2},\delta_{2})$, we have that
\[
\Vert \mathfrak{b}_{i} \Vert_{\mathcal{L}_{\sigma}\times\cdots\times\mathcal{L}_{\Theta}} \leq D \left( \delta_{i}^{2} + \Vert \mathring\lambda_{i}\Vert_{\mathcal{L}_{\lambda}}^{2}\right)\]
and
\[ 
\left\Vert \mathfrak{b}_1-\mathfrak{b}_2\right\Vert_{\mathcal{L}_{\sigma}\times\cdots\times\mathcal{L}_{\Theta}}\leq D (\delta_{1}+ \delta_{2} + \Vert\mathring\lambda_{1}\Vert_{\mathcal{L}_{\lambda}}  + \Vert\mathring\lambda_{2}\Vert_{\mathcal{L}_{\lambda}} )\left( |\delta_{1}-\delta_{2}| + \Vert \mathring\lambda_{1}-\mathring \lambda_{2}\Vert_{\mathcal{L}_{\lambda}} \right)
\]

Finally, for $(\psi_{i},\mu_{i}^{2})$, we have
\[\left\Vert \psi_1-\psi_2\right\Vert_{\mathcal{L}_{\psi}} \leq D(\delta_{1}+\delta_{2})\left\Vert\mathring{\lambda}_1 - \mathring{\lambda}_2\right\Vert_{\mathcal{L}_{\lambda}} + D|\delta_{1}-\delta_{2}| \]
\[\left|\mu_1^2-\mu_2^2\right| \leq D\left\Vert \mathring{\lambda}_1 - \mathring{\lambda}_2 \right\Vert_{\mathcal{L}_{\lambda}} + D|\delta_{1}-\delta_{2}|.\]
where $(\psi_1,\mu_1^2)$ and $(\psi_2,\mu_2^2)$ correspond to $\mathring{\lambda}_1$ and $\mathring{\lambda}_2$ respectively. 
\end{proposition}

\begin{remark}
We note that the $\lambda_{i}$ dependence in the above estimates for $\mathfrak{b}_{i}$ could be improved by a more detailed analysis to
\[
\Vert \mathfrak{b}_{i} \Vert_{\mathcal{L}_{\sigma}\times\cdots\times\mathcal{L}_{\Theta}} \leq D  \delta_{i}^{2} 
\]
and
\[ 
\left\Vert \mathfrak{b}_1-\mathfrak{b}_2\right\Vert_{\mathcal{L}_{\sigma}\times\cdots\times\mathcal{L}_{\Theta}}\leq D (\delta_{1} + \delta_{2})\left( |\delta_{1}-\delta_{2}| + \Vert \mathring\lambda_{1}-\mathring \lambda_{2}\Vert_{\mathcal{L}_{\lambda}} \right).
\]
This is due to the way $\lambda$ enters into the non-linear terms (in every equation besides the scalar field equation, it enters in the form $e^{2\lambda}\psi^{2}$). The estimates proved above suffice here. 
\end{remark}

\begin{proof} For now, we consider $\delta\geq 0$ fixed small. Below, we will consider the dependence of the solutions on $\delta$. 

First we apply Proposition~\ref{fixth} to define a map
\[\mathfrak{L} : B_{\epsilon}\left(\mathcal{L}_{\psi}\right) \times B_{\epsilon}\left(\mathcal{L}_{\mu^2}\right)\times B_{\epsilon}\left(\mathcal{L}_{\lambda}\right) \to B_{\epsilon}\left(\mathcal{L}_{\sigma}\right)\times \cdots \times B_{\epsilon}\left(\mathcal{L}_{\Theta}\right),\]
which solves for
\[\left(\mathring{\sigma},\left(B^{(N)}_{\chi},B^{(S)}_{\chi'},B^{(A)}_z\right),\left(\mathring{X},\mathring{Y}\right),\mathring{\Theta}\right),\]
in terms of
\[\left(\psi,\mathring{\mu}^2,\mathring{\lambda}\right).\]

By Proposition~\ref{fixth}, we have that
\[
\Vert \mathfrak{L}(\psi,\mathring\mu^{2},\mathring\lambda)\Vert_{\mathcal{L}_{\sigma}\times\dots\mathcal{L}_{\Theta}} \leq D \left( \Vert \psi \Vert_{\mathcal{L}_{\psi}}^{2} + \mathring \mu^{4} + \Vert \mathring\lambda \Vert_{\mathcal{L}_{\lambda}}^{2}\right)  
\]
and 
\begin{equation}\label{eqn:diff-data-from-scalar-field-lambda}
\Vert \mathfrak{L}(\psi_{1},\mathring\mu^{2}_{1},\mathring\lambda_{1}) - \mathfrak{L}(\psi_{2},\mathring\mu^{2}_{2},\mathring\lambda_{2}) \Vert_{\mathcal{L}_{\sigma}\times\dots\mathcal{L}_{\Theta}} \leq D \epsilon  \left( \Vert \psi_{1} -  \psi_{2} \Vert_{\mathcal{L}_{\psi}} + |\mathring \mu^{2}_{1} - \mathring\mu_{2}^{2}| + \Vert \mathring\lambda_{1}-  \mathring\lambda_{2} \Vert_{\mathcal{L}_{\lambda}} \right).
\end{equation}
We will use below that this holds true even if we allowed $\mathfrak{L}$ to depend on $\delta$ (since changing $\delta$ does not directly affect any of the unknowns besides $\psi$).

Now, using Proposition~\ref{solveforpsi} we define a map
\[\mathfrak{M} : B_{\epsilon}\left(\mathcal{L}_{\sigma}\right)\times \cdots \times B_{\epsilon}\left(\mathcal{L}_{\lambda} \right) \to B_{\epsilon}\left(\mathcal{L}_{\psi}\right) \times B_{\epsilon}\left(\mathcal{L}_{\mu^2}\right),\]
which takes a set of renormalized quantities and solves for the scalar field $\psi$ and Klein--Gordon mass $\mu^2$. Proposition~\ref{solveforpsi} implies that if $\mathfrak{M}(\mathfrak{a}_{i}) = ( \psi_{i}^{\mathfrak{M}}, (\mathring \mu^{2}_{i})^{\mathfrak{M}})$ for $i=1,2$, then
\[
\Vert \psi_{i}^{\mathfrak{M}}\Vert_{\mathcal{L}_{\psi}} \leq C \delta, \qquad  ( \mathring \mu_{i}^{2})^{\mathfrak{M}} \leq D \Vert \mathfrak{a}_{i}\Vert_{\mathcal{L}_{\sigma}\times\dots\times \mathcal{L}_{\lambda}}
\]
and
\[
\Vert \psi_{1}^{\mathfrak{M}}- \psi_{2}^{\mathfrak{M}}\Vert_{\mathcal{L}_{\psi}} \leq C \delta  \Vert \mathfrak{a}_{1} - \mathfrak{a}_{2}\Vert_{\mathcal{L}_{\sigma}\times\dots\times \mathcal{L}_{\lambda}} , \qquad  |(\mathring \mu_{1}^{2})^{\mathfrak{M}} - (\mathring \mu_{2}^{2})^{\mathfrak{M}}| \leq D \Vert \mathfrak{a}_{1} - \mathfrak{a}_{2}\Vert_{\mathcal{L}_{\sigma}\times\dots\times \mathcal{L}_{\lambda}}
\]

 Finally, we define
\[\mathfrak{T} : B_{\epsilon}\left(\mathcal{L}_{\psi}\right) \times B_{\epsilon}\left(\mathcal{L}_{\mu^2}\right) \times B_{\epsilon}\left(\mathcal{L}_{\lambda}\right) \to B_{\epsilon}\left(\mathcal{L}_{\psi}\right) \times B_{\epsilon}\left(\mathcal{L}_{\mu^2}\right)\]
by $\mathfrak{T} = \mathfrak{M}\circ\left(\mathfrak{L},\mathring{\lambda}\right)$. 

Note that if $\mathfrak{T}(\psi_{i},\mathring \mu_{i},\mathring\lambda_{i}) = (\psi_{i}^{\mathfrak{T}}, (\mathring\mu_{i}^{2})^{\mathfrak{T}})$ for $i=1,2$, by combining the estimates above we conclude that 
\[
\Vert \psi_{i}^{\mathfrak{T}}\Vert_{\mathcal{L}_{\psi}} \leq C \delta, \qquad  ( \mathring \mu_{i}^{2})^{\mathfrak{T}} \leq D \left( \Vert \psi_{i} \Vert_{\mathcal{L}_{\psi}}^{2} + (\mathring \mu_{i}^{2})^{2} + \Vert \mathring\lambda_{i} \Vert_{\mathcal{L}_{\lambda}}^{2}\right) \leq D \epsilon^{2}
\]
as well as
\[
\Vert \psi_{1}^{\mathfrak{T}}- \psi_{2}^{\mathfrak{T}}\Vert_{\mathcal{L}_{\psi}} \leq C \delta \epsilon \left( \Vert \psi_{1} -  \psi_{2} \Vert_{\mathcal{L}_{\psi}} + |\mathring \mu^{2}_{1} - \mathring\mu_{2}^{2}| + \Vert \mathring\lambda_{1}-  \mathring\lambda_{2} \Vert_{\mathcal{L}_{\lambda}} \right)
\]
and
\[
|(\mathring \mu_{1}^{2})^{\mathfrak{M}} - (\mathring \mu_{2}^{2})^{\mathfrak{M}}| \leq D \epsilon \left( \Vert \psi_{1} -  \psi_{2} \Vert_{\mathcal{L}_{\psi}} + |\mathring \mu^{2}_{1} - \mathring\mu_{2}^{2}| + \Vert \mathring\lambda_{1}-  \mathring\lambda_{2} \Vert_{\mathcal{L}_{\lambda}} \right).
\]
Taking first $\epsilon>0$ sufficiently small, we can arrange that $\mathring \mu^{2} \in B_{\epsilon}(\mathcal{L}_{\mu^{2}})$ and so that the constants in the previous two equations (i.e., the contraction map estimate) are strictly less than $1$ (for $\delta \leq 1$). Now, choosing $\delta$ small enough, we can guarantee that $\psi^{\mathfrak{T}} \in B_{\epsilon}(\mathcal{L}_{\psi})$. Thus, for $\mathring\lambda_{1} = \mathring \lambda_{2}\in B_{\epsilon}(\mathcal{L}_{\lambda})$, the map $\mathfrak{T}(\cdot,\mathring\lambda)$ is a contraction map and thus has a unique fixed point. 

The fixed point defines $\psi$ and $\mathring \mu^{2}$ (depending on $\delta,\mathring\lambda$). Then using the map $\mathfrak{L}$ above, we may solve for the remaining unknowns. Now, the asserted estimates with $\delta_{1}=\delta_{2}$ all follow immediately from those established above. 

It thus remains to establish the dependence of the estimates on $\delta$. Consider $(\mathring \lambda_{i},\delta_{i})$ with associated solutions $\mathfrak{b}_{i},\psi_{i},\mu_{i}^{2}$ for $i=1,2$. We may assume that $\delta_{2} \not = 0$ (if $\delta_{1}=\delta_{2}=0$ the estimates follow trivially). Note that $\frac{\delta_{1}}{\delta_{2}}\psi_{2}$ solves the scalar field equation with $\delta=\delta_{1}$ and data $\mathfrak{b}_{2}$. Thus, Proposition~\ref{solveforpsi} implies that 
\[
\left\Vert \psi_{1}- \frac{\delta_{1}}{\delta_{2}} \psi_{2}\right\Vert_{\mathcal{L}_{\psi}} \leq D \delta_{1} \left( \Vert \mathfrak{b}_{1} - \mathfrak{b}_{2}\Vert_{\mathcal{L}_{\sigma}\times \dots \mathcal{L}_{\Theta}} + \Vert \mathring \lambda_{1} - \mathring \lambda_{2}\Vert_{\mathcal{L}_{\lambda}}\right)
\]
and
\[
|\mathring\mu_{1}^{2} - \mathring \mu_{2}^{2}| \leq D \left( \Vert \mathfrak{b}_{1} - \mathfrak{b}_{2}\Vert_{\mathcal{L}_{\sigma}\times \dots \mathcal{L}_{\Theta}} + \Vert \mathring \lambda_{1} - \mathring \lambda_{2}\Vert_{\mathcal{L}_{\lambda}}\right)
\]
We can remove the fraction in the $\psi$ bound as follows:
\begin{align*}
\Vert \psi_{1} - \psi_{2} \Vert_{\mathcal{L}_{\psi}} & \leq \left\Vert \psi_{1}- \frac{\delta_{1}}{\delta_{2}} \psi_{2}\right\Vert_{\mathcal{L}_{\psi}} + \delta_{2}^{-1}|\delta_{2}-\delta_{1}| \Vert \psi_{2}\Vert_{\mathcal{L}_{\psi}}\\
& \leq D \delta_{1} \left( \Vert \mathfrak{b}_{1} - \mathfrak{b}_{2}\Vert_{\mathcal{L}_{\sigma}\times \dots \mathcal{L}_{\Theta}} + \Vert \mathring \lambda_{1} - \mathring \lambda_{2}\Vert_{\mathcal{L}_{\lambda}}\right) + D |\delta_{1}-\delta_{2}|. 
\end{align*}
Combined with \eqref{eqn:diff-data-from-scalar-field-lambda} (and after taking $\epsilon$ smaller if necessary, to absorb the $\psi,\mu^{2}$ terms back into the left hand side), we conclude the desired estimates for $\psi,\mu^{2}$. Finally, the bounds for the remaining quantities follow by combining the bounds on $\psi,\mu^{2}$ with bounds with the ``nonlinear dependence on parameters'' proven in Proposition~\ref{fixth}. 
\end{proof}
\section{Solving for $\mathring{\lambda}$}\label{asectionforlambda}
The goal of this section is to prove the following proposition.
\begin{proposition}\label{fixlam}Let $\epsilon > 0$ be sufficiently small. Then, given $\delta \geq 0$ sufficiently small,
we may find
\[\left(\mathring{\sigma},\left(B^{(N)}_{\chi},B^{(S)}_{\chi'},B^{(A)}_z\right),\left(\mathring{X},\mathring{Y}\right),\mathring{\Theta},\psi,\mathring{\mu}^2 ,\mathring{\lambda}\right)\in B_{\epsilon}\left(\mathcal{L}_{\sigma}\right)\times \cdots \times B_{\epsilon}\left(\mathcal{L}_{\psi}\right)\times B_{\epsilon}\left(\mathcal{L}_{\lambda}\right),\]
which solve~(\ref{renormalsigmaeqn}), (\ref{Btripleeqn}), ~(\ref{X1}) and~(\ref{Y1}),~(\ref{thetaringeqnrho}) and~(\ref{thetaringeqnz}),~(\ref{theeqn10}), and~\eqref{lamdef-rho} and~\eqref{lambdef-z}.

Furthermore, $\mathring{\sigma}$, $\left(B^{(N)}_{\chi},B^{(S)}_{\chi'},B^{(A)}_z\right)$, $(\mathring{X},\mathring{Y})$, $\mathring{\Theta}$, $\mathring{\lambda}$  have ``continuous nonlinear dependence on the parameter'' $\delta$, in the sense of Lemma~\ref{fixit} and Remark~\ref{contnonlindep}, and, the scalar field $\psi$ and Klein--Gordon mass $\mu^2$ are Lipschitz continuous in $\delta$. 
\end{proposition}

\subsection{A modified equation for $\mathring{\lambda}$}

We would like to solve for $\mathring \lambda$ by integrating the first order equation $d\lambda = \alpha - \frac 12 d\log X$, cf.\ \eqref{lamdef-rho} and \eqref{lambdef-z}. However, (at this stage of the argument) we are unable to directly check that $d\alpha = 0$ (see Theorem 1.2 in \cite{HBH:geometric}). Instead, we will first solve the following system
\begin{align}
\label{lamdef-rho-doable} \partial_{\rho}\mathring\lambda & = \alpha_{\rho} - (\alpha_{K})_{\rho} - \frac 12 \partial_{\rho}\log (1+\mathring X) \\
 \label{lambdef-z-doable} \partial_{z}\mathring\lambda & = \alpha_{z} - (\alpha_{K})_{z} - \frac 12 \partial_{z} \log(1+\mathring X)- \int_{\rho}^{\infty} \left(\beta_{2}\wedge \left(d\mathring\lambda - (\alpha-(\alpha_{K}) - \frac 12 d \log (1+\mathring X) \right)\right)_{\rho,z}(\tau,z) d\tau,
\end{align}
\index{Metric Quantities!$\beta_{2}$}where the $1$-form $\beta_{2} = (\beta_{2})_{\rho}d\rho + (\beta_{2})_{z}dz$
\begin{align*}
\left( (\partial_{\rho}\sigma)^{2} + (\partial_{z}\sigma)^{2} \right) (\beta_{2})_{\rho}  & = \frac{1}{2} \left((\partial_{z}\sigma)(\partial^{3}_{\rho,z,z}\sigma - \partial^{3}_{\rho}\sigma) +(\partial_{\rho}\sigma)(\partial^{3}_{z}\sigma - \partial^{3}_{\rho,\rho,z}\sigma + 2\partial^{2}_{\rho}\sigma+ 2\partial^{2}_{z}\sigma)\right), \\
\left( (\partial_{\rho}\sigma)^{2} + (\partial_{z}\sigma)^{2} \right) (\beta_{2})_{z}  & = \frac{1}{2} \left((\partial_{\rho}\sigma)(\partial^{3}_{\rho,\rho,z}\sigma - \partial^{3}_{z}\sigma) +(\partial_{z}\sigma)(\partial^{3}_{\rho}\sigma - \partial^{3}_{\rho,z,z}\sigma + 2\partial^{2}_{\rho}\sigma+ 2\partial^{2}_{z}\sigma)\right)
\end{align*}
comes from the compatibility condition, Theorem 1.2 in \cite{HBH:geometric}.

\subsection{Integrating the first order equations} The goal of this subsection is to show that after taking $\delta > 0$ sufficiently small, there is $\mathring \lambda \in B_{\epsilon}(\mathcal{L}_{\lambda})$ so that if we use Proposition~\ref{fixpsi} to solve for
\[
\left( \mathring\sigma, B,(\mathring X,\mathring Y),\mathring \Theta,\psi,\mathring{\mu}^2\right),
\]
then $\mathring\lambda$ solves \eqref{lamdef-rho-doable} and \eqref{lambdef-z-doable}. We will always assume that we have taken $\delta>0$ sufficiently small so that $|\partial\sigma| \not = 0$ on $\mathscr{B}$. The following lemma is an immediate consequence of Theorem 1.2 in \cite{HBH:geometric}.
\begin{lemma}\label{bianch-eqns-satisfied}
Suppose that for $\mathring\lambda$ fixed,
\[
\left( \mathring\sigma, B,(\mathring X,\mathring Y),\mathring \Theta,\psi,\mathring{\mu}^2\right),
\]
solve their respective equations. Then the $1$-form $\alpha$ defined in Section \ref{subsubsection-mathring-lambda} satisfies
\[
d\alpha = \beta_{2}\wedge \left(d\lambda -\alpha - \frac 12 d\log X\right),
\]
where $\beta_{2}$ is defined above.
\end{lemma}

Now, for $\epsilon > 0$ fixed sufficiently small and $\delta\geq0$ sufficiently small, we define a map
\[
\mathfrak{L}(\cdot, \delta): B_{\epsilon}(\mathcal{L}_{\lambda}) \to B_{\epsilon}(\mathcal{L}_{\lambda})
\]
as follows: given $\mathring \lambda \in B_{\epsilon}(\mathcal{L}_{\lambda})$, we use Proposition \ref{fixpsi} to solve for the other parameters
\[
\left( \mathring\sigma, B,(\mathring X,\mathring Y),\mathring \Theta,\psi,\mathring{\mu}^2\right),
\]
then we integrate \eqref{lamdef-rho-doable} to define \index{Metric Quantities!$\mathfrak{L}(\mathring\lambda)$}$\mathfrak{L}(\mathring\lambda)$ (which will be a function of $\rho$ and $z$) by the expression
\begin{equation}\label{def-lambda-fp}
\mathfrak{L}(\mathring\lambda,\delta) \doteq -\int_{\rho}^{\infty}\left(\alpha_{\rho} - (\alpha_{K})_{\rho} - \frac12 \partial_{\rho}\log(1+\mathring X) \right)(\tau,z) d\tau.
\end{equation}
The next lemma checks, among other things, that this integral converges as long as $\epsilon>0$ is sufficiently small.

\begin{lemma}\label{lemm:lam-bound-mathfrakL}
For $\epsilon>0$ and $\delta\geq0$ sufficiently small, the integral \eqref{def-lambda-fp} converges, and the resulting expression solves both \eqref{lamdef-rho-doable} and  \eqref{lambdef-z-doable} with $\partial_{\rho}\mathring{\lambda}$ and $\partial_z\mathring{\lambda}$ replaced by $\partial_{\rho}\mathfrak{L}\left(\mathring{\lambda}\right)$ and $\partial_z\mathfrak{L}\left(\mathring{\lambda}\right)$ respectively. Moreover, we have the bounds
\[
\Vert \mathfrak{L}(\mathring\lambda,\delta) \Vert_{\mathcal{L}_{\lambda}} \leq C \left( \delta^{2}+ \Vert \mathring\lambda\Vert_{\mathcal{L}_{\lambda}}^{2} \right)
\]
and 
\[
\Vert \mathfrak{L}(\mathring\lambda_{1},\delta_{1})-\mathfrak{L}(\mathring\lambda_{2},\delta_{2})\Vert_{\mathcal{L}_{\lambda}}\leq C( \delta_{1} + \delta_{2} + \Vert \mathring\lambda_{1}\Vert_{\mathcal{L}_{\lambda}}+ \Vert \mathring\lambda_{2}\Vert_{\mathcal{L}_{\lambda}}) \left( |\delta_{1} - \delta_{2}| + \Vert \mathring\lambda_{1}-\mathring\lambda_{2}\Vert_{\mathcal{L}_{\lambda}} \right).
\]
\end{lemma}
\begin{proof}
For now, we will consider the parameters
\[
\left( \mathring\sigma, B,(\mathring X,\mathring Y),\mathring \Theta,\psi,\mathring{\mu}^2\right),
\]
simply as a fixed element of $B_{\epsilon}(\mathcal{L}_{\sigma})\times \dots B_{\epsilon}(\mathcal{L}_{\mu^{2}})$. We claim that the following estimate holds
\begin{align*}
& \Vert r^{2}(\alpha -\alpha_{K}) \Vert_{\hat C^{1,\alpha_{0}}(\overline{\mathscr{B}_{A}}\cup\overline{\mathscr{B}_{H}})} + \Vert r^{2}d \log (1+\mathring X) \Vert_{\hat C^{1,\alpha_{0}}(\overline{\mathscr{B}_{A}}\cup\overline{\mathscr{B}_{H}})} \\
&  \leq C \left\Vert \left( \mathring\sigma, B,(\mathring X,\mathring Y),\mathring \Theta\right)\right\Vert_{\mathcal{L}_{\sigma}\times\dots \mathcal{L}_{\Theta}} + C \Vert \psi\Vert_{\mathcal{L}_{\psi}}^2.
\end{align*}
We consider the bounds for $\alpha - \alpha_{K}$ first. We explain how to bound the $\rho$-component, since the $z$ component is similar. The only difficult term in $\alpha_{\rho}$ is
\[
\frac 14 |\partial\sigma|^{-2} (\partial_{\rho}\sigma)\sigma X^{-2} (\partial_{\rho}X)^{2} = \frac 14 |\partial\sigma|^{-2}  (\partial_{\rho}\log X)^{2},
\]
since $\partial_{\rho}\log X \sim \rho^{-1}$. Note, however, that $(\alpha_{K})_{\rho}$ contains a corresponding term. Thus, it suffices to bound 
\begin{align*}
& \frac 14 |\partial\sigma|^{-2} (\partial_{\rho}\sigma)\sigma (\partial_{\rho}\log X)^{2} - \frac 14 (\partial_{\rho}\log X_{K})^{2} \\
& = \frac 1 4 (1- \rho |\partial\sigma|^{-2} (1+\mathring\sigma + \rho |\partial\mathring\sigma|^{2})) \rho (\partial_{\rho} \log (1+\mathring X))^{2} - \frac 1 4  |\partial\sigma|^{-2} (1+\mathring\sigma + \rho |\partial\mathring\sigma|^{2}) \rho^{2} (\partial_{\rho} \log X_{K} )^{2}
\end{align*} 
The first term is easily bounded in $\hat C^{1,\alpha_{0}}(\overline{\mathscr{B}_{A}}\cup\overline{\mathscr{B}_{H}})$ as claimed, while the second term is similarly bounded after observing that the $\rho^{2}$ cancels the divergent $\partial_{\rho}\log X_{K}\sim \frac 2\rho$ contribution. The remaining estimates easily follow (the terms involving $\partial_{\rho}X$ are handled in a similar manner). 

Note that this shows that the integral defining $\mathfrak{L}(\mathring\lambda)$ converges, and justifies differentiating under the integral sign. Given this, \eqref{lamdef-rho-doable}, with $\partial_{\rho}\mathring{\lambda}$ replaced by $\partial_{\rho}\mathfrak{L}\left(\mathring{\lambda}\right)$, is automatically satisfied and we have the bounds
\begin{equation}\label{boundawaypNpS}
\Vert \mathfrak{L}(\mathring\lambda) \Vert_{\hat C^{1,\alpha_{0}}(\overline{\mathscr{B}_{A}} \cup \overline{\mathscr{B}_{H}} )} \leq C \left\Vert \left( \mathring\sigma, B,(\mathring X,\mathring Y),\mathring \Theta \right)\right\Vert_{\mathcal{L}_{\sigma}\times\dots \mathcal{L}_{\Theta}} + C \Vert \psi\Vert_{\mathcal{L}_{\psi}}^2.
\end{equation}

From now on, we will assume that the parameters
\[
\left( \mathring\sigma, B,(\mathring X,\mathring Y),\mathring \Theta,\psi,\mathring{\mu}^2\right),
\]
solve their respective equations for fixed $\mathring\lambda$ and $\delta>0$. Then, by Lemma \ref{bianch-eqns-satisfied}, we find that
\[
d(\alpha-\alpha_{K}) = \beta_{2} \wedge \left(d\mathring\lambda - (\alpha-\alpha_{K}) - \frac 12 d \log (1+\mathring X)\right).
\]
Since~\eqref{boundawaypNpS} is easily seen to hold on the complement of any compact set containing $p_N$ and $p_S$, we see that, away from the intersection of the axis and the horizon, $\mathfrak{L}(\mathring\lambda)$ solves \eqref{lambdef-z-doable}, with $\partial_z\mathring{\lambda}$ replaced by $\partial_z\mathfrak{L}\left(\mathring{\lambda}\right)$:
\begin{align*}
\partial_{z}\mathfrak{L}(\mathring\lambda)(\rho,z) & = -\int_{\rho}^{\infty}\partial_{z}\left(\alpha_{\rho} - (\alpha_{K})_{\rho} - \frac12 \partial_{\rho}\log (1+\mathring X)\right)(\tau,z)d\tau\\
& = \alpha_{z}(\rho,z)-(\alpha_{K})_{z} (\rho,z)  - \frac 12 \partial_{z}\log(1+\mathring X)- \int_{\rho}^{\infty} (d\alpha)_{\rho z}(\tau,z) d\tau\\
& = \alpha_{z}(\rho,z)-(\alpha_{K})_{z} (\rho,z) - \frac12 \partial_{z}\log(1+\mathring X) \\
& \qquad - \int_{\rho}^{\infty} \left(\beta_{2}\wedge \left(d\mathring\lambda - (\alpha-(\alpha_{K}) - \frac 12 d\log(1+\mathring X) \right)\right)_{\rho,z}(\tau,z) d\tau.
\end{align*}

We now turn to extending the estimate~\eqref{boundawaypNpS} to the region $\overline{\mathscr{B}_{N}}$. It is sufficient to estimate $d\mathfrak{L}\left(\mathring{\lambda}\right)$ in $\hat C^{0,\alpha_{0}}(\overline{\mathscr{B}_{N}})$. The first step is to compute the change of coordinates from $(\rho,z)$ to $(s,\chi)$ for the form $\alpha-\alpha_{K}$. This may be carried out directly, or alternatively, one may notice that if we consider this as a coordinate change for the full metric ansatz, we obtain
\[
g =  -V dt^{2} + 2W dtd\phi + X d\phi^{2} + (s^{2}+\chi^{2})e^{2\lambda}(ds^{2}+d\chi^{2}),
\]
and thus we may read off the expression for $\alpha$ (shifted by $\frac 12 \log(s^{2}+\chi^{2})$)  from Theorem 1.3 in \cite{HBH:geometric} by simply replacing $\rho$ by $s$ and $z$ by $\chi$. Carrying this out, we find
\begin{align*}
\left( (\partial_{s}\sigma)^{2} + (\partial_{\chi}\sigma)^{2} \right) \left( \alpha_{s} - \frac{s}{s^{2}+\chi^{2}} \right)& = \frac 14 (\partial_{s}\sigma)\sigma X^{-2} \left((\partial_{s}X)^{2} - (\partial_{\chi}X)^{2} \right) +\frac 12 (\partial_{\chi}\sigma)\sigma X^{-2} \left((\partial_{s}X)(\partial_{\chi}X)\right)\\
& + \partial_{s}\sigma (\partial^{2}_{s}\sigma - \partial^{2}_{\chi}\sigma) + \partial_{\chi}\sigma (\partial^{2}_{s,\chi}\sigma)\\
& + \frac 14 (\partial_{s}\sigma)\sigma X^{-2} \left((\theta_{s})^{2} - (\theta_{\chi})^{2} \right) +\frac 12 (\partial_{\chi}\sigma)\sigma X^{-2} \left( (\theta_{s})(\theta_{\chi}) \right)\\
& + \frac 12 (\partial_{s}\sigma)\sigma ((\partial_{s}\psi)^{2}-(\partial_{\chi}\psi)^{2}) + (\partial_{\chi}\sigma)\sigma \left((\partial_{s}\psi)(\partial_{\chi}\psi) \right),
\end{align*}
as well as
\begin{align*}
\left( (\partial_{s}\sigma)^{2} + (\partial_{\chi}\sigma)^{2} \right) \left( \alpha_{\chi} - \frac{\chi}{s^{2}+\chi^{2}}\right) & = \frac 1 4 (\partial_{\chi}\sigma) \sigma X^{-2} \left( (\partial_{\chi}X)^{2} - (\partial_{s}X)^{2} \right) + \frac 12 (\partial_{s}\sigma)\sigma X^{-2}\left( (\partial_{s}X)(\partial_{\chi}X) \right)\\
& + \partial_{\chi}\sigma (\partial^{2}_{\chi}\sigma - \partial^{2}_{s}\sigma) + \partial_{s}\sigma ( \partial^{2}_{s,\chi}\sigma)\\
& + \frac 1 4 (\partial_{\chi}\sigma) \sigma X^{-2} \left( (\theta_{\chi})^{2} - (\theta_{s})^{2} \right) + \frac 12 (\partial_{s}\sigma)\sigma X^{-2}\left( (\theta_{s})(\theta_{\chi}) \right)\\
& + \frac 12 (\partial_{\chi}\sigma)\sigma\left( (\partial_{\chi}\psi)^{2} - (\partial_{s}\psi)^{2} \right) + (\partial_{s}\sigma)\sigma \left( (\partial_{s}\psi) (\partial_{\chi}\psi) \right).
\end{align*}
Note that $\alpha_{K}$ will cancel the $\frac 12 d\log(s^{2}+\chi^{2})$ terms, so these are not relevant to $\mathring\lambda$. We claim that
\[
\Vert \alpha -\alpha_{K}\Vert_{\hat C^{0,\alpha_{0}}(\overline{\mathscr{B}_{N}})} + \Vert d\log(1+\mathring X)\Vert_{\hat C^{0,\alpha_{0}}(\overline{\mathscr{B}_{N}})} \leq C \left\Vert \left( \mathring\sigma, B,(\mathring X,\mathring Y),\mathring \Theta\right)\right\Vert_{\mathcal{L}_{\sigma}\times\dots \mathcal{L}_{\mu^2}} + C \Vert \psi\Vert_{\mathcal{L}_{\psi}}^2.
\]
This follows from a similar argument as before. Note that again the terms in $\alpha$ involving $\partial_{s} \log X$ are singular, but are cancelled by the corresponding terms in $\alpha_{K}$. The computation is nearly identical as above (except now in $s,\chi$ coordinates) so we omit it. 

Finally, we need to also control the $1$-form\index{Miscellaneous!$\mathfrak{I}$}
\[
\mathfrak{I} \doteq \left( \int_{\rho}^{\infty} \left(\beta_{2}\wedge \left(d\mathring\lambda - (\alpha-(\alpha_{K}) - \frac 12 d\log(1+\mathring X) \right)\right)_{\rho,z}(\tau,z) d\tau \right)dz.
\]
We have seen that the $\rho,z$ derivatives $\partial\mathring\sigma,\partial^{2}\mathring\sigma,\partial^{3}\sigma$ are all controlled in $\hat C^{0,\alpha_{0}}(\overline{\mathscr{B}_{N}})$ by the $\mathcal{L}_{\sigma}$ norm. Thus the $\rho,z$ coefficients of $\beta_{2}$ are controlled in $\hat C^{0,\alpha_{0}}(\overline{\mathscr{B}_{N}})$. In particular, note that
\[
\beta_{2} = (\chi (\beta_{2})_{\rho} - s(\beta_{2})_{z}) ds + (s(\beta_{2})_{\rho} + \chi (\beta_{2})_{z})d\chi
\]
Combining this with the above bounds, we have that\index{Miscellaneous!$\mathfrak{f}$}
\[
\left(\beta_{2}\wedge \left(d\mathring\lambda - (\alpha-(\alpha_{K}) - \frac 12 d\log(1+\mathring X) \right)\right)_{s,\chi} = s \mathfrak{f}_{1} + \chi \mathfrak{f}_{2},
\]
where
\[
\Vert \mathfrak{f}_{1} \Vert_{\hat C^{0,\alpha_{0}}(\overline{\mathscr{B}_{N}})} + \Vert \mathfrak{f}_{2} \Vert_{\hat C^{0,\alpha_{0}}(\overline{\mathscr{B}_{N}})} \leq C \Vert \mathring\sigma \Vert_{\mathcal{L}_{\sigma}}\Vert \mathring\lambda\Vert_{\mathcal{L}_{\lambda}} + C \left\Vert \left( \mathring\sigma, B,(\mathring X,\mathring Y),\mathring \Theta\right)\right\Vert_{\mathcal{L}_{\sigma}\times\dots \mathcal{L}_{\mu^2}} + C \Vert \psi\Vert_{\mathcal{L}_{\psi}}^2.
\]
Observing that
\[
\left(\beta_{2}\wedge \left(d\mathring\lambda - (\alpha-(\alpha_{K}) - \frac 12 d\log(1+\mathring X) \right)\right)_{\rho,z} = \frac{s}{s^{2}+\chi^{2}} \mathfrak{f}_{1} + \frac{\chi}{s^{2}+\chi^{2}} \mathfrak{f}_{2},
\]
if we set
\[
\mathfrak{I}_{1} \doteq \int_{\rho}^{\infty} \frac{s}{s^{2}+\chi^{2}} \mathfrak{f}_{1} (\tau,z) d\tau, \qquad \mathfrak{I}_{1} \doteq \int_{\rho}^{\infty} \frac{\chi }{s^{2}+\chi^{2}} \mathfrak{f}_{2} (\tau,z) d\tau,
\]
we see that (because $\partial_{\rho} = (s^{2}+\chi^{2})^{-1}(\chi\partial_{s}+s\partial_{\chi})$),
\[
\chi\partial_{s}\mathfrak{I}_{1} + s\partial_{\chi}\mathfrak{I}_{1} = s \mathfrak{f}_{1}, \qquad \chi\partial_{s}\mathfrak{I}_{2} + s\partial_{\chi}\mathfrak{I}_{2} = \chi \mathfrak{f}_{2}.
\]
These expressions are easily integrated using the observation that
\[
\partial_{\tau} \left( \mathfrak{I}_{1}(\sqrt{\tau^{2} + c},\tau) \right) = \left(\frac \chi s \partial_{s} \mathfrak{I}_{1} + \partial_{\chi}\mathfrak{I}_{1} \right)\Big|_{(s,\chi) = (\sqrt{\tau^{2}+c},\tau)} = \mathfrak{f}_{1}(\sqrt{\tau^{2}+c},\tau),
\]
with a similar form for $\mathfrak{I}_{2}$.

Putting this together, we obtain
\begin{align*}
\mathfrak{I}_{1}(s,\chi) & = \mathfrak{I}_{1}(\sqrt{\chi_{0}^{2}+s^{2}-\chi^{2}},\chi_{0}) - \int_{\chi}^{\chi_{0}} \mathfrak{f}_{1}(\sqrt{\tau^{2}+s^{2}-\chi^{2}},\tau) d\tau,\\
\mathfrak{I}_{2}(s,\chi) & = \mathfrak{I}_{2}(s_{0},\sqrt{s_{0}^{2}+\chi^{2}-s^{2}}) - \int_{s}^{s_{0}} \mathfrak{f}_{2}(\tau,\sqrt{\tau^{2}+\chi^{2}-s^{2}}) d\tau.
\end{align*}
Choosing $(s_{0},\chi_{0})$ appropriately, we may assure that $(\sqrt{\chi_{0}^{2}+s^{2}-\chi^{2}},\chi_{0})$ and $(s_{0},\sqrt{s_{0}^{2}+\chi^{2}-s^{2}})$ are bounded away from $\{s=\chi = 0\}$. Then, the bounds for $\mathfrak{f}_{1},\mathfrak{f}_{2}$ given above allow us to bound both terms in $\mathfrak{I}_{1}$ and $\mathfrak{I}_{2}$ in $\hat C^{0,\alpha_{0}}(\overline{\mathscr{B}_{N}})$. Finally, thanks to the extra factor of $s$ and $\chi$ in the formula
\[
\mathfrak{I} = (\mathfrak{I}_{1}+\mathfrak{I}_{2}) (sds - \chi d\chi),
\]
we may bound the contribution to the (Cartesian) first-derivative of $\mathring\lambda$. Putting this all together, we obtain
\[
\Vert \mathfrak{L}(\mathring\lambda) \Vert_{\hat C^{1,\alpha_{0}}(\overline{\mathscr{B}})} \leq C \Vert \mathring\sigma \Vert_{\mathcal{L}_{\sigma}}\Vert \mathring\lambda\Vert_{\mathcal{L}_{\lambda}} + C \left\Vert \left( \mathring\sigma, B,(\mathring X,\mathring Y),\mathring \Theta  \right)\right\Vert_{\mathcal{L}_{\sigma}\times\dots \mathcal{L}_{\Theta}} + C \Vert \psi\Vert_{\mathcal{L}_{\psi}}^2.
\]
Combined with Proposition \ref{fixpsi}, the first claimed estimate follows. The second estimate follows from a similar argument.
\end{proof}

Thus, we have seen for $\epsilon> 0$ and $\delta\geq0$ sufficiently small, $\mathfrak{L} : B_{\epsilon}(\mathcal{L}_{\lambda}) \to B_{\epsilon}(\mathcal{L}_{\lambda})$ is a contraction map. Thus, we may find, for each $\delta\geq 0$ fixed small, $\mathring\lambda$, and associated (by Proposition \ref{fixpsi}) parameters
\[
\left( \mathring\sigma, B,(\mathring X,\mathring Y),\mathring \Theta,\psi,\mathring{\mu}^2\right)
\]
so that $\mathring\lambda$ solves \eqref{lamdef-rho-doable} and \eqref{lambdef-z-doable}.

\subsection{Upgrading the equations for $\mathring\lambda$}
Now that we have solved the fixed point, we would like to conclude that $\mathring\lambda$ satisfies the correct equations given in \eqref{lamdef-rho} and \eqref{lambdef-z}, i.e.
\[
d\mathring\lambda = \alpha - \alpha_{K} - \frac 12 d\log(1+\mathring X).
\]
For $z$ fixed, we set $f(\rho) \doteq \left(\partial_{z} \mathring\lambda - (\alpha_{z}-(\alpha_{K})_{z}) -\frac 12 \partial_{z}\log(1+\mathring X)\right)(\rho,z)$. Then, we see that
\[
f(\rho) = \int_{\rho}^{\infty} (\beta_{2})_{\rho} (\tau,z) f(\tau) d\tau
\]
and because $\mathring\sigma \in \mathscr{L}_{\sigma}$, we have that $|(\beta_{2})_{\rho}|\leq Cr ^{-3}$. For $\rho\geq \rho_{0}$ (and $z$ considered fixed) we have $c_{0} = c_{0}(\rho_{0})$ so that $|(\beta_{2})_{\rho}|\leq c_{0}\rho^{-3}$. Moreover, we have that $|f(\rho)|\leq c_{1}$ for some $c_{1}$. Thus, for $\rho\geq \rho_{0}$, we find
\[
|f(\rho)|\leq c_{0}\int_{\rho}^{\infty} |f(\tau)|  \tau^{-3} d\tau.
\]
Hence,
\[
|f(\rho)| \leq \frac{c_{0}c_{1}}{2} \rho^{-2}.
\]
Iterating this, we find
\[
|f(\rho)| \leq \frac{c_{0}^{j}c_{1}}{2 \cdot 4 \cdot \dots \cdot (2j)} \rho^{-2j} =  c_{1} \left( \frac{c_{0}}{2\rho^{2}} \right)^{j} \frac{1}{j!}.
\]
By Stirling's approximation, this tends to $0$ as $j\to\infty$. Thus $f(\rho) = 0$ for $\rho\geq \rho_{0}$ (which was arbitrary). Thus, we see that
\begin{equation}\label{eq:lam-full-eqn}
d\mathring\lambda = \alpha -\alpha_{K} - \frac 12 d\log(1+\mathring X).
\end{equation}
on $\overline{\mathscr{B}}$, as desired. 

\subsection{Lipschitz continuity with respect to $\delta$} Here we record the dependence of the solution 
\[
\index{Metric Quantities!$\mathfrak{c}$} \mathfrak{c}_{i} = \left( \mathring\sigma_{i}, B,(\mathring X_{i},\mathring Y_{i}),\mathring \Theta_{i},\mathring\lambda_{i}\right)
\]
and $(\psi_{i},\mathring \mu_{i}^{2})$ on $\delta_{i}\geq 0$ (small) for $i=1,2$ (we have written the expression in this form, since the scalar field behaves differently with respect to $\delta$ than the other parameters). 
\begin{lemma}
For $\epsilon>0$ and $\delta_{i}\geq 0$ sufficiently small, we have 
\[
\Vert \mathfrak{c}_{i} \Vert_{\mathcal{L}_{\sigma}\times\dots\times \mathcal{L}_{\lambda}} \leq D\delta_{i}^{2}, \qquad \Vert\psi_{i} \Vert_{\mathcal{L}_{\psi}} \leq D \delta_{i}, \qquad \mathring\mu_{i}^{2} \leq D \delta_{i}
\]
and
\[
\Vert \mathfrak{c}_{1} - \mathfrak{c}_{2} \Vert_{\mathcal{L}_{\sigma}\times\dots\times \mathcal{L}_{\lambda}} \leq D(\delta_{1}+\delta_{2})|\delta_{1} - \delta_{2}|, \qquad \Vert\psi_{1}-\psi_{2} \Vert_{\mathcal{L}_{\psi}} \leq D |\delta_{1}-\delta_{2}|, \qquad |\mathring\mu_{1}^{2} - \mathring\mu_{2}^{2}|\leq D |\delta_{1}-\delta_{2}|
\]
\end{lemma}
\begin{proof}
This follows immediately by combining the the $\delta$-dependence proven in Proposition \ref{fixpsi} with Lemma \ref{lemm:lam-bound-mathfrakL}. 
\end{proof}

This completes the proof of Proposition~\ref{fixlam}.

\section{Proof of the Main Result}
Proposition~\ref{fixlam} and Theorem \ref{thm:HBH-geometric-main} (cf.\ Theorem 1.3 in \cite{HBH:geometric}) immediately imply that we have actually solved the Einstein--Klein--Gordon equations in $\mathcal{M}$.
\begin{proposition}Let $\delta \geq 0$ be sufficiently small and let $\left( \mathring\sigma, B,(\mathring X,\mathring Y),\mathring \Theta,\psi,\mathring{\mu}^2,\mathring{\lambda}\right)$ be produced by Proposition~\ref{fixlam}. Then, the corresponding metric $g_{\delta}$ (see Remark~\ref{unrenormalize}) is a solution to the Einstein--Klein--Gordon equations on $\mathcal{M}$.
\end{proposition}

There are three remaining steps in the proof of Theorem~\ref{timeperiodicsoln}:
\begin{enumerate}
    \item We must show that the metric $g_{\delta}$ is smooth and that $(\mathcal{M},g_{\delta})$ is ``extendable to a regular black hole spacetime'' in the sense of Definition 2.5 in \cite{HBH:geometric}.
    \item We must show that the extended solution $(\tilde{\mathcal{M}},\tilde{g}_{\delta})$ is asymptotically flat.
    \item We must show that the scalar field $\psi$ decays exponentially with respect to $r$ and smoothly extends to $(\tilde{\mathcal{M}},\tilde{g}_{\delta})$.
    \item By introducing variations in the Kerr parameters $(a,M)$, we must arrange for the Klein--Gordon mass $\mu^2$ to be constant in $\delta$.
\end{enumerate}
The next four sections will resolve each of these issues. We remark that arguments similar to those used to handle (1) and (2) can be found in \cite{chrusciellopes}.

\subsection{Boundary Conditions and Regularity}
We begin by checking that $\mathring{\lambda}$ satisfies the desired compatibility conditions along $\{\rho = 0\}$. Recall that in Corollary \ref{omconst} we have seen that $W/X$ is constant on $\mathscr{H}$. We denote this constant by $-\Omega$.
\begin{proposition}\label{compatible}Let $\delta$ be sufficiently small and $\left( \mathring\sigma, B,(\mathring X,\mathring Y),\mathring \Theta,\psi,\mathring{\mu}^2,\mathring{\lambda}\right)$ be produced by Proposition~\ref{fixlam}. Then there exists a constant \index{Metric Quantities!$\kappa$}$\kappa > 0$ such that $\mathring{\lambda}$ satisfies the following:
\begin{enumerate}
    \item $\left(e^{2\lambda} - \rho^{-2}X\right)|_{\mathscr{A}} = 0$.
    \item There exists $\kappa > 0$ such that $\left(e^{2\lambda} - \kappa^{-2}\rho^{-2}\left(V-2\Omega W - \Omega^2X\right)\right)|_{\mathscr{H}} = 0$.
    \item On the set $\overline{\mathscr{B}_N}$ we have $\left(\left(\chi^2+s^2\right)e^{2\lambda} - s^{-2}X\right)|_{\{s=0\}\cup \{\chi = 0\}} = 0$.
    \item On the set $\overline{\mathscr{B}_N}$ we have $\left(\left(\chi^2+s^2\right)e^{2\lambda} - \kappa^{-2}\chi^{-2}\left(V-2\Omega W - \Omega^2X\right)\right)|_{\{s=0\}\cup \{\chi = 0\}} = 0$.
    \item On the set $\overline{\mathscr{B}_S}$ we have $\left(\left((\chi')^2+(s')^2\right)e^{2\lambda} - (s')^{-2}X\right)|_{\{s'=0\}\cup \{\chi' = 0\}} = 0$.
    \item On the set $\overline{\mathscr{B}_S}$ we have $\left(\left((\chi')^2+(s')^2\right)e^{2\lambda} - \kappa^{-2}(\chi')^{-2}\left(V-2\Omega W - \Omega^2X\right)\right)|_{\{s'=0\}\cup \{\chi' = 0\}} = 0$.
\end{enumerate}
\end{proposition}
\begin{proof}
Note that for $\delta$ sufficiently small, it follows by comparison with $X_K$ that $\rho^{-2}X|_{\mathscr{A}} > 0$. We then define a function $u \doteq \log\left(\rho^{-2}X\right)$. Observe that it follows easily from our previous estimates for $X$ that $\partial_zu$ continuously extends to $\mathscr{A}$.

Then, evaluating the $z$-component of $\lambda$'s equation from Theorem 1.3 of \cite{HBH:geometric} along the axis $\mathscr{A}$ (keep in mind that the right hand side of $\lambda$'s equation does not vanish identically on $\mathscr{A}$) yields
\[\partial_z\left(\lambda - \frac{1}{2}u\right)\Big|_{\mathscr{A}} = 0.\]
Since $\lambda - \frac{1}{2}u$ is easily seen to vanish at $z = \pm\infty$, we conclude that $\lambda- \frac{1}{2}u$ vanishes everywhere along $\mathscr{A}$. In terms of $X$ we thus have $\left(e^{2\lambda} - \rho^{-2}X\right)|_{\mathscr{A}} = 0$.

Now we turn to the horizon. In this case evaluating the $z$-component of $\lambda$'s equation \eqref{eq:lam-full-eqn} yields
\[\partial_z\left(\lambda + \frac{1}{2}\log(X)-\frac{1}{2}\log\left(\left(1+\mathring{\sigma}\right)^2\right)\right)\Big|_{\mathscr{H}} = 0,\]
where we recall that $\sigma = \rho\left(1+\mathring{\sigma}\right)$. We conclude that there exists a constant $c$ such that
\[e^{2\lambda}|_{\mathscr{H}} = c\frac{\left(1+\mathring{\sigma}\right)^2}{X}\Big|_{\mathscr{H}}.\]

From the definition of $\sigma$, it follows that
\[V = \frac{\sigma^2}{X} - \frac{W^2}{X}.\]
In particular,
\[V - 2\Omega W - \Omega^2X = \frac{\sigma^2}{X} - X^{-1}\left(W + \Omega X\right)^2.\]
It then is easy to establish
\[\rho^{-2}\left(V - 2\Omega W - \Omega^2X\right)|_{\mathscr{H}} = \frac{\left(1+\mathring{\sigma}\right)^2}{X}\Big|_{\mathscr{H}},\]
which implies
\[\left(e^{2\lambda} - c\rho^{-2}\left(V-2\Omega W - \Omega^2X\right)\right)|_{\mathscr{H}} = 0.\]
The fact that $c$ is positive for sufficiently small $\delta$ follows by comparison with its value on the Kerr spacetime.

The remaining statements are easily proved by repeating the above strategies in $(s,\chi)$ and $(s',\chi')$ coordinates. We omit the details.
\end{proof}

The following two lemmas will be useful when we study the regularity of our solutions.
\begin{lemma}\label{regu}Let $\mathbb{R}^2$ be covered by Cartesian coordinates $(x,y)$ and polar coordinates $(r,\theta)$. For $m\in\ZZ$ and $k\in\mathbb{N}$, if the function defined in polar coordinates
\[
u_{k,m}(r,\theta) \doteq r^{k}e^{im\theta}
\]
is in $C^{k}_{loc}(\mathbb{R}^{2})$, then $k \geq |m|$.
\end{lemma}
\begin{proof}
This is easily checked for $k=0,1$. In general, observe that
\[
\Delta u_{k,m} = \frac 1 r \frac{\partial}{\partial r} \left( r \frac{\partial u_{k,m}}{\partial r} \right) + \frac{1}{r^{2}} \frac{\partial^{2}u_{k,m}}{\partial \theta^{2}} = (k^{2}-m^{2}) u_{k-2,m}.
\]
Thus, if $u_{k,m}\in C^{k}_{loc}(\RR^{2})$, then we see that
\[
C^{k-2\lfloor\frac k 2 \rfloor}_{loc}(\mathbb{R}^{2}) \ni (\Delta)^{\lfloor \frac k 2 \rfloor} u_{k,m} =  \begin{cases} (k^{2}-m^{2})((k-2)^{2}-m^{2})\cdots(-m)^{2} r e^{im\theta} & k \text{ even}\\(k^{2}-m^{2})((k-2)^{2}-m^{2})\cdots(1-m^{2})e^{im\theta} & k \text{ odd} . \end{cases}
\]
From the $k=0,1$ case, we see that one of the terms $(k-2j)^{2}-m^{2}$ must vanish. Thus, $k\geq |m|$.
\end{proof}

\begin{lemma}\label{usefullemma}Let $\mathbb{R}^2$ be covered by Cartesian coordinates $(x,y)$ and polar coordinates $(r,\theta)$. Consider a function $g(x,y) \doteq e^{im\theta}f(r)$. Then $g \in C^{k,\alpha}\left(\mathbb{R}^2\right)$ implies that
\begin{enumerate}
    \item $f \in C^{k,\alpha}\left(\mathbb{R}\right)$.
    \item $\frac{d^jf}{d^jr}(0) = 0$ for all $0 \leq j \leq {\rm min}\left(|m| - 1,k\right)$.
    \item $\frac{d^jf}{d^jr}(0) = 0$ for all $|m| + 1 \leq j \leq k$ such that $j-|m|$ is odd.
\end{enumerate}
\end{lemma}
\begin{proof}Assume that $g \in C^{k,\alpha}\left(\mathbb{R}^2\right)$. First of all, since we have $g(x,0) = f(x)$, we immediately conclude that $f \in C^{k,\alpha}\left(\mathbb{R}\right)$.

We now turn to the second assertion. Fix $|m|$ and $k$; we proceed by induction in $j$. For $j=0$, the assertion easily follows from the continuity of $f$ and $g$. Now, assuming that the second assertion holds for $j-1$, then Taylor's theorem implies that there is $\tilde f \in C^{k,\alpha}(\RR)$ with $f(r) = r^{j}\tilde f(r)$. If $\frac{d^{j}f}{dr^{j}}(0) \not = 0$, then $\tilde f(0) \not = 0$, implying that
\[
e^{im\theta}r^{j} \in C^{k,\alpha}_{loc}(\RR^{2}).
\]
This contradicts Lemma \ref{regu}.

The third assertion is proved with a similar argument; we omit the details.
\end{proof}

We need one last definition before we are ready for the full regularity argument.
\begin{definition}We say that \index{Function Spaces!$\hat{C}^{k,\alpha}_{\rm m,azi}\left(\overline{\mathscr{B}}\right)$}$f \in \hat{C}^{k,\alpha}_{\rm m,azi}\left(\overline{\mathscr{B}}\right)$ if
\begin{enumerate}
    \item $f \in \hat{C}^{k,\alpha}\left(\overline{\mathscr{B}}\right)$
    \item After the introduction of Cartesian coordinates $(x,y,z)$ in $\overline{\mathscr{B}_A}$ by $x = \rho\cos\phi$ and $y = \rho\sin\phi$, the function $e^{im\phi}f|_{\overline{\mathscr{B}_A}} \in C^{k,\alpha}\left(\mathbb{R}^3\right)$.
    \item After the introduction of Cartesian coordinates $(x,y,w,v)$ in $\overline{\mathscr{B}_N}$ by $x = s\cos\phi$, $y = s\sin\phi$, $w = \chi\sin\phi$, $v = \chi\cos\phi$, the function $e^{im\phi}f|_{\overline{\mathscr{B}_N}} \in C^{k,\alpha}\left(\mathbb{R}^4\right)$.
    \item After the introduction of Cartesian coordinates $(x,y,w,v)$ in $\overline{\mathscr{B}_S}$ by $x = s'\cos\phi$, $y = s'\sin\phi$, $w = \chi'\sin\phi$, $v = \chi'\cos\phi$, the function $e^{im\phi}f|_{\overline{\mathscr{B}_N}} \in C^{k,\alpha}\left(\mathbb{R}^4\right)$.
\end{enumerate}
Recall that the value of $m$ has been fixed in Remark~\ref{fixmalpha0}.
\end{definition}

Now we show how the regularity for each of the renormalized unknowns can be upgraded.
\begin{proposition}\label{prop:full-reg}Let $\delta$ be sufficiently small and $\left( \mathring\sigma, B,(\mathring X,\mathring Y),\mathring \Theta,\psi,\mathring{\mu}^2,\mathring{\lambda}\right)$ be produced by Proposition~\ref{fixlam}. Then
\[\left( \mathring\sigma, B,(\mathring X,\mathring Y),\mathring \Theta,\mathring{\lambda}\right) \in \hat{C}^{\infty}\left(\overline{\mathscr{B}}\right)\times \cdots\times \hat{C}^{\infty}\left(\overline{\mathscr{B}}\right),\]
\[\psi \in \hat{C}^{\infty}_{\rm m,azi}\left(\overline{\mathscr{B}}\right)\]
\end{proposition}
\begin{proof}The proof is by induction. The induction hypothesis is that for every natural number $j \geq 0$, we have
\[\left( \mathring\sigma, B,(\mathring X,\mathring Y),\mathring \Theta,\mathring{\lambda}\right) \in\] \[\hat{C}^{2+j,\alpha_0}\times \left(\hat{C}^{1+j,\alpha_0}\times\cdots\times\hat{C}^{1+j,\alpha_0}\right)\times\left(\hat{C}^{2+j,\alpha_0}\times \hat{C}^{2+j,\alpha_0}\right) \times \hat{C}^{2+j,\alpha_0}\times \hat{C}^{1+j,\alpha_0},\]
\[\psi \in \hat{C}_{\rm m,azi}^{2+j,\alpha_0}.\]
where we have dropped, and will continue to, the $\left(\overline{\mathscr{B}}\right)$ from the function spaces for typographical reasons.

We start with the base case $j=0$. Proposition~\ref{fixlam} establishes everything except for the statement $\psi \in \hat{C}_{\rm m,azi}^{2,\alpha_0}$. For $\psi$ we simply need to revisit the proof of Lemma~\ref{psireg}, with the results of Proposition~\ref{compatible} in hand. In particular, the beginning of the proof of Lemma~\ref{psireg} along with Proposition~\ref{compatible} and Lemma~\ref{usefullemma} now imply (see Remark \ref{fullregpsi}) that when~(\ref{psimeqn}) is expressed in Cartesian coordinates near the axis, the corresponding equation is a uniformly elliptic equation with $C^{0,\alpha_0}$ coefficients. Thus, near the axis, the desired estimate follows immediately from Schauder theory. One argues similarly near the points $p_N$ and $p_S$, near the horizon, and in the region $\{\rho > 0\}$.

Let's assume the induction hypothesis holds for some $j \geq 0$, and we will show that it also holds for $j+1$. We start $\mathring{\sigma}$. We recall the equation:
\[\Delta_{\mathbb{R}^4}\mathring{\sigma} = \rho^{-1}\sigma^{-1}Xe^{2\lambda}\left(\left(\omega+X^{-1}Wm\right)^2 - \sigma^2\left(\frac{\mu^2}{X} + \frac{m^2}{X^2}\right)\right)\psi^2.\]
Away from $\{\rho > 0\}$ the desired regularity is trivial to obtain. Let's consider $\mathring{\sigma}$ near $\{\rho = 0\}$ but away from the points $p_N$ and $p_S$. We have $\rho^{-1}\sigma^{-1}X \in \hat{C}^{2+j,\alpha_0}$. However, near the axis, the term $\sigma^2X^{-2} \sim \rho^{-2}$, prevents a naive application of Schauder theory. In order to deal with this singularity, we simply observe that
\begin{equation}\label{regpsiyes}
\rho^{-2}m^2\psi^2 = \rho^{-2}\left|\partial_{\phi}\psi_m\right|^2 = \left|\partial_x\psi_m\right|^2 + \left|\partial_y\psi_m\right|^2 - \left(\partial_{\rho}\psi\right)^2,
\end{equation}
where $(x,y)$ are the Cartesian coordinates associated to the regularity argument of $\psi_m$. The right hand side of~(\ref{regpsiyes}) is easily seen to be $\hat{C}^{1+j,\alpha_0}$. Near the points $p_N$ and $p_S$, one may use $(s,\chi)$ or $(s',\chi')$ coordinates and argue in the same fashion. We eventually conclude,  via a standard Schauder estimate, that $\mathring{\sigma} \in \hat{C}^{3+j,\alpha_0}$.

It immediately follows from $B$'s equations and the arguments from Section~\ref{solvingforb} that $B \in \left(\hat{C}^{2+j,\alpha_0}\times\cdots\times\hat{C}^{2+j,\alpha_0}\right)$. It is also straightforward (keeping~(\ref{rhotos}) in mind) to establish
\begin{equation}\label{B111}
\left(1-\xi_N-\xi_S\right)\frac{B_{\rho}\partial_{\rho}X_K + B_z\partial_zX_K}{X_K^3}\in \hat{C}^{1+j,\alpha_0},
\end{equation}
\begin{equation}\label{B222}
\left(\chi^2+s^2\right)\xi_N\frac{B_{\rho}\partial_{\rho}X_K+B_z\partial_zX_K}{X_K^3} \in \hat{C}^{1+j,\alpha_0},
\end{equation}
\begin{equation}\label{B333}
\left((\chi')^2+(s')^2\right)\xi_S\frac{B_{\rho}\partial_{\rho}X_K+B_z\partial_zX_K}{X_K^3} \in \hat{C}^{1+j,\alpha_0}.
\end{equation}

Now we turn to $(\mathring{X},\mathring{Y})$. Using~(\ref{B111}), (\ref{B222}), (\ref{B333}), the observation~(\ref{regpsiyes}), and our already established renormalization scheme, one easily shows that $(\mathring{X},\mathring{Y}) \in \hat{C}^{3+j,\alpha_0}\times \hat{C}^{3+j,\alpha_0}$.

The desired regularity statement for $\mathring{\Theta}$ and $\psi$ are straightforward given the analysis we have already carried out. We omit the details. It remains to consider $\mathring\lambda$. As above, the only issue is at $\partial\overline{\mathscr{B}}$. Instead of the first-order equation used in the bootstrap, we now rely on the Liouville equation
\begin{align*}
2\partial^{2}_{\rho}\mathring\lambda + 2\partial^{2}_{z}\mathring\lambda & = - \partial^{2}_{\rho}\log (1+\mathring X) - \partial^{2}_{z}\log (1+\mathring X) +  \sigma^{-1}(\partial_{\rho}^{2}\sigma + \partial^{2}_{z}\sigma)\\
& - e^{2\lambda} \mu^{2} \psi^{2} - (\partial_{\rho}\psi)^{2} - (\partial_{z}\psi)^{2} - X^{-1}\left( 2m^{2} + \mu^{2} X\right)e^{2\lambda}\psi^2\\
& - \frac 12 X^{-2}\left[ (\partial_{\rho}X)^{2} + (\partial_{z}X)^{2} + (\theta_{\rho})^{2} + (\theta_{z})^{2} \right]\\
& + \frac 12 X_{K}^{-2}\left[ (\partial_{\rho}X_{K})^{2} + (\partial_{z}X_{K})^{2} + (\partial_{\rho}Y_{K})^{2} + (\partial_{z}Y_{K})^{2} \right],
\end{align*}
which holds in $\overline{\mathscr{B}_{A}}\cup\overline{\mathscr{B}_{H}}$ thanks to Theorem 1.2 in \cite{HBH:geometric}. Arguing as above, we may easily see that the right hand side of this equation is in $\hat C^{1+j,\alpha_{0}}$. In particular, we may obtain $C^{3+j,\alpha_{0}}$-regularity with respect to the $\rho,z$ coefficients by two-dimensional elliptic regularity. Finally, we may check by induction that $\partial^{2k+1}_{\rho}\partial^{l}_{z}\mathring\lambda |_{\{\rho = 0\}} =0$ for $2k+l\leq j$. When $j=0$, this follows immediately from the proof of Proposition \ref{fixlam}. The $\partial^{2k+1}_{\rho}\partial^{l}_{z}$ derivatives of the right hand side of $\mathring\lambda$'s equation vanish by the regularity of the other coefficients proven above. Thus, differentiating and evaluating at $\{\rho =0 \}$, we find that $\partial^{2(k+1)+1}_{\rho}\partial^{l}_{z}\mathring\lambda |_{\{\rho = 0\}} =0$ for $2k+l\leq j$. Combined with a similar argument near $p_{N}$ and $p_{S}$, this shows $\mathring\lambda \in \hat C^{2+j,\alpha_{0}}$.
\end{proof}

Now we may check that $(\mathcal{M},g_{\delta})$ is ``extendable to a regular black hole spacetime.''

\begin{lemma}\label{lemm:rot-sym-extend}
Suppose that $f(\rho,z)$ is a smooth function on $\overline{\mathscr{B}_{A}}$, resp.\ $\overline{\mathscr{B}_{H}}$, that is smooth when considered as a function on $\RR^{n}$ with the metric $d\rho^{2}+\rho^{2}d\mathbb{S}^{n-2}+dz^{2}$ for some $n>2$. Then, $f(\rho,z) = g(\rho^{2},z)$ for some smooth function on $\overline{\mathscr{B}_{A}}$, resp.\ $\overline{\mathscr{B}_{H}}$.

Similarly, if $f(s,\chi)$ is a smooth function on $\overline{\mathscr{B}_{N}}$ that is smooth when considered as a function on $\RR^{n}$ with the metric $ds^{2} + s^{2}d\phi_{1}^{2} + d\chi^{2} +\chi^{2}d\phi_{2}^{2}$, then $f(s,\chi) = g(s^{2},\chi^{2})$ for some smooth function $g$ on $\overline{\mathscr{B}_{N}}$. An analogous statement holds for $\overline{\mathscr{B}_{S}}$.
\end{lemma}
\begin{proof}
In the first case, the assumption implies that $\partial^{2j+1}_{\rho}\partial^{k}_{z} f(0,z) = 0$ for all integers $j,k$. The claim follows from this via a straightforward argument using Taylor series. The other cases follow similar arguments.
\end{proof}

\begin{proposition}Let $\delta$ be sufficiently small and $\left( \mathring\sigma, B,(\mathring X,\mathring Y),\mathring \Theta,\psi,\mathring{\mu}^2,\mathring{\lambda}\right)$ be produced by Proposition~\ref{fixlam}. Then the corresponding spacetime $(\mathcal{M},g_{\delta})$ is ``extendable to a regular black hole spacetime'' in the sense of Definition 2.5 in \cite{HBH:geometric}.
\end{proposition}
\begin{proof}
Using Proposition \ref{prop:full-reg} and Lemma \ref{lemm:rot-sym-extend}, we immediately see that in an open set around the axis, we can write $\mathring \sigma(\rho,z) = \mathring\sigma_{\mathscr{A}}(\rho^{2},z)$, $\mathring\Theta (\rho,z) = \mathring\Theta_{\mathscr{A}}(\rho^{2},z)$, $\mathring X(\rho,z) = \mathring X_{\mathscr{A}}(\rho^{2},z)$, $\mathring\lambda(\rho,z) = \mathring \lambda_{\mathscr{A}}(\rho^{2},z)$, for smooth functions $ \mathring\sigma_{\mathscr{A}}, \mathring\Theta_{\mathscr{A}}, \mathring X_{\mathscr{A}}, \mathring \lambda_{\mathscr{A}}$. In the remainder of the proof, when we write an expression of the form $f(x,y) = g(x^{2},y)$, the smoothness of $g$ will always be implied.

The equation
\[
V = \frac{\rho^{2}}{X_{K}} \frac{(1+\mathring\sigma)^{2}}{1+\mathring X} - \left( \frac{W_{K}}{X_{K}} + \mathring\Theta \right)^{2} X_{K}(1+\mathring X)
\]
immediately implies that $V(\rho,z) = V_{\mathscr{A}}(\rho^{2},z)$. That $V_{\mathscr{A}}(0,z) > 0$ for $\delta > 0$ sufficiently small follows from the corresponding fact for Kerr (see Section 3 of \cite{HBH:geometric}). Similarly, using
\[
W = X_{K}\mathring\Theta (1+\mathring X) + W_{K}(1+\mathring X),
\]
we see that $W(\rho,z)=W_{\mathscr{A}}(\rho^{2},z)$, by using the above properties, as well as the properties of $X_{K},W_{K}$ near the axis (see Section 3 in \cite{HBH:geometric}). That $X(\rho,z) = \rho^{2}X_{\mathscr{A}}(\rho^{2},z)$ with $X_{\mathscr{A}}(0,z)>0$ follows by the same argument. Finally, because we have shown in Proposition \ref{compatible} that $e^{2\lambda}-\rho^{-2}X$ vanishes at the axis, and we know that each term can be written as a smooth function of $\rho^{2}$, the fact that
\[
e^{2\lambda} = X_{\mathscr{A}}(\rho^{2},z) + \rho^{2}\Sigma_{\mathscr{A}}(\rho^{2},z)
\]
near the axis follows easily. Putting this together, we see that $(\mathcal{M},g)$ is extendible across the axis in the sense of Definition 2.2 in \cite{HBH:geometric}.

Extendability across the horizon, as well as the north and south poles follows by a similar argument, using the other parts of Proposition \ref{prop:full-reg} and Lemma \ref{lemm:rot-sym-extend}.
\end{proof}

Given this, Proposition 2.2.1 in \cite{HBH:geometric} implies that we may extend $(\mathcal{M},g)$ to a Lorentzian manifold with corners \index{Coordinates!$\tilde{\mathcal{M}}$}$(\tilde{\mathcal{M}},\tilde{g})$, which is stationary and axisymmetric, whose boundary corresponds to a bifurcate Killing event horizon.

\subsection{Asymptotic Flatness}

\begin{proposition}
The manifold $(\tilde{\mathcal{M}},\tilde g)$ is asymptotically flat in the sense that in the coordinates from Appendix A.1 in \cite{HBH:geometric}, $\tilde g$ has the smooth expansion\footnote{By smooth expansion, we mean that after every application of a $(t,x,y,z)$ derivative, the error decays one power faster.}
\[
\tilde g = \left(1+O\left(r^{-1}\right)\right)\left(-dt^{2} + dx^{2}+dy^{2}+dz^{2}\right) + O\left(r^{-2}\right)\left(dtdx + dtdy+dxdy\right).
\]
as $r\to\infty$.
\end{proposition}
\begin{proof}
We will only prove the expansion holds for in $C^{1}$; the higher derivatives follow via an analogous argument.

First, note that we may upgrade the asymptotic falloff for $\tilde Y$ to $|\tilde Y| + r|\partial \tilde Y| \leq Cr^{-4}$ by using Lemma \ref{newt1} and \eqref{tildeeqnY}, along with the fact that the vector field $e_{A}$ in $\partial h = \frac{2}{\rho} \partial_{\rho} + e_{A}$ decays like $|e_{A}|\leq Cr^{-2}$ (see \eqref{hsing}). This easily yields the improved estimate $|\mathring\Theta| + r |d\mathring\Theta| \leq C r^{-3}$.

We now to turn to $\mathring{\lambda}$. First of all, one may easily establish $|\mathring\lambda| + r|d\mathring\lambda| \leq Cr^{-1}$ simply by re-running the argument used in the proof of Lemma \ref{lemm:lam-bound-mathfrakL}; however, this is not quite sufficient to establish the desired control of $\tilde g_{xx}$, $\tilde g_{yy}$ and $\tilde g_{xy}$. Fortunately, a further term by term inspection of the first order equation for $\lambda$ from Theorem 1.3 of~\cite{HBH:geometric} yields the additional estimate $\left|\partial\left(\lambda- \frac{u}{2}\right)\right| \leq C\rho r^{-4}$ when $r$ is sufficiently large. (Recall that $u \doteq \log\left(\rho^{-2}X\right)$ was introduced in the proof of Proposition~\ref{compatible}.)

 Putting this together with the decay $|\mathring\sigma| + r|\partial \mathring\sigma|\leq Cr^{-2}$ and the decay estimates for the various Kerr quantities proven in Section 3 of \cite{HBH:geometric}, the claim easily follows.
\end{proof}

\subsection{Exponential Decay and Smooth Horizon Extension of the Scalar Field}

By iterating Lemma \ref{lem:poly-decay-scalar}, it is not hard to show that $\psi$ decays faster than any polynomial power as $r\to\infty$. Here, we upgrade this to show that $\psi$ decays exponentially fast as $r\to\infty$.
\begin{proposition}
The scalar field $\psi$ decays exponentially for large $r$.
\end{proposition}
\begin{proof}This is a relatively straighforward application of the maximum principle: Pick $k > 0$ satisfying $0 < k^2 < \mu^2-\omega^2$. We begin by observing that when $r$ is sufficiently large (depending on $k$), $\psi$'s equation and the fact that $\psi > 0$ imply
\begin{equation}\label{psilarger}
\sigma^{-1}\partial_{\rho}\left(\sigma\partial_{\rho}\psi\right) + \sigma^{-1}\partial_z\left(\sigma\partial_z\psi\right) - \rho^{-2}m^2\psi - k^2\psi > 0.
\end{equation}

Next, let $S\left(x\right)$ denote the unique solution to the ODE
\[\frac{d}{dx}\left(\left(1-x^2\right)\frac{dS}{dx}\right) - \left(\frac{m^2}{1-x^2}-|m|(|m|+1)\right)S = 0,\]
which satisfies
\[S \sim \left(1\pm x\right)^{|m|/2}\text{ as }x \to \mp 1.\]
(Recall that $e^{im\phi}S\left(\cos\theta\right)$ is a spherical harmonic and that $S > 0$.)

Now define coordinates $(\hat{r},\theta)$ in terms of $(\rho,z)$ coordinates by considering $(\hat{r},\theta,\phi)$ as spherical coordinates and $(\rho,z,\phi)$ as cylindrical coordinates.  Then, for $k'$ satisfying $0 < k' < k$, define \index{Metric Quantities!$\mathfrak{W}$}$\mathfrak{W} : \mathscr{B} \to \mathbb{R}$ by
\[\mathfrak{W}\left(\rho,z\right) \doteq S\left(\cos\theta\right)\exp\left(-k'\hat{r}\right)\hat{r}^{-1}.\]
Using that $S\left(\cos\theta\right)e^{im\phi}$ is a spherical harmonic and the asymptotics of $\sigma$ as $\rho \to 0$, one may check that for $r$ sufficiently large we have
\[\sigma^{-1}\partial_{\rho}\left(\sigma\partial_{\rho}\mathfrak{W}\right) + \sigma^{-1}\partial_z\left(\sigma\partial_z\mathfrak{W}\right) - \rho^{-2}m^2\mathfrak{W} - k^2\mathfrak{W} < 0.\]

In particular, for any constant $C > 0$,  we have
\begin{equation}\label{getreadytomax}
\sigma^{-1}\partial_{\rho}\left(\sigma\partial_{\rho}\left(\psi - C\mathfrak{W}\right)\right) + \sigma^{-1}\partial_z\left(\sigma\partial_z\left(\psi - C\mathfrak{W}\right)\right) - \rho^{-2}m^2\left(\psi-C\mathfrak{W}\right) - k^2\left(\psi - C\mathfrak{W}\right) > 0.
\end{equation}

Next, let $\Gamma$ denote a curve of constant $\rho^2+z^2 = R$. We note that the asymptotic behavior of $S$ is easily seen to imply that
\[\mathfrak{W}|_{\{\rho \leq 1\} \cup \Gamma} \geq b\left(\Gamma\right)\rho^{|m|}.\]
Using Lemma~\ref{usefullemma} one finds that when $r$ is sufficiently large,
\[\psi|_{\{\rho \leq 1\} \cup \Gamma} \leq B\left(\Gamma\right)\rho^{|m|}.\]
In particular, picking $R$ sufficiently large, so that~(\ref{getreadytomax}) holds for $r \geq R$, and then $C$ sufficiently large so that $\left(\psi - C\mathfrak{W}\right)|_{\Gamma \cup \{\rho \leq 1\}} < 0$, it follows immediately from the maximum principle that $\psi \leq C\mathfrak{W}$ for $r \geq R$.
\end{proof}
\begin{remark}The argument given above does not produce the sharp decay rates. Though we will not pursue this here, it would be interesting to see if the actual decay rate is as fast as that of the time periodic solutions to Klein--Gordon on the exact Kerr background (where one can exploit separation of variables and o.d.e. techniques to easily find the sharp decay rate).
\end{remark}

Finally, one easily checks that $e^{-it\omega}e^{im\phi}\psi$ extends smoothly to the extended spacetime $(\tilde{\mathcal{M}},\tilde g_{\delta})$.
\begin{proposition}The function $\Psi\left(t,\phi,\rho,z\right) = e^{-it\omega}e^{im\phi}\psi$ extends smoothly to the extended spacetime $(\tilde{\mathcal{M}},\tilde g_{\delta})$.
\end{proposition}
\begin{proof}This is immediate after expressing $\Psi$ in Kruskal coordinates, see the appendix of~\cite{HBH:geometric}.
\end{proof}
\subsection{Arranging for a Constant Klein--Gordon Mass}\label{arrange}
The analysis up to this point establishes the following theorem.
\begin{theorem}\label{whatwehavesofar}For every choice of $(a,M)$ satisfying $0 < |a| < M$, there exists a Lipschitz continuous $1$-parameter family of smooth spacetimes $(\mathcal{M},g_{\delta})$, scalar fields $\Psi_{\delta} : \mathcal{M} \to \mathbb{C}$, and Klein--Gordon masses $\mu^2\left(\delta\right)$ indexed by $\delta \in [0,\epsilon)$ such that
\begin{enumerate}
    \item For each $\delta \geq 0$ the pair $(\mathcal{M},g_{\delta})$ and $\Psi_{\delta}$ yields a solution to the Einstein--Klein--Gordon equations with mass $\mu^2\left(\delta\right)$.
    \item The spacetimes $(\mathcal{M},g_{\delta})$ are all stationary, axisymmetric, asymptotically flat, and posses a non-degenerate bifurcate event horizon.
    \item For $\delta > 0$ the scalar field $\Psi_{\delta}$ is non-zero, \underline{time-periodic}, and decays exponentially along any asymptotically flat Cauchy hypersurface.
    \item The $1$-parameter family bifurcates off the Kerr family in the sense that $(\mathcal{M},g_0)$ is isometric to the sub-extremal Kerr exterior spacetime of mass $M$ and angular momentum $aM$, and $\lim_{\delta\to 0}\delta^{-1}\Psi_{\delta} = \hat{\Psi}$, where \index{Metric Quantities!$\hat\Psi$}$\hat{\Psi}$ is a non-zero time-periodic solution to the Klein--Gordon equation~(\ref{kg}) of mass $\mu^2(0)$ on $(\mathcal{M},g_0)$.
\end{enumerate}
\end{theorem}

Our goal in this section is to vary $a$ and $M$ with $\delta$ so $\mu^2\left(a\left(\delta\right),M\left(\delta\right),\delta\right)$ is constant as a function of $\delta$.

First we discuss the regularity of our spacetimes, scalar fields, and Klein--Gordon masses upon varying $a$ and $M$.
\begin{lemma}Recall that the renormalized unknowns are all functions of $a$, $M$, and $\delta$ (we have often suppressed the dependence on $a$ and $M$ in the notation). Pick a value $\gamma^2 > 0$, define $M\left(a\right) \doteq \sqrt{\gamma^2+a^2}$, and then, noting that this choice of $M$ fixes the value of $\sqrt{M^2-a^2}$, consider the renormalized unknowns as functions of $a$ and $\delta$ all defined on the same conformal manifold with corners $\overline{\mathscr{B}^{(\gamma)}}$ (see~\cite{HBH:geometric}). Then the renormalized unknowns are Lipschitz functions of both $a$ and $\delta$.
\end{lemma}
\begin{proof}When $\delta = 0$, the Lipschitz continuity in $a$ simply follows from the fact that the Kerr family is a smooth $2$-parameter family. For $\delta > 0$, one simply re-runs the proof of Theorem~\ref{whatwehavesofar} adding $a$ as a parameter in all of the equations. 
\end{proof}

Next, we note that all of our unknowns are easily seen to be differentiable in $\delta$ at $\delta = 0$.
\begin{lemma}\label{derdelta} All of the renormalized unknowns are differentiable in $\delta$ at $\delta = 0$; in fact, all of these derivatives vanish except for the derivative of the scalar field.
\end{lemma}
\begin{proof}It already follows immediately from Theorem~\ref{whatwehavesofar} that $\frac{\partial\psi}{\partial\delta}|_{\delta = 0}$ exists. Next, we recall that the estimates we established in the course of the proof of Theorem~\ref{whatwehavesofar} have shown that, with the exception of $\mathring{\mu}^2$, all of the renormalized unknowns are  $O\left(\delta^2\right)$ as $\delta \to 0$. In particular, they are differentiable in $\delta$ at $\delta = 0$, and the derivative vanishes.

Lastly, we turn to the Klein--Gordon mass $\mu^2$. Using that the renormalized unknowns, except for $\psi$, are all $O\left(\delta^2\right)$, it follows immediately from the proof of Lemma~\ref{itisindeedquitecontinuous} that $\mathring{\mu}^2$ is also $O\left(\delta^2\right)$; hence it is differentiable at $\delta = 0$, and the derivative there with respect to $\delta$ vanishes.
\end{proof}

Finally we are ready to complete the proof of our main result, Theorem~\ref{timeperiodicsoln}.
\begin{proof}
We begin by observing that
\begin{equation}\label{tozero}
\lim_{a\to 0} \mu^2(a,M\left(a\right),0) = 0,
\end{equation}
which essentially encodes the fact that there is no supperradiance on Schwarzschild. In the context of this paper, one may see this by tracing through the proof of Lemma \ref{neg}, after observing that on Kerr, the parameter $\omega \to 0$ as $a\to 0$. On the other hand, for $a>0$, we have seen that $\mu^2(a,\sqrt{\gamma^{2}+a^{2}},0) > \omega^2 > 0$.

Because $a\mapsto \mu^2(a,M(a),0)$ is Lipchitz, and thus absolutely continuous, it is differentiable almost everywhere, and the fundamental theorem of calculus holds. In particular, it follows from~(\ref{tozero}) that we can find an set of positive measure $\mathfrak{D} \subset (0,\infty)$ so that $a \in \mathfrak{D}$ implies that $\frac{d}{da}\left(\mu^2\left(a,M(a),0\right) \right)$ exists and is non-zero. Using also Lemma~\ref{derdelta} we may appeal to a version of the implicit function theorem valid for Lipschitz functions (see Theorem~\ref{weakimp} in the appendix) so that for every $a \in \mathfrak{D}$, we may find a function $a(\delta)$ defined for $\delta\geq 0$ sufficiently small, so that $a(0) = a$, $\mu(a(\delta),M(a(\delta)),\delta)$ is constant with respect to $\delta$, and so that $a(\delta)$ is differentiable at $0$. Using this choice of $a(\delta)$ in the previously constructed solutions from Theorem~\ref{whatwehavesofar} yields a $1$-parameter family of metrics with the properties asserted in Theorem~\ref{timeperiodicsoln}.
\end{proof}

\appendix
\section{Newton Potential Estimates}
The following basic Newton potential estimates are quite useful and are easily proven using standard techniques. See~\cite{giltru}.
\begin{lemma}\label{newt1}Let $n \geq 3$ and $F : \mathbb{R}^n \to \mathbb{R}$ satisfy $\left|F(x)\right| \leq C\langle x\rangle^{-k}$ for some $k > 2$. Then let $u : \mathbb{R}^n \to \mathbb{R}$ be the corresponding Newton potential, i.e.~
\[u\left(x\right) \doteq \int_{\mathbb{R}^n}\left|x-y\right|^{2-n}F\left(y\right)\, dy.\]
Then
\[\left|u(x)\right| \leq C\sup_{y \in \mathbb{R}^n}\left|\langle y\rangle^kF(y)\right|\left[\langle x\rangle^{2-n} + \langle x\rangle^{2-k}\right],\]
\[\left|\partial u(x)\right| \leq C\sup_{y \in \mathbb{R}^n}\left|\langle y\rangle^kF(y)\right|\left[\langle x\rangle^{1-n} + \langle x\rangle^{1-k}\right].\]
\end{lemma}

\begin{lemma}\label{newt2}Let $n \geq 3$ and $F : \mathbb{R}^n \to \mathbb{R}$ satisfy $\left|F(x)\right| + \sup_{|x-y| \leq 1}\frac{F(x)-F(y)}{|x-y|^{\alpha_0}}\leq C\langle x\rangle^{-k}$ for some $k > 2$. Then let $u : \mathbb{R}^n \to \mathbb{R}$ be the corresponding Newton potential, i.e.~
\[u\left(x\right) \doteq \int_{\mathbb{R}^n}\left|x-y\right|^{2-n}F\left(y\right)\, dy.\]
Then
\[\left|\partial^2 u\right|(x) + \sup_x\left[\sup_{|y-x| \leq 1}\frac{\left|\partial^2u(y)-\partial^2u(z)\right|}{|y-x|^{\alpha_0}}\right]\leq \]
\[C\left(\sup_{y \in \mathbb{R}^n}\left|\langle y\rangle^kF(y)\right| + \sup_y\left(\langle y\rangle^k\left[\sup_{|y-z| \leq 1}\frac{\left|F(y)-F(z)\right|}{|y-z|^{\alpha_0}}\right]\right)\right)\log\left(4\langle x\rangle\right)\left[\langle x\rangle^{-k} + \langle x\rangle^{-n}\right].\]
\end{lemma}

\section{An Implicit Function Theorem for Lipschitz Functions}
The following version of the implicit function theorem is proven in~\cite{halkin}; we provide the proof (of the form that we use above) for completeness.
\begin{theorem}\label{weakimp}
Suppose that $f(x,y) : \{(x,y) : x \in (-1,1), y \in [0,1)\}\to \RR$ is Lipschitz and differentiable at $(0,0)$ with $f(0,0) = 0$ and $\partial_{x}f(0,0) \not = 0$. Then, for some $\eta_{0}>0$, there is a function $X :[0,\eta_{0})\to (-1,1)$, so that $X(0) = 0$, $f(X(y),y) = 0$ for $y\in [0,\eta_{0})$, and so that $X$ is differentiable at $0$.
\end{theorem}
\begin{proof}
By scaling, we may ensure that $\partial_{x}f(0,0) = 1$. Then, by definition of the derivative, and Lipschitz continuity, there is $x_{0},y_{0} \in (0,1)$ so that for $x \in (-x_{0},x_{0})$, $y\in[0,y_{0})$
\[
|f(x,0) - x| \leq \frac {x_{0}}{2} \qquad \text{and} \qquad |f(x,0)-f(x,y)| \leq \frac {x_{0}}{2}.
\]
For $y \in [0,y_{0})$, let $F_{y}(x) = x - f(x,y)$. Note that $F_{y}$ is continuous and maps $F_{y} : (-x_{0},x_{0}) \to (-x_{0},x_{0})$, since
\[
|F_{y}(x)| \leq |x-f(x,0)| + |f(x,0) - f(x,y)| \leq a.
\]
The Brouwer fixed point theorem guarantees a fixed point $X(y)$ for $F_{y}$, which clearly satisfies $f(X(y),y) = 0$.

Because $f$ is differentiable at $(0,0)$, we have that
\[
|X(y) + \partial_{y}f(0,0) y|  = |X(y) + \partial_{y}f(0,0) y - f(X(y),y) | \leq o(|X(y)| + |y|)
\]
as $y\to 0$. This implies that $X(y)$ is Lipschitz at $0$, and then using this in the inequality, we find that
\[
|X(y) + \partial_{y}f(0,0) y|  \leq o(|y|),
\]
showing that $X(y)$ is differentiable at $0$.
\end{proof}

%\printindex

%%ADDED HERE BECAUSE CMP DID NOT ACCEPT IDX FILE. REMOVE THIS AND USE PRINTINDEX COMMAND TO GET CORRECTED INDEX
\begin{theindex}

  \item Coordinates
    \subitem $K_{\minus}$, \hyperpage{59}
    \subitem $\beta $, \hyperpage{13}
    \subitem $\chi$, \hyperpage{14}
    \subitem $\chi'$, \hyperpage{15}
    \subitem $\mathcal{M}$, \hyperpage{13}
    \subitem $\mathscr{A}$, \hyperpage{13}
    \subitem $\mathscr{A}_{N}$, \hyperpage{13}
    \subitem $\mathscr{A}_{S}$, \hyperpage{13}
    \subitem $\mathscr{B}$, \hyperpage{13}
    \subitem $\mathscr{H}$, \hyperpage{13}
    \subitem $\overline  {\mathscr  {B}_{A}}$, \hyperpage{15}
    \subitem $\overline  {\mathscr  {B}_{H}}$, \hyperpage{15}
    \subitem $\overline  {\mathscr  {B}_{N}}$, \hyperpage{15}
    \subitem $\overline  {\mathscr  {B}_{S}}$, \hyperpage{15}
    \subitem $\overline{\mathscr{B}}$, \hyperpage{13}
    \subitem $\rho$, \hyperpage{12, 13}
    \subitem $\tilde r$, \hyperpage{12}
    \subitem $\tilde{\mathcal{M}}$, \hyperpage{71}
    \subitem $f_{\mathbb{R}^{4}}$, \hyperpage{16}
    \subitem $p_{N}$, \hyperpage{15}
    \subitem $p_{S}$, \hyperpage{15}
    \subitem $r$, \hyperpage{17}
    \subitem $s$, \hyperpage{14}
    \subitem $s'$, \hyperpage{15}
    \subitem $z$, \hyperpage{12, 13}
    \subitem Boyer--Lindquist, \hyperpage{12}

  \indexspace

  \item Function Spaces
    \subitem $C^{k,\alpha}_{0,\rm axi}\left(\overline{\mathscr{B}}\right) $, 
		\hyperpage{15}
    \subitem $C^{k,\alpha}_{\rm axi}\left(\overline{\mathscr{B}}\right)$, 
		\hyperpage{15}
    \subitem $L_{u}^{p}$, \hyperpage{37}
    \subitem $L_{v}^{p}$, \hyperpage{37}
    \subitem $W^{k,p}_{\rm axi}\left(\overline{\mathscr{B}}\right)$, 
		\hyperpage{15}
    \subitem $\dot{W}^{k,p}_{\rm axi}\left(\overline{\mathscr{B}}\right)$, 
		\hyperpage{15}
    \subitem $\hat{C}^{\infty}$, \hyperpage{16}
    \subitem $\hat{C}^{\infty}_{0}$, \hyperpage{16}
    \subitem $\hat{C}^{k,\alpha}_{0,\rm axi}\left(\overline{\mathscr{B}}\right)$, 
		\hyperpage{16}
    \subitem $\hat{C}^{k,\alpha}_{\rm axi}\left(\overline{\mathscr{B}}\right)$, 
		\hyperpage{16}
    \subitem $\hat{C}^{k,\alpha}_{\rm m,azi}\left(\overline{\mathscr{B}}\right)$, 
		\hyperpage{69}
    \subitem $\hat{C}^{k}$, \hyperpage{16}
    \subitem $\hat{C}^{k}_{0}$, \hyperpage{16}
    \subitem $\hat{W}^{k,p}_{\rm axi}\left(\overline{\mathscr{B}}\right)$, 
		\hyperpage{16}
    \subitem $\hat{\dot{W}}_{\rm axi}^{k,p}\left(\overline{\mathscr{B}}\right)$, 
		\hyperpage{16}
    \subitem $\mathcal{L}_{B}$, \hyperpage{30}
    \subitem $\mathcal{L}_{X,Y}$, \hyperpage{30}
    \subitem $\mathcal{L}_{X}$, \hyperpage{30}
    \subitem $\mathcal{L}_{Y}$, \hyperpage{30}
    \subitem $\mathcal{L}_{\Theta}$, \hyperpage{31}
    \subitem $\mathcal{L}_{\lambda}$, \hyperpage{31}
    \subitem $\mathcal{L}_{\mu^{2}}$, \hyperpage{31}
    \subitem $\mathcal{L}_{\psi}$, \hyperpage{31}
    \subitem $\mathcal{L}_{\sigma}$, \hyperpage{29}
    \subitem $\mathcal{N}_{B}$, \hyperpage{30}
    \subitem $\mathcal{N}_{X,Y}$, \hyperpage{30}
    \subitem $\mathcal{N}_{X}$, \hyperpage{30}
    \subitem $\mathcal{N}_{Y}$, \hyperpage{30}
    \subitem $\mathcal{N}_{\Theta}$, \hyperpage{31}
    \subitem $\mathcal{N}_{\sigma}$, \hyperpage{29}

  \indexspace

  \item Metric Quantities
    \subitem $B$, \hyperpage{24}
    \subitem $B^{(A)}$, \hyperpage{25}
    \subitem $B^{(N)}$, \hyperpage{25}
    \subitem $B^{(S)}$, \hyperpage{25}
    \subitem $H^{(1)}_{B}$, \hyperpage{33}
    \subitem $H^{(2)}_{B}$, \hyperpage{33}
    \subitem $H^{(3)}_{B}$, \hyperpage{33}
    \subitem $H_{X}$, \hyperpage{35}
    \subitem $H_{Y}$, \hyperpage{35}
    \subitem $H_{\Theta}$, \hyperpage{46}
    \subitem $H_{\sigma}$, \hyperpage{31}
    \subitem $M$, \hyperpage{26}
    \subitem $T$, \hyperpage{12}
    \subitem $V$, \hyperpage{22}
    \subitem $V_{K}$, \hyperpage{12}
    \subitem $W$, \hyperpage{22}
    \subitem $W_{K}$, \hyperpage{12}
    \subitem $X$, \hyperpage{22}
    \subitem $X_{K}$, \hyperpage{12}
    \subitem $Y$, \hyperpage{23, 24}
    \subitem $\Delta $, \hyperpage{12}
    \subitem $\Phi$, \hyperpage{12}
    \subitem $\Pi $, \hyperpage{12}
    \subitem $\Psi$, \hyperpage{22}
    \subitem $\Sigma ^{2}$, \hyperpage{12}
    \subitem $\beta_{2}$, \hyperpage{62}
    \subitem $\check{d}_{A}$, \hyperpage{38}
    \subitem $\gamma$, \hyperpage{13}
    \subitem $\hat\Psi$, \hyperpage{73}
    \subitem $\hat{\mu}^{2}$, \hyperpage{50}
    \subitem $\kappa$, \hyperpage{67}
    \subitem $\lambda _{K}$, \hyperpage{13}
    \subitem $\lambda$, \hyperpage{22}
    \subitem $\mathfrak{A}$, \hyperpage{57}
    \subitem $\mathfrak{L}(\mathring\lambda)$, \hyperpage{63}
    \subitem $\mathfrak{W}$, \hyperpage{72}
    \subitem $\mathfrak{a}$, \hyperpage{48}, \hyperpage{57}
    \subitem $\mathfrak{b}$, \hyperpage{60}
    \subitem $\mathfrak{c}$, \hyperpage{66}
    \subitem $\mathring X$, \hyperpage{25}
    \subitem $\mathring Y$, \hyperpage{25}
    \subitem $\mathring \Theta$, \hyperpage{25}
    \subitem $\mathring \omega$, \hyperpage{26}, \hyperpage{48}
    \subitem $\mathring{\sigma}$, \hyperpage{25}
    \subitem $\mu^{2}$, \hyperpage{25}
    \subitem $\mu_{K}$, \hyperpage{25}
    \subitem $\omega$, \hyperpage{22}
    \subitem $\psi$, \hyperpage{25}
    \subitem $\sigma$, \hyperpage{22}
    \subitem $\theta$, \hyperpage{22}
    \subitem $\tilde Y$, \hyperpage{38}
    \subitem $\tilde r_{\pm}$, \hyperpage{12}
    \subitem $\tilde{d}_{A}$, \hyperpage{38}
    \subitem $\tilde{e}_N$, \hyperpage{39}
    \subitem $\tilde{e}_{A}$, \hyperpage{38}
    \subitem $a$, \hyperpage{26}
    \subitem $g_{(a,M)}$, \hyperpage{12}
    \subitem $h$, \hyperpage{17}
    \subitem $m$, \hyperpage{22}, \hyperpage{26}, \hyperpage{56}
    \subitem $x_{k}$, \hyperpage{18}
    
   \indexspace
    
  \item Miscellaneous
    \subitem $A$, \hyperpage{59}
    \subitem $B_{r}(\mathcal{B})$, \hyperpage{20}
    \subitem $\alpha_{0}$, \hyperpage{29}, \hyperpage{56}
    \subitem $\chi_{N}$, \hyperpage{16}
    \subitem $\chi_{S}$, \hyperpage{16}
    \subitem $\delta$, \hyperpage{26}
    \subitem $\hat{\partial}f$, \hyperpage{16}
    \subitem $\langle r \rangle$, \hyperpage{36}
    \subitem $\mathbb{R}^7_A$, \hyperpage{38}
    \subitem $\mathbb{R}^{4}_{N}$, \hyperpage{39}
    \subitem $\mathcal{L}(\mathring X,\mathring Y)$, \hyperpage{35}
    \subitem $\mathcal{L}_{0}(\mathring X,\mathring Y)$, \hyperpage{36}
    \subitem $\mathfrak{I}$, \hyperpage{65}
    \subitem $\mathfrak{f}$, \hyperpage{65}
    \subitem $\mathscr{L}_{\mu^{2}}$, \hyperpage{49}
    \subitem $\nu$, \hyperpage{50}
    \subitem $\partial f$, \hyperpage{16}
    \subitem $\underline{\partial}f$, \hyperpage{16}
    \subitem $b$, \hyperpage{59}

\end{theindex}

\bibliography{bib}

\providecommand{\bysame}{\leavevmode\hbox to3em{\hrulefill}\thinspace}
\providecommand{\MR}{\relax\ifhmode\unskip\space\fi MR }
\providecommand{\MRhref}[2]{%
  \href{http://www.ams.org/mathscinet-getitem?mr=#1}{#2}
}
\providecommand{\href}[2]{#2}
\begin{thebibliography}{10}

\bibitem{agmon}
S.~Agmon, \emph{Bounds on exponential decay of eigenfunctions of
  {S}chr{\"o}dinger operators}, Schr{\"o}dinger operators, Springer, 1985,
  pp.~1--38.

\bibitem{aik}
S.~Alexakis, A.~Ionescu, and S.~Klainerman, \emph{Hawking's local rigidity
  theorem without analyticity}, Geometric and Functional Analysis \textbf{20}
  (2010), no.~4, 845--869.

\bibitem{aik2}
\bysame, \emph{Uniqueness of smooth stationary black holes in vacuum: small
  perturbations of the {K}err spaces}, Communications in Mathematical Physics
  \textbf{299} (2010), no.~1, 89--127.

\bibitem{periodic}
S.~Alexakis and V.~Schule, \emph{Non-existence of time-periodic vacuum
  spacetimes}, to appear in J.\ Diff.\ Geom., available at
  \url{http://arxiv.org/abs/1504.04592} (2015).

\bibitem{unique}
Spyros Alexakis, Volker Schlue, and Arick Shao, \emph{Unique continuation from
  infinity for linear waves}, Adv. Math. \textbf{286} (2016), 481--544.
  \MR{3415691}

\bibitem{anblue}
Lars Andersson and Pieter Blue, \emph{Hidden symmetries and decay for the wave
  equation on the {K}err spacetime}, Ann. of Math. (2) \textbf{182} (2015),
  no.~3, 787--853. \MR{3418531}

\bibitem{aretakisKerr}
S.~Aretakis, \emph{Decay of axisymmetric solutions of the wave equation on
  extreme {K}err backgrounds}, Journal of Functional Analysis \textbf{263}
  (2012), no.~9, 2770--2831.

\bibitem{aretakisHor}
Stefanos Aretakis, \emph{Horizon instability of extremal black holes}, Adv.
  Theor. Math. Phys. \textbf{19} (2015), no.~3, 507--530. \MR{3418509}

\bibitem{bartnik}
R.~Bartnik, \emph{The mass of an asymptotically flat manifold}, Comm. Pure and
  Appl. Math. \textbf{39} (1986), no.~5, 661--693.

\bibitem{EYMnum}
R.~Bartnik and J.~McKinnon, \emph{Particlelike solutions of the
  {E}instein--{Y}ang--{M}ills equations}, Phys. Rev. Lett. \textbf{61} (1988),
  no.~2, 141.

\bibitem{bchr}
C.~Benone, L.~Crispino, C.~Herdeiro, and E.~Radu, \emph{Kerr--{N}ewman scalar
  clouds}, Phys. Rev. D \textbf{90} (2014).

\bibitem{bw}
P.~Bizo{\'n} and A.~Wasserman, \emph{On existence of mini-boson stars}, Comm.
  Math Phys. \textbf{215} (2000), no.~2, 357--373.

\bibitem{bhr}
Y.~Brihaye, C.~Herdeiro, and E.~Radu, Phys. Lett. B \textbf{739} (2014), 1--7.

\bibitem{super}
R.~Brito, V.~Cardoso, and P.~Pani, \emph{Superradiance}, 2015.

\bibitem{bunting}
G.~Bunting, \emph{Proof of the uniqueness conjecture for black holes}, Ph.D.
  thesis, University of New England, 1983.

\bibitem{carter3}
B.~Carter, \emph{Axisymmetric black hole has only two degrees of freedom},
  Phys. Rev. Lett. \textbf{26} (1971), no.~6, 331.

\bibitem{HBH:geometric}
O.~Chodosh and Y.~Shlapentokh-Rothman, \emph{Stationary axisymmetric black
  holes with matter}, preprint (2015).

\bibitem{cklinear}
D.~Christodoulou and S.~Klainerman, \emph{Asymptotic properties of linear field
  equations in {M}inkowski space}, Comm. Pure and Appl. Math. \textbf{43}
  (1990), no.~2, 137--199.

\bibitem{ck}
\bysame, \emph{The global nonlinear stability of the {M}inkowski space},
  Princeton University Press, Princeton, 1993.

\bibitem{chrusciellopes}
P.~Chru{\'s}ciel and J.~Costa, \emph{On uniqueness of stationary vacuum black
  holes}, Ast\'erisque (2008), no.~321, 195--265.

\bibitem{reviewunique}
P.~Chru{\'s}ciel, J.~Costa, and M.~Heusler, \emph{Stationary black holes:
  uniqueness and beyond}, Living Rev. Relativity \textbf{15} (2012), no.~7.

\bibitem{costa}
J.~Costa, \emph{On black hole uniqueness theorems}, Ph.D. thesis, Oxford
  University, 2010.

\bibitem{chrr}
Pedro V.~P. Cunha, Carlos A.~R. Herdeiro, Eugen Radu, and Helgi~F. R\'unarsson,
  \emph{Shadows of kerr black holes with scalar hair}, Phys. Rev. Lett.
  \textbf{115} (2015), 211102.

\bibitem{bholescatter}
M.~Dafermos, G.~Holzegel, and I.~Rodnianski, \emph{A scattering theory
  construction of dynamical vacuum black holes}, to appear in J.\ Diff.\ Geom.,
  available at \url{http://arxiv.org/abs/1306.5364} (2013).

\bibitem{dr7}
M.~Dafermos and I.~Rodnianski, \emph{Decay for solutions of the wave equation
  on {K}err exterior spacetimes {I--II}: {T}he cases $|a|\ll {M}$ or
  axisymmetry}, preprint, available at \url{http://arxiv.org/abs/1010.5132}
  (2010).

\bibitem{icmp}
\bysame, \emph{A new physical-space approach to decay for the wave equation
  with applications to black hole spacetimes}, X{VI}th {I}nternational
  {C}ongress on {M}athematical {P}hysics, World Sci. Publ., Hackensack, NJ,
  2010, pp.~421--432.

\bibitem{stabi}
\bysame, \emph{The black hole stability problem for linear scalar
  perturbations}, Proceedings of the Twelfth Marcel Grossmann Meeting on
  General Relativity, T. Damour et al (ed.), World Scientific, Singapore,
  available at \url{http://arxiv.org/abs/1010.5137} (2011), 132--189.

\bibitem{scatter}
M.~Dafermos, I.~Rodnianski, and Y.~Shlapentokh-Rothman, \emph{A scattering
  theory for the wave equation on {K}err black hole exteriors}, to appear in
  Ann.\ Sci.\ \'ec.\ Norm. Sup\'er., available at
  \url{http://arxiv.org/abs/1412.8379} (2014).

\bibitem{waveKerr}
Mihalis Dafermos, Igor Rodnianski, and Yakov Shlapentokh-Rothman, \emph{Decay
  for solutions of the wave equation on {K}err exterior spacetimes {III}: {T}he
  full subextremal case {$|a|<M$}}, Ann. of Math. (2) \textbf{183} (2016),
  no.~3, 787--913. \MR{3488738}

\bibitem{dain1}
S.~Dain, \emph{Angular-momentum-mass inequality for axisymmetric black holes},
  Phys. Rev. Lett. \textbf{96} (2006), 101101.

\bibitem{dain2}
\bysame, \emph{Proof of the angular momentum-mass inequality for axisymmetric
  black holes}, J. Diff. Geom. \textbf{79} (2008), 33--67.

\bibitem{ddr}
T.~Damour, N.~Deruelle, and R.~Ruffini, \emph{On {Q}uantum {R}esonances in
  {S}tationary {G}eometries}, Lett. Al Nuovo Cimento \textbf{15} (1976), no.~8,
  257--262.

\bibitem{det}
S.~Detweiler, \emph{Klein--{G}ordon equation and rotating black holes}, Phys.
  Rev. D \textbf{22} (1980), no.~10, 2323--2326.

\bibitem{dyatlov2}
S.~Dyatlov, \emph{Exponential energy decay for {K}err--de {S}itter black holes
  beyond event horizons}, Math. Res. Lett. \textbf{18} (2011), no.~5,
  1023--1035.

\bibitem{dyatlov1}
\bysame, \emph{Quasi-normal modes and exponential energy decay for the
  {K}err-de {S}itter black hole}, Comm. Math. Phys. \textbf{306} (2011), no.~1,
  119--163.

\bibitem{dyatlov-last}
\bysame, \emph{Asymptotics of linear waves and resonances with applications to
  black holes}, Comm. Math. Phys. \textbf{335} (2015), no.~3, 1445--1485.

\bibitem{eelslemaire}
J.~Eells and L.~Lemaire, \emph{Selected topics in harmonic maps}, vol.~50,
  American Mathematical Society, Providence, RI, 1983.

\bibitem{FKSY}
F.~Finster, N.~Kamran, J.~Smoller, and S.-T. Yau, \emph{A rigorous treatment of
  energy extraction from a rotating black hole}, Comm. Math. Phys. \textbf{287}
  (2009), no.~3, 829--847.

\bibitem{flp}
R.~Friedberg, T.-D. Lee, and Y.~Pang, \emph{Mini-soliton stars}, Phys. Rev. D
  \textbf{35} (1987), 3640--3657.

\bibitem{flp2}
\bysame, \emph{Scalar soliton stars and black holes}, Phys. Rev. D \textbf{35}
  (1987), 3658--3677.

\bibitem{giltru}
D.~Gilbarg and N.~Trudinger, \emph{Elliptic partial differential equations of
  second order}, Springer-Verlag, Berlin, 2001.

\bibitem{halkin}
H.~Halkin, \emph{Implicit functions and optimization problems without
  continuous differentiability of the data}, SIAM J. Control \textbf{12}
  (1974), 229--236, Collection of articles dedicated to the memory of Lucien W.
  Neustadt. \MR{0406524 (53 \#10311)}

\bibitem{hawking}
S.~Hawking, \emph{Black holes in general relativity}, Comm. Math. Phys.
  \textbf{25} (1972), 152--166.

\bibitem{hr}
C.~Herdeiro and E.~Radu, \emph{Ergosurfaces for {K}err black holes with scalar
  hair}, Phys. Rev. D \textbf{89} (2014).

\bibitem{hairy}
\bysame, \emph{Kerr black holes with scalar hair}, Phys. Rev. Lett.
  \textbf{112} (2014), 221101.

\bibitem{hrr}
C.~Herdeiro, E.~Radu, and H.~R\'{u}narsson, \emph{Non-linear {Q}-clouds around
  {K}err black holes}, Phys. Lett. B \textbf{739} (2014).

\bibitem{holz-smul}
G.~Holzegel and J.~Smulevici, \emph{Decay properties of {K}lein-{G}ordon fields
  on {K}err-{A}d{S} spacetimes}, Comm. Pure Appl. Math. \textbf{66} (2013),
  no.~11, 1751--1802.

\bibitem{holz-smul2}
\bysame, \emph{Quasimodes and a lower bound on the uniform energy decay rate
  for {K}err-{A}d{S} spacetimes}, Anal. PDE \textbf{7} (2014), no.~5,
  1057--1090.

\bibitem{ionkla}
Alexandru~D. Ionescu and Sergiu Klainerman, \emph{On the global stability of
  the wave-map equation in {K}err spaces with small angular momentum}, Ann. PDE
  \textbf{1} (2015), no.~1, Art. 1, 78. \MR{3479066}

\bibitem{israel}
W.~Israel, \emph{Event horizons in static vacuum space-times}, Phys. Rev.
  \textbf{164} (1967), 1776--1779.

\bibitem{israel2}
\bysame, \emph{Event horizions in static electrovac space-times}, Comm. Math.
  Phys. \textbf{8} (1968), 245--260.

\bibitem{kaup}
D.~Kaup, \emph{Klein--{G}ordon {G}eons}, Phys. Rev. \textbf{172} (1968),
  1331--1342.

\bibitem{lp}
T.-D. Lee and Y.~Pang, \emph{Stability of {Mini-Boson} stars}, Nuclear Phys. B
  \textbf{315} (1989), 477--516.

\bibitem{bosonreview}
S.~Liebling and C.~Palenzuela, \emph{Dynamical boson stars}, Living Rev.
  Relativity \textbf{15} (2012).

\bibitem{mazur}
P.~Mazur, \emph{Proof of uniqueness of the {Kerr-Neuman} black hole solution},
  J. Math. Phys. \textbf{15} (1982), 3173--3180.

\bibitem{EYMprove}
J.~Mcleod, J.~Smoller, A.~Wasserman, and S.-T. Yau, \emph{Smooth static
  solutions of the {E}instein/{Y}ang-{M}ills equations}, Comm. Math. Phys.
  \textbf{143} (1991), no.~1, 115--147.

\bibitem{pressteuk}
W.~Press and S.~Teukolsky, \emph{Floating {O}rbits, {S}uperradiant {S}cattering
  and the {B}lack-hole {B}omb}, Nature \textbf{238} (1972), 211--212.

\bibitem{robinson}
D.~Robinson, \emph{Uniqueness of the {K}err black hole}, Phys. Rev. Lett.
  \textbf{34} (1975), 905--906.

\bibitem{rufbon}
R.~Ruffini and S.~Bonazzola, \emph{Systems of self-gravitating particles in
  general relativity}, Phys. Rev. \textbf{187} (1969), 1767--1783.

\bibitem{schoenzhou}
R.~Schoen and X.~Zhou, \emph{Convexity of reduced energy and mass angular
  momentum inequalities}, Ann. Henri Poincar\'e \textbf{14} (2013), no.~7,
  1747--1773.

\bibitem{rotbos}
F.~Schunk and E.~Mielke, \emph{Rotating boson stars}, iRelativity and
  Scientific Computing: Computer Algebra, Numerics, Visualization, 152nd
  WE-Heraeus seminar on Relativity and Scientific computing, Bad Honnef,
  Germany, Septeber 18 \textbf{22} (1995), 138--151.

\bibitem{shlapgrow}
Y.~Shlapentokh-Rothman, \emph{Exponentially growing finite energy solutions for
  the {K}lein-{G}ordon equation on sub-extremal {K}err spacetimes}, Comm. Math.
  Phys. \textbf{329} (2014), no.~3, 859--891.

\bibitem{realmodes}
\bysame, \emph{Quantitative mode stability for the wave equation on the {K}err
  spacetime}, Ann. Henri Poincar\'e \textbf{16} (2015), no.~1, 289--345.

\bibitem{EYMblack}
J.~Smoller, A.~Wasserman, and S.-T. Yau, \emph{Existence of black hole
  solutions for the {E}instein-{Y}ang/{M}ills equations}, Comm. Math. Phys.
  \textbf{154} (1993), no.~2, 377--401.

\bibitem{tattoh}
D.~Tataru and M.~Tohaneanu, \emph{A local energy estimate on {K}err black hole
  backgrounds}, Int. Math. Res. Not. IMRN (2011), no.~2, 248--292.

\bibitem{vasy}
A.~Vasy, \emph{Microlocal analysis of asymptotically hyperbolic and {K}err-de
  {S}itter spaces (with an appendix by {S}emyon {D}yatlov)}, Invent. Math.
  \textbf{194} (2013), no.~2, 381--513.

\bibitem{wald}
R.~Wald, \emph{General relativity}, University of Chicago Press, Chicago, IL,
  1984.

\bibitem{weinstein}
G.~Weinstein, \emph{On rotating black holes in equilibrium in general
  relativity}, Comm. Pure Appl. Math. \textbf{43} (1990), no.~7, 903--948.

\bibitem{weinstein2}
\bysame, \emph{The stationary axisymmetric two-body problem in general
  relativity}, Comm. Pure Appl. Math. \textbf{45} (1992), no.~9, 1183--1203.

\bibitem{wongyu}
W.~Wong and P.~Yu, \emph{Non-existence of multiple-black-hole solutions close
  to {K}err-{N}ewman}, Comm. Math. Phys. \textbf{325} (2014), no.~3, 965--996.

\bibitem{zeldovich}
Y.~Zeldovich, \emph{Generating of {W}aves by a {R}otating {B}ody}, ZhETF
  \textbf{14} (1971), 180--181.

\bibitem{ze}
T.~Zourous and D.~Eardley, \emph{Instabilities of massive scalar perturbations
  of a rotating black hole}, Ann. of Phys. \textbf{118} (1979), no.~1,
  139--155.

\end{thebibliography}
\bibliographystyle{amsplain}

\end{document}